\documentclass[12pt, draftclsnofoot,onecolumn]{IEEEtran}
\ifCLASSINFOpdf

\else

\fi
\usepackage{subeqnarray}
\usepackage{mathrsfs}
\usepackage{amsfonts}
\usepackage{amsmath}
\usepackage{amsthm}
\usepackage{amssymb}
\usepackage{amstext}
\usepackage{graphicx}
\usepackage{subfig}   % for subfloat (IEEE compatible)
\usepackage{indentfirst}
\usepackage{array}
\usepackage{cite}
\usepackage[numbers,sort&compress]{natbib}
\usepackage{enumerate}
\usepackage{stmaryrd}
\usepackage[linesnumbered,ruled,vlined]{algorithm2e}
\usepackage{multirow}
\usepackage{multicol}
\usepackage{verbatim}
\usepackage{stfloats}
\usepackage{color}
\usepackage{bm}
\usepackage{multirow}
\usepackage{multicol}
\usepackage{microtype}

\newtheorem{theorem}{\bf{Theorem}}
\newtheorem{lemma}{\bf{Lemma}}

\newtheorem{corollary}{\bf{Corollary}}

\newtheorem{definition}{\bf{Definition}}
\newtheorem{remark}{\bf{Remark}}
\newtheorem{example}{\bf{Example}}
\usepackage{hyperref}

\makeatletter
\def\endIEEEproof{\hspace*{\fill}~\hollowsquare\par}
\makeatother

\newcommand{\hollowsquare}{\text{$\square$}}

\begin{document}
\title{A Mathematical Theory of Pragmatic Information}

\author{Kai Niu, \IEEEmembership{Senior Member,~IEEE}, Ping Zhang,  \IEEEmembership{Fellow,~IEEE}
\thanks
{
This work is supported by the National Key R\&D Program of China (No. 2025YFF0514400, No. 2025YFF0514404), the Fundamental and Interdisciplinary Disciplines Breakthrough Plan of the Ministry of Education of China (No. JYB2025XDXM119), and the National Natural Science Foundation of China (No. 62293481, No. 62471054).\protect\\
\indent K. Niu is with the Key Laboratory of Universal Wireless Communications, Ministry of Education,
and P. Zhang is with the State Key Laboratory of Networking and Switching Technology, Beijing University of Posts and Telecommunications,  Beijing, 100876, China (e-mail: \{niukai, pzhang\}@bupt.edu.cn). \protect\\
}}

\maketitle

\begin{abstract}
We propose a mathematical theory of pragmatic information that unifies communication, control, and decision-making within a coherent framework. At its core is the isoteleia mapping, which formalizes equifinality: distinct semantic paths that lead to the same optimal action are regarded as pragmatically equivalent. This mapping induces a three-tier hierarchy of syntactic, semantic, and pragmatic information, in which each successive abstraction discards distinctions that are irrelevant to the task. Building on this hierarchy, we develop a set of information-theoretic measures at the pragmatic level, including pragmatic entropy, up/down pragmatic mutual information, channel capacity, and rate-distortion. In addition, we prove three coding theorems---lossless source coding, channel coding, and rate-distortion---which generalize Shannon's classical results. We then introduce pragmatic value of information (VoI) and pragmatic cost of information (CoI) as decision-theoretic duals to rate-distortion and capacity, respectively, and formulate a Lagrangian dual framework for cross-layer optimization. The resulting pragmatic efficiency bound \(\mathcal{E}_p(\lambda) = \sup_R [\Phi_p(R) - \lambda \, \mathrm{CoI}_p(R)]\) characterizes the maximum net utility attainable by a resource-constrained intelligent system under a given resource price, thereby providing a limit on purposeful behavior---a \emph{behavioral capacity}. This limit can be viewed as a goal-directed analogue of Shannon`s symbol-level capacity under the stated utility and resource assumptions. Extensions to continuous messages yield closed-form expressions for Gaussian channels and sources, while dynamic settings are addressed through a Bellman equation for sequential decision-making. Finally, this framework provides a foundation for task-oriented communication, networked control, autonomous systems, and embodied AI, shifting the focus from symbol fidelity to the effectiveness of information in guiding actions, and offering a mathematical language for the design of next-generation intelligent systems.
\end{abstract}


\IEEEpeerreviewmaketitle

\section{Introduction}
\label{section_I}

\subsection{Background and Motivation}

Classic information theory (CIT), established by C. E. Shannon in 1948~\cite{Classicpaper_Shannon,Shannon_Weaver}, was a monumental achievement that laid the mathematical foundations for modern communication. By quantifying information through entropy, mutual information, channel capacity, and rate-distortion functions, CIT provided fundamental limits on data compression and reliable transmission over noisy channels. These theoretical bounds have guided the development of practical coding schemes over the past seven decades, with techniques such as Huffman coding, arithmetic coding, turbo codes, LDPC codes, and polar codes approaching the Shannon limits~\cite{Classicpaper_Shannon,Shannon_Weaver}. In parallel, the emergence of cybernetics and control theory established complementary foundations for understanding and designing dynamical systems. Wiener's cybernetics~\cite{Wiener1948} first unified the concepts of communication and control, while Tsien's engineering cybernetics~\cite{Tsien1954} and Bellman's dynamic programming~\cite{Bellman1954} provided powerful tools for optimal decision-making under uncertainty. Kalman's filtering theory~\cite{Kalman1960} further bridged estimation and control, enabling real-time state estimation from noisy observations. Decision theory, particularly through the work of Stratonovich on the value of information (VoI)~\cite{Paper_VoI,Book_VoI}, addressed the fundamental question of how much utility can be gained from acquiring information before making a decision. Despite their shared intellectual roots, these three pillars—communication theory, control theory, and decision theory—have largely developed along independent trajectories, with limited cross-fertilization among them.

Weaver~\cite{Shannon_Weaver,Semantic_Weaver} first articulated that communication involves not only the accurate transmission of symbols (Level A) but also the conveyance of meaning (Level B) and the effectiveness of conduct (Level C). This layered view has motivated numerous efforts to formalize semantic information. Early attempts to formalize semantic information include Carnap and Bar-Hillel's logical approach based on propositional logic~\cite{Semantic_Carnap}, Floridi's theory of strongly semantic information~\cite{Semantic_Floridi}, and fuzzy entropy formulations by De Luca and Termini~\cite{Entropy_Luca,Fuzzy_Luca} and Wu~\cite{Wuweiling}. Zhong \cite{Zhong1986,Zhong1998} and Lu \cite{Lu1994,Lu2025} independently generalized Shannon's information theory to the semantic and pragmatic levels. Bao et al.~\cite{Semantic_Bao} extended the logical framework to derive semantic source and channel coding theorems. More recently, Liu et al.~\cite{RateD_Liu} and others~\cite{SideInfo_Guo,Theory_Shao,Theory_Tang} investigated rate-distortion frameworks for semantic information. However, a systematic information-theoretic treatment of semantic communication remained elusive until the work of Niu and Zhang~\cite{Paper_SIT,Book_SIT,Paper_Beyond_Shannon}, who established a rigorous mathematical theory of semantic communication centered on the concept of synonymy, the principle that the same meaning can be expressed by many different syntactic realizations. Through the synonymous mapping, they introduced semantic entropy, up/down semantic mutual information, semantic channel capacity, and semantic rate-distortion functions, proving corresponding coding theorems that demonstrate performance gains over classical counterparts.

Despite these pioneering efforts, a systematic information-theoretic treatment of pragmatic information has remained largely absent. Early philosophical inquiries~\cite{Weizsacker1972} and historical surveys~\cite{Gernert2006} provided valuable conceptual foundations but stopped short of developing a rigorous mathematical framework. The quantitative formulation by Weinberger~\cite{Weinberger2002,Weinberger2024}, while introducing a formal definition of pragmatic information as the change in an agent's action distribution, was primarily developed within the context of biological evolution and did not establish fundamental coding theorems, rate-distortion limits, or channel capacity characterizations. More critically, none of these prior works provided a unified framework that integrates pragmatic information with syntactic and semantic layers. Consequently, a comprehensive mathematical theory capable of quantifying, compressing, and reliably transmitting pragmatic information in engineered communication systems has been lacking until now.

The recent surge of interest in semantic communication, driven by advances in deep learning and the growing demand for intelligent communications, has further highlighted the limitations of the classical syntactic paradigm. Deep learning-based semantic communication systems have demonstrated significant performance improvements in tasks such as image transmission, speech communication, and text understanding~\cite{Survey_Shi,Survey_Qin,Survey_Xie,Survey_Gunduz}. These systems extract and transmit semantic features relevant to the task, moving beyond bit-level fidelity to meaning-level effectiveness. The convergence of communication and artificial intelligence has opened new possibilities for task-oriented communication systems that prioritize the end goal over exact symbol reconstruction. Zhang et al.~\cite{Semantic_ZhangPing} and Niu et al.~\cite{Semantic_Niu} have articulated the vision of wisdom-evolutionary and primitive-concise communication networks for 6G, where semantic information plays a central role.

Equally important is the convergence of communication and control, extensively studied in networked control systems (NCS) and cyber-physical systems (CPS), where joint design of communication and control policies is essential for stability and performance under bandwidth constraints. Foundational results on this co-design have been established by Xiao et al.~\cite{Xiao2003Joint}, with comprehensive surveys provided by Park et al.~\cite{Park2018WNCS}, Wang et al.~\cite{Wang2023WNCS}, Zhao et al.~\cite{Zhao2019CoDesign}, Negi and Chakrabortty~\cite{Negi2022Optimal}, and Lu and Guo~\cite{Lu2023Survey}. Tishby and Polani~\cite{Tishby2011} incorporated information-theoretic costs into the Bellman equation from a perception-action perspective, though their focus remained on action complexity rather than physical communication resources. Critically, these works treat information at the syntactic level, emphasizing signal fidelity over semantic or pragmatic content.

At a deeper level, modern information-driven systems, whether they are large language models interacting with users, autonomous vehicles navigating complex environments, embodied AI agents performing physical tasks, or human learners acquiring knowledge, share a common operational structure that can be abstracted as a perception-communication-decision-execution (PCDE) closed loop. In such systems, information is not merely transmitted for its own sake; it is used to guide decisions and actions that ultimately determine system performance. The pragmatic dimension of information, its effect on conduct, is therefore paramount. The following examples illustrate the ubiquity of the PCDE loop and the central role of pragmatic information across diverse domains.

\begin{figure}[htbp]
\setlength{\abovecaptionskip}{0.cm}
\setlength{\belowcaptionskip}{-0.cm}
  \centering{\includegraphics[scale=0.06]{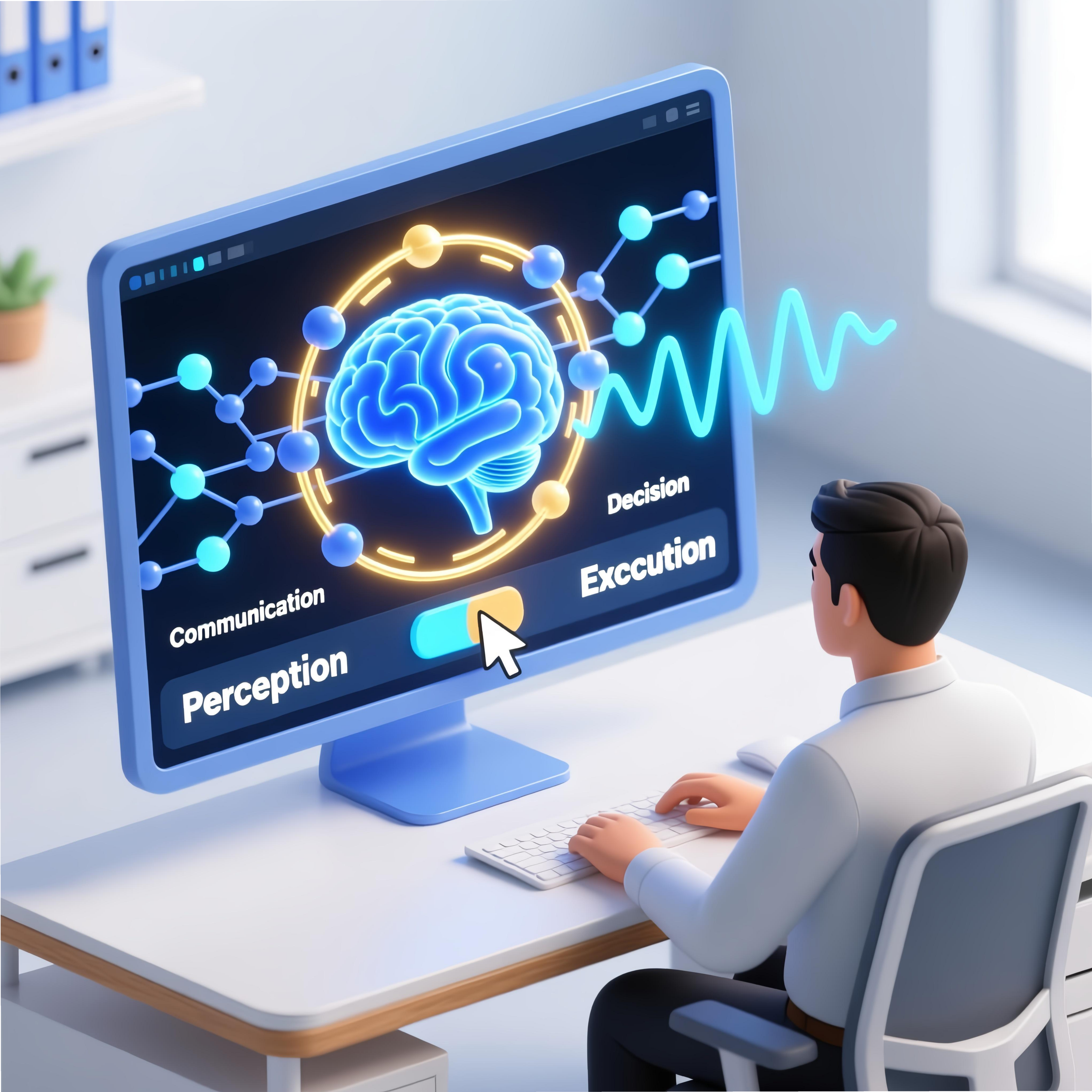}}
\caption{The perception-communication-decision-execution (PCDE) loop in large language model (LLM) and AI agent interactions. The user query (perception) is processed by the LLM (communication/understanding), which generates a response (decision) that drives the user's subsequent action (execution). Pragmatic information resides in the effectiveness of the response in achieving the user's intended goal.}
\label{fig:pcdc-llm}
\end{figure}

\textit{Large Language Models and AI Agents.} In a typical interaction with a large language model (LLM) or AI agent, the user submits a query (perception), the model processes the query through its internal reasoning (communication), generates a response or takes an action (decision), and the user applies the response to their task (execution). The pragmatic value of the response lies not in its syntactic correctness but in its effectiveness in helping the user achieve their goal—whether solving a problem, generating code, or making a decision. The same response may be pragmatically equivalent or entirely different depending on the user's context and intent. Figure~\ref{fig:pcdc-llm} illustrates this PCDE loop, where pragmatic information bridges the gap between meaning and effective action.

\begin{figure}[htbp]
\setlength{\abovecaptionskip}{0.cm}
\setlength{\belowcaptionskip}{-0.cm}
  \centering{\includegraphics[scale=0.06]{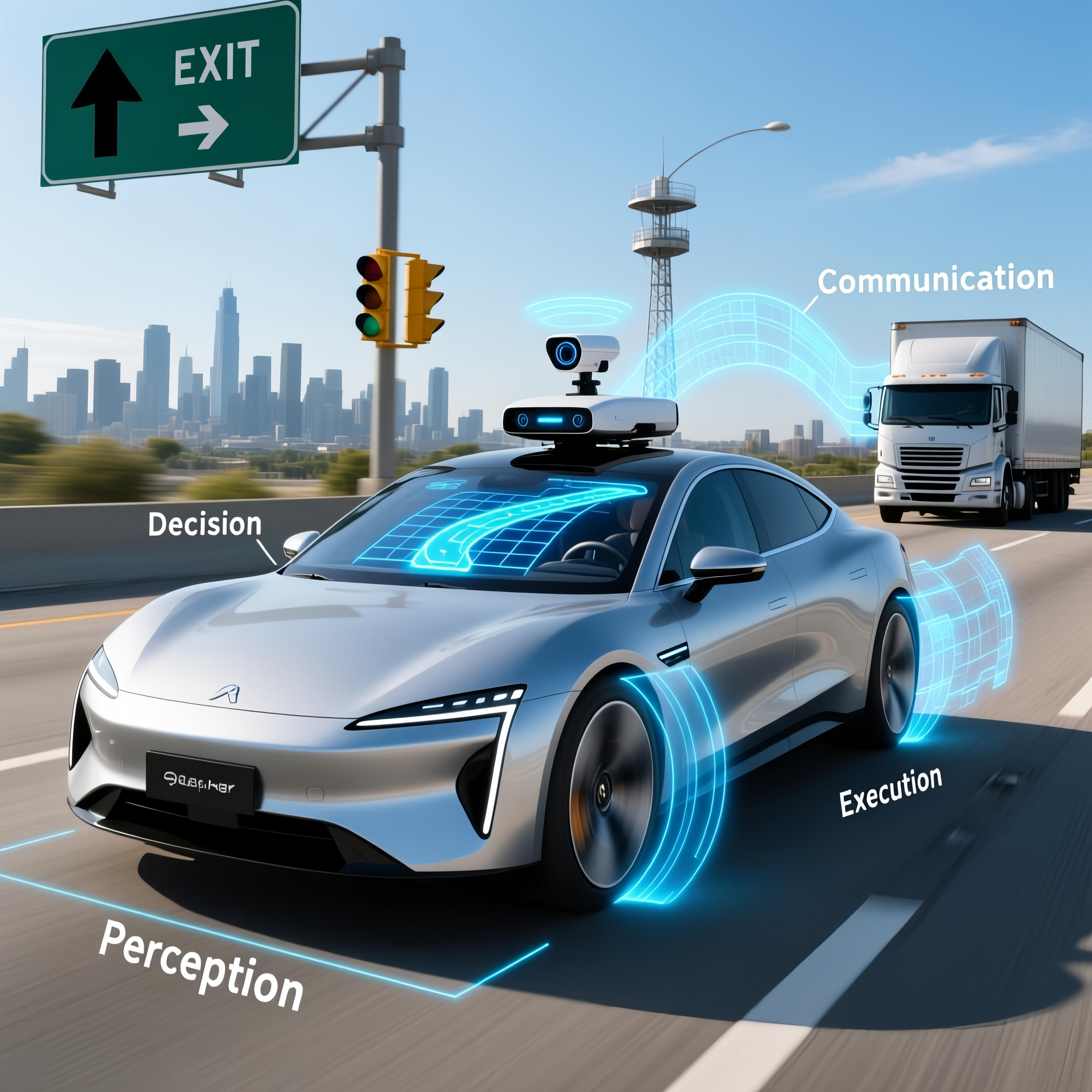}}
\caption{The perception-communication-decision-execution loop in autonomous driving. Sensors perceive the environment (perception), the vehicle communicates with other vehicles and infrastructure (communication), the planning module decides on a trajectory (decision), and the actuators execute the maneuver (execution). Pragmatic information—the correct action to take—is the ultimate objective that transcends individual sensor measurements.}
\label{fig:pcdc-autonomous}
\end{figure}

\textit{Autonomous Driving.} An autonomous vehicle operates within a complex PCDE loop: sensors (cameras, LiDAR, radar) perceive the surrounding environment (perception), the vehicle communicates with other vehicles and infrastructure via V2X (communication), the planning module decides on acceleration, braking, and steering (decision), and the actuators execute the planned actions (execution). The effectiveness of the system is measured by safety, efficiency, and comfort—all pragmatic objectives that depend on the correct interpretation of sensory data and the appropriate selection of actions. The same set of sensor readings may lead to different actions depending on the context (e.g., yielding to a pedestrian vs. proceeding through an intersection), highlighting the need for pragmatic reasoning. Figure~\ref{fig:pcdc-autonomous} depicts this loop, where pragmatic information—the optimal action class—is the key to safe and efficient navigation.

\begin{figure}[htbp]
\setlength{\abovecaptionskip}{0.cm}
\setlength{\belowcaptionskip}{-0.cm}
  \centering{\includegraphics[scale=0.06]{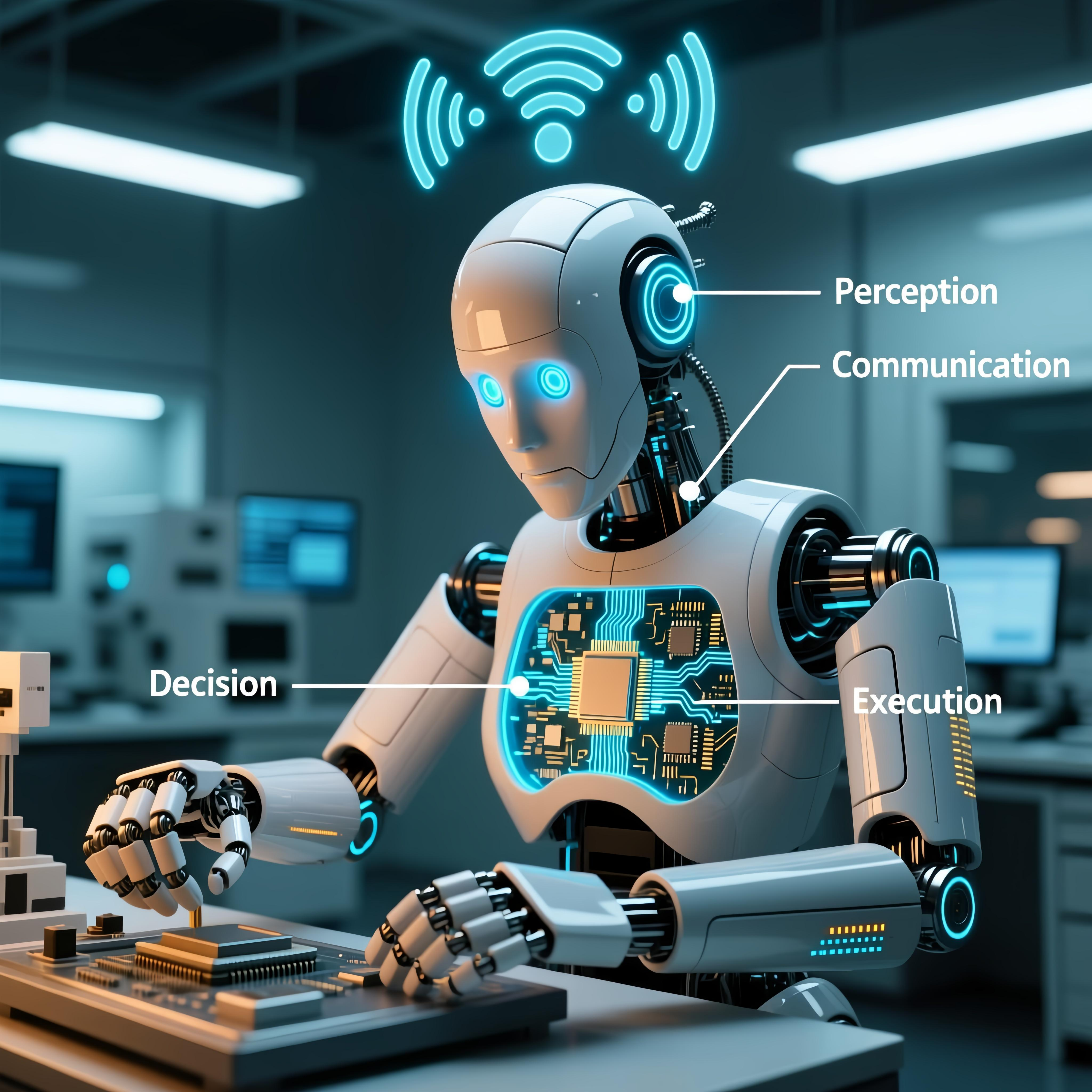}}
\caption{The perception-communication-decision-execution loop in embodied AI and robotics. Sensors perceive the environment (perception), internal representations communicate state information (communication), the policy selects an action (decision), and the actuators interact with the physical world (execution). Pragmatic information is embodied in the action's effectiveness in achieving the task objective.}
\label{fig:pcdc-embodied}
\end{figure}

\textit{Embodied AI and Robotics.} Embodied AI agents and robots operate in physical environments, where perception, communication, decision, and execution are tightly coupled. Sensors perceive the environment (perception), internal representations communicate state and intent across modules (communication), the control policy decides on the next action (decision), and motors and actuators execute the motion (execution). The success of the system depends on the pragmatic effectiveness of the actions—whether the robot grasps the object, navigates to the target, or completes the task. The same sensory input may lead to different actions depending on the task goal, illustrating the context-dependent nature of pragmatic information. Figure~\ref{fig:pcdc-embodied} shows this loop, where pragmatic information is the bridge between perception and effective physical action.

\begin{figure}[htbp]
\setlength{\abovecaptionskip}{0.cm}
\setlength{\belowcaptionskip}{-0.cm}
  \centering{\includegraphics[scale=0.06]{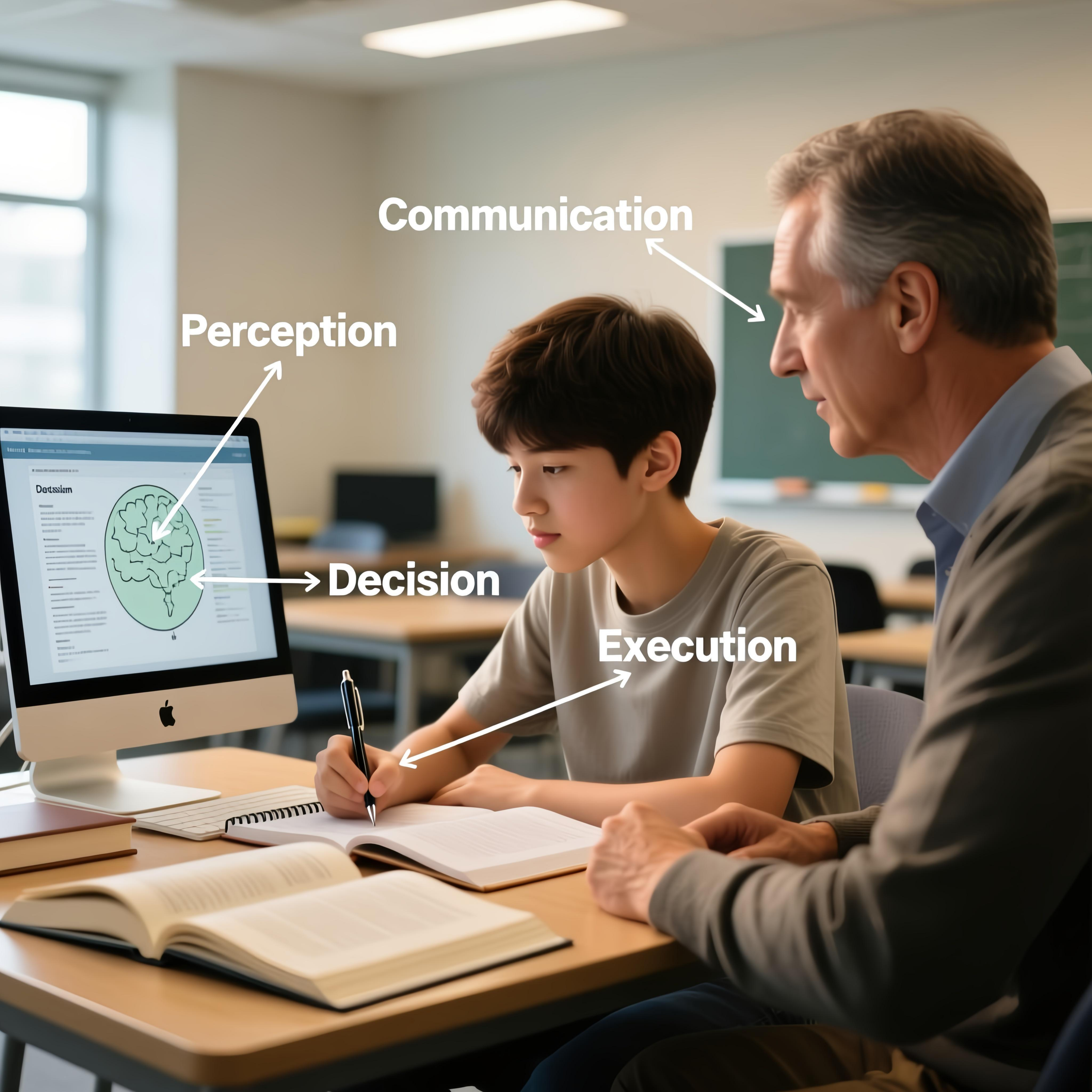}}
\caption{The perception-communication-decision-execution loop in human learning. Instructional content is perceived (perception), communicated through explanation and feedback (communication), the learner decides on a learning strategy (decision), and applies it to practice (execution). Pragmatic information is the knowledge that effectively improves performance, beyond mere factual recall.}
\label{fig:pcdc-learning}
\end{figure}

\textit{Human Learning and Education.} Human learning follows a PCDE loop where instructional content is perceived (perception), communicated through explanation, demonstration, and feedback (communication), the learner decides on a learning strategy (decision), and applies it through practice and application (execution). The effectiveness of instruction depends on its pragmatic value—whether it enables the learner to perform better in real-world tasks. The same instructional content may have high pragmatic value for one learner but low value for another, depending on prior knowledge and learning context. Figure~\ref{fig:pcdc-learning} illustrates this loop, where pragmatic information is the knowledge that translates into improved performance.

These four examples, spanning AI, autonomous systems, robotics, and human cognition, share a common underlying structure: information flows through a perception-communication-decision-execution (PCDE) loop, and the ultimate measure of success is the effectiveness of the resulting actions. In each domain, pragmatic information, the information that makes a difference to decisions and actions, is the central concern. The value of information (VoI) quantifies the utility gain from information, while the cost of information (CoI) captures the resources needed to acquire and transmit it. Despite the ubiquity of this structure across diverse domains, no unified mathematical framework exists for describing, analyzing, and optimizing pragmatic information systems. Classical information theory emphasizes symbol fidelity; control theory emphasizes system stability; and decision theory emphasizes utility maximization. None provides a comprehensive account of the interplay among perception, communication, decision, and execution. A unified theory is therefore needed to bridge these disciplines and to provide a rigorous foundation for the design of modern information-driven systems.

\subsection{Contributions and Organization}

This paper presents a mathematical theory of pragmatic information that unifies communication, control, and decision-making within a single coherent framework. The proposed theory extends classical information theory by incorporating the pragmatic dimension of information, its effect on conduct, while subsuming semantic information theory and classical information theory as special cases. By introducing the isoteleia mapping, which formalizes the principle of equifinality (distinct semantic paths converging to the same optimal action), we establish a three-tier hierarchy of syntactic, semantic, and pragmatic information that provides a complete characterization of information in goal-directed systems.

\begin{figure}[htbp]
\setlength{\abovecaptionskip}{0.cm}
\setlength{\belowcaptionskip}{-0.cm}
  \centering{\includegraphics[scale=0.95]{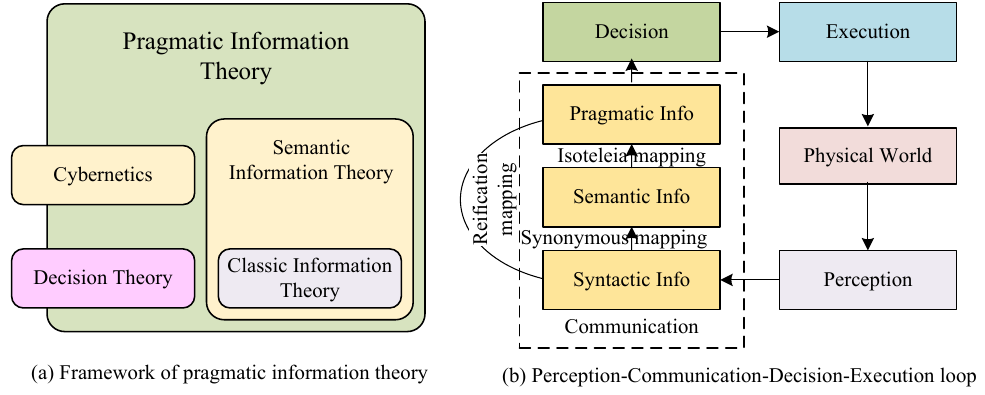}}
\caption{Three-tier framework of pragmatic information theory and its integration with communication, control, and decision.}
\label{fig:information-theory-framework}
\end{figure}

Figure~\ref{fig:information-theory-framework} illustrates the three-tier hierarchy of information and its integration with the PCDE (perception-communication-decision-execution) closed loop. As shown in Fig.~\ref{fig:information-theory-framework}(a), this unified framework positions pragmatic information theory as a natural generalization that subsumes both classical and semantic information theories, while simultaneously integrating decision theory and cybernetics into a single coherent structure. The three theories are nested and compatible: pragmatic information theory contains semantic information theory as a special case when the isoteleia mapping is trivial, and both contain classical information theory when the synonymous mapping is also trivial. This hierarchical compatibility ensures backward compatibility with existing theories while extending their reach into decision-making and control. Fig.~\ref{fig:information-theory-framework}(b) depicts the perception-communicatio-decision-execution (PCDE) loop, where syntactic information, at the bottom, is the domain of classical information theory and concerns the faithful transmission of symbols. Semantic information, the middle layer captured by the synonymous mapping, formalizes the principle that the same meaning can be expressed by many different syntactic realizations. At the top lies pragmatic information, the focus of this paper, which concerns the effectiveness of information in guiding actions. The isoteleia mapping groups semantic classes that lead to the same optimal terminal action, embodying the principle of equifinality. Their composition, the reification mapping, directly bridges pragmatic purposes to syntactic signals, enabling the materialization of intent into physical communication. The PCDE loop connects the Physical World through perception, communication, decision, and execution, with pragmatic information serving as the essential link that closes the loop, answering not only ``What signal should be sent?'' (syntax) and ``What does it mean?'' (semantics), but fundamentally ``What difference does this information make to the decisions and actions of the receiver?'' (pragmatics). 

Figure~\ref{fig:pragmatic-information-theory-framework} illustrates the mathematical framework of pragmatic information theory. 
\begin{figure}[htbp]
\setlength{\abovecaptionskip}{0.cm}
\setlength{\belowcaptionskip}{-0.cm}
  \centering{\includegraphics[scale=0.85]{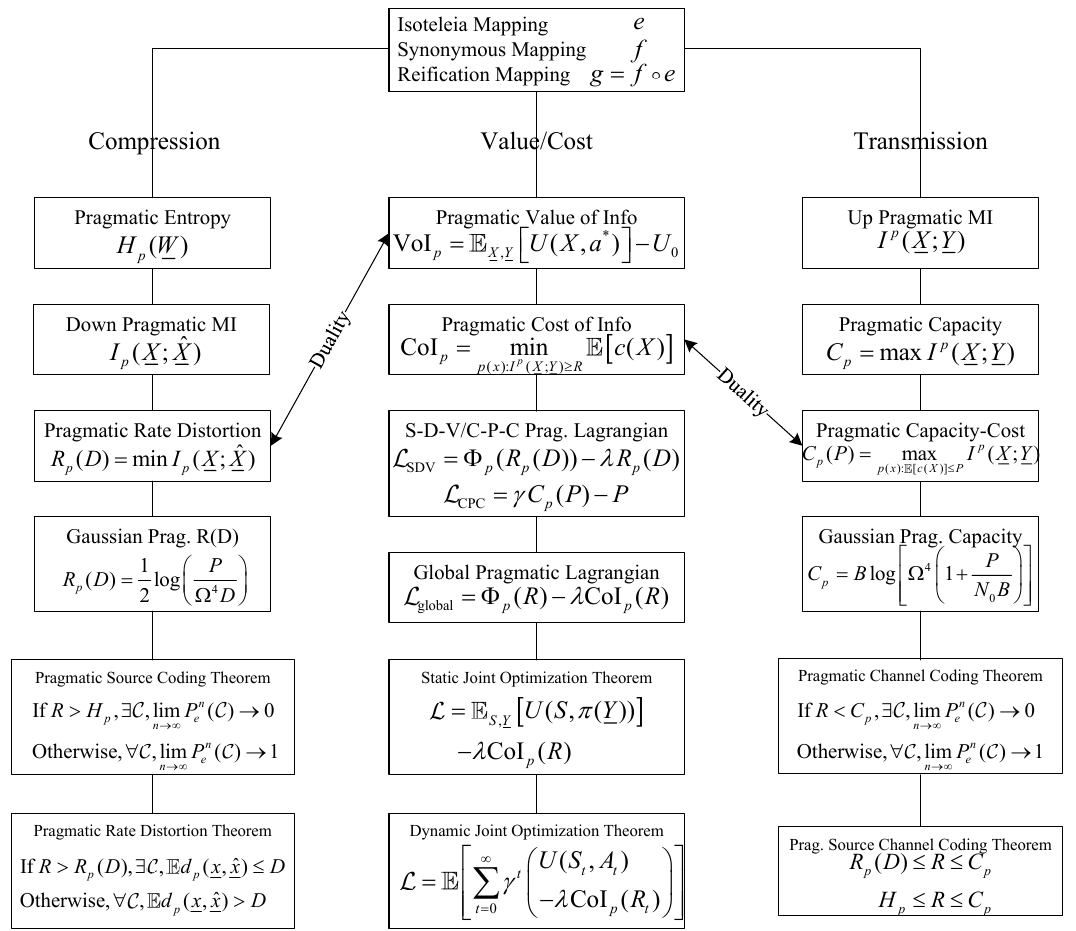}}
\caption{The mathematical framework of pragmatic information theory.}
\label{fig:pragmatic-information-theory-framework}
\end{figure}

The main contributions of this paper are summarized as follows:

\begin{enumerate}
    \item \textbf{Unified Pragmatic Information System Framework.} We develop a comprehensive system model for pragmatic information, integrating Shannon's communication paradigm with cybernetic control. The system is structured as a perception-communication-decision-execution (PCDE) closed loop, comprising ten core modules organized into two functional subsystems: the Cognitive-Regulation Subsystem (CRS), which handles reasoning, intent formulation, and value/cost evaluation, and the Symbolic-Conduction Subsystem (SCS), which handles encoding, transmission, and decoding. The reification mapping $g = f\circ e$ serves as the critical bridge between these subsystems, enabling the direct materialization of pragmatic intent into syntactic signals.

    \item \textbf{Information-Theoretic Measures for Pragmatic Information.} We define a complete set of information-theoretic measures at the pragmatic level, including pragmatic entropy $H_p(\underline{W})$, up/down pragmatic mutual information $I^p(\underline{X};\underline{Y})$ and $I_p(\underline{X};\underline{Y})$, pragmatic channel capacity $C_p$, and pragmatic rate-distortion function $R_p(D)$. We establish the fundamental hierarchies that relate these measures to their syntactic and semantic counterparts: $H_p(\underline{W}) \le H_s(\tilde{W}) \le H(W)$, $C_p \ge C_s \ge C$, and $R_p(D) \le R_s(D) \le R(D)$. These hierarchies quantify the performance gains achievable by exploiting pragmatic abstraction.

    \item \textbf{Pragmatic Value of Information and Pragmatic Cost of Information.} We introduce the pragmatic value of information (VoI) as the utility gain obtained from task-relevant information, and the pragmatic cost of information (CoI) as the minimum resource expenditure required to convey information at a given pragmatic rate. We establish the dualities between VoI and the rate-distortion function, and between CoI and the channel capacity, and develop single-sided extensions for perceptual and expressive paths. These measures provide a decision-theoretic and economic interpretation of the information-theoretic limits.

  \item \textbf{Lagrangian Dual Framework and Pragmatic Shadow Price for Cross-Layer Optimization.} We develop a unified Lagrangian dual framework that integrates the rate-distortion-value duality (down loop) and the capacity-cost duality (up loop) into a coherent cross-layer optimization principle. The global pragmatic Lagrangian $\mathcal{L}_{\mathrm{global}}(R; \lambda) = \Phi_p(R) - \lambda \, \mathrm{CoI}_p(R)$ expresses the net benefit as value minus cost, with the optimality condition $\lambda\,\mathrm{CoI}_p'(R^*) = \Phi_p'(R^*)$ balancing marginal value against marginal cost. The resulting pragmatic efficiency functional \(\mathcal{E}_p(\lambda) = \sup_R [\Phi_p(R) - \lambda \, \mathrm{CoI}_p(R)]\) defines a fundamental behavioral limit for any resource-constrained intelligent system—a \emph{behavioral capacity} that parallels Shannon's physical capacity, but replaces symbol fidelity with the effectiveness of information in guiding actions. Central to this framework is the shadow price \(\lambda\), which serves as a unified currency that converts physical resource consumption—such as transmit power, bandwidth, and latency—into utility-equivalent units. In static settings, \(\lambda\) determines the optimal operating point where marginal value equals marginal cost; in dynamic settings, it acts as a real-time scarcity signal that continuously adapts system behavior to varying channel conditions, power budgets, and task priorities. This framework provides a principled, practical, and dimensionally consistent approach for resource allocation and system design in task-oriented communication systems, bridging information-theoretic limits with decision-theoretic objectives without ad hoc heuristics.
    
    \item \textbf{Coding Theorems for Pragmatic Communication.} We prove three fundamental coding theorems for pragmatic information systems: (i) the pragmatic lossless source coding theorem, establishing that the minimum rate for lossless pragmatic compression is the pragmatic entropy $H_p(\underline{W})$; (ii) the pragmatic channel coding theorem, establishing that the maximum reliable transmission rate is the pragmatic capacity $C_p$; and (iii) the pragmatic rate-distortion coding theorem, establishing that the minimum rate for lossy pragmatic compression with distortion $D$ is $R_p(D)$. These theorems generalize Shannon's classical coding theorems \cite{Classicpaper_Shannon,Shannon_Weaver} and the semantic coding theorems of Niu and Zhang \cite{Paper_SIT,Book_SIT}.

    \item \textbf{Continuous-Domain Pragmatic Information Measures.} We extend the pragmatic information measures to continuous messages through the concept of isoteleic volumes $\Omega$, which are the continuous analogues of pragmatic equivalence classes. We derive closed-form expressions for the pragmatic capacity of the Gaussian channel $C_p = \frac{1}{2}\log(\Omega^4(1+P/\sigma^2))$ and the pragmatic rate-distortion function for Gaussian sources $R_p(D) = \frac{1}{2}\log(P/(\Omega^4 D))$, along with the corresponding CoI and VoI functions. We also derive the band-limited capacity $C_p = B\log(\Omega^4(1+P/(N_0B)))$.

    \item \textbf{Joint Optimization of Communication, Control, and Decision.} We formulate the joint optimization of information, control, and decision-making in both static and dynamic settings. For the static case, we derive the optimality conditions for the global pragmatic Lagrangian, showing that the optimal policy is the posterior-maximizing Bayes decision rule and the optimal rate satisfies marginal value equals marginal cost. For the dynamic case, we derive the Bellman equation for the dynamic pragmatic Lagrangian, providing a principled framework for co-designing communication and control in sequential decision-making systems.

    \item \textbf{Generalization and Unification.} We demonstrate that the pragmatic information theory framework seamlessly subsumes classical communication systems (when semantic and pragmatic layers are trivialized), closed physical systems (with explicit utility functions), embodied and interactive systems (with multiple agents), and open social and cognitive systems (with implicit and evolving utility functions). This establishes the pragmatic information system as a primary reference model for analyzing and designing a system that involves goal-directed information processing under uncertainty.

\end{enumerate}

The remainder of the paper is organized as follows. Section~\ref{section_II} presents the pragmatic information system model, including the system architecture, the three-tier hierarchy, the reification mapping, and the axiomatic foundations. Section~\ref{section_III} defines pragmatic entropy and its joint and conditional counterparts, establishing the fundamental hierarchies. Section~\ref{section_IV} introduces pragmatic relative entropy and up/down pragmatic mutual information. Section~\ref{section_V} defines pragmatic channel capacity and pragmatic rate-distortion functions, including single-sided extensions. Section~\ref{section_VI} establishes the pragmatic value of information and pragmatic cost of information, along with their Lagrangian dual framework. Section~\ref{section_VII} proves the pragmatic lossless source coding theorem. Section~\ref{section_VIII} proves the pragmatic channel coding theorem. Section~\ref{section_IX} proves the pragmatic rate-distortion coding theorem. Section~\ref{section_X} extends the pragmatic information measures to continuous messages, deriving the pragmatic capacity of Gaussian channels and the pragmatic rate-distortion function for Gaussian sources, along with the corresponding CoI and VoI functions. Section~\ref{section_XI} addresses the joint optimization of communication, control, and decision-making in pragmatic information systems. Section~\ref{section_XII} concludes the paper.

\section{Pragmatic Information System and Isoteleia Mapping}
\label{section_II}

This section presents the theoretical foundations of the pragmatic information system. We establish a unified framework that integrates Shannon's communication paradigm with cybernetic control, introducing the three-tier hierarchy of syntactic, semantic, and pragmatic information. Central to this framework is the Isoteleia Mapping, which formalizes the principle of equifinality, that is, distinct semantic paths converging to the same optimal action.

\subsection{Notation Conventions}
\label{subsec:notation}

We adopt standard information-theoretic notation throughout. Calligraphic letters ($\mathcal{X}, \mathcal{Y}, \mathcal{W}, \mathcal{V}$) denote alphabets; uppercase letters ($X,Y,W$) are random variables, lowercase ($x,y,w$) their realizations. Semantic and pragmatic counterparts are marked with a tilde ($\tilde{W}$) and an underline ($\underline{W}$), respectively. Probabilities are written as $P(\cdot)$ or $p(\cdot)$, expectations as $\mathbb{E}[\cdot]$, and $\log(\cdot)$ denotes the logarithm to base~2. The main symbols are summarized in Table~\ref{tab:main_symbols_extended}; additional notation is introduced as needed in the text.

\begin{table}[htbp]
\centering
\small
\caption{Principal symbols, meanings, and intuitive explanations}
\label{tab:main_symbols_extended}
\begin{tabular}{p{4.0cm} p{11.0cm}}
\hline
\textbf{Symbol} & \textbf{Meaning and Intuitive Explanation} \\
\hline
$\mathcal{X}, \mathcal{Y}, \mathcal{W}, \mathcal{V}$ & Alphabets: collections of possible symbols \\
$X, Y, W, V$ / $x, y, w, v$ & Random variables / realizations at syntactic layer \\
$\tilde{X}, \tilde{Y}, \tilde{W}, \tilde{V}$ / $\underline{X}, \underline{Y}, \underline{W}, \underline{V}$ & Semantic layer (meanings) / Pragmatic layer (actions) \\
$P(\cdot)$ or $p(\cdot)$, $\mathbb{E}[\cdot]$, $\log(\cdot)$ & Probability, expectation, log (bits/sebits/prabits) \\
\hline
$H(W)$, $H_s(\tilde{W})$, $H_p(\underline{W})$ & Entropy hierarchy: symbol $\to$ meaning $\to$ action uncertainty \\
$H(\cdot|\cdot)$ & Conditional entropy: residual uncertainty given context \\
$I(X;Y)$ & Syntactic mutual information: symbol-level dependence \\
$I_s(\tilde{X};\tilde{Y})$, $I^s(\tilde{X};\tilde{Y})$ & Down/Up semantic MI: meaning dependence (marginal/joint) \\
$I_p(\underline{X};\underline{Y})$, $I^p(\underline{X};\underline{Y})$ & Down/Up pragmatic MI: action dependence (marginal/joint) \\
$C$, $C_s$, $C_p$ & Capacity hierarchy: max reliable symbol/meaning/action rate \\
$R(D)$, $R_s(D)$, $R_p(D)$ & Rate-distortion hierarchy: min symbol/meaning/action rate for distortion $D$ \\
\hline
$f: \tilde{\mathcal{W}} \to 2^{\mathcal{W}}$ & Synonymous mapping: one meaning $\to$ many signals \\
$e: \underline{\mathcal{W}} \to 2^{\tilde{\mathcal{W}}}$ & Isoteleia mapping: one action $\to$ many meanings \\
$g = f \circ e: \underline{\mathcal{W}} \to 2^{\mathcal{W}}$ & Reification mapping: one action $\to$ all realizing signals \\
$g^n$, $f^n$, $u^n$, $\mathcal{U}^n$ & Sequential extensions: block-level mappings and sequences \\
$\Omega$ & Average isoteleic volume: avg size of action-equivalence set (continuous) \\
\hline
$\mathrm{VoI}_p(R)$ or $\Phi_p(R)$ & Pragmatic value of information: max utility gain from rate $R$ \\
$\mathrm{CoI}_p(R)$ & Pragmatic cost of information: min resource (power/bandwidth) for rate $R$ \\
$\mathcal{L}_{\mathrm{SDV}}$, $\mathcal{L}_{\mathrm{CPC}}$, $\mathcal{L}_{\mathrm{global}}$ & Lagrangians: down-loop value$-$cost, up-loop rate$-$cost, global net benefit \\
$\lambda, \gamma$ & Lagrange multipliers (shadow prices): convert resource cost to utility \\
$V(S)$ & Value function in dynamic programming: optimal expected future reward \\
\hline
\end{tabular}
\end{table}

For sequences and conditional quantities, we adopt the usual conventions: $u^n = (u_1,\dots,u_n)$, $\mathcal{U}^n$ denotes the $n$-fold Cartesian product, and $f^n$ is the sequential extension of a mapping $f$. Conditional entropy and mutual information are written as $H(\cdot|\cdot)$ and $I(\cdot;\cdot|\cdot)$, with the convention $0\log 0=0$. Syntactic quantities are measured in bits, semantic in sebits, and pragmatic in prabits.

\subsection{Pragmatic Information System}
\label{subsec:system-model}

The pragmatic information system is a perception-communication-decision-execution (PCDE) closed-loop architecture that integrates the classic Shannon communication paradigm with a cybernetic control framework. Its fundamental purpose is not merely to transmit symbols, but to ensure that the received information leads to terminal actions that maximize a predefined task utility. In this section, we present the system model in detail, describe its processing flow with state-space dynamics, formulate its axiomatic foundations, and discuss its broad applicability across diverse domains.

\subsubsection{System Structure and Processing Flow}

The system consists of ten core modules that operate in a sequential cascade, forming a complete perception–communication–decision–execution loop. These modules are: (1) Physical World / Task Environment, (2) Task Intent Initiator, (3) Source, (4) Reification Mapping, (5) Encoder, (6) Channel, (7) Decoder, (8) Demapping, (9) Destination, and (10) Task Execution Responder. Figure~\ref{Fig_Pragmatic_information_system} provides a conceptual illustration of the overall architecture.

\begin{figure*}[htbp]
\setlength{\abovecaptionskip}{0.cm}
\setlength{\belowcaptionskip}{-0.cm}
  \centering{\includegraphics[scale=0.9]{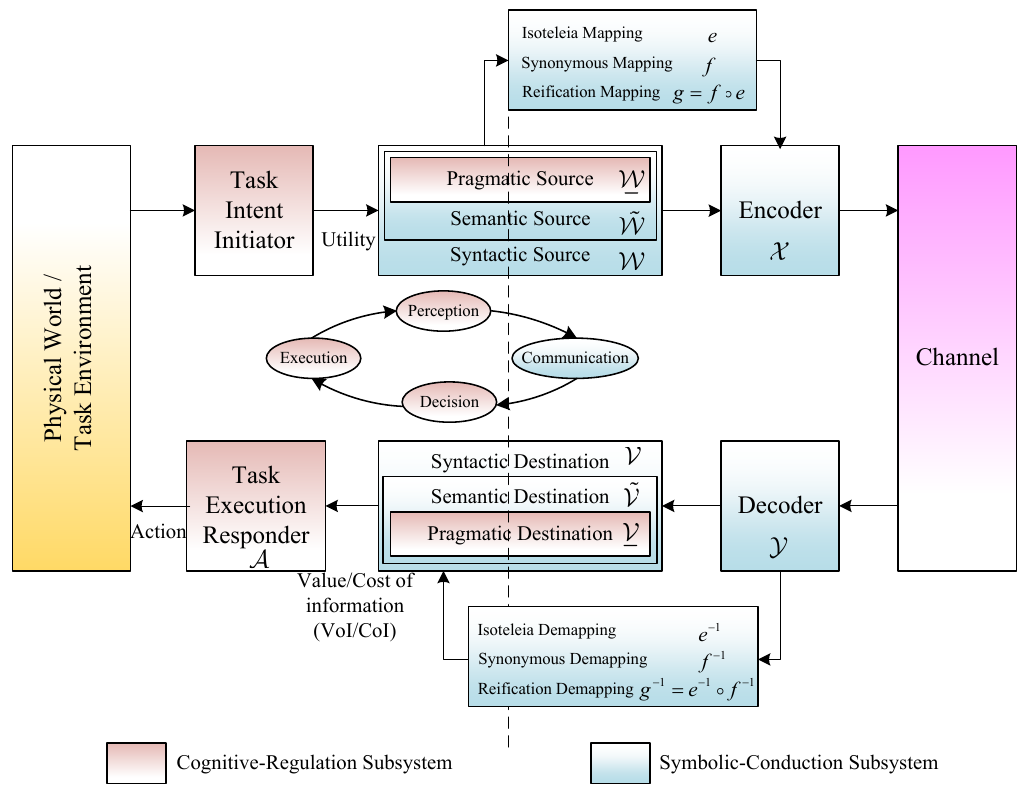}}
  \caption{The block diagram of Pragmatic Information System with PCDE loop, showing the interaction between the Physical World, the Cognitive-Regulation Subsystem (CRS), and the Symbolic-Conduction Subsystem (SCS).}\label{Fig_Pragmatic_information_system}
\end{figure*}

We now describe each module and the information flow that connects them.

\paragraph*{Module 1: Physical World / Task Environment}

The system operates within a Physical World characterized by a Plant with state $S_t \in \mathcal{S}$ at discrete time $t$, governed by:
\begin{equation}
    S_{t+1} = F_m(S_t, A_t) + Z_t,
    \label{eq:state-evolution}
\end{equation}
where $A_t \in \mathcal{A}$ is the action applied by the actuator, $F_m: \mathcal{S} \times \mathcal{A} \to \mathcal{S}$ is the nominal transition dynamics, and $Z_t$ represents exogenous disturbances. Observations are obtained through a sensor model:
\begin{equation}
    Y_t = F_o(S_t) + V_t,
    \label{eq:sensor-model}
\end{equation}
with $F_o: \mathcal{S} \to \mathcal{Y}$ the observation function and $V_t$ measurement noise.

\paragraph*{Module 2: Task Intent Initiator}

The Task Intent Initiator compares the observed state (or belief) $\hat{S}_t$ with a task goal $G_t$ to generate a task intent:
\begin{equation}
    T_t = \Phi(\hat{S}_t, G_t),
    \label{eq:intent-generation}
\end{equation}
where $\Phi$ encodes the system's strategy. This module answers: \emph{``What should be done?''}

\paragraph*{Module 3: Source}

The Source is a three-tier hierarchical structure corresponding to the pragmatic, semantic, and syntactic levels \cite{Semantic_Weaver,Shannon_Weaver}. The Pragmatic Source $\underline{\mathcal{W}}$ first maps the intent $T_t$ to a pragmatic symbol $\underline{w} \in \underline{\mathcal{W}}$ representing the optimal action class, answering: \emph{``What is the best terminal action?''} This is then expanded by the Semantic Source $\tilde{\mathcal{W}}$ into a semantic symbol $\tilde{w} \in \tilde{\mathcal{W}}$ justifying the action, answering: \emph{``Why is this action justified?''} Finally, the Syntactic Source $\mathcal{W}$ maps to a syntactic symbol $w \in \mathcal{W}$ for transmission, answering: \emph{``What signal should be sent?''}

\paragraph*{Module 4: Reification Mapping}

The Reification Mapping bridges the Pragmatic and Syntactic layers as the composition of the isoteleia and synonymous mappings:
\begin{equation}
    g \triangleq f \circ e: \underline{\mathcal{W}} \to \mathcal{P}(\mathcal{P}(\mathcal{W})),
    \label{eq:reification-def}
\end{equation}
where $e: \underline{\mathcal{W}} \to \mathcal{P}(\tilde{\mathcal{W}})$ maps a pragmatic symbol to all semantic justifications (equifinality), and $f: \tilde{\mathcal{W}} \to \mathcal{P}(\mathcal{W})$ maps a semantic symbol to all syntactic realizations (synonymy). Thus,
\begin{equation}
    g(\underline{w}) = \bigcup_{\tilde{w} \in e(\underline{w})} f(\tilde{w}) \subset \mathcal{W},
    \label{eq:reification-expanded}
\end{equation}
enumerating all syntactic realizations that fulfill a given pragmatic purpose.

\paragraph*{Module 5: Encoder}

The Encoder maps a block of $n$ syntactic symbols to channel input symbols:
\begin{equation}
    \phi: \mathcal{W}^n \to \mathcal{X}^n,
    \label{eq:encoder}
\end{equation}
performing a unified source-channel encoding operation designed to maximize end-to-end task utility.

\paragraph*{Module 6: Channel}

The Channel is the physical medium characterized by the conditional distribution $P(Y_t | X_t)$, subject to resource constraints such as power, bandwidth, and delay.

\paragraph*{Module 7: Decoder}

The Decoder recovers an estimate of the transmitted syntactic sequence:
\begin{equation}
    \psi: \mathcal{Y}^n \to \mathcal{W}^n,
    \label{eq:decoder}
\end{equation}
ensuring that the recovered symbols preserve semantic and pragmatic content for effective decision-making.

\paragraph*{Module 8: Demapping}

The Demapping module reverses the Reification process: recovered syntactic symbols are first mapped to estimated semantic symbols via an inverse synonymous mapping, then to estimated pragmatic intentions via an inverse isoteleia mapping, recovering meaning and purpose from the received signals.

\paragraph*{Module 9: Destination}

The Destination is the hierarchical counterpart to the Source. The Syntactic Destination $\mathcal{V}$ receives decoded syntactic symbols; the Semantic Destination $\tilde{\mathcal{V}}$ reconstructs semantic interpretations; and the Pragmatic Destination $\underline{\mathcal{V}}$ reconstructs the intended pragmatic purpose, answering: \emph{``What was the intended purpose?''}

\paragraph*{Module 10: Task Execution Responder}

The Task Execution Responder (Actuator) translates the recovered pragmatic intention $\hat{\underline{w}}$ into a physical action:
\begin{equation}
    A_t = \mu(\hat{\underline{w}}_t),
    \label{eq:actuator-policy}
\end{equation}
where $\mu: \underline{\mathcal{V}} \to \mathcal{A}$ maps pragmatic intentions to actions, closing the perception-communication-decision-execution loop via Eq.~(\ref{eq:state-evolution}).

The effectiveness of the entire process is evaluated by the Value of Information (VoI) and Cost of Information (CoI). The VoI quantifies the expected utility gain:
\begin{equation}
    \text{VoI} = \mathbb{E}_{S,Y}[U(S_t, \mu(\hat{\underline{w}}_t))] - \mathbb{E}_{S}[U(S_t, \mu(\underline{w}_0^*))],
    \label{eq:voi-def}
\end{equation}
where $\underline{w}_0^*$ is the optimal action based on prior beliefs alone. The CoI quantifies the minimum resource expenditure required to convey pragmatic information at a given rate. Their joint consideration enables principled trade-off optimization between task performance and resource consumption, a central theme of the pragmatic Lagrangian framework developed in Section~\ref{section_VI}.

\subsubsection{Functional Subsystems}

The ten modules are naturally grouped into two functional subsystems: the \textbf{Cognitive-Regulation Subsystem (CRS)} and the \textbf{Symbolic-Conduction Subsystem (SCS)}.

The CRS handles high-level reasoning, intent formulation, and information value evaluation. It comprises the Task Intent Initiator (generating intent from state and goals), the Pragmatic Source (formalizing intent as pragmatic symbols), the Pragmatic Destination (reconstructing pragmatic intention), and the Task Execution Responder (translating intention into action). The isoteleia mapping $e$ embodies equifinality by grouping semantic justifications that lead to the same optimal action. The CRS answers: \emph{``What should be done, and is it worth doing?''}

The SCS handles encoding, transmission, and decoding. It comprises the Semantic Source (generating semantic symbols from pragmatic intent), the Syntactic Source (generating syntactic symbols from semantic meanings), the Encoder (converting symbols to physical signals), the Decoder (recovering symbols from received signals), and the Syntactic and Semantic Destinations (reconstructing meanings along the reverse path). The synonymous mapping $f$ captures synonymy by mapping each meaning to multiple syntactic realizations. The SCS answers: \emph{``How should it be expressed and transmitted?''}

The reification mapping $g = f \circ e$ serves as the critical bridge between these subsystems, allowing CRS intent to be directly materialized into syntactic signals by the SCS.

\subsubsection{Axiomatic Foundations}

The operation of the pragmatic information system is governed by three fundamental axioms that ensure its mathematical tractability, general applicability, and consistency with both information theory and control theory.

\paragraph*{\textbf{Axiom 1: The Probability Axiom}}

The system operates in a world where all observations, states, and intentions are subject to irreducible uncertainty. This uncertainty is quantified by probability measures. Formally, for any time $t$, the joint distribution over the system variables is:
\begin{equation}
    P(S_t, W_t, Y_t) \in \Delta(\mathcal{S} \times \mathcal{W} \times \mathcal{Y}),
    \label{eq:joint-distribution}
\end{equation}
where $\Delta$ denotes the space of all probability distributions. The channel is characterized by the conditional distribution $P(Y_t | X_t)$. This axiom ensures that all information-theoretic quantities—entropy, mutual information, channel capacity, rate-distortion functions, and VoI/VoC—are well-defined and computable.

\paragraph*{\textbf{Axiom 2: The Decision Axiom}}

All deliberate actions of the system are governed by the principle of maximizing expected utility. Formally, the \textbf{utility function} is defined as:
\begin{equation}
    U: \mathcal{S} \times \mathcal{A} \to \mathbb{R},
    \label{eq:utility-def}
\end{equation}
which assigns a real-valued scalar to each state-action pair $(S_t, A_t)$, representing the immediate reward or cost associated with taking action $A_t$ in state $S_t$. For any decision node where the system has access to information $I$, it chooses an action $a^*$ such that:
\begin{equation}
    a^* = \arg\max_{a \in \mathcal{A}} \mathbb{E}_{P(S_t | I)}[U(S_t, a)].
    \label{eq:decision-axiom}
\end{equation}
This axiom underpins the isoteleia mapping, as it defines the equivalence classes of semantic states that lead to the same optimal action (the same \emph{telos}). It also underpins the Value of Information (VoI), which measures the utility gain obtained from additional information, and the Pragmatic Lagrangian Functional, which formalizes the trade-off between information value and communication cost.

\paragraph*{\textbf{Axiom 3: The Markov Axiom}}

The system's dynamics satisfy the Markov property:
\begin{equation}
    P(S_{t+1} | S_t, S_{t-1}, \dots, A_t, A_{t-1}, \dots) = P(S_{t+1} | S_t, A_t).
    \label{eq:markov-axiom}
\end{equation}
This property allows the system to be modeled as a Partially Observable Markov Decision Process (POMDP). Furthermore, the Markov property implies that the system can be decomposed across time: by removing the time indices, the dynamic control problem can be reduced to a static communication and encoding problem. In this static formulation, the channel coding and source coding theorems can be applied directly to characterize the fundamental limits of information transmission, treating each time step independently while preserving the overall structure of the pragmatic system. This temporal decomposition forms the basis for the coding theorems developed in subsequent sections.

\subsubsection{Generality and Scope: A Unified Framework for Information-Driven Systems}

The pragmatic information system provides a mathematical language for describing a class of systems that involve goal-directed information processing under uncertainty. The framework spans a continuous spectrum of systems, ranging from pure symbol transmission to tightly coupled physical control loops, and further to open-ended social and cognitive interactions.

We organize the discussion into four conceptual classes, ordered according to the \emph{degree of closure} and the \emph{nature of the utility function}. At one extreme lies the classical communication paradigm, where utility is defined solely by symbol fidelity. Moving along the spectrum, we encounter closed physical systems with well-defined utility functions; embodied and interactive systems with multiple agents; and finally, open social and cognitive systems where utility is implicit, dynamic, and subject to interpretation. This progression reflects increasing complexity in the relationship between information, meaning, and value.

\paragraph*{\textbf{Case I: Classical Communication Systems}}

When the Semantic and Pragmatic layers are trivialized---i.e., when utility depends solely on symbol fidelity (mean squared error, Hamming distance, or bit error rate)---the isoteleia and synonymous mappings reduce to identity mappings. The three-tier source collapses to a single syntactic source, and the system reverts to the familiar source-channel separation architecture. All classical results, including Shannon's source coding, channel coding, and rate-distortion theorems, emerge as special cases. The framework thus serves not as a replacement but as a rigorous generalization of classical information theory, extending its reach from ``how accurately can symbols be transmitted?'' to ``how effectively does the received meaning affect conduct in the desired way?''

\paragraph*{\textbf{Case II: Closed Physical Systems}}

Here the Physical World is governed by well-defined physical laws, the utility function is explicitly specified and measurable (e.g., tracking error, energy consumption, or task success rate), and system behavior is constrained by strict physical and temporal limits. The framework naturally embeds the control-theoretic feedback loop via the Plant, Task Intent Initiator, and Actuator. Representative domains include networked control systems, robotics and industrial automation, autonomous vehicles and intelligent transportation, UAV swarm control, smart grids, and healthcare cyber-physical systems. In all these cases, VoI provides a rigorous metric for quantifying the trade-off between communication resource consumption and closed-loop performance, while the isoteleia mapping captures how diverse sensory inputs converge to the same control or navigation command, enabling robust decision-making under noise and delays.

\paragraph*{\textbf{Case III: Embodied and Interactive Systems}}

In this class, the physical world remains central, but the system now involves multiple agents (human or artificial) with possibly distinct sensors, actuators, and world-views. Utility may be shared or conflicting, and communication becomes strategic or collaborative rather than a mere data pipe. Representative domains include embodied AI and autonomous agents, multi-agent systems and collaborative AI, human-machine interaction (HMI), and intelligent decision support systems. The framework provides a reference model for signaling games, negotiation protocols, and shared intentionality. The isoteleia mapping captures how different perceptual or contextual states lead to the same joint action or collaborative outcome, while VoI quantifies the benefit of sharing information among agents with potentially asymmetric information and conflicting objectives.

\paragraph*{\textbf{Case IV: Open Social and Cognitive Systems}}

Here the physical world recedes into the background, and the focus shifts to internal cognitive processes or emergent social dynamics. Utility is no longer a single measurable objective but may be implicit, dynamic, subjective, or even contradictory across agents. Representative domains include human learning and education, human communication and language, social and economic systems (where agents exchange signals to coordinate actions under uncertainty), biological and cognitive systems (aligning with predictive processing and active inference), and recommendation and content delivery systems. The framework may provide a useful starting point for modeling information exchange in open social and cognitive systems, but the present theory relies on explicit probabilistic and utility assumptions. Relaxing these assumptions is left for future work. The isoteleia mapping may offer a conceptual analogy for the human cognitive ability to generalize from diverse experiences to coherent principles, while VoI suggests a way to quantify the effectiveness of communication in reducing uncertainty and improving decisions under subjective and evolving objectives.

\paragraph*{\textbf{The Unity of the Framework}}

Despite the apparent diversity across these four cases, they can be discussed within a common conceptual core: a state space $\mathcal{S}$, a utility function $U: \mathcal{S} \times \mathcal{A} \to \mathbb{R}$, a three-tier information hierarchy (Syntactic, Semantic, Pragmatic), communication constraints, and a VoI metric guiding resource allocation and decision-making. The spectrum from Case I to Case IV is not a rigid partition but a continuum: classical communication systems can be enhanced with semantic and pragmatic layers; closed physical systems can involve human operators; and social systems can incorporate physical sensors and actuators. The pragmatic information system can serve as a reference model for analyzing a class of goal-directed systems, ranging from classical communication and closed-loop control to embodied and interactive systems. Extensions to social, economic, and cognitive domains are conceptually appealing but are not developed formally in this paper and remain open research directions.

\subsection{Design Principles of Pragmatic Information Systems}
\label{subsec:design-principles}

Building upon the system architecture and the axiomatic foundations established in the preceding subsections, we now articulate three fundamental design principles that guide the analysis, synthesis, and operation of pragmatic information systems. These principles are not arbitrary engineering heuristics; they are direct consequences of the mathematical structure of the pragmatic framework and the nature of information as it flows from the physical world through semantic interpretation to pragmatic action.

\paragraph*{\textbf{Principle I: Pragmatic Imperceptibility and Isoteleia}}

Pragmatic information, the ultimate purpose or terminal action that a communication system is intended to achieve, is not directly observable. It cannot be read from the received signal as one would read a bit string. Instead, pragmatic information must be inferred or evaluated indirectly through the lens of a utility function, which measures the consequences of actions taken in response to received information. Consequently, the fundamental characteristic of pragmatic information is \emph{Isoteleia}: different semantic interpretations, and even different syntactic realizations, are pragmatically equivalent if they lead to the same optimal terminal action under the given utility. This is precisely the equivalence relation captured by the isoteleia mapping $e: \underline{\mathcal{W}} \to \mathcal{P}(\tilde{\mathcal{W}})$, which partitions the semantic space into classes of meanings that share the same optimal action. The Value of Information (VoI) serves as the primary metric for quantifying pragmatic information, measuring the utility gain obtained from information. This principle implies that the system must be designed around the utility function, not around the fidelity of symbol reconstruction; performance metrics should be defined in terms of the quality of the resulting actions; and communication resources should be allocated according to the marginal Value of Information they provide.

\paragraph*{\textbf{Principle II: Semantic Imperceptibility and Synonymy}}

Semantic information, the meaning or interpretation conveyed by a signal, is also not directly observable. It cannot be extracted from the signal by simple demodulation or decoding. Instead, semantic information must be inferred from the observed syntactic symbols through a process of pattern recognition, contextual reasoning, and knowledge-based interpretation. Consequently, the fundamental characteristic of semantic information is \emph{synonymy}: the same meaning can be expressed by many different syntactic realizations. This equivalence relation is captured by the synonymous mapping $f: \tilde{\mathcal{W}} \to \mathcal{P}(\mathcal{W})$, which partitions the syntactic space into equivalence classes of symbols that share the same semantic interpretation. This principle implies that the system should exploit synonymy by allowing lossy compression at the syntactic level, as long as the semantic content is preserved; that the decoder must incorporate contextual information to resolve ambiguities; and that the encoder should extract and transmit semantic features relevant to the task, rather than attempting to reconstruct the original signal.

\paragraph*{\textbf{Principle III: Backward Compatibility and Unification}}

The pragmatic information system must be backward compatible with both classical communication systems and classical control systems. This means that the pragmatic framework is not a replacement for existing theories but a strict generalization that contains them as special cases. When the pragmatic and semantic layers are trivialized, that is, when the utility function depends solely on symbol fidelity and the isoteleia and synonymous mappings become identity mappings, the pragmatic framework must collapse exactly into the classical Shannon communication paradigm. Similarly, when communication is ignored, the framework must collapse into classical control theory. This principle imposes that the three-tier hierarchy must be implemented in a modular fashion, allowing each layer to be independently analyzed, optimized, or bypassed; that the system must include explicit mechanisms to detect when the utility function reduces to syntactic fidelity and automatically switch to classical operation; and that the encoder and decoder must support both pragmatic encoding and syntactic encoding to interface with legacy infrastructure.

Together, these three principles establish the philosophical and mathematical foundation for the design of pragmatic information systems. They ensure that the system is simultaneously \emph{task-effective} (Principle I), \emph{meaning-aware} (Principle II), and \emph{infrastructure-compatible} (Principle III). 

\subsection{Pragmatic Information and Isoteleia Mapping}
\label{subsec:isoteleia-mapping}

We now turn to the purely structural account of the mappings that bridge the three layers of information: Syntactic, Semantic, and Pragmatic.

\subsubsection{The Nature of Pragmatic Information}

In classical information theory, information is defined syntactically: it concerns the statistical properties of symbols and their transmission, independent of meaning or purpose. Semantic information, as formalized in the framework of Niu and Zhang~\cite{Paper_SIT,Book_SIT}, extends this by considering the meaning conveyed by symbols: it concerns the relationship between symbols and their interpretations.

\textbf{Pragmatic information}, in contrast, concerns the \emph{use} or \emph{purpose} of information. It addresses the fundamental question: \emph{``What difference does this information make to the decisions and actions of the receiver?''} Pragmatic information is inherently value-laden: it is defined not by what symbols mean, but by what they \emph{achieve}. Two messages that convey different syntactic symbols and even different semantic meanings may carry the same pragmatic information if they lead the receiver to the same optimal terminal action. Conversely, two messages that convey identical semantic content may carry different pragmatic information if they are interpreted in different contexts or lead to different decisions.

The defining characteristic of pragmatic information is therefore \textbf{Isoteleia}, the principle of \emph{equifinality}: distinct semantic paths converging to the same ultimate end, or \emph{telos}. This principle captures the essence of goal-directed communication: the receiver's objective is not to reconstruct the sender's exact message, but to extract from it the guidance necessary to achieve a desired outcome.

\subsubsection{Foundations: Power Sets, Equivalence Relations, and Fibers}

To rigorously formalize the relationships among syntactic, semantic, and pragmatic information, by using the concepts in \cite{Book_RandomSet}, we introduce the mathematical structures of power sets, equivalence relations, and fibers. Let \(\mathcal{W}\) be the \textbf{syntactic alphabet} (the set of all possible physical symbols or signals). A \textbf{semantic partition} of \(\mathcal{W}\) is a collection of disjoint subsets whose union is \(\mathcal{W}\), where each subset groups symbols that share the same meaning. Equivalently, a semantic partition is induced by an equivalence relation \(\sim_s\) on \(\mathcal{W}\): \(w_1 \sim_s w_2\) iff they convey the same semantic meaning. The quotient set \(\mathcal{W} / \sim_s\) is the \textbf{semantic alphabet}, denoted by \(\tilde{\mathcal{W}}\).

The collection of all subsets of \(\mathcal{W}\) is its \textbf{power set}, denoted \(\mathcal{P}(\mathcal{W})\). The semantic partition \(\Pi_{sem}\) is a subset of \(\mathcal{P}(\mathcal{W})\):
\begin{equation}
    \Pi_{sem} \triangleq \mathcal{W} / \sim_s = \bigl\{ \tilde{w}_1, \tilde{w}_2, \dots, \tilde{w}_K \bigr\} \subset \mathcal{P}(\mathcal{W}).
\end{equation}
Each element \(\tilde{w}_k \in \Pi_{sem}\) is a \textbf{semantic fiber} (or synonymous class): it is a subset of \(\mathcal{W}\) containing all syntactic symbols that carry the same meaning.

Similarly, a \textbf{pragmatic partition} is a collection of disjoint subsets of the semantic alphabet, induced by an equivalence relation \(\sim_p\) on \(\tilde{\mathcal{W}}\): \(\tilde{w}_i \sim_p \tilde{w}_j\) iff they lead to the same optimal terminal action. The quotient set \(\tilde{\mathcal{W}} / \sim_p\) is the \textbf{pragmatic alphabet}, denoted by \(\underline{\mathcal{W}}\). The pragmatic partition \(\Pi_{prag}\) is a subset of the power set of \(\tilde{\mathcal{W}}\):
\begin{equation}\label{eq:pragmatic_partition}
    \Pi_{prag} \triangleq \tilde{\mathcal{W}} / \sim_p = \bigl\{ \underline{w}_1, \underline{w}_2, \dots, \underline{w}_L \bigr\} \subset \mathcal{P}(\tilde{\mathcal{W}}) \subset \mathcal{P}(\mathcal{P}(\mathcal{W})).
\end{equation}
Thus, each element \(\underline{w}_l \in \Pi_{prag}\) is a set of semantic fibers, i.e., a subset of \(\mathcal{P}(\mathcal{W})\), forming a \textbf{second-order power set}. The fibers at each level form a hierarchy: syntactic fibers are individual elements \(w \in \mathcal{W}\); semantic fibers are subsets \(\tilde{w} \subset \mathcal{W}\); and pragmatic fibers are subsets \(\underline{w} \subset \tilde{\mathcal{W}} \subset \mathcal{P}(\mathcal{W})\), i.e., collections of semantic fibers. This hierarchy is the mathematical backbone of pragmatic information theory.

\subsubsection{Formal Definitions of the Mappings}

We now formally define the two fundamental mappings—the synonymous mapping and the isoteleia mapping—along with their composite, the reification mapping.

\begin{definition}[Synonymous Mapping]
Let \(\mathcal{W}\) be the syntactic alphabet and \(\tilde{\mathcal{W}}\) be the semantic alphabet. Let \(\sim_s\) be an equivalence relation on \(\mathcal{W}\) defined by meaning identity. The \textbf{Synonymous Mapping}
\begin{equation}
    f: \tilde{\mathcal{W}} \to 2^{\mathcal{W}}
\end{equation}
is defined for each semantic symbol \(\tilde{w} \in \tilde{\mathcal{W}}\) as:
\begin{equation}
    f(\tilde{w}) \triangleq \left\{ w \in \mathcal{W} \;\middle|\; w \sim_s \tilde{w} \right\} \subset \mathcal{W}.
\end{equation}
Here, \(2^{\mathcal{W}}\) denotes the power set of \(\mathcal{W}\). The mapping is \emph{one-to-many} in the forward direction: a single meaning corresponds to a set of syntactic realizations.
\end{definition}

\begin{definition}[Isoteleia Mapping]\label{definition:Isoteleia-Mapping}
Let \(\tilde{\mathcal{W}}\) be the semantic alphabet and let \(U: \mathcal{W} \times \mathcal{A} \to \mathbb{R}\) be a utility function. For each semantic symbol \(\tilde{w} \in \tilde{\mathcal{W}}\), define its induced optimal action as:
\begin{equation}
    a^*(\tilde{w}) \triangleq \arg\max_{a \in \mathcal{A}} \mathbb{E}_{W \sim P(W|\tilde{w})}\left[ U(W, a) \right].
\end{equation}
Define an equivalence relation \(\sim_p\) on \(\tilde{\mathcal{W}}\) by \( \tilde{w}_i \sim_p \tilde{w}_j \iff a^*(\tilde{w}_i) = a^*(\tilde{w}_j) \). The \textbf{Isoteleia Mapping}
\begin{equation}
    e: \underline{\mathcal{W}} \to 2^{\tilde{\mathcal{W}}}
\end{equation}
is defined for each pragmatic class \(\underline{w} \in \underline{\mathcal{W}} \triangleq \tilde{\mathcal{W}} / \sim_p\) as:
\begin{equation}
    e(\underline{w}) \triangleq \left\{ \tilde{w} \in \tilde{\mathcal{W}} \;\middle|\; a^*(\tilde{w}) = \underline{w} \right\} \subset \tilde{\mathcal{W}}.
\end{equation}
The mapping is \emph{one-to-many} in the forward direction: a single pragmatic purpose corresponds to a set of semantic justifications.
\end{definition}

\begin{definition}[Dynamic Compatible Isoteleia Mapping]\label{definition:Dynamic-Isoteleia-Mapping}
For sequential decision-making problems (Section~\ref{section_XI}), the optimal action must account for 
future consequences. Let $V: \Delta(\mathcal{S}) \to \mathbb{R}$ be the value function over belief states. 
We define the \emph{dynamic optimal action} for a semantic symbol $\tilde{w}$ as:
\begin{equation}
a^*(\tilde{w}, V) \triangleq \arg\max_{a\in \mathcal{A}} 
\mathbb{E}_{S|\tilde{w}}\left[ U(S,a) + \gamma \sum_{s'} P(s'|S,a) V(b_{s'}') \right],
\end{equation}
where $b_{s'}'$ is the updated belief after transitioning to $s'$. 
Consequently, the dynamic isoteleia mapping is:
\begin{equation}
e_V(\underline{w}) \triangleq \left\{ \tilde{w} \in \tilde{\mathcal{W}} \mid a^*(\tilde{w}, V) = \underline{w} \right\}.
\end{equation}
\end{definition}
When $\gamma = 0$ (static decision), we have $e_V(\underline{w}) \equiv e(\underline{w})$, recovering the static Definition \ref{definition:Isoteleia-Mapping}. All dynamic Bellman equations in Section~\ref{subsec:joint-optimization-dynamic} implicitly adopt this conditional extension of the isoteleia mapping.

\begin{definition}[Reification Mapping]
The \textbf{Reification Mapping}
\begin{equation}
    g: \underline{\mathcal{W}} \to 2^{\mathcal{W}}
\end{equation}
is defined as the composition of the isoteleia and synonymous mappings:
\begin{equation}
    g(\underline{w}) \triangleq \bigcup_{\tilde{w} \in e(\underline{w})} f(\tilde{w}) \subset \mathcal{W}.
\end{equation}
This mapping directly associates a pragmatic purpose with the complete set of syntactic symbols that can realize it.
\end{definition}

Figure \ref{Fig_Isoteleia_mapping} illustrates the three-tier quotient mapping structure. The isoteleia mapping \(e: \underline{\mathcal{W}} \to 2^{\tilde{\mathcal{W}}}\) groups semantic classes that lead to the same optimal action, forming pragmatic equivalence classes. The synonymous mapping \(f: \tilde{\mathcal{W}} \to 2^{\mathcal{W}}\) associates each semantic meaning with a set of syntactic realizations. Their composition, the reification mapping \(g = f \circ e: \underline{\mathcal{W}} \to 2^{\mathcal{W}}\), directly links a pragmatic purpose to all possible syntactic signals. The resulting pragmatic classes are evaluated by the utility function \(U\) and mapped to physical actions \(A_t \in \mathcal{A}\), forming the closed loop of pragmatic information processing.

\begin{figure*}[htbp]
\setlength{\abovecaptionskip}{0.cm}
\setlength{\belowcaptionskip}{-0.cm}
  \centering{\includegraphics[scale=0.85]{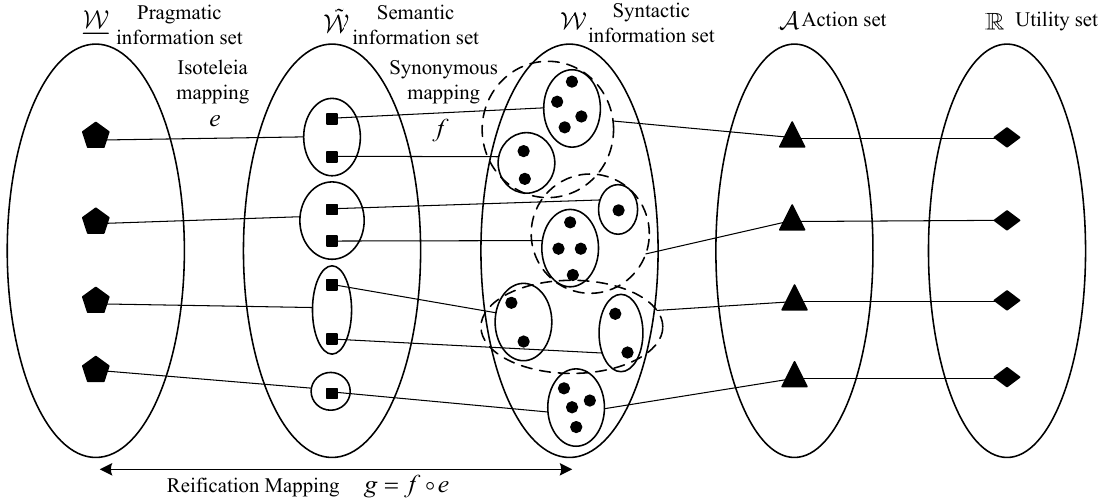}}
  \caption{The hierarchical mapping from pragmatic classes to semantic classes and syntactic symbols, with correspondence to the action space $\mathcal{A}$ and the utility evaluation $U$.}\label{Fig_Isoteleia_mapping}
\end{figure*}

\subsubsection{The Five Relations Between Layers}

The three mappings give rise to five fundamental relations. The isoteleia mapping \(e\) is a pragmatic-to-semantic relation: a single pragmatic class corresponds to multiple semantic classes. Conversely, each semantic class maps to exactly one pragmatic class via the inverse isoteleia, ensuring the partition is well-defined. The synonymous mapping \(f\) relates semantic to syntactic: a single meaning corresponds to multiple syntactic symbols, while each syntactic symbol belongs to exactly one semantic class via the inverse synonymous mapping. Finally, the reification mapping \(g\) directly relates pragmatic purposes to syntactic signals, associating a single purpose with the union of all syntactic symbols in its associated semantic classes.

\subsubsection{Consistency Condition: Semantic Refinement}

A critical condition must be satisfied for the isoteleia mapping to be well-defined: every semantic class \(\tilde{w}\) must be \emph{pragmatically homogeneous}, meaning that for any two syntactic symbols \(w_1, w_2 \in \tilde{w}\), we must have \(a^*(w_1) = a^*(w_2)\). If this fails, the semantic equivalence is too coarse, and we must \emph{refine} the semantic partition by splitting \(\tilde{w}\) into smaller classes, each homogeneous with respect to the optimal action. This operation, termed \textbf{semantic refinement}, is defined as:
\begin{equation}
    \tilde{w} \to \{\tilde{w}^{(1)}, \tilde{w}^{(2)}, \dots\}, \quad \tilde{w}^{(j)} \triangleq \left\{ w \in \tilde{w} \;\middle|\; a^*(w) = a_j \right\}.
\end{equation}
Only after such refinement is the isoteleia mapping uniquely defined.

\subsubsection{Quotient-Space Structure}

The syntactic, semantic, and pragmatic layers are successive quotients of the same original set \(\mathcal{W}\): \(\mathcal{W} \xrightarrow{\sim_s} \tilde{\mathcal{W}} \xrightarrow{\sim_p} \underline{\mathcal{W}}\). Syntactic information is the finest description, retaining all symbol distinctions. Semantic information is the quotient by meaning equivalence, discarding differences that do not affect meaning. Pragmatic information is the quotient by action equivalence, discarding semantic differences that do not affect optimal actions. Thus, the three layers are different levels of abstraction of the same underlying signal, with the pragmatic layer representing the coarsest, most task-relevant description. This quotient-space perspective provides the mathematical foundation for all subsequent information-theoretic analyses.

\section{Pragmatic Entropy}
\label{section_III}

In this section, we begin with pragmatic entropy, the fundamental measure of decision uncertainty, and then extend to joint and conditional pragmatic entropies, establishing their compatibility with syntactic and semantic counterparts.

\subsection{Formal Definition of Pragmatic Entropy}
\label{subsec:pragmatic-entropy}

Pragmatic entropy is the fundamental measure of decision uncertainty in a pragmatic information system. Unlike syntactic entropy, which measures uncertainty about which symbol will be transmitted, or semantic entropy, which measures uncertainty about which meaning is conveyed, pragmatic entropy measures the uncertainty about which terminal action will be optimal, that is, the uncertainty that directly affects the system's ability to achieve its task objectives. This measure is derived from the utility function \(U\) through the Isoteleia Mapping \(e: \underline{\mathcal{W}} \to 2^{\tilde{\mathcal{W}}}\), which groups semantic symbols into pragmatic equivalence classes based on their induced optimal actions.

Let \(\mathcal{W}\) be the syntactic alphabet with probability mass function \(P(W)\), and let \(\tilde{\mathcal{W}}\) be the semantic alphabet induced by the Synonymous Mapping \(f: \tilde{\mathcal{W}} \to 2^{\mathcal{W}}\). Let the Isoteleia Mapping \(e: \underline{\mathcal{W}} \to 2^{\tilde{\mathcal{W}}}\) partition the semantic space according to the utility function \(U: \mathcal{W} \times \mathcal{A} \to \mathbb{R}\), as defined in Eq.~\eqref{eq:utility-def}.

For each pragmatic symbol \(\underline{w} \in \underline{\mathcal{W}}\), its probability is given by:
\begin{equation}
    P(\underline{w}) = \sum_{\tilde{w} \in e(\underline{w})} P(\tilde{w}) = \sum_{\tilde{w} \in e(\underline{w})} \sum_{w \in f(\tilde{w})} P(w).
    \label{eq:pragmatic-probability}
\end{equation}

\begin{definition}[Pragmatic Entropy]
The \textbf{pragmatic entropy} of a pragmatic source \(\underline{\mathcal{W}}\) is defined as:
\begin{equation}
    H_p(\underline{W}) \triangleq -\sum_{\underline{w} \in \underline{\mathcal{W}}} P(\underline{w}) \log P(\underline{w}),
    \label{eq:pragmatic-entropy}
\end{equation}
where \(P(\underline{w})\) is given by Eq.~\eqref{eq:pragmatic-probability}. The unit of pragmatic entropy is the \textbf{pragmatic bit} (\textbf{Prabit}), reflecting that it measures uncertainty about terminal actions rather than about symbols or meanings.
\end{definition}

Equivalently, pragmatic entropy can be expressed directly in terms of the syntactic distribution and the composite Reification Mapping \(g = f \circ e: \underline{\mathcal{W}} \to 2^{\mathcal{W}}\):
\begin{equation}
    H_p(\underline{W}) = -\sum_{\underline{w} \in \underline{\mathcal{W}}} \left( \sum_{w \in g(\underline{w})} P(w) \right) \log \left( \sum_{w \in g(\underline{w})} P(w) \right).
    \label{eq:pragmatic-entropy-reification}
\end{equation}

The pragmatic entropy quantifies the minimum average number of pragmatic bits required to specify the optimal terminal action. It represents the irreducible decision uncertainty that must be resolved through communication and control.

We now establish the fundamental properties of pragmatic entropy. These properties are direct consequences of the quotient-space structure and the deterministic nature of the Isoteleia and Synonymous mappings.

\begin{lemma}[Non-negativity of Pragmatic Entropy]
\label{lem:pragmatic-nonnegativity}
The pragmatic entropy is non-negative:
\begin{equation}
    H_p(\underline{W}) \ge 0.
\end{equation}
Equality holds if and only if the pragmatic alphabet has a single element, i.e., there is only one possible optimal action regardless of the state.
\end{lemma}

\begin{proof}
By definition, \(P(\underline{w}) \in [0, 1]\) for all \(\underline{w} \in \underline{\mathcal{W}}\), and \(\sum_{\underline{w}} P(\underline{w}) = 1\). Thus \(-\log P(\underline{w}) \ge 0\), and the weighted sum is non-negative. Equality holds if and only if \(P(\underline{w}) = 1\) for some \(\underline{w}\) and \(P(\underline{w}') = 0\) for all \(\underline{w}' \neq \underline{w}\), which means the pragmatic alphabet has cardinality 1. \qedhere
\end{proof}

\begin{lemma}[Pragmatic Entropy Hierarchy]
\label{lem:pragmatic-hierarchy}
Let \(H(W)\) denote the syntactic entropy (classical Shannon entropy), \(H_s(\tilde{W})\) denote the semantic entropy as defined in semantic information theory, and \(H_p(\underline{W})\) denote the pragmatic entropy. Then:
\begin{equation}
        H_p(\underline{W}) \le H_s(\tilde{W}) \le H(W).
    \label{eq:entropy-hierarchy}
\end{equation}
\end{lemma}

\begin{proof}
We prove the two inequalities separately.

\textbf{First inequality: \(H_p(\underline{W}) \le H_s(\tilde{W})\).}

Since the Isoteleia Mapping \(e: \underline{\mathcal{W}} \to 2^{\tilde{\mathcal{W}}}\) partitions the semantic alphabet \(\tilde{\mathcal{W}}\) into disjoint equivalence classes, the pragmatic symbol \(\underline{W}\) is a deterministic function of the semantic symbol \(\tilde{W}\). That is, there exists a deterministic function \(q: \tilde{\mathcal{W}} \to \underline{\mathcal{W}}\) such that \(\underline{W} = q(\tilde{W})\). By the data processing inequality for entropy, the entropy of a function of a random variable cannot exceed the entropy of the original variable:
\begin{equation}
    H_p(\underline{W}) = H(q(\tilde{W})) \le H(\tilde{W}) = H_s(\tilde{W}).
\end{equation}
Equality holds if and only if the isoteleia mapping is injective (no two semantic classes are merged into the same pragmatic class).

\textbf{Second inequality: \(H_s(\tilde{W}) \le H(W)\).}
This is the fundamental result of semantic information theory, established in \cite{Paper_SIT,Book_SIT}.

Combining both inequalities yields the desired hierarchy.\qedhere
\end{proof}

This hierarchy captures the essence of pragmatic information theory: each successive layer of abstraction reduces uncertainty by discarding distinctions that are irrelevant to the current level of analysis. The syntactic layer retains all symbol distinctions; the semantic layer discards distinctions that do not affect meaning; and the pragmatic layer discards distinctions that do not affect the optimal terminal action.

We now establish the maximum entropy principle for pragmatic entropy, which characterizes the upper bound of decision uncertainty given the pragmatic alphabet.

\begin{lemma}[Maximum Pragmatic Entropy]
\label{lem:pragmatic-maximum}
Let \(\underline{\mathcal{W}}\) be the pragmatic alphabet with cardinality \(|\underline{\mathcal{W}}| = L\). The pragmatic entropy satisfies:
\begin{equation}
    H_p(\underline{W}) \le \log L.
    \label{eq:pragmatic-maximum}
\end{equation}
Equality holds if and only if the pragmatic symbols are uniformly distributed, i.e., \(P(\underline{w}) = 1/L\) for all \(\underline{w} \in \underline{\mathcal{W}}\).
\end{lemma}

\begin{proof}
Applying the standard entropy maximization argument, we know that the entropy is maximized when the distribution is uniform. Therefore, \(H_p(\underline{W}) \le \log L\), with equality if and only if \(P(\underline{w}) = 1/L\) for all \(\underline{w} \in \underline{\mathcal{W}}\). \qedhere
\end{proof}

\begin{corollary}[Hierarchy of Maximum Entropies]
\label{cor:pragmatic-max-hierarchy}
Let \(L_p = |\underline{\mathcal{W}}|\), \(L_s = |\tilde{\mathcal{W}}|\), and \(L = |\mathcal{W}|\) denote the cardinalities of the pragmatic, semantic, and syntactic alphabets, respectively. Since the Isoteleia and Synonymous mappings are quotient mappings, we have \(L_p \le L_s \le L\). Consequently:
\begin{equation}
    \log L_p \le \log L_s \le \log L.
\end{equation}
Combining with Lemma~\ref{lem:pragmatic-hierarchy}, we obtain the full hierarchy of maximum entropies:
\begin{equation}
    H_p(\underline{W}) \le \log L_p \le \log L_s \le \log L.
    \label{eq:max-entropy-hierarchy}
\end{equation}
\end{corollary}

\begin{remark}
The maximum pragmatic entropy \(\log L_p\) is strictly less than the maximum semantic entropy \(\log L_s\) whenever the Isoteleia Mapping is non-injective (i.e., when multiple semantic classes merge into the same pragmatic class). This reflects the fundamental principle that the goal-directed compression of semantic distinctions reduces the theoretical upper bound of decision uncertainty. In the extreme case where all semantic classes map to a single pragmatic class (i.e., the same action is optimal for all states), we have \(L_p = 1\) and \(H_p(\underline{W}) = 0\), indicating that there is no decision uncertainty whatsoever.
\end{remark}

The maximum pragmatic entropy lemma provides a theoretical limit on the amount of decision uncertainty that can exist in a pragmatic information system. It also serves as a design guideline: to maximize the information content of pragmatic communication, the system should be designed such that the pragmatic symbols are as uniformly distributed as possible, within the constraints imposed by the utility function and the environment.

\begin{example}[Autonomous Vehicle Decision-Making]
\leavevmode\newline\indent
We now present a concrete numerical example to illustrate the calculation and interpretation of pragmatic entropy, comparing it with syntactic and semantic entropies. Consider an autonomous vehicle approaching an intersection. The vehicle's perception system observes three possible traffic light states, which serve as the syntactic alphabet \(\mathcal{W} = \{w_1, w_2, w_3\}\) with probabilities \(P(W) = \{0.6, 0.3, 0.1\}\), representing green, yellow, and red lights, respectively.
\end{example}

The vehicle's semantic interpretation groups traffic lights into two meaning classes: \(\tilde{w}_1 =\) ``Proceed with caution'' (green and yellow) with \(P(\tilde{w}_1) = 0.9\), and \(\tilde{w}_2 =\) ``Stop'' (red) with \(P(\tilde{w}_2) = 0.1\). The decision utility \(U(w, a)\) for actions \(a_1 = \text{Go}\) and \(a_2 = \text{Stop}\) is given in Table~\ref{tab:utility-example}.

\begin{table}[htbp]
\centering
\caption{Utility function \(U(w, a)\) for autonomous vehicle decision example.}
\label{tab:utility-example}
\begin{tabular}{c|c|c}
\hline
State \(w\) & Go (\(a_1\)) & Stop (\(a_2\)) \\ \hline
\(w_1\) (green) & +10 & -5 \\
\(w_2\) (yellow) & -2 & +1 \\
\(w_3\) (red) & -100 & +5 \\
\hline
\end{tabular}
\end{table}

For each semantic class, the expected utility is \(\mathbb{E}[U(\text{Go})|\tilde{w}_1] = (0.6/0.9)\cdot 10 + (0.3/0.9)\cdot (-2) = 6\), \(\mathbb{E}[U(\text{Stop})|\tilde{w}_1] = (0.6/0.9)\cdot (-5) + (0.3/0.9)\cdot (+1) = -3\), \(\mathbb{E}[U(\text{Go})|\tilde{w}_2] = (0.1/0.1)\cdot (-100) = -100\) and \(\mathbb{E}[U(\text{Stop})|\tilde{w}_2] = (0.1/0.1)\cdot (+5) = +5\). Thus, the optimal action for \(\tilde{w}_1\) is Go (\(6 > -3\)), while for \(\tilde{w}_2\) it is Stop (\(-100 < +5\)). The Isoteleia Mapping induces two pragmatic classes: \(\underline{w}_1 =\) ``Go'' containing \(\tilde{w}_1\) with \(P(\underline{w}_1) = 0.9\), and \(\underline{w}_2 =\) ``Stop'' containing \(\tilde{w}_2\) with \(P(\underline{w}_2) = 0.1\).

The syntactic entropy is \(H(W) = -[0.6\log 0.6 + 0.3\log 0.3 + 0.1\log 0.1] \approx 1.295\) bits. The semantic entropy is \(H_s(\tilde{W}) = -[0.9\log 0.9 + 0.1\log 0.1] \approx 0.469\) sebits. The pragmatic entropy is \(H_p(\underline{W}) = -[0.9\log 0.9 + 0.1\log 0.1] \approx 0.469\) prabits. In this example, \(H_p = H_s\) because each semantic class maps to a distinct pragmatic class (the Isoteleia Mapping is injective), yet the hierarchy \(H_p = 0.469 < 1.295 = H(W)\) still holds.

The syntactic entropy of 1.295 bits represents the uncertainty about the specific traffic light color observed. The semantic entropy of 0.469 sebits captures the uncertainty about whether the vehicle should ``Proceed with caution'' or ``Stop'' based on the meaning of the observation. The pragmatic entropy of 0.469 prabits reflects the uncertainty about the optimal terminal action, which ultimately determines the vehicle's behavior. The reduction from syntactic to pragmatic entropy (from 1.295 to 0.469 bits) demonstrates that the vehicle does not need to distinguish between green and yellow lights to make the correct decision; it only needs to distinguish between ``Go'' and ``Stop'' situations.

This example demonstrates how the pragmatic entropy measure provides a principled way to quantify the minimum decision uncertainty that must be resolved through communication and control, thereby establishing the theoretical foundation for resource allocation and system design in pragmatic information systems.

\subsection{Pragmatic Joint Entropy and Conditional Entropy}
\label{subsec:pragmatic-joint-entropy}

Building upon the definition of pragmatic entropy for a single variable, we now extend the framework to multiple semantic/pragmatic variables. This extension is essential for analyzing the dependencies and information flows between different components of a pragmatic information system, such as the relationship between the transmitted pragmatic intention and the received pragmatic interpretation. 

To formalize the joint behavior of two pragmatic sources, we first extend the Isoteleia Mapping to the joint and conditional cases. Let \(\underline{\mathcal{W}}\) and \(\underline{\mathcal{V}}\) be two pragmatic alphabets, with associated semantic alphabets \(\tilde{\mathcal{W}}\) and \(\tilde{\mathcal{V}}\), respectively. Let the corresponding syntactic alphabets be \(\mathcal{W}\) and \(\mathcal{V}\), with a joint distribution \(P(W,V)\).

\begin{definition}[Joint Isoteleia Mapping]
The \textbf{joint Isoteleia Mapping}
\begin{equation}
    e_{wv}: \underline{\mathcal{W}} \times \underline{\mathcal{V}} \longrightarrow 2^{\tilde{\mathcal{W}} \times \tilde{\mathcal{V}}}
\end{equation}
is defined by its preimage: for each pair of pragmatic symbols \((\underline{w}, \underline{v}) \in \underline{\mathcal{W}} \times \underline{\mathcal{V}}\), the image is the set of all semantic pairs \((\tilde{w}, \tilde{v})\) that jointly induce the same optimal terminal actions:
\begin{equation}
    e_{wv}(\underline{w}, \underline{v}) \triangleq \left\{ (\tilde{w}, \tilde{v}) \in \tilde{\mathcal{W}} \times \tilde{\mathcal{V}} \;\middle|\; 
    \arg\max_{a \in \mathcal{A}} \mathbb{E}[U(W,a) | \tilde{w}] = \underline{w},\ 
    \arg\max_{b \in \mathcal{B}} \mathbb{E}[U(V,b) | \tilde{v}] = \underline{v}
    \right\}.
\end{equation}
This mapping partitions the joint semantic space \(\tilde{\mathcal{W}} \times \tilde{\mathcal{V}}\) into equivalence classes according to the pair of optimal actions.
\end{definition}

\begin{definition}[Conditional Isoteleia Mapping]
Given a specific pragmatic symbol \(\underline{w} \in \underline{\mathcal{W}}\), the \textbf{conditional Isoteleia Mapping}
\begin{equation}
    e_{v|w}(\underline{v} | \underline{w}): \underline{\mathcal{V}} \longrightarrow 2^{\tilde{\mathcal{V}}}
\end{equation}
is defined as the set of semantic symbols \(\tilde{v}\) that, together with the given \(\underline{w}\), lead to the optimal action \(\underline{v}\):
\begin{equation}
    e_{v|w}(\underline{v} | \underline{w}) \triangleq \left\{ \tilde{v} \in \tilde{\mathcal{V}} \;\middle|\; 
    \arg\max_{b \in \mathcal{B}} \mathbb{E}[U(V,b) | \tilde{v}, \underline{w}] = \underline{v}
    \right\}.
\end{equation}
This mapping captures the pragmatic equivalence of the second variable conditioned on the first pragmatic class.
\end{definition}

The joint and conditional synonymous mappings have been established in the semantic information theory literature~\cite{Paper_SIT,Book_SIT}. The joint synonymous mapping partitions the syntactic product space according to meaning identity, while the conditional synonymous mapping partitions the syntactic space of one variable given the semantic class of the other. We will leverage these existing definitions in the subsequent entropy formulations.

With the joint and conditional isoteleia mappings defined, we can now introduce the corresponding entropy measures.

\begin{definition}[Pragmatic Joint Entropy]
Let \((\underline{W}, \underline{V})\) be a pair of pragmatic random variables derived from the syntactic pair \((W,V)\) via the joint isoteleia mapping \(e_{wv}\). The \textbf{pragmatic joint entropy} is defined as the Shannon entropy of the joint pragmatic distribution:
\begin{equation}
    H_p(\underline{W}, \underline{V}) \triangleq -\sum_{\underline{w} \in \underline{\mathcal{W}}} \sum_{\underline{v} \in \underline{\mathcal{V}}} P(\underline{w}, \underline{v}) \log P(\underline{w}, \underline{v}),
\end{equation}
where the joint pragmatic probabilities are obtained by aggregating over all semantic and syntactic fibers:
\begin{equation}
    P(\underline{w}, \underline{v}) = \sum_{(\tilde{w}, \tilde{v}) \in e_{uv}(\underline{w}, \underline{v})} P(\tilde{w}, \tilde{v}) = \sum_{(\tilde{w}, \tilde{v}) \in e_{uv}(\underline{w}, \underline{v})} \sum_{(w,v) \in f_{uv}(\tilde{w}, \tilde{v})} P(w,v),
\end{equation}
with \(f_{uv}\) denoting the joint Synonymous Mapping.
\end{definition}

\begin{definition}[Pragmatic Conditional Entropy]
There are two natural variants of pragmatic conditional entropy, depending on whether the conditioning variable is syntactic or pragmatic.

\textbf{(a) Pragmatic conditional entropy of \(\underline{V}\) given the syntactic variable \(W\):}
\begin{equation}
    H_p(\underline{V} | W) \triangleq -\sum_{w \in \mathcal{W}} \sum_{\underline{v} \in \underline{\mathcal{V}}} P(w) P(\underline{v} | w) \log P(\underline{v} | w),
\end{equation}
where \(P(\underline{v} | w) = \sum_{\tilde{v} \in e(\underline{v})} P(\tilde{v} | w)\) is the pragmatic distribution of \(\underline{V}\) conditioned on a specific syntactic observation \(w\).

\textbf{(b) Pragmatic conditional entropy of \(\underline{V}\) given the pragmatic variable \(\underline{W}\):}
\begin{equation}
    H_p(\underline{V} | \underline{W}) \triangleq -\sum_{\underline{w} \in \underline{\mathcal{W}}} \sum_{\underline{v} \in \underline{\mathcal{V}}} P(\underline{w}, \underline{v}) \log P(\underline{v} | \underline{w}),
\end{equation}
where \(P(\underline{v} | \underline{w}) = P(\underline{w}, \underline{v}) / P(\underline{w})\).
\end{definition}

Similarly, one can define \(H_p(\underline{W} | V)\) and \(H_p(\underline{W} | \underline{V})\) by symmetry. The choice between conditioning on syntactic or pragmatic variables depends on the application: \(H_p(\underline{V} | W)\) measures the residual decision uncertainty after observing the exact syntactic signal, while \(H_p(\underline{V} | \underline{W})\) measures the uncertainty of one decision given the other decision.

Analogous to semantic entropy, pragmatic entropy does not satisfy the exact additive chain rule of classical entropy; instead, it satisfies a series of inequalities that reflect the coarsening of information. We first establish the chain rule for a pair of pragmatic variables.

\begin{theorem}[Chain Rule for Binary Pragmatic Entropy]
For any two pragmatic variables \(\underline{W}\) and \(\underline{V}\) derived from syntactic variables \(W\) and \(V\) via the isoteleia and synonymous mappings, the following chain of inequalities holds:
\begin{equation}
H_{p}(\underline{W}) + H_{p}(\underline{V}|W) \le H_{p}(\underline{W},\underline{V}) \le H(W) + H_{p}(\underline{V}|W),
\label{eq:binary-pragmatic-chain}
\end{equation}
and symmetrically,
\begin{equation}
H_{p}(\underline{V}) + H_{p}(\underline{W}|V) \le H_{p}(\underline{W},\underline{V}) \le H(V) + H_{p}(\underline{W}|V).
\label{eq:binary-pragmatic-chain2}
\end{equation}
\end{theorem}

\begin{proof}
Let \(g_W: \underline{\mathcal{W}} \to 2^{\mathcal{W}}\) and \(g_V: \underline{\mathcal{V}} \to 2^{\mathcal{V}}\) be the reification mappings, so that \(\underline{W} = g_W^{-1}(W)\) and \(\underline{V} = g_V^{-1}(V)\) are deterministic functions of \(W\) and \(V\), respectively.

For any \(\underline{w} \in \underline{\mathcal{W}}\), define its preimage set
\[
\mathcal{W}_{\underline{w}} \triangleq \{ w \in \mathcal{W} : g_W(w) = \underline{w} \},
\]
and similarly \(\mathcal{V}_{\underline{v}} \triangleq \{ v \in \mathcal{V} : g_V(v) = \underline{v} \}\). The aggregated probabilities are
\[
p(\underline{w}) = \sum_{w \in \mathcal{W}_{\underline{w}}} p(w), \qquad
p(\underline{w}, \underline{v}) = \sum_{w \in \mathcal{W}_{\underline{w}}} \sum_{v \in \mathcal{V}_{\underline{v}}} p(w, v).
\]

For a fixed \(\underline{w}\), the conditional distribution of \(\underline{V}\) given \(\underline{W} = \underline{w}\) is
\begin{equation}
p(\underline{v} \mid \underline{w}) = \frac{p(\underline{w}, \underline{v})}{p(\underline{w})}
= \sum_{w \in \mathcal{W}_{\underline{w}}} \frac{p(w)}{p(\underline{w})} \cdot p(\underline{v} \mid w),
\label{eq:mix1}
\end{equation}
where \(p(\underline{v} \mid w) \triangleq \sum_{v \in \mathcal{V}_{\underline{v}}} p(v \mid w)\). Thus, for each \(\underline{w}\), the distribution \(p(\cdot \mid \underline{w})\) is a mixture of the distributions \(\{ p(\cdot \mid w) : w \in \mathcal{W}_{\underline{w}} \}\) with weights \(\frac{p(w)}{p(\underline{w})}\).

Since the Shannon entropy is concave,
\begin{equation}
H\big(p(\cdot \mid \underline{w})\big) \ge \sum_{w \in \mathcal{W}_{\underline{w}}} \frac{p(w)}{p(\underline{w})} H\big(p(\cdot \mid w)\big).
\label{eq:concavity}
\end{equation}
Multiplying by \(p(\underline{w})\) and summing over \(\underline{w}\):
\begin{equation}
\sum_{\underline{w}} p(\underline{w}) H\big(p(\cdot \mid \underline{w})\big)
\ge
\sum_{\underline{w}} \sum_{w \in \mathcal{W}_{\underline{w}}} p(w) H\big(p(\cdot \mid w)\big).
\label{eq:sum_mix}
\end{equation}
The left-hand side is \(H_p(\underline{V} \mid \underline{W})\), and the right-hand side is \(H_p(\underline{V} \mid W)\). Hence,
\begin{equation}
H_p(\underline{V} \mid \underline{W}) \ge H_p(\underline{V} \mid W).
\label{eq:cond_ineq}
\end{equation}

Since \(\underline{W}\) and \(\underline{V}\) are ordinary random variables, the classical chain rule applies:
\begin{equation}
H_p(\underline{W}, \underline{V}) = H_p(\underline{W}) + H_p(\underline{V} \mid \underline{W}).
\label{eq:classical_chain}
\end{equation}
Substituting \eqref{eq:cond_ineq} into \eqref{eq:classical_chain} yields the left inequality of \eqref{eq:binary-pragmatic-chain}:
\begin{equation}
H_p(\underline{W}, \underline{V}) \ge H_p(\underline{W}) + H_p(\underline{V} \mid W).
\label{eq:left_ineq}
\end{equation}

For the right inequality of \eqref{eq:binary-pragmatic-chain}, since \(\underline{W}\) is a function of \(W\),
\begin{equation}
H_p(\underline{W}, \underline{V}) = H(\underline{W}, \underline{V}) \le H(W, \underline{V}) = H(W) + H(\underline{V} \mid W) = H(W) + H_p(\underline{V} \mid W),
\label{eq:right_ineq}
\end{equation}
which proves the right inequality.

The symmetric chain \eqref{eq:binary-pragmatic-chain2} follows by interchanging \(W\) and \(V\), noting that \(\underline{V}\) is a function of \(V\), and applying the same mixture-entropy argument to \(p(\cdot \mid \underline{v})\). This completes the proof.
\end{proof}

The chain rule (Eq.~\ref{eq:binary-pragmatic-chain}) reveals that pragmatic joint entropy lies between two bounds: the lower bound is the sum of the marginal pragmatic entropy of \(\underline{W}\) and the pragmatic uncertainty of \(\underline{V}\) given the syntactic observation \(W\); the upper bound is the sum of the syntactic entropy of \(W\) and the same pragmatic conditional term. This intermediate position reflects the fact that pragmatic abstraction reduces uncertainty, but not as much as semantic abstraction.

\begin{theorem}[Hierarchy of Binary Joint Entropies]
For any two syntactic variables \(W\) and \(V\) with their associated semantic variables \(\tilde{W}, \tilde{V}\) and pragmatic variables \(\underline{W}, \underline{V}\), the following inequality holds:
\begin{equation}
     H_p(\underline{W}, \underline{V}) \;\le\; H_s(\tilde{W}, \tilde{V}) \;\le\; H(W, V).
    \label{eq:binary-joint-hierarchy}
\end{equation}
Moreover, equality in the first inequality holds if and only if the isoteleia mapping is injective on the joint semantic space (i.e., no two distinct semantic pairs are merged into the same pragmatic pair). Equality in the second holds if and only if the joint synonymous mapping is injective (no two distinct syntactic pairs share the same semantic meaning).
\end{theorem}

\begin{proof}
We prove the two inequalities separately.

\textbf{First inequality: \(H_p(\underline{W}, \underline{V}) \le H_s(\tilde{W}, \tilde{V})\)}.

By definition, the pragmatic variables are deterministic functions of the semantic variables via the joint Isoteleia mapping:
\begin{equation}
    (\underline{W}, \underline{V}) = e_{uv}^{-1}(\tilde{W}, \tilde{V}),
\end{equation}
where \(e_{uv}^{-1}\) is the inverse of the isoteleia mapping that collapses each semantic equivalence class to a single pragmatic symbol. Since \((\underline{W}, \underline{V})\) is a deterministic function of \((\tilde{W}, \tilde{V})\), the joint entropy of a function of a random variable cannot exceed the joint entropy of the original variable. Applying the data processing inequality for entropy:
\begin{equation}
    H_p(\underline{W}, \underline{V}) = H(e_{uv}^{-1}(\tilde{W}, \tilde{V})) \le H(\tilde{W}, \tilde{V}) = H_s(\tilde{W}, \tilde{V}).
\end{equation}
Equality holds if and only if the mapping is injective, i.e., no two distinct semantic pairs are merged.

\textbf{Second inequality: \(H_s(\tilde{W}, \tilde{V}) \le H(W, V)\)}.

Similarly, the semantic variables are deterministic functions of the syntactic variables via the joint synonymous mapping:
\begin{equation}
    (\tilde{W}, \tilde{V}) = f_{uv}^{-1}(W, V),
\end{equation}
where \(f_{uv}^{-1}\) collapses each syntactic equivalence class to a single semantic symbol. Hence,
\begin{equation}
    H_s(\tilde{W}, \tilde{V}) = H(f_{uv}^{-1}(W, V)) \le H(W, V).
\end{equation}
Equality holds if and only if the joint synonymous mapping is injective.

Combining both inequalities yields the desired hierarchy.\qedhere
\end{proof}

We now extend this result to sequences of pragmatic variables. To do so, we first define the sequential extension of the joint reification mapping.

\begin{definition}[Sequential Joint Reification Mapping]
Let \(g_{WV}: \underline{\mathcal{W}} \times \underline{\mathcal{V}} \to 2^{\mathcal{W} \times \mathcal{V}}\) be the joint Reification Mapping that partitions the joint syntactic space \(\mathcal{W} \times \mathcal{V}\) into pragmatic equivalence classes. Its \(n\)-th extension
\begin{equation}
    g_{WV}^n: \underline{\mathcal{W}}^n \times \underline{\mathcal{V}}^n \longrightarrow 2^{\mathcal{W}^n \times \mathcal{V}^n}
\end{equation}
is defined by
\begin{equation}
    g_{WV}^n(\underline{w}^n, \underline{v}^n) \triangleq \prod_{i=1}^{n} g_{WV}(\underline{w}_i, \underline{v}_i),
\end{equation}
where \(\underline{w}^n = (\underline{w}_1, \dots, \underline{w}_n)\) and \(\underline{v}^n = (\underline{v}_1, \dots, \underline{v}_n)\) are pragmatic sequences, and the product denotes the Cartesian product of sets. In other words, a pragmatic sequence pair \((\underline{w}^n, \underline{v}^n)\) corresponds to the set of all syntactic sequence pairs \((w^n, v^n)\) such that for every position \(i\), \((w_i, v_i) \in g_{WV}(\underline{w}_i, \underline{v}_i)\).

The marginal sequential pragmatic variables \(\underline{W}^n\) and \(\underline{V}^n\) are then obtained by projecting from the joint pragmatic sequence. The probability of a pragmatic sequence pair is given by
\begin{equation}
    P(\underline{w}^n, \underline{v}^n) = \sum_{(w^n, v^n) \in g_{WV}^n(\underline{w}^n, \underline{v}^n)} P(w^n, v^n).
\end{equation}
\end{definition}

\begin{theorem}[Chain Rule for Sequential Pragmatic Entropy]\label{theorem:Chain-Rule-Sequential-Pragmatic-Entropy}
Let \((\underline{W}_1, \underline{W}_2, \dots, \underline{W}_n)\) be a sequence of pragmatic variables derived from the syntactic sequence \((W_1, W_2, \dots, W_n)\) via the sequential joint Reification Mapping. Define the sequential partial entropy
\begin{equation}\label{eq:sequential-partial-entropy}
    \tilde{H}_p(\underline{W}_1^m, W_{m+1}^n) \triangleq \sum_{k=1}^{m} H_p(\underline{W}_k | \underline{W}_1^{k-1}) + \sum_{k=m+1}^{n} H_p(W_k | \underline{W}_1^m, W_{m+1}^{k-1}),
\end{equation}
where the notation \(\underline{W}_1^m\) denotes the block \((\underline{W}_1, \dots, \underline{W}_m)\), and \(W_{m+1}^n\) denotes the block \\ \((W_{m+1}, \dots, W_n)\). Then the following chain of inequalities holds:
\begin{equation}
      \sum_{k=1}^{n} H_p(\underline{W}_k | W_1^{k-1}) \;\le\; H_p(\underline{W}_1^n) \;\le\; \tilde{H}_p(\underline{W}_1^{n-1}, W_n) \;\le\; \cdots \;\le\; \tilde{H}_p(\underline{W}_1, W_2^n) \;\le\; H_p(W_1^n),
\end{equation}
where the last term \(H_p(W_1^n)\) denotes the syntactic entropy of the full sequence.
\end{theorem}

\begin{proof}
The proof follows by iteratively applying the binary chain rule to each adjacent pair in the sequence. Starting from the full pragmatic joint entropy \(H_p(\underline{W}_1^n)\), we use the binary chain rule to bound it in terms of the entropy of the first \(n-1\) variables and the last variable, either conditioned on the syntactic or pragmatic prefix. By gradually replacing pragmatic conditioning with syntactic conditioning, we obtain the chain of inequalities. The details are analogous to the semantic sequential chain rule established in~\cite{Paper_SIT,Book_SIT}, with the substitution of pragmatic entropies.
\end{proof}

\begin{theorem}[Hierarchy of Sequential Joint Entropies]
Let \(W^n = (W_1, \dots, W_n)\) be a syntactic sequence, with associated semantic sequence \(\tilde{W}^n\) and pragmatic sequence \(\underline{W}^n\) obtained by applying the Synonymous and Isoteleia mappings element-wise (or, equivalently, by the sequential joint Reification Mapping). Then:
\begin{equation}
    H_p(\underline{W}^n) \;\le\; H_s(\tilde{W}^n) \;\le\; H(W^n),
    \label{eq:sequential-joint-hierarchy}
\end{equation}
where \(H_p(\underline{W}^n)\), \(H_s(\tilde{W}^n)\), and \(H(W^n)\) denote the joint entropies of the entire sequences.
\end{theorem}

\begin{proof}
The proof follows directly from the element-wise deterministic mappings. Since the sequential joint reification mapping is defined as the product of the individual mappings, the entire pragmatic sequence \(\underline{W}^n\) is a deterministic function of the semantic sequence \(\tilde{W}^n\), and \(\tilde{W}^n\) is a deterministic function of the syntactic sequence \(W^n\). Therefore:
\begin{equation}
    H_p(\underline{W}^n) = H(e^{-n}(\tilde{W}^n)) \le H(\tilde{W}^n) = H_s(\tilde{W}^n),
\end{equation}
and
\begin{equation}
    H_s(\tilde{W}^n) = H(f^{-n}(W^n)) \le H(W^n).
\end{equation}
Combining these yields the desired hierarchy.
\end{proof}

\begin{corollary}[Hierarchy of Sequential Conditional Entropies]
For any \(m\) with \(1 \le m \le n\), the sequential partial entropies also satisfy the same ordering:
\begin{equation}
    \tilde{H}_p(\underline{W}_1^m, W_{m+1}^n) \;\le\; \tilde{H}_s(\tilde{W}_1^m, W_{m+1}^n) \;\le\; H(W^n),
\end{equation}
where the partial entropies are defined as in the chain rule (Eq.~(\ref{eq:sequential-partial-entropy})). This follows because the partial entropies are linear combinations of conditional entropies, each of which obeys the single-variable hierarchy.
\end{corollary}

\begin{remark}
The hierarchy theorems reveal a layered structure of uncertainty: syntactic entropy captures all symbol-level distinctions; semantic entropy removes meaningless syntactic variations; and pragmatic entropy further removes decision-irrelevant semantic distinctions. Each successive coarsening reduces the joint entropy, providing a principled foundation for multi-level source coding and compression in pragmatic communication systems.
\end{remark}

\begin{example}[Multi-Vehicle Autonomous Driving Scenario]
\label{example:multi-vehicle-driving}
\leavevmode\newline\indent
We now present a concrete numerical example to illustrate the computation and interpretation of pragmatic joint and conditional entropies, with particular emphasis on distinguishing pragmatic entropy from semantic entropy. Consider two autonomous vehicles, Vehicle A and Vehicle B, approaching an intersection. Each vehicle must decide whether to ``Go'' or ``Stop'' based on its observed traffic light state and, potentially, the state of the other vehicle.
\end{example}

Let the syntactic state of Vehicle A be \(W \in \{w_1, w_2, w_3, w_4\}\) representing green, yellow, red-short, and red-long, respectively; Vehicle B's state \(V\) has the same interpretation. The joint probability distribution \(P(W,V)\) is given in Table~\ref{tab:joint-prob}, with marginals \(P(w_1)=0.23,\;P(w_2)=0.22,\;P(w_3)=0.25,\;P(w_4)=0.30\), and similarly for \(V\).

\begin{table}[htbp]
\centering
\caption{Joint probability distribution \(P(W,V)\) for two vehicles.}
\label{tab:joint-prob}
\begin{tabular}{c|cccc}
\hline
\(P(W,V)\) & \(v_1\) (green) & \(v_2\) (yellow) & \(v_3\) (red-short) & \(v_4\) (red-long) \\ \hline
\(w_1\) (green)       & 0.12 & 0.06 & 0.03 & 0.02 \\
\(w_2\) (yellow)      & 0.06 & 0.10 & 0.04 & 0.02 \\
\(w_3\) (red-short)   & 0.03 & 0.04 & 0.12 & 0.06 \\
\(w_4\) (red-long)    & 0.02 & 0.02 & 0.06 & 0.20 \\ \hline
\end{tabular}
\end{table}

The synonymous mapping groups the syntactic states into three semantic classes: \(\tilde{w}_1 = \tilde{v}_1 =\) ``Proceed with caution'' (green and yellow) with probability 0.45, \(\tilde{w}_2 = \tilde{v}_2 =\) ``Stop-short'' (red-short) with probability 0.25, and \(\tilde{w}_3 = \tilde{v}_3 =\) ``Stop-long'' (red-long) with probability 0.30. Under a safety-first utility, the optimal action for \(\tilde{w}_1\) is ``Go'', while for both \(\tilde{w}_2\) and \(\tilde{w}_3\) it is ``Stop''. The isoteleia mapping therefore merges the two stop classes into a single pragmatic class: \(\underline{w}_1 =\) ``Go'' containing \(\tilde{w}_1\) with probability 0.45, and \(\underline{w}_2 =\) ``Stop'' containing \(\tilde{w}_2\) and \(\tilde{w}_3\) with probability 0.55. The same mapping applies to Vehicle B.

Aggregating the syntactic probabilities yields the semantic joint distribution \(P(\tilde{W},\tilde{V})\) (Table~\ref{tab:semantic-joint-dist}) and the pragmatic joint distribution \(P(\underline{W},\underline{V})\) (Table~\ref{tab:pragmatic-joint-dist}). The syntactic joint entropy is \(H(W,V) \approx 3.6314\) bits; the semantic joint entropy is \(H_s(\tilde{W},\tilde{V}) \approx 2.7563\) sebits; and the pragmatic joint entropy is \(H_p(\underline{W},\underline{V}) \approx 1.7509\) prabits, confirming the hierarchy \(H_p \le H_s \le H\). Marginal entropies are \(H_s(\tilde{W}) \approx 1.5395\) bits and \(H_p(\underline{W}) \approx 0.9928\) prabits.

\begin{table}[htbp]
\centering
\caption{Joint semantic distribution \(P(\tilde{W},\tilde{V})\).}
\label{tab:semantic-joint-dist}
\begin{tabular}{c|ccc}
\hline
\(P(\tilde{W},\tilde{V})\) & \(\tilde{v}_1\) (Proceed) & \(\tilde{v}_2\) (Stop-short) & \(\tilde{v}_3\) (Stop-long) \\ \hline
\(\tilde{w}_1\) (Proceed) & 0.34 & 0.07 & 0.04 \\
\(\tilde{w}_2\) (Stop-short) & 0.07 & 0.12 & 0.06 \\
\(\tilde{w}_3\) (Stop-long)  & 0.04 & 0.06 & 0.20 \\ \hline
\end{tabular}
\end{table}

\begin{table}[htbp]
\centering
\caption{Joint pragmatic distribution \(P(\underline{W},\underline{V})\).}
\label{tab:pragmatic-joint-dist}
\begin{tabular}{c|cc}
\hline
\(P(\underline{W},\underline{V})\) & \(\underline{v}_1\) (Go) & \(\underline{v}_2\) (Stop) \\ \hline
\(\underline{w}_1\) (Go)   & 0.34 & 0.11 \\
\(\underline{w}_2\) (Stop) & 0.11 & 0.44 \\ \hline
\end{tabular}
\end{table}

The pragmatic conditional entropy is \(H_p(\underline{V}|\underline{W}) = H_p(\underline{W},\underline{V}) - H_p(\underline{W}) = 0.7581\) prabits, while \(H_p(\underline{V}|W) \approx 0.7435\) prabits. The binary chain rule is verified: \(H_p(\underline{W}) + H_p(\underline{V}|W) = 1.7363 \le H_p(\underline{W},\underline{V}) = 1.7509\), and \(H_p(\underline{W},\underline{V}) = 1.7509 \le H(W) + H_p(\underline{V}|W) = 1.9893 + 0.7435 = 2.7328\). These results (summarized in Table~\ref{tab:entropy-comparison}) demonstrate that pragmatic abstraction, by merging semantically distinct but action-equivalent states, reduces joint uncertainty beyond semantic abstraction.

\begin{table}[htbp]
\centering
\caption{Comparison of joint and conditional entropies.}
\label{tab:entropy-comparison}
\begin{tabular}{c|c}
\hline
\textbf{Entropy Measure} & \textbf{Value (bits/sebits/prabits)} \\ \hline
Syntactic Joint Entropy \(H(W,V)\) & 3.6314 \\
Semantic Joint Entropy \(H_s(\tilde{W},\tilde{V})\) & 2.7563 \\
Pragmatic Joint Entropy \(H_p(\underline{W},\underline{V})\) & 1.7509 \\
Pragmatic Marginal Entropy \(H_p(\underline{W})\) & 0.9928 \\
Pragmatic Conditional Entropy \(H_p(\underline{V}|\underline{W})\) & 0.7581 \\
Pragmatic Conditional Entropy \(H_p(\underline{V}|W)\) & 0.7435 \\ \hline
\end{tabular}
\end{table}

This example illustrates the key insight: \textbf{pragmatic entropy is strictly less than semantic entropy when the isoteleia mapping merges distinct semantic classes that share the same optimal action}. The two types of ``stop'' (short red and long red) are semantically distinct but pragmatically equivalent, resulting in a joint pragmatic entropy (1.7509 prabits) that is significantly lower than the joint semantic entropy (2.7563 sebits). This reduction quantifies the decision-irrelevant information that can be discarded in a pragmatic communication system.

\section{Pragmatic Relative Entropy and Pragmatic Mutual Information}
\label{section_IV}

We now introduce the information-theoretic analogues of relative entropy and mutual information at the pragmatic level. These measures quantify the divergence between probability distributions and the statistical dependence between variables, respectively, when the comparison and dependency are evaluated at the level of terminal actions rather than symbols or meanings. 

\subsection{Pragmatic Relative Entropy}
\label{subsec:pragmatic-relative-entropy}

In classical information theory, the relative entropy (Kullback-Leibler divergence) \cite{Book_Cover} between two probability mass functions \(p(w)\) and \(q(w)\) over a syntactic alphabet \(\mathcal{W}\) is defined as \(D(p\|q) = \sum_{w} p(w) \log \frac{p(w)}{q(w)}\). Semantic information theory \cite{Paper_SIT,Book_SIT} extends this by introducing semantic relative entropies that aggregate probabilities over synonymous sets, thereby measuring divergence at the meaning level. In pragmatic information theory, we further aggregate over pragmatic equivalence classes induced by the isoteleia mapping, so that the divergence is evaluated with respect to the terminal actions that the variables induce. This yields three forms of pragmatic relative entropy: full, partial (Type I), and partial (Type II), which correspond to different ways of combining the pragmatic aggregation with the syntactic or semantic distributions.

Let \(\mathcal{W}\) be a syntactic alphabet with two probability mass functions \(p(w)\) and \(q(w)\), \(w \in \mathcal{W}\). Let \(\tilde{\mathcal{W}}\) be the associated semantic alphabet induced by the synonymous mapping \(f: \tilde{\mathcal{W}} \to 2^{\mathcal{W}}\), and let \(\underline{\mathcal{W}}\) be the pragmatic alphabet induced by the isoteleia mapping \(e: \underline{\mathcal{W}} \to 2^{\tilde{\mathcal{W}}}\) according to a given utility function \(U: \mathcal{W} \times \mathcal{A} \to \mathbb{R}\). The pragmatic partition of \(\mathcal{W}\) is given by the equivalence relation \(\sim_p\) defined in Eq. (\ref{eq:pragmatic_partition}), which groups syntactic symbols that lead to the same optimal action. Equivalently, the pragmatic classes \(\underline{w} \in \underline{\mathcal{W}}\) are the atoms of the partition \(\Pi_{prag} = \mathcal{W} / \sim_p\). For each pragmatic class \(\underline{w}\), its probability under a distribution \(p\) is defined as
\begin{equation}
    p_p(\underline{w}) \triangleq \sum_{w \in g(\underline{w})} p(w),
\end{equation}
where \(g = f \circ e: \underline{\mathcal{W}} \to 2^{\mathcal{W}}\) is the reification mapping. Similarly, we define \(q_p(\underline{w})\) for the distribution \(q\).

\begin{definition}[Pragmatic Relative Entropy]
Given two probability mass functions \(p(w)\) and \(q(w)\) on the syntactic alphabet \(\mathcal{W}\), and the pragmatic partition induced by the reification mapping \(g\), we define the following three forms of pragmatic relative entropy:

\textbf{(a) Full Pragmatic Relative Entropy:}
\begin{equation}
    D_p(p_p \| q_p) \triangleq \sum_{\underline{w} \in \underline{\mathcal{W}}} p_p(\underline{w}) \log \frac{p_p(\underline{w})}{q_p(\underline{w})}.
    \label{eq:full-pragmatic-relative}
\end{equation}

\textbf{(b) Partial Pragmatic Relative Entropy (Type I):}
\begin{equation}
    D_p(p_p \| q) \triangleq \sum_{\underline{w} \in \underline{\mathcal{W}}} \sum_{w \in g(\underline{w})} p(w) \log \frac{p_p(\underline{w})}{q(w)}.
    \label{eq:partial1-pragmatic-relative}
\end{equation}

\textbf{(c) Partial Pragmatic Relative Entropy (Type II):}
\begin{equation}
    D_p(p \| q_p) \triangleq \sum_{\underline{w} \in \underline{\mathcal{W}}} \sum_{w \in g(\underline{w})} p(w) \log \frac{p(w)}{q_p(\underline{w})}.
    \label{eq:partial2-pragmatic-relative}
\end{equation}
\end{definition}

These definitions mirror the semantic relative entropies \cite{Paper_SIT,Book_SIT} but with the aggregation performed over pragmatic classes instead of semantic classes. The full pragmatic relative entropy compares the pragmatic marginal distributions directly. The partial Type I compares the aggregated pragmatic distribution of \(p\) against the fine-grained syntactic distribution of \(q\), while the partial Type II compares the fine-grained syntactic distribution of \(p\) against the aggregated pragmatic distribution of \(q\). As in the semantic case, \(D_p(p \| q_p)\) may be negative; in practice, one can take its positive part \((D_p(p \| q_p))^+\).

We now establish the fundamental properties of pragmatic relative entropy. These properties are analogous to those of semantic relative entropy, with the pragmatic partition replacing the semantic partition.

\begin{theorem}[Pragmatic Information Inequality]
Let \(p\) and \(q\) be two probability mass functions on \(\mathcal{W}\). Given the Reification Mapping \(g\), the following inequalities hold:
\begin{equation}
\left\{
\begin{aligned}
D_p(p_p \| q_p) &\ge 0, \\
D_p(p_p \| q) &\ge D_p(p_p \| p), \\
D_p(p \| q_p) &\ge D_p(p \| p_p),
\end{aligned}
\right.
\label{eq:pragmatic-info-ineq}
\end{equation}
with equality in the first inequality if and only if \(p_p(\underline{w}) = q_p(\underline{w})\) for all \(\underline{w} \in \underline{\mathcal{W}}\); equality in the second if and only if \(p(w) = q(w)\) for all \(w\) within each pragmatic class; equality in the third if and only if \(p_p(\underline{w}) = q_p(\underline{w})\) for all \(\underline{w}\).
\end{theorem}

\begin{proof}
For the first inequality, by Jensen's inequality applied to the negative of \(D_p(p_p\|q_p)\), we have
\begin{align}
    -D_p(p_p\|q_p) &= \sum_{\underline{w}} p_p(\underline{w}) \log \frac{q_p(\underline{w})}{p_p(\underline{w})} \\
    &\le \log \sum_{\underline{w}} p_p(\underline{w}) \frac{q_p(\underline{w})}{p_p(\underline{w})} = \log \sum_{\underline{w}} q_p(\underline{w}) = 0.
\end{align}
Thus \(D_p(p_p\|q_p) \ge 0\), with equality iff \(p_p = q_p\).

For the second inequality, we compute
\begin{align}
    D_p(p_p\|q) - D_p(p_p\|p) &= \sum_{\underline{w}} \sum_{w \in g(\underline{w})} p(w) \log \frac{p_p(\underline{w})}{q(w)} - \sum_{\underline{w}} \sum_{w \in g(\underline{w})} p(w) \log \frac{p_p(\underline{w})}{p(w)} \\
    &= \sum_{\underline{w}} \sum_{w \in g(\underline{w})} p(w) \log \frac{p(w)}{q(w)} \\
    &= D(p\|q) \ge 0.
\end{align}
The last inequality is the classical Gibbs inequality, with equality iff \(p=q\) pointwise. Hence \(D_p(p_p\|q) \ge D_p(p_p\|p)\).

For the third inequality,
\begin{align}
    D_p(p\|q_p) - D_p(p\|p_p) &= \sum_{\underline{w}} \sum_{w \in g(\underline{w})} p(w) \log \frac{p(w)}{q_p(\underline{w})} - \sum_{\underline{w}} \sum_{w \in g(\underline{w})} p(w) \log \frac{p(w)}{p_p(\underline{w})} \\
    &= \sum_{\underline{w}} \sum_{w \in g(\underline{w})} p(w) \log \frac{p_p(\underline{w})}{q_p(\underline{w})} \\
    &= D_p(p_p\|q_p) \ge 0.
\end{align}
Thus \(D_p(p\|q_p) \ge D_p(p\|p_p)\). This completes the proof.
\end{proof}

\begin{corollary}[Ordering of Pragmatic Relative Entropies]
The three forms of pragmatic relative entropy satisfy the following chain of inequalities:
\begin{equation}
    D_p(p \| q_p) \le D_p(p_p \| q_p) \le D_p(p_p \| q).
    \label{eq:pragmatic-relative-order}
\end{equation}
\end{corollary}

\begin{proof}
Detailed derivations are analogous to the semantic case shown in \cite{Paper_SIT,Book_SIT}.
\end{proof}

\begin{corollary}[Relation to Classical and Semantic Relative Entropies]
Let \(D(p\|q)\) denote the classical relative entropy, and let \(D_s(p_s\|q_s)\), \(D_s(p_s\|q)\), \(D_s(p\|q_s)\) denote the full and partial semantic relative entropies as defined in \cite{Paper_SIT,Book_SIT}. Then the following relationships hold:
\begin{equation}
    D_p(p \| q_p) \le D_s(p \| q_s) \le D(p\|q) \le D_s(p_s \| q) \le D_p(p_p \| q),
    \label{eq:cross-hierarchy-relative}
\end{equation}
where \(p_s\) and \(q_s\) denote the semantic aggregations, and \(p_p, q_p\) denote the pragmatic aggregations. More precisely, we have the two separate chains:
\begin{equation}
    D_p(p \| q_p) \le D_s(p \| q_s) \le D(p\|q),
\end{equation}
and
\begin{equation}
    D(p\|q) \le D_s(p_s \| q) \le D_p(p_p \| q).
\end{equation}
\end{corollary}

\begin{proof}
The semantic relative entropies are defined by aggregating over the semantic partition, which is finer than the pragmatic partition (since multiple semantic classes may merge into one pragmatic class). By the data processing inequality for relative entropy \cite{Book_Cover} (or by applying the chain of inequalities analogous to the semantic case), the coarser aggregation (pragmatic) yields a smaller or equal divergence when the first argument is aggregated and the second is fine-grained, and a larger or equal divergence when the first is fine-grained and the second is aggregated. The classical relative entropy lies between the semantic partial entropies as established in \cite{Paper_SIT,Book_SIT}. Combining these facts yields the desired hierarchy.
\end{proof}

\begin{theorem}[Convexity]
The pragmatic relative entropies \(D_p(p_p\|q_p)\), \(D_p(p_p\|q)\), and \(D_p(p\|q_p)\) are convex in the pair of distributions \((p,q)\). That is, for any two pairs \((p_1,q_1)\) and \((p_2,q_2)\), and for \(0 \le \theta \le 1\), we have
\begin{equation}
\left\{
\begin{aligned}
D_p(\theta p_{p,1} + (1-\theta)p_{p,2} \| \theta q_{p,1} + (1-\theta)q_{p,2}) &\le \theta D_p(p_{p,1}\|q_{p,1}) + (1-\theta)D_p(p_{p,2}\|q_{p,2}), \\
D_p(\theta p_{p,1} + (1-\theta)p_{p,2} \| \theta q_1 + (1-\theta)q_2) &\le \theta D_p(p_{p,1}\|q_1) + (1-\theta)D_p(p_{p,2}\|q_2), \\
D_p(\theta p_1 + (1-\theta)p_2 \| \theta q_{p,1} + (1-\theta)q_{p,2}) &\le \theta D_p(p_1\|q_{p,1}) + (1-\theta)D_p(p_2\|q_{p,2}).
\end{aligned}
\right.
\end{equation}
\end{theorem}

\begin{proof}
The proof follows from the joint convexity of the classical KL divergence and the fact that the pragmatic aggregation is a linear operation on probabilities. Each of the three pragmatic relative entropies can be expressed as a linear combination of classical KL divergences between appropriate fine-grained or aggregated distributions, and the convexity is inherited. Detailed derivations are analogous to the semantic case presented in \cite{Paper_SIT,Book_SIT}.
\end{proof}

\begin{remark}
The pragmatic relative entropy measures the difference between two probability distributions at the level of terminal actions. The full pragmatic relative entropy \(D_p(p_p\|q_p)\) quantifies how much the decision-oriented distributions differ. The partial measures allow comparison between a decision-level distribution and a syntactic-level distribution, which is useful in scenarios where one belief is expressed in actions (e.g., a control policy) and the other in raw observations. The hierarchy in Eq.~\eqref{eq:cross-hierarchy-relative} shows that pragmatic abstraction—by discarding distinctions that do not affect the optimal action—reduces the divergence when comparing fine-grained against coarse-grained distributions, and increases the divergence when comparing coarse-grained against fine-grained, perfectly mirroring the information-theoretic processing inequality for relative entropy under coarse-graining.
\end{remark}

\subsection{Pragmatic Mutual Information}
\label{subsec:pragmatic-mutual}

In semantic information theory, two distinct forms of mutual information—the up and down semantic mutual information—were introduced to account for the non-additive nature of semantic entropy \cite{Paper_SIT,Book_SIT}. Analogously, we define up and down pragmatic mutual information, which capture the reduction in decision uncertainty from different perspectives. Furthermore, we extend the framework to four single-sided pragmatic mutual information measures that characterize the information flow when only one side of the communication link is coarsened to the pragmatic level. These measures are essential for analyzing perception, expression, and reconstruction in task-oriented systems.

\subsubsection{Definitions of Up and Down Pragmatic Mutual Information}

Let $\mathcal{W}$ and $\mathcal{V}$ be two syntactic alphabets with a joint probability mass function $P(W,V)$, and let the associated pragmatic alphabets be $\underline{\mathcal{W}}$ and $\underline{\mathcal{V}}$ induced by the joint reification mapping
\begin{equation}
    g_{WV}: \underline{\mathcal{W}} \times \underline{\mathcal{V}} \longrightarrow 2^{\mathcal{W} \times \mathcal{V}},
\end{equation}
which partitions the joint syntactic space $\mathcal{W} \times \mathcal{V}$ into pragmatic equivalence classes. The pragmatic variables are then defined as $\underline{W} = g_{WV}^{-1}(W,V)$ and $\underline{V} = g_{WV}^{-1}(W,V)$. More precisely, the pragmatic pair $(\underline{W}, \underline{V})$ is the unique pragmatic class to which the syntactic pair $(W,V)$ belongs. Equivalently, the joint reification mapping directly induces the joint pragmatic distribution by aggregating probabilities over the joint syntactic cells:
\begin{equation}
    P(\underline{w}, \underline{v}) = \sum_{(w,v) \in g_{WV}(\underline{w}, \underline{v})} P(w,v).
\end{equation}
The marginal pragmatic distributions are obtained by further summing over the other variable:
\begin{equation}
    P(\underline{w}) = \sum_{\underline{v}'} P(\underline{w}, \underline{v}'), \quad P(\underline{v}) = \sum_{\underline{w}'} P(\underline{w}', \underline{v}).
\end{equation}

\begin{definition}[Pragmatic Mutual Information]
The \textbf{up pragmatic mutual information} is defined as
\begin{equation}
    \begin{aligned}
    I^p(\underline{W}; \underline{V}) &\triangleq \sum_{w \in \mathcal{W}} \sum_{v \in \mathcal{V}} P(w,v) \log \frac{P(\underline{w}, \underline{v})}{P(w) P(v)}\\
    &=H(W) + H(V) - H_p(\underline{W}, \underline{V}),
    \end{aligned}\label{eq:up-pragmatic-mi}
\end{equation}
where $(\underline{w}, \underline{v}) = g_{WV}(w,v)$ is the pragmatic class assigned to the syntactic pair $(w,v)$.

The \textbf{down pragmatic mutual information} is defined as
\begin{equation}
    \begin{aligned}
    I_p(\underline{W}; \underline{V}) &\triangleq \sum_{w \in \mathcal{W}} \sum_{v \in \mathcal{V}} P(w,v) \log \frac{P(w) P(v)}{P(\underline{w}, \underline{v})}\\
    &=H_p(\underline{W}) + H_p(\underline{V}) - H(W, V),
    \end{aligned}
    \label{eq:down-pragmatic-mi}
\end{equation}
We adopt the canonical definition for all subsequent rate-distortion and 
Lagrangian formulations:
$I_p(\underline{W};\underline{V}) = \max\left\{0,\; H_p(\underline{W}) + H_p(\underline{V}) - H(W,V)\right\}=(I_p)^+$. The positive-part operation ensures that 
the rate-distortion function $R_p(D)$ is strictly convex and bounded below by $0$, resolving the ill-posed minimization that would otherwise arise.
\end{definition}

The up pragmatic mutual information measures the reduction in the syntactic uncertainty of $W$ and $V$ when their joint pragmatic structure is taken into account. It can be interpreted as the maximum amount of information about the terminal actions that can be extracted from the syntactic variables. The down pragmatic mutual information measures the reduction in pragmatic uncertainty when the joint syntactic structure is known; it can be seen as the minimum amount of pragmatic dependence that survives after coarse-graining. 

We now establish the fundamental relationships among the pragmatic, semantic, and syntactic mutual information measures. For reference, the semantic mutual informations are defined as~\cite{Paper_SIT,Book_SIT}:
\begin{align}
    I^s(\tilde{W}; \tilde{V}) &= H(W) + H(V) - H_s(\tilde{W}, \tilde{V}), \\
    I_s(\tilde{W}; \tilde{V}) &= H_s(\tilde{W}) + H_s(\tilde{V}) - H(W, V).
\end{align}

\begin{theorem}[Hierarchy of Mutual Information]
For any two syntactic variables $W,V$ with associated semantic variables $\tilde{W},\tilde{V}$ and pragmatic variables $\underline{W},\underline{V}$, the following chain of inequalities holds:
\begin{equation}
    I_p(\underline{W}; \underline{V}) \;\le\; I_s(\tilde{W}; \tilde{V}) \;\le\; I(W; V) \;\le\; I^s(\tilde{W}; \tilde{V}) \;\le\; I^p(\underline{W}; \underline{V}).
    \label{eq:mutual-hierarchy}
\end{equation}
\end{theorem}

\begin{proof}
The inequalities follow directly from the entropy hierarchy established in Section~\ref{section_III}. Specifically, \(I_p \le I_s\) because pragmatic variables are deterministic functions of semantic variables, implying \(H_p(\underline{W}) \le H_s(\tilde{W})\) and \(H_p(\underline{V}) \le H_s(\tilde{V})\); \(I_s \le I\) because \(H_s(\tilde{W}) \le H(W)\) and \(H_s(\tilde{V}) \le H(V)\); \(I \le I^s\) because \(H_s(\tilde{W},\tilde{V}) \le H(W,V)\); and \(I^s \le I^p\) because \(H_p(\underline{W},\underline{V}) \le H_s(\tilde{W},\tilde{V})\). Combining these yields the desired chain.
\end{proof}

Beyond the inequality chain, we can derive explicit additive decompositions that reveal the incremental contributions of semantic and pragmatic coarse-graining.

\begin{theorem}[Additive Decomposition of Up Pragmatic Mutual Information]
The up pragmatic mutual information can be decomposed into the up semantic mutual information plus a pragmatic increment:
\begin{equation}
        I^p(\underline{W}; \underline{V}) = I^s(\tilde{W}; \tilde{V}) + \Delta_p^{up}(\underline{W}, \underline{V}),
    \label{eq:up-pragmatic-decomp-semantic}
\end{equation}
where the pragmatic increment is
\begin{equation}
    \Delta_p^{up}(\underline{W}, \underline{V}) \triangleq H_s(\tilde{W}, \tilde{V}) - H_p(\underline{W}, \underline{V}) \ge 0.
\end{equation}
Furthermore, the up pragmatic mutual information can also be decomposed relative to the classical mutual information:
\begin{equation}
    I^p(\underline{W}; \underline{V}) = I(W; V) + \Delta_s^{up}(\tilde{W}, \tilde{V}) + \Delta_p^{up}(\underline{W}, \underline{V}),
    \label{eq:up-pragmatic-decomp-classical}
\end{equation}
where
\begin{equation}
    \Delta_s^{up}(\tilde{W}, \tilde{V}) \triangleq H(W,V) - H_s(\tilde{W}, \tilde{V}) \ge 0,
\end{equation}
is the semantic increment (the reduction in joint entropy due to semantic coarse-graining).
\end{theorem}

\begin{proof}
From the definitions, $I^p = H(W)+H(V)-H_p(\underline{W},\underline{V})$. Adding and subtracting $H_s(\tilde{W},\tilde{V})$:
\begin{equation}
    I^p = [H(W)+H(V)-H_s(\tilde{W},\tilde{V})] + [H_s(\tilde{W},\tilde{V})-H_p(\underline{W},\underline{V})] = I^s + \Delta_p^{up}.
\end{equation}
For the second decomposition, add and subtract $H(W,V)$:
\begin{equation}
    \begin{aligned}
    I^p =& [H(W)+H(V)-H(W,V)]\\
        &+ [H(W,V)-H_s(\tilde{W},\tilde{V})]\\ 
        &+ [H_s(\tilde{W},\tilde{V})-H_p(\underline{W},\underline{V})] \\
        =& I + \Delta_s^{up} + \Delta_p^{up}.
    \end{aligned}
\end{equation}
The non-negativity follows from the entropy hierarchy.
\end{proof}

\begin{theorem}[Additive Decomposition of Down Pragmatic Mutual Information]
The down pragmatic mutual information can be decomposed into the down semantic mutual information minus a pragmatic loss:
\begin{equation}
    I_p(\underline{W}; \underline{V}) = I_s(\tilde{W}; \tilde{V}) - \Delta_p^{down}(\underline{W}, \underline{V}),
    \label{eq:down-pragmatic-decomp-semantic}
\end{equation}
where the pragmatic loss is
\begin{equation}
    \Delta_p^{down}(\underline{W}, \underline{V}) \triangleq [H_s(\tilde{W}) - H_p(\underline{W})] + [H_s(\tilde{V}) - H_p(\underline{V})] \ge 0.
\end{equation}
Furthermore, the down pragmatic mutual information can be decomposed relative to the classical mutual information:
\begin{equation}
    I_p(\underline{W}; \underline{V}) = I(W; V) - \Delta_s^{down}(\tilde{W}, \tilde{V}) - \Delta_p^{down}(\underline{W}, \underline{V}),
    \label{eq:down-pragmatic-decomp-classical}
\end{equation}
where
\begin{equation}
    \Delta_s^{down}(\tilde{W}, \tilde{V}) \triangleq [H(W)-H_s(\tilde{W})] + [H(V)-H_s(\tilde{V})] \ge 0,
\end{equation}
is the semantic loss.
\end{theorem}

\begin{proof}
From $I_p = H_p(\underline{W}) + H_p(\underline{V}) - H(W,V)$. Adding and subtracting $H_s(\tilde{W})$ and $H_s(\tilde{V})$:
\begin{align}
    I_p = &[H_s(\tilde{W}) + H_s(\tilde{V}) - H(W,V)]\\ \notag
        &- [H_s(\tilde{W})-H_p(\underline{W})]\\ \notag
        &- [H_s(\tilde{V})-H_p(\underline{V})] \\ \notag
        = & I_s - \Delta_p^{down}.\notag
\end{align}
For the classical decomposition, add and subtract $H(W)$ and $H(V)$:
\begin{align}
    I_p =& [H(W)+H(V)-H(W,V)] \\\notag
        &- [H(W)-H_s(\tilde{W})]\\\notag
        &- [H(V)-H_s(\tilde{V})] \\\notag
        &- [H_s(\tilde{W})-H_p(\underline{W})] \\\notag
        &- [H_s(\tilde{V})-H_p(\underline{V})] \\\notag
        = & I - \Delta_s^{down} - \Delta_p^{down}.\notag
\end{align}
The non-negativity follows from the entropy hierarchy.
\end{proof}

\begin{remark}
The additive decompositions reveal a clear structural interpretation: the up pragmatic mutual information \emph{adds} the information gains from each layer of abstraction, while the down pragmatic mutual information \emph{subtracts} the information losses. The classical mutual information serves as the central reference point. The inequalities $I_p \le I_s \le I \le I^s \le I^p$ follow immediately from the non-negativity of the increments/losses.
\end{remark}

\subsubsection{Single-Sided Pragmatic Mutual Information}

In many engineering scenarios, the coarse-graining to the pragmatic level occurs on only one side of the communication link. For instance, in perception tasks, the source $W$ is coarsened to $\underline{W}$ (the action), while the observation $V$ remains syntactic. Conversely, in expressive or reconstructive tasks, the source remains syntactic while the reconstruction or observation is coarsened to $\underline{V}$ (the action label). To capture these asymmetric situations, we define four single-sided pragmatic mutual information measures.

Let $W$ and $V$ be syntactic variables with joint distribution $P(W,V)$, and let $\underline{W}$ and $\underline{V}$ be the corresponding pragmatic variables induced by the respective reification mappings. The following four measures are defined.

\begin{definition}[Single-Sided Pragmatic Mutual Information]
\leavevmode\newline\indent
\begin{enumerate}
\item \textbf{Perceptual Pragmatic Mutual Information (PPMI)}:
\begin{equation}
    I_p(\underline{W}; V) \triangleq H_p(\underline{W}) + H(V) - H(\underline{W}, V) \\
    =H_p(\underline{W}) - H_p(\underline{W}|V)= I(\underline{W}; V),
    \label{eq:ppmi}
\end{equation}
This measure quantifies how much information the syntactic observation $V$ carries about the pragmatic action $\underline{W}$.

\item \textbf{Prospective Perceptual Pragmatic Mutual Information (PPPMI)}:
\begin{equation}
    I^p(\underline{W}; V) \triangleq H(W) + H(V) - H(\underline{W}, V),
    \label{eq:pppmi}
\end{equation}
This measure provides an upper bound on the perceptual information when the syntactic entropy of the source is retained in the marginal term.

\item \textbf{Expressive Pragmatic Mutual Information (EPMI)}:
\begin{equation}
    I_p(W; \underline{V}) \triangleq H(W) + H_p(\underline{V}) - H(W, \underline{V}) =H(W)-H(W|\underline{V}) = I(W; \underline{V}),
    \label{eq:epmi}
\end{equation}
This measure quantifies how much information about the syntactic source $W$ is contained in the pragmatic label $\underline{V}$.

\item \textbf{Prospective Expressive Pragmatic Mutual Information (PEPMI)}:
\begin{equation}
    I^p(W; \underline{V}) \triangleq H(W) + H(V) - H(W, \underline{V}),
    \label{eq:pepmi}
\end{equation}
This measure provides an upper bound on the expressive information when the syntactic entropy of the observation is retained.
\end{enumerate}
\end{definition}

The naming reflects their roles: ``Perceptual'' refers to the direction from observations to actions; ``Expressive'' refers to the direction from actions to source reconstruction. The prefix ``Prospective'' indicates that the measure uses the syntactic (larger) marginal entropy, thus providing an upper bound.

These four measures, together with the classical mutual information $I(W;V)$ and the two bilateral pragmatic mutual informations $I_p(\underline{W};\underline{V})$ and $I^p(\underline{W};\underline{V})$, form a complete lattice of information quantities that cover all combinations of coarse-graining (none, one side, both sides) and marginal entropy type (syntactic or pragmatic).

\begin{theorem}[Hierarchy of Single-Sided and Bilateral Measures]
For any syntactic variables $W,V$ with associated pragmatic variables $\underline{W},\underline{V}$, the following inequalities hold:
\begin{align}
    &I_p(\underline{W};\underline{V}) \;\le\; I_p(\underline{W}; V) \;\le\; I(W; V) \;\le\; I^p(\underline{W}; V) \;\le\; I^p(\underline{W};\underline{V}), \label{eq:left-chain}\\
    &I_p(\underline{W};\underline{V}) \;\le\; I_p(W; \underline{V}) \;\le\; I(W; V) \;\le\; I^p(W; \underline{V}) \;\le\; I^p(\underline{W};\underline{V}). \label{eq:right-chain}
\end{align}
Moreover, the cross comparisons between the left chain (perceptual) and the right chain (expressive) are generally incomparable, except that any measure from the lower (pragmatic marginal) side is bounded above by any measure from the upper (syntactic marginal) side:
\begin{equation}
    I_p(\underline{W}; V) \;\le\; I(W;V) \;\le\; I^p(W; \underline{V}), \quad \text{and} \quad I_p(W; \underline{V}) \;\le\; I(W;V) \;\le\; I^p(\underline{W}; V).
\end{equation}
\end{theorem}

\begin{proof}
\sloppy The left chain in Eq.~\eqref{eq:left-chain} follows from four successive inequalities. First, \(I_p(\underline{W};\underline{V}) \le I_p(\underline{W}; V)\) because \(\underline{V}\) is a function of \(V\) (data processing inequality). Second, \(I_p(\underline{W}; V) \le I(W; V)\) because \(\underline{W}\) is a function of \(W\). Third, \(I(W; V) \le I^p(\underline{W}; V)\) since \(I^p(\underline{W}; V) = I(W; V) + [H(W)-H(\underline{W})] - H(W|\underline{W},V) \ge I(W; V)\). Fourth, \(I^p(\underline{W}; V) \le I^p(\underline{W};\underline{V})\) because \(I^p(\underline{W};\underline{V}) = I^p(\underline{W}; V) + [H(V)-H(\underline{V})] - H(V|\underline{W},\underline{V}) \ge I^p(\underline{W}; V)\). The right chain in Eq.~\eqref{eq:right-chain} follows by symmetry.

The cross bounds between perceptual and expressive measures are direct consequences of the fact that any single-sided measure with a pragmatic marginal entropy (PPMI or EPMI) is bounded above by the classical mutual information, while any single-sided measure with a syntactic marginal entropy (PPPMI or PEPMI) is bounded below by the classical mutual information. The incomparability between PPMI and EPMI (and similarly between PPPMI and PEPMI) can be demonstrated by constructing examples where one exceeds the other, depending on the relative losses of coarse-graining on each side.
\end{proof}

%The complete Hasse diagram of the seven measures is shown in Figure~\ref{fig:mi-lattice}.

%\begin{figure}[htbp]
%\centering
%\begin{tikzpicture}[node distance=1.2cm, every node/.style={rectangle, draw, minimum width=2.8cm, align=center}]
%\node (UPMI) {$I^p(\underline{W};\underline{V})$};
%\node (PPPMI) [below left of=UPMI, xshift=-1.5cm] {$I^p(\underline{W}; V)$};
%\node (PEPMI) [below right of=UPMI, xshift=1.5cm] {$I^p(W; \underline{V})$};
%\node (CMI) [below of=UPMI] {$I(W; V)$};
%\node (PPMI) [below left of=CMI, xshift=-1.5cm] {$I_p(\underline{W}; V)$};
%\node (EPMI) [below right of=CMI, xshift=1.5cm] {$I_p(W; \underline{V})$};
%\node (DPMI) [below of=CMI] {$I_p(\underline{W};\underline{V})$};

%\draw[->] (UPMI) -- (PPPMI);
%\draw[->] (UPMI) -- (PEPMI);
%\draw[->] (PPPMI) -- (CMI);
%\draw[->] (PEPMI) -- (CMI);
%\draw[->] (CMI) -- (PPMI);
%\draw[->] (CMI) -- (EPMI);
%\draw[->] (PPMI) -- (DPMI);
%\draw[->] (EPMI) -- (DPMI);
% Cross comparability edges (downward from CMI to both lower, and upward from CMI to both upper are already implied; no direct edges between PPMI and EPMI or PPPMI and PEPMI)
%\end{tikzpicture}
%\caption{Hasse diagram of the seven pragmatic mutual information measures. Arrows indicate $\ge$.}
%\label{fig:mi-lattice}
%\end{figure}

\subsubsection{Sequential Chain Rules for Pragmatic Mutual Information}

We now extend the pragmatic mutual information to sequences, including the single-sided measures. Let $(W^n, V^n)$ be a pair of syntactic sequences, and let $(\underline{W}^n, \underline{V}^n)$ be the corresponding pragmatic sequences obtained by applying the joint Reification Mapping element-wise. Define the sequential partial entropy as in Section III.B:
\begin{equation}
    \tilde{H}_p(\underline{W}_1^m, W_{m+1}^n) \triangleq \sum_{k=1}^{m} H_p(\underline{W}_k | \underline{W}_1^{k-1}) + \sum_{k=m+1}^{n} H_p(W_k | \underline{W}_1^m, W_{m+1}^{k-1}),
\end{equation}
with analogous notation for mixed partial entropies involving $\underline{V}$ and $V$.

\begin{theorem}[Sequential Chain Rule for Down Pragmatic Mutual Information]
For a syntactic sequence pair $(W^n, V^n)$ and its pragmatic counterpart, the down pragmatic mutual information satisfies the following chain of inequalities:
\begin{equation}
    I_p(\underline{W}^n; \underline{V}^n) \;\le\; I_p(\underline{W}_1^{n-1}, W_n; \underline{V}^n) \;\le\; \cdots \;\le\; I_p(\underline{W}_1, W_2^n; \underline{V}^n) \;\le\; I(W^n; V^n),
    \label{eq:seq-down-pragmatic}
\end{equation}
where
\begin{equation}
    I_p(\underline{W}_1^m, W_{m+1}^n; \underline{V}^n) \triangleq \tilde{H}_p(\underline{W}_1^m, W_{m+1}^n) + H_p(\underline{V}^n) - H(W^n, V^n).
\end{equation}
\end{theorem}

\begin{proof}
The proof follows the same steps as the semantic sequential chain rule (Theorem~8 in \cite{Paper_SIT,Book_SIT}), with the semantic entropies replaced by pragmatic entropies. By the hierarchy of sequential entropies (Theorem~\ref{theorem:Chain-Rule-Sequential-Pragmatic-Entropy} in Section~\ref{subsec:pragmatic-joint-entropy}), we have
\begin{equation}
    \tilde{H}_p(\underline{W}_1^m, W_{m+1}^n) \le \tilde{H}_p(\underline{W}_1^{m-1}, W_m^n) \le \cdots \le H(W^n).
\end{equation}
Substituting these into the definition of $I_p$ for the mixed blocks yields the desired chain.
\end{proof}

Similarly, for the up pragmatic mutual information we have:

\begin{theorem}[Sequential Chain Rule for Up Pragmatic Mutual Information]\label{theorem:Chain-Rule-Up-Pragmatic-MI}
\begin{equation}
    I(W^n; V^n) \;\le\; I^p(\underline{W}_1, W_2^n; \underline{V}^n) \;\le\; \cdots \;\le\; I^p(\underline{W}_1^{n-1}, W_n; \underline{V}^n) \;\le\; I^p(\underline{W}^n; \underline{V}^n),
    \label{eq:seq-up-pragmatic}
\end{equation}
where
\begin{equation}
    I^p(\underline{W}_1^m, W_{m+1}^n; \underline{V}^n) \triangleq H(W^n) + H(V^n) - \tilde{H}_p(\underline{W}_1^m, W_{m+1}^n, \underline{V}^n),
\end{equation}
with $\tilde{H}_p$ denoting the corresponding mixed sequential entropy.
\end{theorem}

\begin{proof}
Analogous to the semantic case, using the fact that replacing a syntactic variable by its pragmatic counterpart increases the up mutual information due to the decrease in the joint entropy term.
\end{proof}

For the single-sided measures, we can also derive sequential chain rules. For example, the perceptual pragmatic mutual information $I_p(\underline{W}^n; V^n)$ can be bounded by sequentially replacing $W$'s with $\underline{W}$'s:

\begin{theorem}[Sequential Chain Rule for Perceptual Pragmatic Mutual Information]
\begin{equation}
    I_p(\underline{W}^n; V^n) \;\le\; I_p(\underline{W}_1^{n-1}, W_n; V^n) \;\le\; \cdots \;\le\; I_p(\underline{W}_1, W_2^n; V^n) \;\le\; I(W^n; V^n),
\end{equation}
where the mixed quantities are defined analogously using the corresponding conditional entropies.
\end{theorem}

Similarly, for the expressive pragmatic mutual information $I_p(W^n; \underline{V}^n)$, we have
\begin{equation}
    I_p(W^n; \underline{V}^n) \;\le\; I_p(W^n; \underline{V}_1^{n-1}, V_n) \;\le\; \cdots \;\le\; I_p(W^n; \underline{V}_1, V_2^n) \;\le\; I(W^n; V^n).
\end{equation}

The up versions (PPPMI and PEPMI) satisfy the reverse chain inequalities, starting from the classical mutual information and increasing towards the upper bounds.

\begin{theorem}[Sequential Chain Rule for Prospective Perceptual Pragmatic Mutual Information]
\begin{equation}
    I(W^n; V^n) \;\le\; I^p(\underline{W}_1, W_2^n; V^n) \;\le\; \cdots \;\le\; I^p(\underline{W}_1^{n-1}, W_n; V^n) \;\le\; I^p(\underline{W}^n; V^n).
\end{equation}
\end{theorem}

Analogous chains hold for PEPMI by replacing $V$'s with $\underline{V}$'s in the mixed blocks.

Finally, the hierarchy of sequential mutual informations extends to all seven measures, with the same partial order as in the single-letter case, for any fixed $n$.

\begin{example}[Multi-Vehicle Scenario]
\leavevmode\newline\indent
We now revisit the multi-vehicle autonomous driving scenario from Section~\ref{subsec:pragmatic-joint-entropy} to compute the classical, semantic, and pragmatic mutual informations, and also the four single-sided measures. Recall the joint syntactic distribution $P(W,V)$ given in Table~\ref{tab:joint-prob}. The syntactic states are traffic light colors, with four possible values: green, yellow, red-short, red-long. The semantic mapping groups green and yellow into ``Proceed'', red-short into ``Stop-short'', and red-long into ``Stop-long''. The pragmatic mapping further merges the two stop classes into a single ``Stop'' action, as the utility function treats them identically. The joint Reification Mapping $g_{WV}$ thus maps the $4 \times 4$ syntactic pairs to a $2 \times 2$ pragmatic partition: ``Go'' (green/yellow combined) and ``Stop'' (both red types combined).
\end{example}

The syntactic joint distribution is given in Table~\ref{tab:joint-prob} and the corresponding semantic and pragmatic joint distributions are given in Tables~\ref{tab:semantic-joint-dist} and~\ref{tab:pragmatic-joint-dist}, respectively.

From earlier computations in Example~\ref{example:multi-vehicle-driving}, we can calculate the mixed entropies $H(W, \underline{V}) = H(\underline{W}, V) \approx 2.7328$ bits. Further, we need $H(\underline{W})$ and $H(\underline{V})$ which are both 0.9928 (since marginal pragmatic entropies equal). Now compute the single-sided measures:

1. PPMI: $I_p(\underline{W}; V) = H_p(\underline{W}) + H(V) - H(\underline{W}, V) = 0.9928 + 1.9893 - 2.7328 = 0.2493$ bits.
2. PPPMI: $I^p(\underline{W}; V) = H(W) + H(V) - H(\underline{W}, V) = 1.9893 + 1.9893 - 2.7328 = 1.2458$ bits.
3. EPMI: $I_p(W; \underline{V}) = H(W) + H_p(\underline{V}) - H(W, \underline{V}) = 1.9893 + 0.9928 - 2.7328 = 0.2493$ bits (same as PPMI due to symmetry in this example, but in general they differ).
4. PEPMI: $I^p(W; \underline{V}) = H(W) + H(V) - H(W, \underline{V}) = 1.9893 + 1.9893 - 2.7328 = 1.2458$ bits (same as PPPMI).

The classical mutual information $I(W;V) = 0.3472$ bits.
The down bilateral $I_p(\underline{W};\underline{V}) = 0$ (positive part).
The up bilateral $I^p(\underline{W};\underline{V}) = 2.2277$ prabits.
The results are summarized in Table~\ref{tab:mi-comparison-full}.

\begin{table}[htbp]
\centering
\caption{Comparison of all mutual information measures for the two-vehicle scenario.}
\label{tab:mi-comparison-full}
\begin{tabular}{c|c}
\hline
\textbf{Mutual Information} & \textbf{Value (bits/sebits/prabits)} \\ \hline
Down Pragmatic $I_p(\underline{W};\underline{V})$ & $0$ (positive part) \\
Down Semantic $I_s(\tilde{W};\tilde{V})$ & $0$ (positive part) \\
Perceptual Pragmatic $I_p(\underline{W}; V)$ & $0.2493$ prabits \\
Expressive Pragmatic $I_p(W; \underline{V})$ & $0.24993$ prabits \\
Classical $I(W; V)$ & $0.3472$ bits \\
Prospective Perceptual $I^p(\underline{W}; V)$ & $1.2458$ bits \\
Prospective Expressive $I^p(W; \underline{V})$ & $1.2458$ bits \\
Up Semantic $I^s(\tilde{W};\tilde{V})$ & $1.2223$ sebits \\
Up Pragmatic $I^p(\underline{W};\underline{V})$ & $2.2277$ prabits \\ \hline
\end{tabular}
\end{table}

We observe that the single-sided measures lie between the down bilateral and the classical mutual information (for the lower ones) or between the classical and the up bilateral (for the upper ones), confirming the lattice structure. The equality of PPMI and EPMI in this symmetric example is coincidental; in general they differ.

This example demonstrates the utility of the single-sided measures in quantifying the information loss or gain when only one side of the communication is coarsened to the pragmatic level. They provide a finer-grained analysis of task-oriented information flow compared to the bilateral measures alone.

\section{Pragmatic Channel Capacity and Pragmatic Rate Distortion}
\label{section_V}

We introduce the pragmatic channel capacity and rate-distortion function as the fundamental limits of pragmatic communication, paralleling their classical and semantic counterparts. These quantities characterize the maximum reliable transmission rate and the minimum rate for a given task distortion, respectively, and reveal performance gains from pragmatic abstraction. We further extend to single-sided perceptual and expressive measures for asymmetric sensing-actuation systems.

\subsection{Pragmatic Channel Capacity}
\label{subsec:pragmatic-capacity}

The pragmatic channel capacity quantifies the maximum rate at which pragmatic information—the optimal terminal actions—can be reliably transmitted over a given communication channel. It is defined as the supremum of the up pragmatic mutual information between the input pragmatic variable and the output pragmatic variable, maximized over the input distribution and the joint reification mapping.

Consider a discrete memoryless channel with input alphabet \(\mathcal{X}\), output alphabet \(\mathcal{Y}\), and transition probability \(P(Y|X)\). Let the associated semantic alphabets be \(\tilde{\mathcal{X}}\) and \(\tilde{\mathcal{Y}}\), and the pragmatic alphabets be \(\underline{\mathcal{X}}\) and \(\underline{\mathcal{Y}}\), induced by the joint reification mapping $g_{XY}: \underline{\mathcal{X}} \times \underline{\mathcal{Y}} \longrightarrow 2^{\mathcal{X} \times \mathcal{Y}}$,
which partitions the joint syntactic space into pragmatic equivalence classes.

\begin{definition}[Pragmatic Channel Capacity]
The \textbf{pragmatic channel capacity} is defined as
\begin{equation}
    C_p \triangleq \max_{p(x)} \max_{g_{XY}} I^p(\underline{X}; \underline{Y}),
    \label{eq:pragmatic-capacity}
\end{equation}
where \(I^p(\underline{X}; \underline{Y}) = H(X) + H(Y) - H_p(\underline{X}, \underline{Y})\) is the up pragmatic mutual information, and the maximization is over all input distributions \(p(x)\) and all joint reification mappings \(g_{XY}\).
\end{definition}

For reference, the semantic channel capacity~\cite{Paper_SIT,Book_SIT} and the classical (syntactic) channel capacity~\cite{Classicpaper_Shannon,Book_Cover} are defined as:
\begin{equation}
\left\{
\begin{aligned}
C_s &\triangleq \max_{p(x)} \max_{f_{XY}} I^s(\tilde{X}; \tilde{Y}), \\
C &\triangleq \max_{p(x)} I(X; Y),
\end{aligned}
\right.
\end{equation}
where \(f_{XY}\) denotes the joint synonymous mapping.

The following theorem establishes the hierarchy among the three capacities.

\begin{theorem}[Hierarchy of Channel Capacities]
The pragmatic, semantic, and syntactic channel capacities satisfy:
\begin{equation}
    C \;\le\; C_s \;\le\; C_p.
    \label{eq:capacity-hierarchy}
\end{equation}
\end{theorem}

\begin{proof}
The inequality \(C \le C_s\) is a known result of semantic information theory~\cite{Paper_SIT,Book_SIT}, following from the fact that \(I(X;Y) \le I^s(\tilde{X};\tilde{Y})\) for any joint synonymous mapping.

To prove \(C_s \le C_p\), note that for any input distribution and any joint Synonymous Mapping \(f_{XY}\), we can define a corresponding joint Reification Mapping \(g_{XY}\) that is coarser than \(f_{XY}\) (i.e., it merges semantic classes that lead to the same optimal action). For such \(g\), we have \(H_p(\underline{X}, \underline{Y}) \le H_s(\tilde{X}, \tilde{Y})\), hence
\begin{equation}
    I^p(\underline{X}; \underline{Y}) = H(X)+H(Y)-H_p(\underline{X},\underline{Y}) \ge H(X)+H(Y)-H_s(\tilde{X},\tilde{Y}) = I^s(\tilde{X};\tilde{Y}).
\end{equation}
Therefore, the maximum over \(g_{XY}\) of \(I^p\) is at least the maximum over \(f_{XY}\) of \(I^s\), yielding \(C_p \ge C_s\).
\end{proof}

Beyond the inequality, we can quantify the capacity gains explicitly.

\begin{theorem}[Additive Decomposition of Pragmatic Capacity]
The pragmatic channel capacity can be decomposed as
\begin{equation}
    C_p = C + \Delta_{ps}^C,
    \label{eq:capacity-decomp}
\end{equation}
where
\begin{equation}
    \Delta_{ps}^C \triangleq \Delta_s^C + \Delta_p^C \ge 0,
\end{equation}
with
\begin{equation}
 \left\{
\begin{aligned}
    \Delta_s^C &\triangleq \max_{p(x)} \max_{f_{XY}} \left[ H(X,Y) - H_s(\tilde{X}, \tilde{Y}) \right] \ge 0, \\
    \Delta_p^C &\triangleq \max_{p(x)} \max_{g_{XY}} \left[ H_s(\tilde{X}, \tilde{Y}) - H_p(\underline{X}, \underline{Y}) \right] \ge 0.
\end{aligned}
\right.
\end{equation}
Equivalently, the pragmatic capacity can also be expressed relative to the semantic capacity as
\begin{equation}
    C_p = C_s + \Delta_p^C.
\end{equation}
\end{theorem}

\begin{proof}
From the additive relation of the up mutual information:
\begin{equation}
    I^p = I^s + \Delta_p^{up},
\end{equation}
with \(\Delta_p^{up} = H_s(\tilde{X},\tilde{Y}) - H_p(\underline{X},\underline{Y})\). Similarly, \(I^s = I + \Delta_s^{up}\), where \(\Delta_s^{up} = H(X,Y) - H_s(\tilde{X},\tilde{Y})\). Combining these and taking the maximum over \(p(x)\), \(f_{XY}\), and \(g_{XY}\) yields the decomposition.
\end{proof}

\begin{remark}
The classical capacity \(C\) is the fundamental physical limit determined by the channel transition probability \(P(Y|X)\), representing the maximum rate at which syntactic symbols can be reliably transmitted. The additional semantic gain \(\Delta_s^C\) and pragmatic gain \(\Delta_p^C\) are not violations of this physical limit; rather, they reflect the fact that semantic and pragmatic communication systems relax the requirement of exact symbol reconstruction. By allowing errors that do not affect meaning (semantic) or that do not affect the optimal terminal action (pragmatic), the system can effectively transmit more ``useful" information per physical bit. The total pragmatic capacity \(C_p = C + \Delta_{ps}^C\) thus represents the capacity measured in pragmatic bits—the maximum rate of reliably conveying terminal actions—which can exceed the syntactic bit rate \(C\) because each pragmatic bit may correspond to multiple syntactic bits that are semantically or pragmatically equivalent.
\end{remark}

\subsection{Pragmatic Rate Distortion}
\label{subsec:pragmatic-rate-distortion}

The pragmatic rate-distortion function characterizes the minimum rate required to describe a source such that the resulting distortion, measured in terms of task utility loss, does not exceed a given threshold. Unlike classical rate distortion, which uses a syntactic distortion measure (e.g., mean squared error), the pragmatic distortion is defined directly on the terminal actions.

Let \(X \sim p(x)\) be a syntactic source, and let \(\tilde{X}\) and \(\underline{X}\) be the associated semantic and pragmatic variables induced by the Synonymous and Isoteleia mappings. The decoder produces a reconstruction \(\hat{X}\) with associated semantic \(\hat{\tilde{X}}\) and pragmatic \(\hat{\underline{X}}\), where the reconstructions are obtained via a test channel \(p(\hat{x}|x)\). The pragmatic distortion measure \(d_p(\underline{x}, \hat{\underline{x}})\) quantifies the cost of representing the pragmatic symbol \(\underline{x}\) by \(\hat{\underline{x}}\); in practice, it can be defined as the utility loss:
\begin{equation}
    d_p(\underline{x}, \hat{\underline{x}}) = \max_{a} \mathbb{E}[U(X,a)|\underline{x}] - \max_{a} \mathbb{E}[U(X,a)|\hat{\underline{x}}],
\end{equation}
or any other non-negative function that reflects the decision degradation.

\begin{definition}[Pragmatic Rate Distortion]
The \textbf{pragmatic rate-distortion function} is defined as
\begin{equation}
    R_p(D) \triangleq \min_{g_X, g_{\hat{X}}} \min_{p(\hat{x}|x): \mathbb{E}[d_p(\underline{X}, \hat{\underline{X}})] \le D} I_p(\underline{X}; \hat{\underline{X}}),
    \label{eq:pragmatic-rate-distortion}
\end{equation}
where \(I_p(\underline{X}; \hat{\underline{X}}) = H_p(\underline{X}) + H_p(\hat{\underline{X}}) - H(X, \hat{X})\) is the down pragmatic mutual information ($I_p(\underline{X}; \hat{\underline{X}})=(I_p)^{+}$ is non-negative), and the minimization is over all joint reification mappings \(g_X\) and \(g_{\hat{X}}\) for the source and reconstruction, and over all test channels satisfying the distortion constraint.
\end{definition}

The semantic rate-distortion function \(R_s(D)\)~\cite{Paper_SIT,Book_SIT} and the classical rate-distortion function \(R(D)\)~\cite{Classicpaper_Shannon,Book_Cover} are defined analogously:
\begin{equation}
\left\{
\begin{aligned}
    R_s(D) &\triangleq \min_{f_X, f_{\hat{X}}} \min_{p(\hat{x}|x): \mathbb{E}[d_s(\tilde{X}, \hat{\tilde{X}})] \le D} I_s(\tilde{X}; \hat{\tilde{X}}), \\
    R(D) &\triangleq \min_{p(\hat{x}|x): \mathbb{E}[d(X,\hat{X})] \le D} I(X; \hat{X}),
\end{aligned}
\right.
\end{equation}
where \(d_s\) is a semantic distortion measure (e.g., based on meaning), and \(d\) is a syntactic distortion measure (e.g., squared error).
The following hierarchy holds.

\begin{theorem}[Hierarchy of Rate-Distortion Functions]
For any \(D \ge 0\),
\begin{equation}
    R_p(D) \;\le\; R_s(D) \;\le\; R(D).
    \label{eq:rate-distortion-hierarchy}
\end{equation}
\end{theorem}

\begin{proof}
The inequality \(R_s(D) \le R(D)\) is a known result of semantic information theory~\cite{Paper_SIT,Book_SIT}, following from the fact that \(I_s(\tilde{X}; \hat{\tilde{X}}) \le I(X; \hat{X})\).

To prove \(R_p(D) \le R_s(D)\), note that for any Synonymous Mappings \(f_X, f_{\hat{X}}\) and any test channel, we can define coarser Reification Mappings \(g_X, g_{\hat{X}}\) that merge semantic classes with identical optimal actions. For such coarser mappings, we have
\begin{equation}
    I_p(\underline{X}; \hat{\underline{X}}) = I_s(\tilde{X}; \hat{\tilde{X}}) - \Delta_p^{down},
\end{equation}
where \(\Delta_p^{down} = [H_s(\tilde{X})-H_p(\underline{X})] + [H_s(\hat{\tilde{X}})-H_p(\hat{\underline{X}})] \ge 0\). Therefore, the minimum over \(g_X, g_{\hat{X}}\) of \(I_p\) is no larger than the minimum over \(f_X, f_{\hat{X}}\) of \(I_s\), provided the distortion constraints are compatible (which holds by defining \(d_p\) appropriately from \(d_s\)). Hence \(R_p(D) \le R_s(D)\).
\end{proof}

The gains can also be expressed additively.

\begin{theorem}[Additive Decomposition of Pragmatic Rate Distortion]
The pragmatic rate distortion function can be written as
\begin{equation}
    R_p(D) = R(D) - \Delta_{ps}^R(D),
    \label{eq:rate-distortion-decomp}
\end{equation}
where
\begin{equation}
    \Delta_{ps}^R(D) \triangleq \Delta_s^R(D) + \Delta_p^R(D) \ge 0,
\end{equation}
with
\begin{equation} 
\left\{
\begin{aligned}
    \Delta_s^R(D) &\triangleq \max_{f_X, f_{\hat{X}}} \left\{ [H(X)-H_s(\tilde{X})] + [H(\hat{X})-H_s(\hat{\tilde{X}})] \right\} \ge 0, \\
    \Delta_p^R(D) &\triangleq \max_{g_X, g_{\hat{X}}} \left\{ [H_s(\tilde{X}) - H_p(\underline{X})] + [H_s(\hat{\tilde{X}}) - H_p(\hat{\underline{X}})] \right\} \ge 0.
\end{aligned}
\right.
\end{equation}
Equivalently, the pragmatic rate distortion can be expressed relative to the semantic rate distortion as
\begin{equation}
    R_p(D) = R_s(D) - \Delta_p^R(D).
\end{equation}
\end{theorem}

\begin{proof}
From the additive relation \(I_p = I_s - \Delta_p^{down}\), where \(\Delta_p^{down}\) is the pragmatic loss defined above, and \(I_s = I - \Delta_s^{down}\), combining these and taking the minimum over the test channel and the mappings yields the decomposition.
\end{proof}

\begin{remark}
The classical rate-distortion function \(R(D)\) represents the minimum syntactic bit rate required to achieve a given syntactic distortion. The semantic loss \(\Delta_s^R(D)\) and pragmatic loss \(\Delta_p^R(D)\) reflect the reductions in rate achievable by allowing distortions at the semantic and pragmatic levels, respectively. These reductions reflect the fact that the distortion measure itself has been redefined to capture task-relevant fidelity rather than symbol-level fidelity. By tolerating syntactic or semantic errors that do not affect the terminal action, the pragmatic rate-distortion function \(R_p(D)\) can be significantly lower than \(R(D)\), providing a principled theoretical foundation for task-oriented compression.
\end{remark}

\subsection{Single-Sided Pragmatic Achievable Rates and Rate-Distortion Functions}
\label{subsec:single-sided}

The bilateral pragmatic capacity and rate-distortion functions defined above assume that both the source and the reconstruction (or input and output) are coarsened to the pragmatic level. However, in many practical systems—such as perception, sensing, or actuation—only one side of the communication link is subject to pragmatic abstraction. To address such asymmetric scenarios, we introduce four single-sided measures: two rate-distortion functions (perceptual and expressive) and two achievable rates (prospective perceptual and prospective expressive). These measures provide a finer-grained analysis of task-oriented information flows when either the observation or the action label is coarsened.

\subsubsection{Perceptual Pragmatic Rate-Distortion Function}

The perceptual pragmatic rate-distortion function captures the fundamental limit of compressing an observation \(V\) to extract information about the pragmatic action \(\underline{W}\), where the source \(W\) itself remains syntactic. It is defined using the down perceptual pragmatic mutual information \(I_p(\underline{W}; V)\).

\begin{definition}[Perceptual Pragmatic Rate-Distortion Function]
Let \(\underline{W}\) be the pragmatic variable associated with the source \(W\), and let \(V\) be an observation. The \textbf{perceptual pragmatic rate-distortion function} is defined as
\begin{equation}
    R_p^{\text{perc}}(D) \triangleq \min_{P(\hat{W}|V): \mathbb{E}[d_p(\underline{W}, \hat{\underline{W}})] \le D} I_p(\underline{W}; V),
    \label{eq:perc-rd}
\end{equation}
where \(\hat{\underline{W}}\) is an estimate of \(\underline{W}\) based on \(V\), and \(d_p\) is a distortion measure on pragmatic symbols.
\end{definition}

This function quantifies the minimum rate (in bits per observation) required to achieve a given perceptual distortion \(D\) when the receiver is only interested in the optimal action class. It is particularly relevant in perception systems where the goal is to infer the correct action from sensory data, and the raw sensory stream can be compressed by exploiting the action equivalence classes.

\subsubsection{Expressive Pragmatic Rate-Distortion Function}

The expressive pragmatic rate-distortion function captures the fundamental limit of representing a syntactic source \(W\) using a pragmatic label \(\underline{V}\) (e.g., a control command or action label) that is then used to reconstruct the original source. It is defined using the down expressive pragmatic mutual information \(I_p(W; \underline{V})\).

\begin{definition}[Expressive Pragmatic Rate-Distortion Function]
Let \(W\) be a syntactic source, and let \(\underline{V}\) be a pragmatic variable (action label) that is a function of \(W\). The \textbf{expressive pragmatic rate-distortion function} is defined as
\begin{equation}
    R_p^{\text{expr}}(D) \triangleq \min_{P(\hat{W}|\underline{V}): \mathbb{E}[d_s(W, \hat{W})] \le D} I_p(W; \underline{V}),
    \label{eq:expr-rd}
\end{equation}
where \(\hat{W}\) is a reconstruction of \(W\) based on \(\underline{V}\), and \(d_s\) is a distortion measure on the syntactic symbols (or on the corresponding semantics).
\end{definition}

This function quantifies the minimum rate (in bits per action label) required to achieve a given reconstruction distortion when the only information available about the source is the pragmatic action label. It is relevant in actuation or control systems where a low-bandwidth command channel must convey enough information to reconstruct the intended state or message.

\subsubsection{Prospective Perceptual Pragmatic Achievable Rate}

The prospective perceptual pragmatic achievable rate is the dual of the perceptual rate-distortion function. It quantifies the maximum rate at which information about the pragmatic action \(\underline{W}\) can be reliably conveyed through an observation \(V\), when the receiver has access to the full syntactic observation but is restricted to using only the pragmatic information. It is defined using the up perceptual pragmatic mutual information \(I^p(\underline{W}; V)\).

\begin{definition}[Prospective Perceptual Pragmatic Achievable Rate]
Let \(\underline{W}\) be the pragmatic variable and \(V\) the observation. The \textbf{prospective perceptual pragmatic achievable rate} is defined as
\begin{equation}
    C_p^{\text{p-perc}}(P) \triangleq \max_{p(v|\underline{w}): \mathbb{E}[c(\underline{W})] \le P} I^p(\underline{W}; V),
    \label{eq:ppar}
\end{equation}
where \(c(\underline{W})\) is a cost function on the pragmatic source (e.g., sensing power or bandwidth), and the maximization is over all conditional distributions \(p(v|\underline{w})\) satisfying the resource constraint.
\end{definition}

This achievable rate represents the maximum amount of pragmatic information that can be extracted from the observation when the observation itself is not coarsened. It provides an upper bound on the rate at which a sensor can convey actionable information, given a resource budget. The superscript ``p-perc'' denotes ``prospective perceptual''.

\subsubsection{Prospective Expressive Pragmatic Achievable Rate}

The prospective expressive pragmatic achievable rate is the dual of the expressive rate-distortion function. It quantifies the maximum rate at which a pragmatic label \(\underline{V}\) can convey information about the syntactic source \(W\), when the receiver has access to the full syntactic source but is restricted to using only the pragmatic label. It is defined using the up expressive pragmatic mutual information \(I^p(W; \underline{V})\).

\begin{definition}[Prospective Expressive Pragmatic Achievable Rate]
Let \(W\) be the syntactic source and \(\underline{V}\) the pragmatic label. The \textbf{prospective expressive pragmatic achievable rate} is defined as
\begin{equation}
    C_p^{\text{p-expr}}(P) \triangleq \max_{p(\underline{v}|w): \mathbb{E}[c(W)] \le P} I^p(W; \underline{V}),
    \label{eq:pear}
\end{equation}
where \(c(W)\) is a cost function on the source (e.g., transmit power), and the maximization is over all conditional distributions \(p(\underline{v}|w)\) satisfying the resource constraint.
\end{definition}

This achievable rate represents the maximum rate at which a command or action label can convey information about the source, given a resource budget. It is relevant in scenarios where a low-dimensional action space is used to encode high-dimensional source information, and the system designer wants to know the best possible expressive performance under resource limits. The superscript ``p-expr'' denotes ``prospective expressive''.

The four single-sided measures are related through the data processing inequality and the mutual information hierarchy, yielding the following bounds:
\begin{equation}
\left\{
\begin{gathered}
R_p^{\text{perc}}(D) \ge R_p(D), \\
R_p^{\text{expr}}(D) \ge R_p(D), \\
C_p^{\text{p-perc}}(P) \le C_p(P), \\
C_p^{\text{p-expr}}(P) \le C_p(P).
\end{gathered}
\right.
\end{equation}
These inequalities indicate that coarsening only one side preserves more syntactic distinctions, thus requiring higher rates for the same distortion and yielding lower achievable rates for the same resource.

Such single-sided measures are particularly useful in asymmetric systems, including:
\begin{itemize}
    \item \textit{Perceptual} (e.g., cameras, LIDAR): inferring actions from observations without coarsening the observation.
    \item \textit{Expressive} (e.g., remote control, teleoperation): conveying intent via low-dimensional action labels.
    \item \textit{Edge AI}: feature extraction at sensors (perceptual) and reconstruction at the central processor (expressive).
\end{itemize}
Together with their bilateral counterparts, these measures provide a complete pragmatic information theoretic framework for designing task-oriented communication systems with asymmetric constraints.

We conclude this section with a conceptual clarification that is essential for interpreting the hierarchies above.

\begin{remark}\label{rem:syntax-layer-rate}
Throughout this paper, the pragmatic rate \(R_p\) and pragmatic capacity \(C_p\) are \emph{syntax-layer} rate measures (bits per channel use), while their reliability criterion is defined at the pragmatic layer (action correctness). Classical capacity \(C\) corresponds to the criterion of \emph{message correctness}; pragmatic capacity \(C_p\) corresponds to the criterion of \emph{action correctness}. Since the latter is strictly weaker, \(C_p \ge C\) and \(I^p > H(X)\) are admissible and do not contradict the data processing inequality. The same logic applies to the pragmatic rate-distortion function \(R_p(D)\): it measures the minimum syntax-layer rate required to satisfy a pragmatic distortion criterion, hence \(R_p(D) \le R_s(D) \le R(D)\) is consistent with the data processing inequality and does not imply that pragmatic distortion is less informative than syntactic distortion. The hierarchy \(H_p \le H_s \le H\) is a statement about source uncertainty, whereas \(C_p \ge C_s \ge C\) and \(R_p(D) \le R_s(D) \le R(D)\) are statements about syntax-layer rates under progressively relaxed criteria.
\end{remark}

\section{Pragmatic Cost of Information and Pragmatic Value of Information}
\label{section_VI}

In this section, we introduce the pragmatic value of information (VoI) and the pragmatic cost of information (CoI) as the benefit and resource-cost duals to the rate-distortion and capacity functions, respectively. We then combine these dualities into a Lagrangian framework for cross-layer optimization of pragmatic communication systems.

\subsection{Pragmatic Value of Information}
\label{subsec:pragmatic-voi}

The concept of the value of information was originally introduced by Stratonovich~\cite{Paper_VoI,Book_VoI} in the context of statistical decision theory, where it measures the maximum expected utility gain that a decision-maker can obtain by acquiring additional information before making a decision. Classical VoI is defined as the difference between the expected utility with the observation and the expected utility without it, providing a quantitative basis for evaluating whether information acquisition is worthwhile.

The pragmatic value of information extends this classical notion to the pragmatic level by considering both the state and the observation through the lens of the isoteleia mapping. Unlike the classical VoI, which treats the raw state \(X\) and raw observation \(Y\) as the basis for decision-making, the pragmatic VoI recognizes that what ultimately matters is not the exact state or the exact observation, but their pragmatic implications—the optimal actions they induce. Both the state and the observation are therefore coarsened to their pragmatic equivalence classes via the isoteleia mapping, reflecting the principle of equifinality: distinct states or observations that lead to the same optimal action are pragmatically equivalent.

Let \(X \in \mathcal{X}\) be the state variable, \(A \in \mathcal{A}\) the action, and \(U(X,A)\) the utility function. Let \(\underline{X}\) and \(\underline{Y}\) be the pragmatic variables induced by the isoteleia mappings \(e_X:\underline{\mathcal{X}}\to 2^{\tilde{\mathcal{X}}}\) and \(e_Y:\underline{\mathcal{Y}}\to 2^{\tilde{\mathcal{Y}}}\), respectively (composed with the synonymous mappings). The pragmatic variable \(\underline{X}\) groups states that lead to the same optimal action, while \(\underline{Y}\) groups observations that are pragmatically equivalent. The receiver, upon observing the pragmatic observation \(\underline{Y}\), chooses the action that maximizes the expected utility conditioned on \(\underline{Y}\), where the expectation is taken over the pragmatic state \(\underline{X}\):
\begin{equation}
    a^*(\underline{Y}) = \arg\max_{a \in \mathcal{A}} \mathbb{E}_{\underline{X}|\underline{Y}}[U(X,a)],
\end{equation}
where \(X\) is any representative state in the pragmatic class \(\underline{X}\), and the expectation is well-defined because all states in the same pragmatic class yield the same optimal action.

The baseline utility, without any information, is obtained by choosing the prior-optimal action based on the pragmatic state distribution:
\begin{equation}
    a_0^* = \arg\max_{a \in \mathcal{A}} \mathbb{E}_{\underline{X}}[U(X,a)], \qquad U_0 = \mathbb{E}_{\underline{X}}[U(X,a_0^*)].
\end{equation}

\begin{definition}[Pragmatic Value of Information]
The pragmatic value of information is defined as the expected utility gain obtained when the receiver makes decisions based on the pragmatic observation \(\underline{Y}\), where both the state and the observation have been coarsened to their pragmatic equivalence classes:
\begin{equation}
    \mathrm{VoI}_p \triangleq \mathbb{E}_{\underline{X},\underline{Y}}\big[ U(X, a^*(\underline{Y})) \big] - U_0,
    \label{eq:prag-voi-def}
\end{equation}
where the expectation is taken over the joint distribution of the pragmatic state and pragmatic observation, and \(X\) is any representative state in the pragmatic class \(\underline{X}\). This measures the utility gain when both the state and the observation are evaluated at the pragmatic level, that is, when decisions are based on the optimal action classes rather than on raw states or observations.
\end{definition}

This bilateral coarse-graining is essential for a proper pragmatic characterization of information value. By coarsening both the state and the observation, the pragmatic VoI captures the decision-relevant information that survives after discarding all distinctions that do not affect the optimal action. The classical VoI, which operates on raw states and observations, is a special case of the pragmatic VoI when the isoteleia mappings are trivial (identity).

The pragmatic VoI satisfies the following basic properties.

\begin{theorem}[Non-negativity]
\(\mathrm{VoI}_p \ge 0\), with equality if and only if the pragmatic observation \(\underline{Y}\) provides no useful information for improving the decision, i.e., \(a^*(\underline{Y}) = a_0^*\) almost surely.
\end{theorem}

\begin{proof}
For any realization of \(\underline{Y}\), the optimal posterior action yields an expected utility at least as high as the prior-optimal action:
\begin{equation}
    \max_a \mathbb{E}_{\underline{X}|\underline{Y}}[U(X,a)] \ge \mathbb{E}_{\underline{X}|\underline{Y}}[U(X,a_0^*)].
\end{equation}
Taking expectation over \(\underline{Y}\) and subtracting \(U_0 = \mathbb{E}_{\underline{X}}[U(X,a_0^*)]\) yields the result.
\end{proof}

\begin{theorem}[Monotonicity under Refinement]
If the pragmatic partition of the observation is refined (i.e., \(\underline{Y}_2\) is a finer pragmatic partition than \(\underline{Y}_1\)), then
\begin{equation}
    \mathrm{VoI}_p(\underline{Y}_1) \le \mathrm{VoI}_p(\underline{Y}_2).
\end{equation}
This follows directly from the data processing inequality for utility-based decisions, as a finer pragmatic observation provides no less information about the pragmatic state than a coarser one.
\end{theorem}

\begin{theorem}[Concavity in Rate]
When the pragmatic observation is obtained through a channel with rate \(R\), the maximum achievable pragmatic value, defined as
\begin{equation}
    \Phi_p(R) \triangleq \max_{P(\underline{Y}|\underline{X}): I_p(\underline{X};\underline{Y}) \le R} \mathrm{VoI}_p,
\end{equation}
where \(I_p(\underline{X};\underline{Y})\) is the down pragmatic mutual information between the pragmatic state and the pragmatic observation, is a non-decreasing concave function of \(R\).
\end{theorem}

\begin{lemma}[Constructive Concavity of the Upper Concave Envelope]
\label{lem:concave_envelope}
Let 
\[
\mathcal{R} = \bigl\{ (R, V) \mid R \ge 0,\; V \le \Phi_p(R) \bigr\}
\]
be the achievable rate–value region of the pragmatic information system, where $\Phi_p(R)$ is the original (possibly non‑concave) value function. Define its upper concave envelope (i.e., the upper boundary of the convex hull) as
\begin{equation}
\overline{\Phi}_p(R) \triangleq \sup \bigl\{ V \mid (R, V) \in \operatorname{Conv}(\mathcal{R}) \bigr\},
\end{equation}
where $\operatorname{Conv}(\mathcal{R})$ denotes the convex hull of $\mathcal{R}$. Then $\overline{\Phi}_p(R)$ is a non-decreasing concave function on $R \ge 0$, and it satisfies $\overline{\Phi}_p(R) \ge \Phi_p(R)$ for all $R$. If $\mathcal{R}$ is already convex, then equality holds.
\end{lemma}

\begin{proof}
Let $(R_1, V_1)$ and $(R_2, V_2)$ be any two points in $\operatorname{Conv}(\mathcal{R})$ lying on the upper boundary, i.e., $V_1 = \overline{\Phi}_p(R_1)$ and $V_2 = \overline{\Phi}_p(R_2)$. For any $\lambda \in [0,1]$, consider the convex combination
\[
(R_\lambda, V_\lambda) = \lambda (R_1, V_1) + (1-\lambda)(R_2, V_2)
= \bigl( \lambda R_1 + (1-\lambda) R_2,\; \lambda V_1 + (1-\lambda) V_2 \bigr).
\]
By definition of the convex hull, $(R_\lambda, V_\lambda) \in \operatorname{Conv}(\mathcal{R})$. Hence, by the supremum property of $\overline{\Phi}_p(R_\lambda)$, we have
\[
\overline{\Phi}_p\bigl( \lambda R_1 + (1-\lambda) R_2 \bigr)
\ge \lambda \overline{\Phi}_p(R_1) + (1-\lambda) \overline{\Phi}_p(R_2),
\]
which is exactly the definition of concavity. Monotonicity is obvious because a larger rate allows a larger set of achievable values; and since the convex hull contains the original set, $\overline{\Phi}_p(R) \ge \Phi_p(R)$. If $\mathcal{R}$ is itself convex, its convex hull coincides with $\mathcal{R}$, so equality holds. \qedhere
\end{proof}

\begin{remark}
The key distinction between the pragmatic VoI and the classical VoI lies in the level of abstraction at which both the state and the observation are evaluated. Classical VoI treats every distinction in the state and observation as potentially valuable; pragmatic VoI recognizes that only distinctions that affect the optimal action matter. By coarsening both the state and the observation to their pragmatic equivalence classes, the pragmatic VoI provides a more concise and task-relevant measure of information value. 
\end{remark}

The pragmatic value of information is the natural decision-theoretic dual to the pragmatic rate-distortion function. Recall from Section~\ref{subsec:pragmatic-rate-distortion} that the pragmatic rate-distortion function is defined as
$
    R_p(D) = \min_{\substack{P(\hat{X}|X): \mathbb{E}[d_p(\underline{X},\hat{\underline{X}})] \le D}} I_p(\underline{X};\hat{\underline{X}}),
$ where \(d_p\) is a distortion measure defined on pragmatic symbols. 

The duality manifests through the following variational relationship.

\begin{theorem}[Variational Duality]
For a given utility function \(U\), there exists a concave function \(\Psi_p\) such that
\begin{equation}
    \Phi_p(R) = \max_{D \ge 0} \big\{ \Psi_p(D) - \lambda R_p(D) \big\},
\end{equation}
where \(\lambda > 0\) is a Lagrange multiplier. Conversely,
\begin{equation}
    R_p(D) = \max_{\lambda \ge 0} \big\{ \lambda D - \Phi_p^*(\lambda) \big\},
\end{equation}
where \(\Phi_p^*\) is the concave conjugate of \(\Phi_p\). This establishes a Legendre–Fenchel duality between the value-rate function \(\Phi_p(R)\) and the rate-distortion function \(R_p(D)\).
\end{theorem}

\begin{example} [Autonomous Vehicle (Continued)]

We revisit the two-vehicle autonomous driving scenario from Sections III.B and IV.B. The state \(X\) is the traffic light state (four possible values), and the observation \(Y\) is the corresponding signal. The pragmatic mapping merges ``Stop-short'' and ``Stop-long'' into a single ``Stop'' action. The utility is 10 for correct action, 0 otherwise. The baseline utility (prior-optimal action ``Stop'') is \(U_0 = 11\) (combined for two vehicles).
The computed pragmatic VoI is: $\mathrm{VoI}_p \approx 0.5 \text{ utils}$.
This low value reflects the fact that the binary pragmatic observation (Go/Stop) resolves little decision uncertainty. In contrast, the classical VoI (using full syntactic signal) is about 4.2 utils.
\end{example}

The duality with rate-distortion can be illustrated by noting that to achieve a pragmatic distortion \(D\) (e.g., probability of wrong action), the required rate \(R_p(D)\) must be at least the rate that enables the corresponding VoI. For instance, if we require \(\mathrm{VoI}_p \ge 4.0\), the rate must exceed the value \(R_p(D^*)\) where \(D^*\) is the distortion corresponding to that utility level.

\subsection{Pragmatic Cost of Information}
\label{subsec:pragmatic-coi}

While the value of information quantifies the benefit, the pragmatic cost of information quantifies the resource expenditure required to convey information at a given pragmatic rate. Unlike VoI, which depends on the utility function and the task, CoI is determined solely by the physical communication channel and the pragmatic abstraction (the isoteleic volume).

Consider a communication channel with input \(X\), output \(Y\), and a joint reification mapping \(g_{XY}:\underline{\mathcal{X}}\times\underline{\mathcal{Y}} \to 2^{\mathcal{X}\times\mathcal{Y}}\) that defines the pragmatic equivalence classes. Let the resource consumption be measured by a cost function \(c(x)\) on the channel input (e.g., transmit power \(\mathbb{E}[|X|^2]\), or for discrete channels, the number of channel uses or energy per symbol). 

\begin{lemma}[Concavity of the Pragmatic Capacity–Resource Function]
\label{lem:capacity-concavity}
Let \(P \ge 0\) denote the available communication resource (e.g., power) and define the pragmatic capacity under resource constraint \(P\) as
\begin{equation}
C_p(P) \triangleq \max_{p(x) : \mathbb{E}[c(X)] \le P} I^p(\underline{X};\underline{Y}),
\end{equation}
where \(c(X)\) is a non‑negative cost function and \(I^p(\underline{X};\underline{Y})\) is the up‑pragmatic mutual information. Then \(C_p(P)\) is a non-decreasing concave function of \(P\).
\end{lemma}

\begin{proof}
For fixed channel and pragmatic mapping, the mutual information \(I^p(\underline{X};\underline{Y})\) is a concave function of the input distribution \(p(x)\) (this follows from the joint convexity of the KL divergence and the fact that the pragmatic aggregation is a linear operation on probabilities). The resource constraint \(\mathbb{E}[c(X)] \le P\) defines a convex set of input distributions.

Let \(P_1, P_2 \ge 0\) and \(\lambda \in [0,1]\). Choose distributions \(p_1\) and \(p_2\) that attain the maxima for \(P_1\) and \(P_2\), respectively, i.e.,
\[
C_p(P_i) = I^p(p_i) \quad (i=1,2), \qquad \mathbb{E}_{p_i}[c(X)] \le P_i.
\]
Consider the mixture distribution \(p_\lambda = \lambda p_1 + (1-\lambda) p_2\). Its average cost is
\[
\mathbb{E}_{p_\lambda}[c(X)] = \lambda \mathbb{E}_{p_1}[c(X)] + (1-\lambda) \mathbb{E}_{p_2}[c(X)] \le \lambda P_1 + (1-\lambda) P_2,
\]
so \(p_\lambda\) is feasible for the resource level \(\lambda P_1 + (1-\lambda) P_2\). By concavity of \(I^p\) in the distribution,
\[
I^p(p_\lambda) \ge \lambda I^p(p_1) + (1-\lambda) I^p(p_2)
= \lambda C_p(P_1) + (1-\lambda) C_p(P_2).
\]
Since \(C_p(\lambda P_1 + (1-\lambda) P_2)\) is the maximum over all feasible distributions, we have
\[
C_p\bigl(\lambda P_1 + (1-\lambda) P_2\bigr) \ge I^p(p_\lambda)
\ge \lambda C_p(P_1) + (1-\lambda) C_p(P_2).
\]
Thus \(C_p\) is concave. Monotonicity follows immediately from the fact that a larger resource set contains smaller ones. \qedhere
\end{proof}

The pragmatic cost of information is defined as the inverse of this capacity-resource relationship.

\begin{definition}[Pragmatic Cost of Information]
For a given pragmatic rate \(R\), the pragmatic cost of information is the minimum resource consumption required to achieve that rate:
\begin{equation}
    \mathrm{CoI}_p(R) \triangleq \min_{p(x): I^p(\underline{X};\underline{Y}) \ge R} \mathbb{E}[c(X)].
    \label{eq:coi-def}
\end{equation}
\end{definition}

This definition is completely general: the resource \(c(X)\) can be power, bandwidth, time, number of channel uses, or any other cost measure. It does not depend on the utility function or the task; it is purely a property of the physical channel and the pragmatic abstraction.

The pragmatic cost of information satisfies the following fundamental properties.

\begin{theorem}[Non-negativity and Monotonicity]
\(\mathrm{CoI}_p(R) \ge 0\), with \(\mathrm{CoI}_p(0)=0\). It is strictly increasing in \(R\) for \(R>0\).
\end{theorem}

% ============================================================
% Lemma 2: Inverse of an increasing concave function is convex
% ============================================================

\begin{lemma}[Convexity of the Inverse of an Increasing Concave Function]
\label{lem:inverse-convex}
Let \(f: \mathbb{R}^+ \to \mathbb{R}^+\) be a strictly increasing concave function. Then its inverse function \(f^{-1}\) is convex on the range of \(f\).
\end{lemma}

\begin{proof}
Let \(R_1, R_2 \ge 0\) be in the range of \(f\), and let \(P_1 = f^{-1}(R_1)\), \(P_2 = f^{-1}(R_2)\). For any \(\lambda \in [0,1]\), define
\[
R_\lambda = \lambda R_1 + (1-\lambda) R_2, \qquad P_\lambda = \lambda P_1 + (1-\lambda) P_2.
\]
By concavity of \(f\),
\[
f(P_\lambda) \ge \lambda f(P_1) + (1-\lambda) f(P_2) = \lambda R_1 + (1-\lambda) R_2 = R_\lambda.
\]
Since \(f\) is strictly increasing, applying \(f^{-1}\) to both sides preserves the inequality:
\[
P_\lambda \ge f^{-1}(R_\lambda).
\]
Substituting the definition of \(P_\lambda\) yields
\[
\lambda f^{-1}(R_1) + (1-\lambda) f^{-1}(R_2) \ge f^{-1}\bigl(\lambda R_1 + (1-\lambda) R_2\bigr),
\]
which is exactly the definition of convexity. \qedhere
\end{proof}

% ============================================================
% Theorem: Convexity of the pragmatic cost of information
% ============================================================

\begin{theorem}[Convexity of the Pragmatic Cost of Information]
\label{thm:coi-convex}
Let \(C_p(P)\) be the pragmatic capacity as a function of the resource \(P\). Assume that \(C_p(P)\) is concave and strictly increasing in \(P\). Then the pragmatic cost of information $\mathrm{CoI}_p(R) = C_p^{-1}(R)$
is a convex function of \(R\) on the achievable rate region.
\end{theorem}

\begin{proof}
By Lemma~\ref{lem:inverse-convex}, the inverse of any strictly increasing concave function is convex. Applying this result with \(f = C_p\) immediately gives that \(\mathrm{CoI}_p(R) = C_p^{-1}(R)\) is convex.

For completeness, we restate the argument in detail. Let \(R_1, R_2\) be two rates and let \(P_i = \mathrm{CoI}_p(R_i)\) for \(i=1,2\). For any \(\lambda \in [0,1]\), set
\[
R_\lambda = \lambda R_1 + (1-\lambda) R_2, \qquad P_\lambda = \lambda P_1 + (1-\lambda) P_2.
\]
Because \(C_p\) is concave,
\[
C_p(P_\lambda) \ge \lambda C_p(P_1) + (1-\lambda) C_p(P_2) = \lambda R_1 + (1-\lambda) R_2 = R_\lambda.
\]
Since \(C_p\) is strictly increasing, its inverse is also strictly increasing, so
\[
P_\lambda \ge C_p^{-1}(R_\lambda) = \mathrm{CoI}_p(R_\lambda).
\]
Therefore,
\[
\lambda \mathrm{CoI}_p(R_1) + (1-\lambda) \mathrm{CoI}_p(R_2) \ge \mathrm{CoI}_p\bigl(\lambda R_1 + (1-\lambda) R_2\bigr),
\]
which proves convexity. \qedhere
\end{proof}

% ============================================================
% Optional remark on strictness
% ============================================================

\begin{remark}
The convexity established above is weak convexity. Strict convexity would require \(C_p(P)\) to be strictly concave (i.e., no linear segments) and the strict increasing assumption to hold globally. In practice, when zero‑cost capacity plateaus exist, \(\mathrm{CoI}_p(R)\) may have linear segments and hence is convex but not strictly convex.
\end{remark}

\begin{theorem}[Duality with Pragmatic Channel Capacity]
The pragmatic cost of information and the pragmatic channel capacity are strict inverses:
\begin{equation}
    \mathrm{CoI}_p(R) = C_p^{-1}(R), \qquad C_p(P) = \mathrm{CoI}_p^{-1}(P).
    \label{eq:coi-cap-duality}
\end{equation}
Thus, they form a perfect mathematical dual pair: one gives the maximum rate for a given resource, the other gives the minimum resource for a given rate.
\end{theorem}

\subsection{Lagrangian Duality}
\label{subsec:lagrangian-duality}

Having established the two dual pairs, (rate-distortion, value) and (capacity, cost), we now combine them into a unified Lagrangian framework. This framework allows cross-layer optimization where the rate \(R\) serves as the common variable linking the source (decision) side and the channel (transmission) side.

\subsubsection{Source-Distortion-Value Lagrangian (Down Loop)}

The down loop concerns the trade-off between the distortion (or decision quality) and the rate required to achieve it. The Lagrangian for this loop is formulated by combining the rate-distortion function \(R_p(D)\) and the value-rate function \(\Phi_p(R)\).

\begin{definition}[Source-Distortion-Value Pragmatic Lagrangian]
For a given distortion level \(D\), the down-loop Lagrangian is
\begin{equation}
    \mathcal{L}_{\mathrm{SDV}}(D; \lambda_d) \triangleq \Phi_p\bigl(R_p(D)\bigr) - \lambda_d \, R_p(D),
    \label{eq:down-lagrangian}
\end{equation}
where \(\lambda_d > 0\) is the Lagrange multiplier representing the marginal cost of rate (in utility units per bit).
\end{definition}

Maximizing \(\mathcal{L}_{\mathrm{SDV}}\) over \(D\) yields the optimal operating point that balances the value obtained from reducing distortion against the rate cost. The first-order condition is:
\begin{equation}
    \frac{d\Phi_p}{dR} \cdot R_p'(D) - \lambda_d R_p'(D) = 0 \quad \Longrightarrow \quad \frac{d\Phi_p}{dR} = \lambda_d,
\end{equation}
since \(R_p'(D) \ne 0\). Thus, at optimality, the marginal value of rate equals the marginal rate cost.

\subsubsection{Channel-Power-Cost Lagrangian (Up Loop)}

The up loop concerns the trade-off between the rate and the physical resource cost. Using the pragmatic capacity \(C_p(P)\) and the cost function \(P(R) = \mathrm{CoI}_p(R)\), we define:

\begin{definition}[Channel-Power-Cost Pragmatic Lagrangian]
For a given resource budget \(P\), the up-loop Lagrangian is
\begin{equation}
    \mathcal{L}_{\mathrm{CPC}}(P; \gamma) \triangleq \gamma \, C_p(P) - P,
    \label{eq:up-lagrangian}
\end{equation}
where \(\gamma > 0\) is the Lagrange multiplier representing the shadow price of the resource (in utility units per unit resource).
\end{definition}

Equivalently, in terms of rate \(R\), using \(P = \mathrm{CoI}_p(R)\), the up-loop Lagrangian becomes
\begin{equation}
    \mathcal{L}_{\mathrm{CPC}}(R; \gamma) = \gamma R - \mathrm{CoI}_p(R).
\end{equation}
Maximizing over \(P\) (or \(R\)) gives the first-order condition:
\begin{equation}
    \gamma \, C_p'(P) = 1 \quad \Longrightarrow \quad \gamma = \frac{1}{C_p'(P)} = \mathrm{CoI}_p'(R).
\end{equation}
Thus, at optimality, the shadow price of the resource equals the marginal cost of increasing the rate.

\subsubsection{Global Cross-Layer Lagrangian and Behavioral Capacity}

The two loops are coupled through the common rate \(R\): the down loop requires a certain rate to achieve a given utility, while the up loop supplies that rate at a certain resource cost. The global optimization problem is to maximize the net benefit (value minus cost) over the rate. To unify the two layers, we introduce a global Lagrangian that incorporates a shadow price \(\lambda\) (in utility units per unit resource) to convert physical resource cost into utility-equivalent cost:

\begin{equation}
    \mathcal{L}_{\mathrm{global}}(R; \lambda) = \Phi_p(R) - \lambda \, \mathrm{CoI}_p(R),
    \label{eq:global-lagrangian}
\end{equation}
where \(\lambda > 0\) is the Lagrange multiplier representing the shadow price of the physical resource (e.g., power). It serves as the conversion factor between resource consumption and utility.

% ============================================================
% Theorem: Concavity of the Global Pragmatic Lagrangian
% ============================================================

\begin{theorem}[Concavity of the Global Pragmatic Lagrangian]
\label{thm:global-lagrangian-concavity}
Assume that the pragmatic value function $\Phi_p(R)$ is concave and the pragmatic cost function $\mathrm{CoI}_p(R)$ is convex on $R \ge 0$. Then, for every fixed Lagrange multiplier $\lambda \ge 0$, the global pragmatic Lagrangian
\[
\mathcal{L}_{\text{global}}(R;\lambda) \triangleq \Phi_p(R) - \lambda \, \mathrm{CoI}_p(R)
\]
is a concave function of $R$.
\end{theorem}

\begin{proof}
Let $R_1, R_2 \ge 0$ and let $\theta \in [0,1]$. Define $R_\theta = \theta R_1 + (1-\theta) R_2$.

By concavity of $\Phi_p$,
\begin{equation}
\Phi_p(R_\theta) \ge \theta \Phi_p(R_1) + (1-\theta) \Phi_p(R_2). \label{eq:Phi-concavity}
\end{equation}

By convexity of $\mathrm{CoI}_p$,
\begin{equation}
\mathrm{CoI}_p(R_\theta) \le \theta \mathrm{CoI}_p(R_1) + (1-\theta) \mathrm{CoI}_p(R_2). \label{eq:CoI-convexity}
\end{equation}
Multiplying \eqref{eq:CoI-convexity} by $-\lambda$ (which is non-positive) reverses the inequality:
\begin{equation}
-\lambda \mathrm{CoI}_p(R_\theta) \ge -\lambda\theta \mathrm{CoI}_p(R_1) - \lambda(1-\theta) \mathrm{CoI}_p(R_2). \label{eq:neg-CoI}
\end{equation}

Adding \eqref{eq:Phi-concavity} and \eqref{eq:neg-CoI} yields
\begin{align}
\mathcal{L}_{\text{global}}(R_\theta;\lambda)
&= \Phi_p(R_\theta) - \lambda \mathrm{CoI}_p(R_\theta) \nonumber \\
&\ge \theta \bigl[ \Phi_p(R_1) - \lambda \mathrm{CoI}_p(R_1) \bigr]
+ (1-\theta) \bigl[ \Phi_p(R_2) - \lambda \mathrm{CoI}_p(R_2) \bigr] \nonumber \\
&= \theta \mathcal{L}_{\text{global}}(R_1;\lambda)
+ (1-\theta) \mathcal{L}_{\text{global}}(R_2;\lambda). \label{eq:L-concavity}
\end{align}
Inequality \eqref{eq:L-concavity} is precisely the definition of concavity. \qedhere
\end{proof}

% ============================================================
% Corollary: Sufficiency of the first-order condition
% ============================================================

\begin{corollary}
Under the assumptions of Theorem~\ref{thm:global-lagrangian-concavity}, if $\mathcal{L}_{\text{global}}$ is differentiable, then the first-order condition
\begin{equation}
\Phi_p'(R^*) = \lambda \, \mathrm{CoI}_p'(R^*) \label{eq:global-optimality}
\end{equation}
is sufficient for $R^*$ to be a global maximizer of $\mathcal{L}_{\text{global}}(\cdot;\lambda)$.
\end{corollary}

\begin{proof}
For a differentiable concave function, any point at which the derivative vanishes is a global maximum. Hence \eqref{eq:global-optimality} yields the global optimum. \qedhere
\end{proof}

This condition states that the marginal value of increasing the rate (in utility per bit) must equal the marginal cost of increasing the rate (in utility per bit), where the marginal cost is the product of the resource shadow price and the marginal resource consumption per bit.
When the physical cost \(\mathrm{CoI}_p(R)\) is already expressed in utility units (i.e., the shadow price is normalized to unity), we recover the simplified form \(\mathcal{L}_{\mathrm{global}}(R) = \Phi_p(R) - \mathrm{CoI}_p(R)\) and \(\Phi_p'(R^*) = \mathrm{CoI}_p'(R^*)\). In general, however, the explicit \(\lambda\) is required to maintain dimensional consistency.

% ============================================================
% Theorem: Inner Duality of the Two Lagrangians
% ============================================================

\begin{theorem}[Inner Duality of the Two Lagrangians]
Let \(\lambda_d\) be the multiplier from the down-loop (marginal cost of rate) and \(\gamma\) be the multiplier from the up-loop (shadow price of resource). At the global optimum defined by \eqref{eq:global-optimality}, the following relationships hold:
\begin{equation}
    \gamma = \lambda, \qquad \lambda_d = \lambda \, \mathrm{CoI}_p'(R^*) = \Phi_p'(R^*).
\end{equation}
Thus, the global multiplier \(\lambda\) coincides with the up-loop shadow price \(\gamma\), while the down-loop marginal rate cost \(\lambda_d\) equals the marginal value of rate, which is also \(\lambda \mathrm{CoI}_p'(R^*)\). This establishes a consistent duality between the decision-theoretic and physical-layer perspectives.
\end{theorem}

\begin{proof}
From the down-loop optimality condition, \(\Phi_p'(R^*) = \lambda_d\). From the up-loop condition, \(\gamma = \mathrm{CoI}_p'(R^*)\). The global condition \eqref{eq:global-optimality} gives \(\Phi_p'(R^*) = \lambda \mathrm{CoI}_p'(R^*)\). Combining, we get \(\lambda_d = \lambda \gamma\), and since \(\gamma = \mathrm{CoI}_p'(R^*)\), it follows that \(\lambda_d = \lambda \gamma\). To satisfy the global condition, we must have \(\lambda = \gamma\) (otherwise the marginal value equation would not hold for all \(R\)). Hence, \(\lambda = \gamma\) and \(\lambda_d = \lambda \mathrm{CoI}_p'(R^*) = \Phi_p'(R^*)\).
\end{proof}

% ============================================================
% Definition: Pragmatic Efficiency Bound
% ============================================================

\begin{definition}[Pragmatic Efficiency Bound (Behavioral Capacity)]
\label{def:pragmatic-efficiency-bound}
For a given resource shadow price \(\lambda \ge 0\), the \emph{pragmatic efficiency functional} \(\mathcal{E}_p(\lambda)\) is defined as the supreme net benefit achievable by the system:
\begin{equation}
    \mathcal{E}_p(\lambda) \triangleq \sup_{R \ge 0} \left[ \Phi_p(R) - \lambda \cdot \mathrm{CoI}_p(R) \right].
    \label{eq:pragmatic-efficiency-bound}
\end{equation}
This quantity represents the maximum net utility that an intelligent agent can extract from its environment per unit of time, after accounting for the physical cost of sensing, communication, and computation. It constitutes the \emph{behavioral capacity} of the system under resource pricing \(\lambda\), and serves as the master performance metric for any resource-constrained goal-directed system.
\end{definition}

\begin{remark}
\label{rem:sup-vs-max}
When the pragmatic rate \(R\) is continuous and the function \(\Phi_p(R) - \lambda \mathrm{CoI}_p(R)\) is strictly concave (which holds under the assumptions of Theorem~\ref{thm:global-lagrangian-concavity}), the supremum in \eqref{eq:pragmatic-efficiency-bound} is attained at a unique point \(R^*\) satisfying the marginal condition \(\Phi_p'(R^*) = \lambda \mathrm{CoI}_p'(R^*)\). Consequently, in all operational contexts, we may write
\[
\mathcal{E}_p(\lambda) = \max_{R \ge 0} \mathcal{L}_{\mathrm{global}}(R; \lambda),
\]
where the maximum is understood in the sense of the unique global maximizer. In the absence of resource constraints (\(\lambda = 0\)), \(\mathcal{E}_p(0)\) reduces to \(\sup_R \Phi_p(R)\), i.e., the system's performance under unlimited resources.
\end{remark}

% ============================================================
% The Lagrangian dual framework provides a principled method for resource allocation
% ============================================================

The Lagrangian dual framework provides a principled method for resource allocation in pragmatic systems:

\begin{enumerate}
    \item \textbf{Decision-driven design}: Given a target utility level, the down-loop Lagrangian determines the minimum rate required. The corresponding distortion \(D\) and rate \(R\) are obtained by solving \(\max_D \mathcal{L}_{\mathrm{SDV}}\).
    \item \textbf{Resource-limited design}: Given a resource budget \(P\), the up-loop Lagrangian determines the maximum achievable rate. The optimal rate is found by maximizing \(\mathcal{L}_{\mathrm{CPC}}\).
    \item \textbf{Global optimum}: When both utility and cost are considered, for a given shadow price \(\lambda\) (or equivalently, a resource budget constraint), the rate \(R^*\) from \eqref{eq:global-optimality} yields the maximal net benefit. This rate can be implemented by choosing the appropriate distortion level (via source coding) and resource allocation (via channel coding) such that \(R_p(D^*) = C_p(P^*) = R^*\).
\end{enumerate}

The framework also enables adaptive operation: as the channel quality or task requirements change, the system can re-optimize the Lagrangian to adjust the operating point.

In summary, the Lagrangian duality between the source-distortion-value and channel-power-cost Lagrangians provides a complete mathematical characterization of the fundamental trade-offs in pragmatic communication systems, unifying Shannon's physical-layer limits with the decision-theoretic value of information, while maintaining dimensional consistency through the explicit shadow price \(\lambda\).

\subsection{Single-Sided Pragmatic Cost of Information and Pragmatic Value of Information}
\label{subsec:single-sided-cost-value}

The single-sided pragmatic rate-distortion functions and achievable rates defined in Section~\ref{subsec:single-sided} have natural dual counterparts in terms of information value and cost. These single-sided VoI and CoI measures capture the utility gains and resource expenditures when only one side of the communication link is subject to pragmatic abstraction. They provide a finer-grained analysis for asymmetric systems, such as sensing-to-action (perceptual) and command-to-execution (expressive) scenarios.

\subsubsection{Perceptual Pragmatic Value of Information}

The perceptual pragmatic value of information quantifies the maximum expected utility gain obtainable from an observation \(V\) when the decision is based on the pragmatic action \(\underline{W}\). It is the dual of the perceptual pragmatic rate-distortion function \(R_p^{\text{perc}}(D)\), and is defined using the down perceptual pragmatic mutual information \(I_p(\underline{W}; V)\).

\begin{definition}[Perceptual Pragmatic Value of Information]
For a given rate constraint \(R\), the perceptual pragmatic value of information is defined as
\begin{equation}
    \mathrm{VoI}_p^{\text{perc}}(R) \triangleq \max_{P(V|\underline{W}): I_p(\underline{W}; V) \le R} \mathrm{VoI}_p,
    \label{eq:voi-perc-def}
\end{equation}
where \(\mathrm{VoI}_p\) is the pragmatic value defined in Eq.~\eqref{eq:prag-voi-def}, and the maximization is over all conditional distributions \(P(V|\underline{W})\) satisfying the rate constraint.
\end{definition}

This quantity represents the maximum utility gain achievable when the observation is constrained to carry at most \(R\) bits of information about the pragmatic action. It is non-decreasing and concave in \(R\), and satisfies the variational duality:
\begin{equation}
    \mathrm{VoI}_p^{\text{perc}}(R) = \max_{D \ge 0} \big\{ \Psi_p^{\text{perc}}(D) - \lambda R_p^{\text{perc}}(D) \big\},
\end{equation}
where \(\Psi_p^{\text{perc}}\) is a concave function determined by the utility, and \(\lambda > 0\) is the Lagrange multiplier.

\subsubsection{Expressive Pragmatic Value of Information}

The expressive pragmatic value of information quantifies the maximum expected utility gain obtainable from a pragmatic label \(\underline{V}\) when reconstructing the syntactic source \(W\). It is the dual of the expressive pragmatic rate-distortion function \(R_p^{\text{expr}}(D)\), and is defined using the down expressive pragmatic mutual information \(I_p(W; \underline{V})\).

\begin{definition}[Expressive Pragmatic Value of Information]
For a given rate constraint \(R\), the expressive pragmatic value of information is defined as
\begin{equation}
    \mathrm{VoI}_p^{\text{expr}}(R) \triangleq \max_{P(\underline{V}|W): I_p(W; \underline{V}) \le R} \mathrm{VoI}_p,
    \label{eq:voi-expr-def}
\end{equation}
where the maximization is over all conditional distributions \(P(\underline{V}|W)\) satisfying the rate constraint, and \(\mathrm{VoI}_p\) is evaluated with respect to the reconstructed source.
\end{definition}

This quantity measures the utility gain achieved by using the pragmatic label as a compressed representation of the source, subject to a rate limit. It is also non-decreasing and concave in \(R\), and satisfies the dual variational relationship with \(R_p^{\text{expr}}(D)\).

\subsubsection{Prospective Perceptual Pragmatic Cost of Information}

The prospective perceptual pragmatic cost of information quantifies the minimum resource expenditure required to achieve a given level of perceptual pragmatic information rate. It is the dual of the prospective perceptual pragmatic achievable rate \(C_p^{\text{p-perc}}(P)\), and is defined using the up perceptual pragmatic mutual information \(I^p(\underline{W}; V)\).

\begin{definition}[Prospective Perceptual Pragmatic Cost of Information]
For a given pragmatic rate \(R\), the prospective perceptual pragmatic cost of information is defined as
\begin{equation}
    \mathrm{CoI}_p^{\text{p-perc}}(R) \triangleq \min_{P(V|\underline{W}): I^p(\underline{W}; V) \ge R} \mathbb{E}[c(\underline{W})],
    \label{eq:coi-pper-def}
\end{equation}
where \(c(\underline{W})\) is the cost function on the pragmatic source (e.g., sensing power), and the minimization is over all conditional distributions \(P(V|\underline{W})\) satisfying the rate constraint.
\end{definition}

This quantity represents the minimum resource cost required to extract at least \(R\) bits of pragmatic information from the observation. It is non-decreasing and convex in \(R\), and is the strict inverse of the prospective perceptual achievable rate:
\begin{equation}
    \mathrm{CoI}_p^{\text{p-perc}}(R) = \bigl(C_p^{\text{p-perc}}\bigr)^{-1}(R).
\end{equation}

\subsubsection{Prospective Expressive Pragmatic Cost of Information}

The prospective expressive pragmatic cost of information quantifies the minimum resource expenditure required to achieve a given level of expressive pragmatic information rate. It is the dual of the prospective expressive pragmatic achievable rate \(C_p^{\text{p-expr}}(P)\), and is defined using the up expressive pragmatic mutual information \(I^p(W; \underline{V})\).

\begin{definition}[Prospective Expressive Pragmatic Cost of Information]
For a given pragmatic rate \(R\), the prospective expressive pragmatic cost of information is defined as
\begin{equation}
    \mathrm{CoI}_p^{\text{p-expr}}(R) \triangleq \min_{P(\underline{V}|W): I^p(W; \underline{V}) \ge R} \mathbb{E}[c(W)],
    \label{eq:coi-pexpr-def}
\end{equation}
where \(c(W)\) is the cost function on the source (e.g., transmit power), and the minimization is over all conditional distributions \(P(\underline{V}|W)\) satisfying the rate constraint.
\end{definition}

This quantity represents the minimum resource cost required to convey at least \(R\) bits of source information through the pragmatic label. It is non-decreasing and convex in \(R\), and is the strict inverse of the prospective expressive achievable rate:
\begin{equation}
    \mathrm{CoI}_p^{\text{p-expr}}(R) = \bigl(C_p^{\text{p-expr}}\bigr)^{-1}(R).
\end{equation}

\subsubsection{Relationships and Duality Structure}

The four single-sided VoI and CoI measures form a complete duality framework that parallels the bilateral case, with the following key relationships:

\begin{enumerate}
    \item \textbf{Rate-Distortion $\leftrightarrow$ Value Duality}: For both perceptual and expressive paths, the value function is the Legendre--Fenchel dual of the corresponding rate-distortion function:
    \begin{equation}
    \left\{
    \begin{aligned}
        \mathrm{VoI}_p^{\text{perc}}(R) &= \max_{D \ge 0} \big\{ \Psi_p^{\text{perc}}(D) - \lambda R_p^{\text{perc}}(D) \big\}, \\
        \mathrm{VoI}_p^{\text{expr}}(R) &= \max_{D \ge 0} \big\{ \Psi_p^{\text{expr}}(D) - \lambda R_p^{\text{expr}}(D) \big\}.
    \end{aligned}
    \right.
    \end{equation}

    \item \textbf{Capacity $\leftrightarrow$ Cost Duality}: For both prospective perceptual and expressive paths, the cost function is the inverse of the corresponding achievable rate:
    \begin{equation}
    \left\{
    \begin{aligned}
        \mathrm{CoI}_p^{\text{p-perc}}(R) &= \bigl(C_p^{\text{p-perc}}\bigr)^{-1}(R), \\
        \mathrm{CoI}_p^{\text{p-expr}}(R) &= \bigl(C_p^{\text{p-expr}}\bigr)^{-1}(R).
    \end{aligned}
    \right.
    \end{equation}

    \item \textbf{Comparison with Bilateral Counterparts}: Single-sided measures are bounded by their bilateral counterparts:
    \begin{equation}
    \left\{
    \begin{gathered}
    \mathrm{VoI}_p^{\text{perc}}(R) \le \mathrm{VoI}_p(R), \quad
    \mathrm{VoI}_p^{\text{expr}}(R) \le \mathrm{VoI}_p(R), \\
    \mathrm{CoI}_p^{\text{p-perc}}(R) \ge \mathrm{CoI}_p(R), \quad
    \mathrm{CoI}_p^{\text{p-expr}}(R) \ge \mathrm{CoI}_p(R),
    \end{gathered}
    \right.
    \end{equation}
    where $\mathrm{VoI}_p(R)$ and $\mathrm{CoI}_p(R)$ are the bilateral functions defined in Sections~VI.A and VI.B.
\end{enumerate}

\begin{table}[htbp]
\centering
\small
\caption{Summary of bilateral and single-sided pragmatic information measures.}
\label{tab:pragmatic-measures-summary}
\begin{tabular}{p{2.4cm} p{1.8cm} p{5.0cm} p{3.6cm}}
\hline
\textbf{Measure} & \textbf{Abbrev.} & \textbf{Definition / Expression} & \textbf{Scenario} \\
\hline
Bilateral VoI & $\mathrm{VoI}_p(R)$ & $\max\limits_{P(\underline{Y}|\underline{X}): I_p \le R} \mathrm{VoI}_p$ & Symmetric, both ends coarsened \\
Bilateral CoI & $\mathrm{CoI}_p(R)$ & $\min\limits_{p(x): I^p \ge R} \mathbb{E}[c(X)]$ & Symmetric, both ends coarsened \\
Perceptual VoI & $\mathrm{VoI}_p^{\text{perc}}(R)$ & $\max\limits_{P(V|\underline{W}): I_p(\underline{W};V) \le R} \mathrm{VoI}_p$ & Sensing: observation to action \\
Expressive VoI & $\mathrm{VoI}_p^{\text{expr}}(R)$ & $\max\limits_{P(\underline{V}|W): I_p(W;\underline{V}) \le R} \mathrm{VoI}_p$ & Actuation: command to reconstruction \\
Prospective Perceptual CoI & $\mathrm{CoI}_p^{\text{p-perc}}(R)$ & $\min\limits_{P(V|\underline{W}): I^p(\underline{W};V) \ge R} \mathbb{E}[c(\underline{W})]$ & Sensing with full syntactic observation \\
Prospective Expressive CoI & $\mathrm{CoI}_p^{\text{p-expr}}(R)$ & $\min\limits_{P(\underline{V}|W): I^p(W;\underline{V}) \ge R} \mathbb{E}[c(W)]$ & Actuation with full syntactic source \\
\hline
\end{tabular}
\end{table}

Table~\ref{tab:pragmatic-measures-summary} consolidates all bilateral and single-sided pragmatic information measures, their defining mutual informations, and their applicability in symmetric versus asymmetric systems. The perceptual path (sensing) coarsens only the source side, whereas the expressive path (actuation) coarsens only the reconstruction side. The prospective variants use the up mutual information and thus provide upper bounds on achievable rates.

\begin{figure}[htbp]
\setlength{\abovecaptionskip}{0.cm}
\setlength{\belowcaptionskip}{-0.cm}
  \centering{\includegraphics[scale=1]{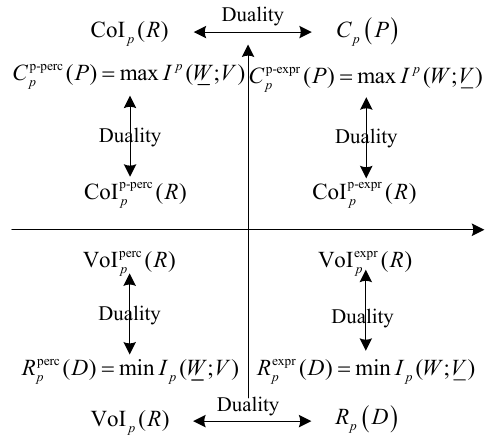}}
  \caption{Duality of pragmatic information measures.}\label{fig:duality-pragmatic}
\end{figure}

\sloppy Figure~\ref{fig:duality-pragmatic} illustrates the duality structure of pragmatic information measures, encompassing both bilateral and single-sided quantities. At the core lie the two fundamental dual pairs: $R_p(D) \leftrightarrow \mathrm{VoI}_p(R)$ and $C_p(P) \leftrightarrow \mathrm{CoI}_p(R)$. Extending this duality to asymmetric scenarios, $R_p^{\mathrm{perc}}(D) \leftrightarrow \mathrm{VoI}_p^{\mathrm{perc}}(R)$ and $R_p^{\mathrm{expr}}(D) \leftrightarrow \mathrm{VoI}_p^{\mathrm{expr}}(R)$ capture the rate--value trade-off when only one side is coarsened. On the capacity--cost side, $C_p^{\mathrm{p-perc}}(P) \leftrightarrow \mathrm{CoI}_p^{\mathrm{p-perc}}(R)$ and $C_p^{\mathrm{p-expr}}(P) \leftrightarrow \mathrm{CoI}_p^{\mathrm{p-expr}}(R)$ complete the duality lattice, which collectively characterizes trade-offs between rate, distortion, capacity, cost, and value in both symmetric and asymmetric pragmatic communication systems.

The single-sided VoI and CoI measures provide a precise toolkit for analyzing asymmetric pragmatic systems:

\begin{itemize}
    \item \textbf{Perceptual systems (sensing)}: $\mathrm{VoI}_p^{\text{perc}}$ measures the utility gain of a sensor under a rate constraint, while $\mathrm{CoI}_p^{\text{p-perc}}$ quantifies the sensing power required to achieve a certain perceptual information rate.
    \item \textbf{Expressive systems (actuation)}: $\mathrm{VoI}_p^{\text{expr}}$ measures the utility gain of a command channel under a rate constraint, while $\mathrm{CoI}_p^{\text{p-expr}}$ quantifies the transmit power required to convey a certain expressive information rate.
    \item \textbf{Joint optimization}: The single-sided measures can be combined with the bilateral Lagrangian framework to design systems where sensing and actuation have different resource constraints and utility contributions.
\end{itemize}

Together with the bilateral value and cost functions, these single-sided counterparts complete the pragmatic information-theoretic spectrum, enabling a comprehensive analysis of task-oriented communication systems with asymmetric constraints.

\section{Pragmatic Lossless Source Coding}
\label{section_VII}

In this section, we investigate the problem of pragmatic lossless source coding. We extend the asymptotic equipartition property (AEP) to the pragmatic domain and introduce the pragmatic typical set, the isoteleia typical set, and the reified typical set. These concepts serve as the mathematical foundation for proving the pragmatic lossless source coding theorem. 

\subsection{Asymptotic Equipartition Property and Pragmatic Typical Set}

Specifically, we establish the syntactic, semantic, and pragmatic AEP theorems. Based on these results, we define the pragmatic typical set, the isoteleia typical set, and the reified typical set, and analyze their cardinality and probability properties. The hierarchical structure of these typical sets reveals that pragmatic abstraction, by merging semantic distinctions that do not affect the optimal terminal action, further reduces the number of typical sequences compared to semantic abstraction, thereby enabling more efficient lossless compression at the pragmatic level.

To establish the asymptotic equipartition property at the pragmatic level, we first introduce the sequential extensions of the fundamental mappings that bridge the syntactic, semantic, and pragmatic layers.

Let \(\mathcal{W}\) be the syntactic alphabet with probability mass function \(P(W)\), and let \(\tilde{\mathcal{W}}\) be the semantic alphabet induced by the synonymous mapping \(f:\tilde{\mathcal{W}}\to 2^{\mathcal{W}}\). Let \(\underline{\mathcal{W}}\) be the pragmatic alphabet induced by the isoteleia mapping \(e:\underline{\mathcal{W}}\to 2^{\tilde{\mathcal{W}}}\) according to a given utility function \(U:\mathcal{W}\times\mathcal{A}\to\mathbb{R}\). The reification mapping \(g = f\circ e:\underline{\mathcal{W}}\to 2^{\mathcal{W}}\) directly bridges the pragmatic and syntactic layers.

\begin{definition}[Sequential Synonymous Mapping]
Let \(f:\tilde{\mathcal{W}}\to 2^{\mathcal{W}}\) be the synonymous mapping. Its \(n\)-th sequential extension $ f^n:\tilde{\mathcal{W}}^n \longrightarrow 2^{\mathcal{W}^n}$
is defined by
\begin{equation}
f^n(\tilde{w}^n) \triangleq \prod_{k=1}^n f(\tilde{w}_k) = \left\{w^n\in\mathcal{W}^n : w_k \in f(\tilde{w}_k),\; k=1,\dots,n\right\},
\end{equation}
where \(\tilde{w}^n = (\tilde{w}_1,\dots,\tilde{w}_n)\).
\end{definition}

\begin{definition}[Sequential Isoteleia Mapping]
Let \(e:\underline{\mathcal{W}}\to 2^{\tilde{\mathcal{W}}}\) be the isoteleia mapping. Its \(n\)-th sequential extension $e^n:\underline{\mathcal{W}}^n \longrightarrow 2^{\tilde{\mathcal{W}}^n}$
is defined by
\begin{equation}
e^n(\underline{w}^n) \triangleq \prod_{k=1}^n e(\underline{w}_k) = \left\{\tilde{w}^n\in\tilde{\mathcal{W}}^n : \tilde{w}_k \in e(\underline{w}_k),\; k=1,\dots,n\right\},
\end{equation}
where \(\underline{w}^n = (\underline{w}_1,\dots,\underline{w}_n)\).
\end{definition}

\begin{definition}[Sequential Reification Mapping]
Let \(g = f\circ e:\underline{\mathcal{W}}\to 2^{\mathcal{W}}\) be the reification mapping. Its \(n\)-th sequential extension$ g^n:\underline{\mathcal{W}}^n \longrightarrow 2^{\mathcal{W}^n}$
is defined as the composition of the sequential isoteleia and synonymous mappings:
\begin{equation}
g^n(\underline{w}^n) \triangleq (f^n\circ e^n)(\underline{w}^n) = \bigcup_{\tilde{w}^n\in e^n(\underline{w}^n)} f^n(\tilde{w}^n) = \prod_{k=1}^n g(\underline{w}_k).
\end{equation}
Equivalently,
\begin{equation}
g^n(\underline{w}^n) = \left\{w^n\in\mathcal{W}^n : w_k \in g(\underline{w}_k),\; k=1,\dots,n\right\}.
\end{equation}
\end{definition}

\subsubsection{Three-Tier AEP Theorems}

We now state the AEP at each of the three layers. The syntactic AEP is the classical result \cite{Classicpaper_Shannon,Book_Cover}. The semantic AEP was established in \cite{Paper_SIT,Book_SIT}. The pragmatic AEP follows analogously.

\begin{theorem}[Three-Tier AEP]\label{thm:three-tier-aep}
Let \((W_1,W_2,\dots)\) be an i.i.d. syntactic sequence drawn according to \(P(W)\), and let \((\tilde{W}_1,\tilde{W}_2,\dots)\) and \((\underline{W}_1,\underline{W}_2,\dots)\) be the associated semantic and pragmatic sequences under the sequential Synonymous Mapping \(f^n\) and the sequential Reification Mapping \(g^n = f^n \circ e^n\), respectively. Then the following convergences hold in probability:
\begin{equation}
\left\{
\begin{aligned}
-\frac{1}{n}\log P(W_1,\dots,W_n) &\longrightarrow H(W), \\
-\frac{1}{n}\log P(\tilde{W}_1,\dots,\tilde{W}_n) &\longrightarrow H_s(\tilde{W}), \\
-\frac{1}{n}\log P(\underline{W}_1,\dots,\underline{W}_n) &\longrightarrow H_p(\underline{W}),
\end{aligned}
\right.
\end{equation}
where
\[
P(W^n)=\prod_{k=1}^n P(W_k),\quad
P(\tilde{w}^n)=\prod_{k=1}^n P(\tilde{w}_k),\quad
P(\tilde{w}_k)=\sum_{w_k\in f(\tilde{w}_k)}P(w_k),
\]
and
\[
P(\underline{w}^n)=\prod_{k=1}^n P(\underline{w}_k),\quad
P(\underline{w}_k)=\sum_{w_k\in g(\underline{w}_k)}P(w_k).
\]
The syntactic convergence is the classical result \cite{Classicpaper_Shannon,Book_Cover}; the semantic and pragmatic convergences follow by applying the weak law of large numbers to the aggregated probabilities over synonymous and reified equivalence classes, respectively.
\end{theorem}

The following two AEP theorems characterize the asymptotic relationship between adjacent layers; they follow directly from the three AEP theorems above.

\begin{theorem}[Isoteleia AEP]
\begin{equation}
-\frac{1}{n}\left[\log P(\tilde{W}^n) - \log P(\underline{W}^n)\right] \longrightarrow H_s(\tilde{W}) - H_p(\underline{W})
\end{equation}
in probability.
\end{theorem}

\begin{theorem}[Reified AEP]
\begin{equation}
-\frac{1}{n}\left[\log P(W^n) - \log P(\underline{W}^n)\right] \longrightarrow H(W) - H_p(\underline{W})
\end{equation}
in probability.
\end{theorem}

\subsubsection{Definitions of the Hierarchical Typical Sets}

We now define the typical sets at each layer and the equivalence classes that partition them.

\begin{definition}[Syntactic, Semantic, and Pragmatic Typical Sets]
For $\epsilon>0$, the syntactically typical set $A_\epsilon^{(n)}$, the semantically typical set $\tilde{A}_\epsilon^{(n)}$, and the pragmatic typical set $\underline{A}_\epsilon^{(n)}$ are respectively defined as
\begin{equation}
\left\{
\begin{aligned}
A_\epsilon^{(n)} &\triangleq \left\{w^n\in\mathcal{W}^n : \left|-\frac{1}{n}\log P(w^n) - H(W)\right| < \epsilon\right\},\\
\tilde{A}_\epsilon^{(n)} &\triangleq \left\{\tilde{w}^n\in\tilde{\mathcal{W}}^n : \left|-\frac{1}{n}\log P(\tilde{w}^n) - H_s(\tilde{W})\right| < \epsilon\right\},\\
\underline{A}_\epsilon^{(n)} &\triangleq \left\{\underline{w}^n\in\underline{\mathcal{W}}^n : \left|-\frac{1}{n}\log P(\underline{w}^n) - H_p(\underline{W})\right| < \epsilon\right\}.
\end{aligned}
\right. \label{eq:three-tier-typical}
\end{equation}
\end{definition}

\begin{definition}[Synonymous Typical Set]
For a given semantic typical sequence \(\tilde{w}^n\in\tilde{A}_\epsilon^{(n)}\), the synonymous typical set \(S_\epsilon^{(n)}(\tilde{w}^n)\) is defined as the set of syntactic sequences that map to \(\tilde{w}^n\) under \(f^n\) and are syntactically typical:
\begin{align}
S_\epsilon^{(n)}(\tilde{w}^n) \triangleq \Bigg\{w^n\in\mathcal{W}^n :\;& w^n\in f^n(\tilde{w}^n), \nonumber \\
& \left|-\frac{1}{n}\log P(w^n) - H(W)\right| < \epsilon, \nonumber \\
& \left|-\frac{1}{n}\log \frac{P(w^n)}{P(\tilde{w}^n)} - \bigl(H(W)-H_s(\tilde{W})\bigr)\right| < \epsilon \Bigg\}.
\end{align}
The collection \(\{S_\epsilon^{(n)}(\tilde{w}^n)\}_{\tilde{w}^n\in\tilde{A}_\epsilon^{(n)}}\) forms a partition of \(A_\epsilon^{(n)}\):
\begin{equation}
A_\epsilon^{(n)} = \bigcup_{\tilde{w}^n\in\tilde{A}_\epsilon^{(n)}} S_\epsilon^{(n)}(\tilde{w}^n),
\end{equation}
with disjointness for distinct \(\tilde{w}^n\).
\end{definition}

\begin{definition}[Isoteleia Typical Set]
For a given pragmatic typical sequence \(\underline{w}^n\in\underline{A}_\epsilon^{(n)}\), the Isoteleia typical set \(E_\epsilon^{(n)}(\underline{w}^n)\) is defined as the set of semantic sequences that map to \(\underline{w}^n\) under \(e^n\) and are semantically typical:
\begin{align}
E_\epsilon^{(n)}(\underline{w}^n) \triangleq \Bigg\{\tilde{w}^n\in\tilde{\mathcal{W}}^n :\;& \tilde{w}^n\in e^n(\underline{w}^n), \nonumber \\
& \left|-\frac{1}{n}\log P(\tilde{w}^n) - H_s(\tilde{W})\right| < \epsilon, \nonumber \\
& \left|-\frac{1}{n}\log \frac{P(\tilde{w}^n)}{P(\underline{w}^n)} - \bigl(H_s(\tilde{W})-H_p(\underline{W})\bigr)\right| < \epsilon \Bigg\}.
\end{align}
The collection \(\{E_\epsilon^{(n)}(\underline{w}^n)\}_{\underline{w}^n\in\underline{A}_\epsilon^{(n)}}\) forms a partition of \(\tilde{A}_\epsilon^{(n)}\):
\begin{equation}
\tilde{A}_\epsilon^{(n)} = \bigcup_{\underline{w}^n\in\underline{A}_\epsilon^{(n)}} E_\epsilon^{(n)}(\underline{w}^n),
\end{equation}
with disjointness for distinct \(\underline{w}^n\).
\end{definition}

\begin{definition}[Reified Typical Set]
For a given pragmatic typical sequence \(\underline{w}^n\in\underline{A}_\epsilon^{(n)}\), the Reified typical set \(B_\epsilon^{(n)}(\underline{w}^n)\) is defined as the set of syntactic sequences that map to \(\underline{w}^n\) under the composite mapping \(g^n = f^n\circ e^n\) and are syntactically typical:
\begin{align}
B_\epsilon^{(n)}(\underline{w}^n) \triangleq \Bigg\{w^n\in\mathcal{W}^n :\;& w^n\in g^n(\underline{w}^n), \nonumber \\
& \left|-\frac{1}{n}\log P(w^n) - H(W)\right| < \epsilon, \nonumber \\
& \left|-\frac{1}{n}\log \frac{P(w^n)}{P(\underline{w}^n)} - \bigl(H(W)-H_p(\underline{W})\bigr)\right| < \epsilon \Bigg\}.
\end{align}
The collection \(\{B_\epsilon^{(n)}(\underline{w}^n)\}_{\underline{w}^n\in\underline{A}_\epsilon^{(n)}}\) forms a partition of \(A_\epsilon^{(n)}\):
\begin{equation}
A_\epsilon^{(n)} = \bigcup_{\underline{w}^n\in\underline{A}_\epsilon^{(n)}} B_\epsilon^{(n)}(\underline{w}^n),
\end{equation}
with disjointness for distinct \(\underline{w}^n\).
\end{definition}

The hierarchical relationships among the pragmatic, semantic, and syntactic typical sets, as well as their corresponding equivalence classes, are illustrated schematically in Fig.~\ref{Fig_Three_Layers_Typical_Sequence}. In this two-level hierarchical partition: the synonymous mapping first refines the syntactic typical set into synonym classes; the isoteleia mapping then groups these classes into coarser Reified classes, each corresponding to a single pragmatic action. This nested structure directly reflects the composition \(g^n = f^n \circ e^n\) and underpins the entropy reduction \(H_p(\underline{W}) \le H_s(\tilde{W}) \le H(W)\). 
Moreover, each reified typical set is the union of synonymous typical sets over all semantic sequences in the corresponding isoteleia typical set:
\[
B_\epsilon^{(n)}(\underline{w}^n) = \bigcup_{\tilde{w}^n\in E_\epsilon^{(n)}(\underline{w}^n)} S_\epsilon^{(n)}(\tilde{w}^n).
\]

\begin{figure}[htbp]
\setlength{\abovecaptionskip}{0.cm}
\setlength{\belowcaptionskip}{-0.cm}
  \centering{\includegraphics[scale=1]{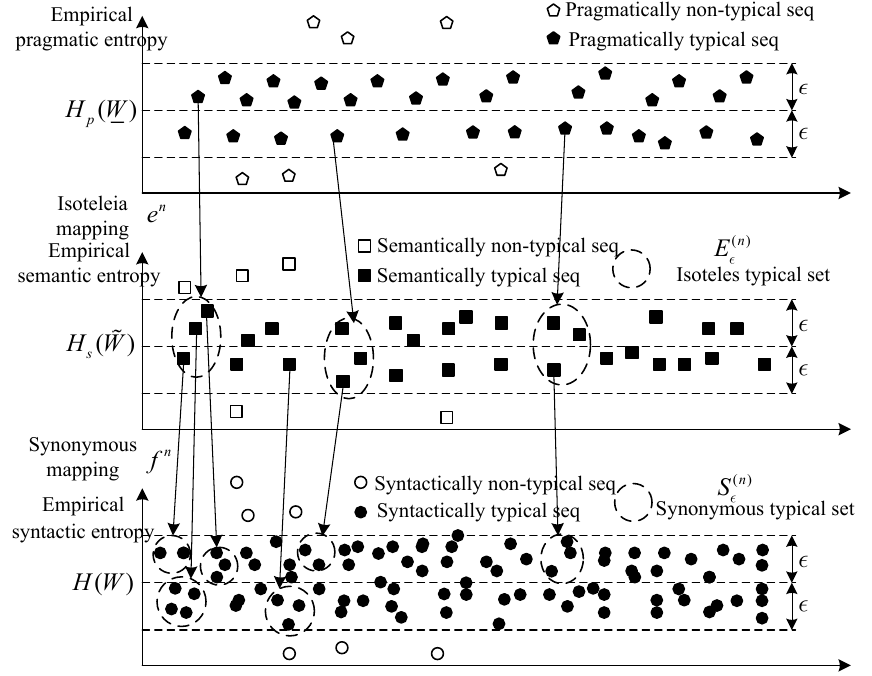}}
  \caption{Hierarchical structure of syntactic, semantic, and pragmatic typical sets.}\label{Fig_Three_Layers_Typical_Sequence}
\end{figure}

\subsubsection{Properties of the Typical Sets}

The syntactic and semantic typical sets satisfy the standard properties established in \cite{Classicpaper_Shannon,Book_Cover} and \cite{Paper_SIT,Book_SIT}, respectively. For completeness, we list them below without proof.

\begin{theorem}[Properties of the Syntactic Typical Set]
For sufficiently large \(n\):
\begin{enumerate}
\item For any \(w^n\in A_\epsilon^{(n)}\), 
\begin{equation}
2^{-n(H(W)+\epsilon)} \le P(w^n) \le 2^{-n(H(W)-\epsilon)}.
\end{equation}
\item \(\Pr\{A_\epsilon^{(n)}\} > 1-\epsilon\).
\item \((1-\epsilon)2^{n(H(W)-\epsilon)} \le |A_\epsilon^{(n)}| \le 2^{n(H(W)+\epsilon)}\).
\end{enumerate}
\end{theorem}

\begin{theorem}[Properties of the Semantic Typical Set]
For sufficiently large \(n\):
\begin{enumerate}
\item For any \(\tilde{w}^n\in\tilde{A}_\epsilon^{(n)}\),
\begin{equation}
2^{-n(H_s(\tilde{W})+\epsilon)} \le P(\tilde{w}^n) \le 2^{-n(H_s(\tilde{W})-\epsilon)}.
\end{equation}
\item \(\Pr\{\tilde{A}_\epsilon^{(n)}\} > 1-\epsilon\).
\item \((1-\epsilon)2^{n(H_s(\tilde{W})-\epsilon)} \le |\tilde{A}_\epsilon^{(n)}| \le 2^{n(H_s(\tilde{W})+\epsilon)}\).
\end{enumerate}
\end{theorem}

We now establish the properties of the pragmatic typical set.

\begin{theorem}[Properties of the Pragmatic Typical Set]\label{thm:pragmatic-typical-properties}
For sufficiently large \(n\):
\begin{enumerate}
\item For any \(\underline{w}^n\in\underline{A}_\epsilon^{(n)}\),
\begin{equation}
2^{-n(H_p(\underline{W})+\epsilon)} \le P(\underline{w}^n) \le 2^{-n(H_p(\underline{W})-\epsilon)}.
\end{equation}
\item \(\Pr\{\underline{A}_\epsilon^{(n)}\} > 1-\epsilon\).
\item \((1-\epsilon)2^{n(H_p(\underline{W})-\epsilon)} \le |\underline{A}_\epsilon^{(n)}| \le 2^{n(H_p(\underline{W})+\epsilon)}\).
\end{enumerate}
\end{theorem}

The proof of Theorem \ref{thm:pragmatic-typical-properties} follows the standard AEP counting argument. For the sake of readability, the detailed summation steps are provided in Appendix \ref{app:a1_pragmatic}.

Next, we state and prove the properties of the isoteleia typical set.

\begin{theorem}[Properties of the Isoteleia Typical Set]\label{thm:isoteleia-typical-properties}
For any \(\underline{w}^n\in\underline{A}_\epsilon^{(n)}\), for sufficiently large \(n\):
\begin{enumerate}
\item For any \(\tilde{w}^n\in E_\epsilon^{(n)}(\underline{w}^n)\),
\begin{equation}
2^{-n(H_s(\tilde{W})-H_p(\underline{W})+\epsilon)} \le \frac{P(\tilde{w}^n)}{P(\underline{w}^n)} \le 2^{-n(H_s(\tilde{W})-H_p(\underline{W})-\epsilon)}.
\end{equation}
\item 
\begin{equation}
2^{n(H_s(\tilde{W})-H_p(\underline{W})-\epsilon)} \le |E_\epsilon^{(n)}(\underline{w}^n)| \le 2^{n(H_s(\tilde{W})-H_p(\underline{W})+\epsilon)}.
\end{equation}
\end{enumerate}
\end{theorem}
The details of the proof are provided in Appendix \ref{app:a2_isoteleia}.

The synonymous typical set properties are known from \cite{Paper_SIT,Book_SIT}; we restate them for reference.

\begin{theorem}[Properties of the Synonymous Typical Set]
For any \(\tilde{w}^n\in\tilde{A}_\epsilon^{(n)}\), for sufficiently large \(n\):
\begin{enumerate}
\item For any \(w^n\in S_\epsilon^{(n)}(\tilde{w}^n)\),
\begin{equation}
2^{-n(H(W)-H_s(\tilde{W})+\epsilon)} \le \frac{P(w^n)}{P(\tilde{w}^n)} \le 2^{-n(H(W)-H_s(\tilde{W})-\epsilon)}.
\end{equation}
\item 
\begin{equation}
2^{n(H(W)-H_s(\tilde{W})-\epsilon)} \le |S_\epsilon^{(n)}(\tilde{w}^n)| \le 2^{n(H(W)-H_s(\tilde{W})+\epsilon)}.
\end{equation}
\end{enumerate}
\end{theorem}

Finally, we prove the properties of the Reified typical set.

\begin{theorem}[Properties of the Reified Typical Set]\label{theorem:Reified-Typical-Set}
For any \(\underline{w}^n\in\underline{A}_\epsilon^{(n)}\), for sufficiently large \(n\):
\begin{enumerate}
\item For any \(w^n\in B_\epsilon^{(n)}(\underline{w}^n)\),
\begin{equation}
2^{-n(H(W)-H_p(\underline{W})+\epsilon)} \le \frac{P(w^n)}{P(\underline{w}^n)} \le 2^{-n(H(W)-H_p(\underline{W})-\epsilon)}.
\end{equation}
\item 
\begin{equation}
2^{n(H(W)-H_p(\underline{W})-\epsilon)} \le |B_\epsilon^{(n)}(\underline{w}^n)| \le 2^{n(H(W)-H_p(\underline{W})+\epsilon)}.
\end{equation}
\end{enumerate}
\end{theorem}
The details of the proof are provided in Appendix \ref{app:a3_reified}.

\begin{corollary}[Hierarchy of Typical Set Sizes]
The cardinalities of the typical sets satisfy
\begin{equation}
|\underline{A}_\epsilon^{(n)}| \le |\tilde{A}_\epsilon^{(n)}| \le |A_\epsilon^{(n)}|,
\end{equation}
and the size of a Reified typical set is the product of the sizes of the corresponding Isoteleia and Synonymous typical sets:
\begin{equation}
|B_\epsilon^{(n)}(\underline{w}^n)| = |E_\epsilon^{(n)}(\underline{w}^n)| \cdot |S_\epsilon^{(n)}(\tilde{w}^n)|,
\end{equation}
for any \(\tilde{w}^n\in E_\epsilon^{(n)}(\underline{w}^n)\), reflecting the composition \(g^n = f^n\circ e^n\).
\end{corollary}

\begin{remark}
The hierarchy \(|\underline{A}_\epsilon^{(n)}| \le |\tilde{A}_\epsilon^{(n)}|\) reflects the fact that the isoteleia mapping merges semantic distinctions that do not affect the optimal action, thereby reducing the number of distinct pragmatic typical sequences. Similarly, the synonymous mapping reduces the number of semantic typical sequences relative to syntactic ones. The reified typical set \(B_\epsilon^{(n)}(\underline{w}^n)\) directly collects all syntactic sequences that realize a given pragmatic purpose, and its size is approximately \(2^{n(H(W)-H_p(\underline{W}))}\), which is larger than a synonymous typical set by a factor of \(2^{n(H_s(\tilde{W})-H_p(\underline{W}))}\). This additional freedom is the source of compression gains in pragmatic lossless source coding.
\end{remark}

\subsection{Pragmatic Lossless Source Coding Theorem}

We establish the fundamental limit of pragmatic lossless source coding using a three-tier architecture that refines symbols through meanings to actions, yet only the pragmatic index is transmitted. The code maps each pragmatic typical sequence to a unique codeword, allowing arbitrary syntactic and semantic realizations. Using random coding and typical set properties, we prove that pragmatic entropy is the limit: rates above it are achievable with vanishing error, while rates below fail. We also relate this to the down pragmatic value of information, showing that the full decision utility can only be obtained when the rate reaches the pragmatic entropy.

\subsubsection{Three-Tier Pragmatic Source Code Model}

We consider a discrete memoryless syntactic source \(W\) with alphabet \(\mathcal{W}\) and probability mass function \(P(W)\). Let \(\tilde{W}\) and \(\underline{W}\) be the associated semantic and pragmatic variables induced by the synonymous mapping \(f\) and the isoteleia mapping \(e\), respectively, with the reification mapping \(g = f\circ e\). The source generates an i.i.d. sequence \(W^n = (W_1,\dots,W_n)\) according to \(P(W^n)=\prod_{k=1}^n P(W_k)\). The semantic sequence \(\tilde{W}^n\) and pragmatic sequence \(\underline{W}^n\) are obtained by applying the sequential mappings \(f^n\) and \(g^n\) element-wise.

In the pragmatic lossless source coding scenario, the encoder may operate in three stages corresponding to the three layers of information, but in the order from pragmatic to semantic to syntactic: it first identifies the pragmatic purpose (optimal action class) from the source, then selects a semantic interpretation that justifies that action, and finally chooses a syntactic realization that conveys that meaning. However, since the reification mapping \(g^n = f^n \circ e^n\) directly gives the pragmatic class of any syntactic sequence, the encoder can equivalently apply \(g^n\) directly and then encode the resulting pragmatic index. The decoder, upon receiving the index, outputs an arbitrary syntactic sequence from the corresponding reified typical set \(B_\epsilon^{(n)}(\underline{w}^n)\), thereby recovering the pragmatic purpose without needing to reconstruct the exact syntactic or semantic details.

\begin{definition}[Pragmatic Lossless Source Code]
An \((M,n)\) pragmatic lossless source code consists of:
\begin{enumerate}
\item A pragmatic index set \(\mathcal{I}_p = \{1,2,\dots,M_p\}\), where each index \(i_p \in \mathcal{I}_p\) corresponds to a distinct pragmatic purpose (optimal action class). Additionally, a reification index set \(\mathcal{I}_r = \{1,2,\dots,M_r\}\), where each index \(i_r \in \mathcal{I}_r\) enumerates the distinct syntactic realizations within each reified typical set. The total number of syntactic sequences that can be represented is \(M = M_p \cdot M_r\).
\item An encoding function that operates in three stages:
  \begin{enumerate}
  \item \textbf{Pragmatic-to-semantic mapping (isoteleia):} For a given pragmatic index \(i_p\), the encoder first selects a semantic interpretation from the corresponding isoteleia typical set \(E_\epsilon^{(n)}(\underline{w}_{i_p})\), where the index \(j\) enumerates the semantic sequences within that set.
  \item \textbf{Semantic-to-syntactic mapping (synonymous):} For the chosen semantic sequence \(\tilde{w}^j\), the encoder then chooses a syntactic realization from the synonymous typical set \(S_\epsilon^{(n)}(\tilde{w}^j)\), where the index \(k\) enumerates the syntactic sequences within that set.
  \item \textbf{Reification index selection:} The encoder finally selects a reification index \(i_r \in \mathcal{I}_r\), which corresponds to a specific syntactic sequence \(w^n\) from the reified typical set \(B_\epsilon^{(n)}(\underline{w}_{i_p})\). The reified typical set is the union of all synonymous typical sets corresponding to semantic sequences in the isoteleia typical set for that pragmatic index.
  \end{enumerate}
  Equivalently, the encoding function \(\phi: \mathcal{W}^n \to \mathcal{I}_p \times \mathcal{I}_r\) maps each syntactic sequence \(w^n\) to a pair of indices \((i_p, i_r)\) such that \(w^n \in B_\epsilon^{(n)}(\underline{w}_{i_p})\) and \(i_r\) identifies the specific syntactic sequence within that reified typical set. Pragmatically, only the pragmatic index \(i_p\) needs to be transmitted to preserve the action; the reification index \(i_r\) is optional and determines syntactic fidelity.
\item A decoding function \(\psi: \mathcal{I}_p \to \mathcal{W}^n\) that, upon receiving an index \(i_p\), outputs a representative syntactic sequence \(\hat{w}^n\) belonging to the reified typical set \(B_\epsilon^{(n)}(\underline{w}_{i_p})\). If the reification index \(i_r\) is also transmitted, the decoder can output the exact syntactic sequence identified by \(i_r\); otherwise, it may choose any syntactic sequence from the corresponding synonymous typical sets, as all share the same pragmatic meaning.
\end{enumerate}
The pragmatic code rate is defined as \(R = \frac{1}{n}\log_2 M_p\) (prabits per source symbol), and the reification rate is defined as \(R_r = \frac{1}{n}\log_2 M_r\) (bits per source symbol). The total syntactic code rate is therefore \(R_{\text{syn}} = R + R_r = \frac{1}{n}\log_2 M\). The error probability for a given code is
\begin{equation}
P_e^{(n)} = \Pr\{\psi(\phi(W^n)) \notin B_\epsilon^{(n)}(\underline{W}^n)\},
\end{equation}
i.e., the probability that the decoded syntactic sequence does not belong to the correct reified typical set, and hence the pragmatic purpose is not preserved.
\end{definition}

This coding architecture is illustrated in Fig.~\ref{fig:pragmatic_source_coding}. The encoder first applies the Reification Mapping \(g^n\) to determine the pragmatic class of the observed syntactic sequence, then encodes that class index. The decoder recovers the index and outputs an arbitrary syntactic sequence from the corresponding reified typical set, which indicates that the decoder's output may differ syntactically from the input but preserves the pragmatic meaning. The intermediate semantic layer is implicit in the composition \(g^n = f^n \circ e^n\), and the system can be viewed as a cascade of synonymous and isoteleia mappings.

\begin{figure*}[htbp]
\setlength{\abovecaptionskip}{0.cm}
\setlength{\belowcaptionskip}{-0.cm}
  \centering{\includegraphics[scale=1]{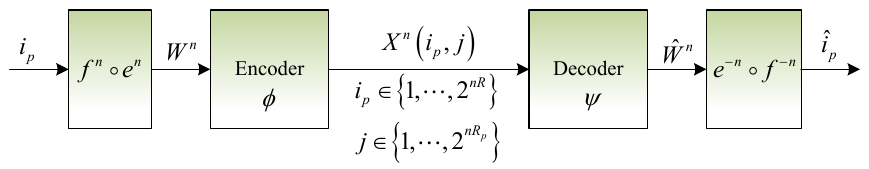}}\caption{Block diagram of pragmatic lossless source coding. The encoder applies the reification mapping $g^n$ map the pragmatic class to its syntactic sequence and encodes the index. The decoder outputs a representative syntactic sequence from the corresponding reified typical set. }
\label{fig:pragmatic_source_coding}
\end{figure*}

\begin{theorem}[Pragmatic Lossless Source Coding Theorem]\label{theorem:Pragmatic-Lossless-Source-Coding-Theorem}
Let \(W\) be a discrete memoryless syntactic source with associated pragmatic variable \(\underline{W}\) having pragmatic entropy \(H_p(\underline{W})\). For any \(\epsilon>0\) and pragmatic code rate \(R\):
\begin{enumerate}
\item \textbf{Achievability:} If \(R > H_p(\underline{W})\), then for sufficiently large \(n\), there exists an \((2^{n(R+R_p)}, n)\) pragmatic lossless source code with error probability \(P_e^{(n)} < \epsilon\), where \(R_p\) is the rate of the reified typical set (i.e., the number of syntactic realizations per pragmatic class).
\item \textbf{Converse:} If \(R < H_p(\underline{W})\), then for any \((2^{n(R+R_p)}, n)\) pragmatic lossless source code, the error probability \(P_e^{(n)}\) tends to \(1\) as \(n\to\infty\).
\end{enumerate}
Moreover, the pragmatic Value of Information (VoI) achieved by the compressed index is maximized (i.e., equals the VoI obtained from the full pragmatic variable \(\underline{W}\)) if and only if \(R \ge H_p(\underline{W})\).
\end{theorem}

\begin{proof}
\textbf{Achievability:} Fix \(\epsilon>0\) and choose \(R > H_p(\underline{W})\). By the properties of the pragmatic typical set (Theorem~\ref{thm:pragmatic-typical-properties}), the pragmatic typical set \(\underline{A}_\epsilon^{(n)}\) satisfies
\[
(1-\epsilon)2^{n(H_p(\underline{W})-\epsilon)} \le |\underline{A}_\epsilon^{(n)}| \le 2^{n(H_p(\underline{W})+\epsilon)}.
\]
Since \(R > H_p(\underline{W})\), we can select \(n\) large enough such that \(2^{nR} \ge |\underline{A}_\epsilon^{(n)}|\). We construct a one-to-one mapping from each pragmatic typical sequence \(\underline{w}^n\in\underline{A}_\epsilon^{(n)}\) to a distinct pragmatic index \(i_p\in\{1,\dots,2^{nR}\}\). For any syntactic sequence \(w^n\), the encoder computes its pragmatic class \(\underline{w}^n\). If \(\underline{w}^n\in\underline{A}_\epsilon^{(n)}\), it sends the assigned index \(i_p\); otherwise, it sends a default index (e.g., 1). 

Each pragmatic index \(i_p\) corresponds to a reified typical set \(B_\epsilon^{(n)}(\underline{w}_{i_p})\) of size approximately \(2^{nR_p}\), where \(R_p = H(W) - H_p(\underline{W})\) is the maximum rate of the reified typical set. The decoder receives the index \(i_p\) and outputs an arbitrary syntactic sequence from the corresponding reified typical set \(B_\epsilon^{(n)}(\underline{w}_{i_p})\) (which is non-empty by definition). This ensures that the pragmatic purpose is correctly recovered whenever \(\underline{W}^n \in \underline{A}_\epsilon^{(n)}\). By Theorem~\ref{thm:pragmatic-typical-properties}, \(\Pr\{\underline{W}^n\in\underline{A}_\epsilon^{(n)}\} > 1-\epsilon\), so the error probability is bounded by \(\epsilon\). Hence, for any \(R > H_p(\underline{W})\), there exists a code with arbitrarily small error probability.

The total syntactic code rate is \(R_{syn} = R + R_p\). When \(R_p = 0\) (i.e., each pragmatic class is represented by a single syntactic sequence), the total rate approaches \(H_p(\underline{W})\). When \(R_p\) increases to its maximum value \(H(W) - H_p(\underline{W})\), the total rate approaches the syntactic entropy \(H(W)\), recovering the classical lossless source coding limit. The intermediate values of \(R_p\) correspond to different levels of syntactic redundancy preservation within each pragmatic equivalence class.

\textbf{Converse:} Assume a code with pragmatic rate \(R < H_p(\underline{W})\) and error probability \(P_e^{(n)}\). Let \(\mathcal{I}_p\) be the set of pragmatic indices, with \(|\mathcal{I}_p| = 2^{nR}\). Define the decoding sets \(\mathcal{D}_{i_p} = \{w^n: \phi(w^n)=i_p\}\), and their corresponding pragmatic classes \(\underline{\mathcal{D}}_{i_p} = \{\underline{w}^n: g^n(w^n)=\underline{w}^n \text{ for some } w^n\in\mathcal{D}_{i_p}\}\). Since each decoding set is assigned to a single pragmatic index, the decoder can correctly recover the pragmatic purpose only if the true \(\underline{W}^n\) belongs to \(\underline{\mathcal{D}}_{i_p}\). The total number of distinct pragmatic classes that can be correctly decoded is at most \(|\mathcal{I}_p| = 2^{nR}\). To achieve small error probability, these classes must cover most of the probability mass of \(\underline{W}^n\). By the properties of the pragmatic typical set, the number of pragmatic typical sequences is approximately \(2^{nH_p(\underline{W})}\). If \(R < H_p(\underline{W})\), then \(2^{nR} < (1-\epsilon)2^{n(H_p(\underline{W})-\epsilon)}\) for large \(n\), so at most a fraction of the pragmatic typical sequences can be assigned distinct indices. Consequently, for any code, the probability that the true pragmatic sequence falls into an unassigned class or is mapped to a wrong index is at least \(1 - 2^{-n(H_p(\underline{W})-R-\epsilon)}\), which tends to 1 as \(n\to\infty\). Thus, \(P_e^{(n)} \to 1\). This proves the converse.

\textbf{Connection to Pragmatic VoI:} The pragmatic Value of Information achieved by the compressed index \(I_p\) is defined as
\[
\mathrm{VoI}_p(I_p) = \mathbb{E}_{W,I_p}\big[U(W, a^*(\underline{W}(I_p)))\big] - U_0,
\]
where \(a^*(\underline{W}(I_p))\) is the optimal action based on the recovered pragmatic class, and \(U_0\) is the baseline utility without any information. Since \(I_p\) is a function of the encoded index, it carries at most \(nR\) bits of information about the source. By the data processing inequality, the mutual information between the pragmatic variable \(\underline{W}\) and \(I_p\) is bounded by
\[
I(\underline{W}; I_p) \le H(I_p) \le R.
\]
To achieve the maximum pragmatic value \(\mathrm{VoI}_p^{\max}\)—i.e., the same utility as if the true pragmatic variable \(\underline{W}\) were directly observed—the index \(I_p\) must preserve all pragmatically relevant distinctions. This requires that \(\underline{W}\) be recoverable from \(I_p\), which in turn necessitates \(I(\underline{W}; I_p) = H_p(\underline{W})\). However, this equality can hold only if \(R \ge H_p(\underline{W})\). If \(R < H_p(\underline{W})\), then \(I(\underline{W}; I_p) \le R < H_p(\underline{W})\), implying that some uncertainty about the optimal action class remains unresolved. Consequently, the expected utility conditioned on \(I_p\) is strictly less than the utility obtained under perfect pragmatic knowledge, and the achievable VoI is strictly below \(\mathrm{VoI}_p^{\max}\). Conversely, when \(R \ge H_p(\underline{W})\), the encoder can simply transmit a lossless description of the pragmatic index, and the decoder recovers \(\underline{W}\) perfectly, thereby attaining the full pragmatic VoI. Thus, the minimum rate for full pragmatic value extraction is exactly \(H_p(\underline{W})\).
\end{proof}

\begin{remark}
The pragmatic lossless source coding theorem generalizes Shannon's classical theorem. The hierarchy \(H_p(\underline{W}) \le H_s(\tilde{W}) \le H(W)\) quantifies the compression gains from pragmatic abstraction. The reified typical set \(B_\epsilon^{(n)}(\underline{w}^n)\) admits a two-level decomposition via the composition \(g^n = f^n \circ e^n\): it partitions into isoteleia typical sets \(E_\epsilon^{(n)}(\underline{w}^n)\) (grouping semantic justifications that lead to the same action) and, within each, synonymous typical sets \(S_\epsilon^{(n)}(\tilde{w}^n)\) (grouping syntactic realizations of the same meaning). This structure underpins the additional rate \(R_p\), which offers a flexible trade-off: setting \(R_p = 0\) achieves the pure pragmatic limit, while increasing \(R_p\) preserves syntactic details within each pragmatic class, recovering the classical limit as \(R_p \to H(W) - H_p(\underline{W})\).
\end{remark}

The optimal encoding strategy compresses directly to the pragmatic index, discarding all syntactic and semantic distinctions that do not affect the optimal action — consistent with the design principles of pragmatic information systems (Section~II). The pragmatic VoI provides a direct performance metric: for a given rate \(R\), the achievable utility is maximized when the code preserves pragmatic equivalence classes, establishing a principled trade-off between communication rate and task performance in applications such as autonomous driving, remote control, and industrial automation.

\section{Pragmatic Channel Coding}
\label{section_VIII}

In this section, we study pragmatic information transmission over noisy channels. Extending the joint AEP via the joint reification mapping, we define the jointly reified typical set that groups syntactic input-output pairs by their common pragmatic purpose. Using random coding and joint typical decoding, we prove the pragmatic channel coding theorem: the pragmatic capacity equals the maximum achievable rate and extends the classical capacity by tolerating errors that do not affect optimal actions.

\subsection{Jointly Asymptotic Equipartition Property and Jointly Reified Typical Set}
\label{subsection:JAEP-JRTS}
To establish the fundamental limits of pragmatic channel coding, we first extend the asymptotic equipartition property to the joint distribution of channel input and output sequences at the pragmatic level. Unlike the semantic approach, which relies on the joint synonymous mapping, we directly employ the joint reification mapping, which is the composition of the joint synonymous and joint isoteleia mappings, to aggregate syntactic sequence pairs into pragmatic equivalence classes. This approach provides a direct characterization of the equivalence of input-output pairs with respect to the optimal terminal actions.

Consider a discrete memoryless channel with input alphabet \(\mathcal{X}\), output alphabet \(\mathcal{Y}\), and transition probability \(P(Y|X)\). Let \(\tilde{\mathcal{X}}\) and \(\tilde{\mathcal{Y}}\) be the associated semantic alphabets, and let \(\underline{\mathcal{X}}\) and \(\underline{\mathcal{Y}}\) be the pragmatic alphabets induced by the isoteleia mappings \(e_X:\underline{\mathcal{X}}\to 2^{\tilde{\mathcal{X}}}\) and \(e_Y:\underline{\mathcal{Y}}\to 2^{\tilde{\mathcal{Y}}}\), respectively. The synonymous mappings \(f_X:\tilde{\mathcal{X}}\to 2^{\mathcal{X}}\) and \(f_Y:\tilde{\mathcal{Y}}\to 2^{\mathcal{Y}}\) map semantic symbols to sets of syntactic symbols. The reification mappings \(g_X = f_X\circ e_X:\underline{\mathcal{X}}\to 2^{\mathcal{X}}\) and \(g_Y = f_Y\circ e_Y:\underline{\mathcal{Y}}\to 2^{\mathcal{Y}}\) directly bridge the pragmatic and syntactic layers.

For the joint behavior of channel input and output, we define the joint reification mapping, the joint synonymous mapping, and the joint isoteleia mapping. Their sequential extensions are defined element-wise.

\begin{definition}[Joint Mappings for Sequences]
Let \(f_{XY}:\tilde{\mathcal{X}}\times\tilde{\mathcal{Y}}\to 2^{\mathcal{X}\times\mathcal{Y}}\) be the joint synonymous mapping, and let \(e_{XY}:\underline{\mathcal{X}}\times\underline{\mathcal{Y}}\to 2^{\tilde{\mathcal{X}}\times\tilde{\mathcal{Y}}}\) be the joint isoteleia mapping. The joint reification mapping \(g_{XY} = f_{XY}\circ e_{XY}:\underline{\mathcal{X}}\times\underline{\mathcal{Y}}\to 2^{\mathcal{X}\times\mathcal{Y}}\) is their composition. Their \(n\)-th sequential extensions are defined by:
\begin{equation}
\left\{
\begin{aligned}
f_{XY}^n(\tilde{x}^n,\tilde{y}^n) &\triangleq \prod_{k=1}^n f_{XY}(\tilde{x}_k,\tilde{y}_k), \\
e_{XY}^n(\underline{x}^n,\underline{y}^n) &\triangleq \prod_{k=1}^n e_{XY}(\underline{x}_k,\underline{y}_k), \\
g_{XY}^n(\underline{x}^n,\underline{y}^n) &\triangleq (f_{XY}^n\circ e_{XY}^n)(\underline{x}^n,\underline{y}^n) = \prod_{k=1}^n g_{XY}(\underline{x}_k,\underline{y}_k).
\end{aligned}\right.
\end{equation}
Thus, \(g_{XY}^n\) partitions the joint syntactic space \(\mathcal{X}^n\times\mathcal{Y}^n\) into equivalence classes according to the pragmatic pair \((\underline{x}^n,\underline{y}^n)\).
\end{definition}

We now define the relevant typical sets. Let \(P(X,Y)=P(X)P(Y|X)\) denote the joint distribution induced by the channel. The syntactically jointly typical set \(A_\epsilon^{(n)}\) is defined as in classical information theory \cite{Classicpaper_Shannon,Book_Cover}:
\begin{equation}
\begin{aligned}
A_\epsilon^{(n)} \triangleq \Big\{(x^n,y^n)\in\mathcal{X}^n\times\mathcal{Y}^n :\;& 
\left|-\frac{1}{n}\log P(x^n) - H(X)\right|<\epsilon,\\
& \left|-\frac{1}{n}\log P(y^n) - H(Y)\right|<\epsilon,\\
& \left|-\frac{1}{n}\log P(x^n,y^n) - H(X,Y)\right|<\epsilon \Big\}.
\end{aligned}
\end{equation}

For the pragmatic layer, we define the pragmatically jointly typical set, which aggregates probabilities over reified equivalence classes.

\begin{definition}[Pragmatically Jointly Typical Set]
The pragmatically jointly typical set \(\underline{A}_\epsilon^{(n)}\) is the set of pragmatic sequence pairs \((\underline{x}^n,\underline{y}^n)\in\underline{\mathcal{X}}^n\times\underline{\mathcal{Y}}^n\) such that
\begin{equation}
\begin{aligned}
\underline{A}_\epsilon^{(n)} \triangleq \Bigg\{(\underline{x}^n,\underline{y}^n) :\;& \left|-\frac{1}{n}\log P(\underline{x}^n) - H_p(\underline{X})\right|<\epsilon,\\
& \left|-\frac{1}{n}\log P(\underline{y}^n) - H_p(\underline{Y})\right|<\epsilon,\\
& \left|-\frac{1}{n}\log P(\underline{x}^n,\underline{y}^n) - H_p(\underline{X},\underline{Y})\right|<\epsilon\Bigg\},
\end{aligned}
\end{equation}
where
\begin{equation}
P(\underline{x}^n,\underline{y}^n) = \prod_{k=1}^n P(\underline{x}_k,\underline{y}_k), \quad 
P(\underline{x}_k,\underline{y}_k) = \sum_{(x_k,y_k)\in g_{XY}(\underline{x}_k,\underline{y}_k)} P(x_k,y_k).
\end{equation}
The marginal pragmatic probabilities \(P(\underline{x}^n)\) and \(P(\underline{y}^n)\) are obtained by summing over the other variable.
\end{definition}

For a given pragmatic pair, the jointly reified typical set collects all syntactic pairs that share that pragmatic interpretation and are syntactically jointly typical. This set is naturally decomposed into synonymous typical sets corresponding to each semantic interpretation in the isoteleia typical set.

\begin{definition}[Jointly Reified Typical Set]\label{def:Jointly-Reified-Typical-Set}
For a given pragmatic typical pair $(\underline{x}^n,\underline{y}^n)\in\underline{A}_\epsilon^{(n)}$, the jointly reified typical set $B_\epsilon^{(n)}(\underline{x}^n,\underline{y}^n)$ with the syntactically jointly typical sequences $(x^n,y^n)$ is defined as the set of $n$-sequence pairs $(x^n,y^n)\in\mathcal{X}^n\times\mathcal{Y}^n$ such that the following conditions hold:
\begin{equation}
\begin{aligned}
B_\epsilon^{(n)}(\underline{x}^n,\underline{y}^n)
= \Bigg\{(x^n,y^n)\in\mathcal{X}^n\times\mathcal{Y}^n :\;&
\left|-\frac{1}{n}\log P(x^n)-H(X)\right|<\epsilon,\\
&
\left|-\frac{1}{n}\log P(y^n)-H(Y)\right|<\epsilon,\\
&
\left|-\frac{1}{n}\log P(x^n,y^n)-H(X,Y)\right|<\epsilon,\\
&
\left|-\frac{1}{n}\log P(\underline{x}^n)-H_p(\underline{X})\right|<\epsilon,\\
&
\left|-\frac{1}{n}\log P(\underline{y}^n)-H_p(\underline{Y})\right|<\epsilon,\\
&
\left|-\frac{1}{n}\log P(\underline{x}^n,\underline{y}^n)-H_p(\underline{X},\underline{Y})\right|<\epsilon,\\
&
\left|-\frac{1}{n}\log P((\underline{x}^n,\underline{y}^n)\to(x^n,y^n))\right.\\
&\qquad\left.-\bigl(H(X,Y)-H_p(\underline{X},\underline{Y})\bigr)\right|<\epsilon,
\Bigg\},
\end{aligned}
\end{equation}
where
\begin{equation}
P((\underline{x}^n,\underline{y}^n)\to(x^n,y^n)) \triangleq
\begin{cases}
\dfrac{P(x^n,y^n)}{P(\underline{x}^n,\underline{y}^n)}, & \text{if } (x^n,y^n)\in g_{XY}^n(\underline{x}^n,\underline{y}^n),\\[6pt]
0, & \text{otherwise}.
\end{cases}
\end{equation}
The first three conditions ensure syntactic joint typicality of the sequence pair $(x^n,y^n)$. The fourth through sixth conditions ensure pragmatic joint typicality of the sequence pair $(\underline{x}^n,\underline{y}^n)$. The seventh condition enforces the reification typicality, i.e., the difference between the syntactic joint entropy and the pragmatic joint entropy is $\epsilon$-close to the entropy reduction $H(X,Y)-H_p(\underline{X},\underline{Y})$.

Equivalently, the jointly reified typical set can be expressed as the union of jointly synonymous typical sets over all semantic pairs in the corresponding isoteleia typical set:
\begin{equation}
B_\epsilon^{(n)}(\underline{x}^n,\underline{y}^n) = \bigcup_{(\tilde{x}^n,\tilde{y}^n)\in E_\epsilon^{(n)}(\underline{x}^n,\underline{y}^n)} S_\epsilon^{(n)}(\tilde{x}^n,\tilde{y}^n),
\end{equation}
where $E_\epsilon^{(n)}(\underline{x}^n,\underline{y}^n)$ denotes the isoteleia typical set, and $S_\epsilon^{(n)}(\tilde{x}^n,\tilde{y}^n)$ denotes the jointly synonymous typical set. Thus, the jointly reified typical set partitions the syntactically jointly typical set according to the equivalence classes induced by the composition of the joint isoteleia and joint synonymous mappings, i.e., the joint reification mapping.
\end{definition}

Under the sequential joint reification mapping \(g_{XY}^n\), the syntactically jointly typical set \(A_\epsilon^{(n)}\) can be partitioned into jointly reified typical sets:
\begin{equation}
A_\epsilon^{(n)} = \bigcup_{(\underline{x}^n,\underline{y}^n)\in\underline{A}_\epsilon^{(n)}} B_\epsilon^{(n)}(\underline{x}^n,\underline{y}^n),
\end{equation}
with disjointness for distinct pragmatic typical pairs.

Fig.~\ref{fig:joint_typical_set_mapping_for_channel_coding} illustrates the hierarchical relationship among the syntactically jointly typical set $A_\epsilon^{(n)}$, the pragmatically jointly typical set $\underline{A}_\epsilon^{(n)}$, and the jointly reified typical sets $B_\epsilon^{(n)}(\underline{x}^n,\underline{y}^n)$. The joint reification mapping $g_{XY}^n = f_{XY}^n \circ e_{XY}^n$ maps each syntactic sequence pair $(x^n,y^n)$ to its corresponding pragmatic pair $(\underline{x}^n,\underline{y}^n)$. Under the joint reification mapping, $A_\epsilon^{(n)}$ is partitioned into disjoint jointly reified typical sets, each corresponding to a distinct pragmatic typical pair $(\underline{x}^n,\underline{y}^n)$ in $\underline{A}_\epsilon^{(n)}$. Each jointly reified typical set $B_\epsilon^{(n)}(\underline{x}^n,\underline{y}^n)$ is itself the union of jointly synonymous typical sets $S_\epsilon^{(n)}(\tilde{x}^n,\tilde{y}^n)$ over all semantic pairs in the isoteleia typical set $E_\epsilon^{(n)}(\underline{x}^n,\underline{y}^n)$, reflecting the composition $g_{XY}^n = f_{XY}^n \circ e_{XY}^n$. The dashed circles within each reified set represent the individual synonymous typical sets, which group syntactic pairs sharing the same semantic interpretation. This nested structure visually demonstrates that the pragmatic abstraction merges semantically distinct but pragmatically equivalent pairs, thereby reducing the number of equivalence classes from $|A_\epsilon^{(n)}|$ to $|\underline{A}_\epsilon^{(n)}|$, consistent with the entropy reduction $H_p(\underline{X},\underline{Y}) \le H_s(\tilde{X},\tilde{Y}) \le H(X,Y)$ established in Section~\ref{section_III}.

\begin{figure*}[htbp]
\setlength{\abovecaptionskip}{0.cm}
\setlength{\belowcaptionskip}{-0.cm}
  \centering{\includegraphics[scale=0.9]{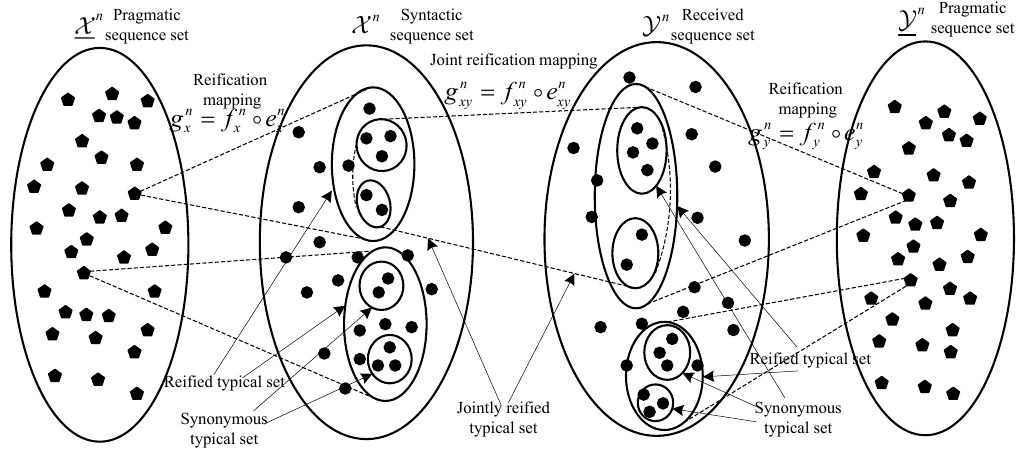}}
\caption{The joint reification mapping \(g_{XY}^n\) maps pragmatic classes \((\underline{x}^n,\underline{y}^n)\) to syntactic pairs \((x^n,y^n)\). Each pragmatic typical pair corresponds to a jointly reified typical set \(B_\epsilon^{(n)}(\underline{x}^n,\underline{y}^n)\), which is a union of synonymous typical sets over the isoteleia typical set.}
\label{fig:joint_typical_set_mapping_for_channel_coding}
\end{figure*}

We now establish the AEP properties of these typical sets. The syntactically joint AEP is the classical result \cite{Classicpaper_Shannon,Book_Cover}; we restate it for reference.

\begin{theorem}[Syntactically Joint AEP]
Let \((X^n,Y^n)\) be a sequence pair of length \(n\) drawn i.i.d. according to \(P(x^n,y^n)\). Then for sufficiently large \(n\):
\begin{enumerate}
\item \(\Pr\{(X^n,Y^n)\in A_\epsilon^{(n)}\} > 1-\epsilon\).
\item \((1-\epsilon)2^{n(H(X,Y)-\epsilon)} \le |A_\epsilon^{(n)}| \le 2^{n(H(X,Y)+\epsilon)}\).
\item If \((\dot{X}^n,\dot{Y}^n)\sim P(x^n)P(y^n)\) are independent sequences with the same marginals, then
\[
(1-\epsilon)2^{-n(I(X;Y)+3\epsilon)} \le \Pr\{(\dot{X}^n,\dot{Y}^n)\in A_\epsilon^{(n)}\} \le 2^{-n(I(X;Y)-3\epsilon)}.
\]
\end{enumerate}
\end{theorem}

The pragmatically joint AEP follows by applying the weak law of large numbers to the aggregated pragmatic probabilities.

\begin{theorem}[Pragmatically Joint AEP for Transmission]\label{theorem:Pragmatically-Joint-AEP-Trans}
Let \((\underline{X}^n,\underline{Y}^n)\) be a pragmatic sequence pair of length \(n\) drawn i.i.d. according to \(P(\underline{x}^n,\underline{y}^n)\), with the associated syntactic sequence pair \((X^n,Y^n)\) under the joint reification mapping \(g_{XY}^n\). Then for sufficiently large \(n\):
\begin{enumerate}
\item \(\Pr\{(\underline{X}^n,\underline{Y}^n)\in\underline{A}_\epsilon^{(n)}\} > 1-\epsilon\).
\item For any \((\underline{x}^n,\underline{y}^n)\in\underline{A}_\epsilon^{(n)}\),
\[
2^{-n(H_p(\underline{X},\underline{Y})+\epsilon)} \le P(\underline{x}^n,\underline{y}^n) \le 2^{-n(H_p(\underline{X},\underline{Y})-\epsilon)}.
\]
\item \((1-\epsilon)2^{n(H_p(\underline{X},\underline{Y})-\epsilon)} \le |\underline{A}_\epsilon^{(n)}| \le 2^{n(H_p(\underline{X},\underline{Y})+\epsilon)}\).
\item If \((\dot{X}^n,\dot{Y}^n)\sim P(x^n)P(y^n)\) are independent sequences with the same marginals as \(X^n\) and \(Y^n\), and \((\dot{\underline{X}}^n,\dot{\underline{Y}}^n)\) are the corresponding pragmatic sequences, then
\begin{equation}
(1-\epsilon)2^{-n(I^p(\underline{X};\underline{Y})+3\epsilon)} \le \Pr\{(\dot{\underline{X}}^n,\dot{\underline{Y}}^n)\in\underline{A}_\epsilon^{(n)}\} \le 2^{-n(I^p(\underline{X};\underline{Y})-3\epsilon)}.
\end{equation}
\end{enumerate}
\end{theorem}
The details of the proof are provided in Appendix \ref{app:a4_Pragmatic-Joint-AEP-Trans}.

Finally, we establish the properties of the jointly reified typical set and its decomposition.

\begin{theorem}[Properties of the Jointly Reified Typical Set]\label{theorem:Jointly-Reified-Typical-Set}
For any pragmatic typical pair \((\underline{x}^n,\underline{y}^n)\in\underline{A}_\epsilon^{(n)}\), for sufficiently large \(n\):
\begin{enumerate}
\item For any \((x^n,y^n)\in B_\epsilon^{(n)}(\underline{x}^n,\underline{y}^n)\),
\[
2^{-n(H(X,Y)-H_p(\underline{X},\underline{Y})+\epsilon)} \le \frac{P(x^n,y^n)}{P(\underline{x}^n,\underline{y}^n)} \le 2^{-n(H(X,Y)-H_p(\underline{X},\underline{Y})-\epsilon)}.
\]
\item 
\[
2^{n(H(X,Y)-H_p(\underline{X},\underline{Y})-\epsilon)} \le |B_\epsilon^{(n)}(\underline{x}^n,\underline{y}^n)| \le 2^{n(H(X,Y)-H_p(\underline{X},\underline{Y})+\epsilon)}.
\]
\end{enumerate}
\end{theorem}
The details of the proof are provided in Appendix \ref{app:a5_Jointly-Reified-Typical-Set}.

\begin{corollary}[Decomposition of the Jointly Reified Typical Set]
The jointly reified typical set can be decomposed as the union of synonymous typical sets over the isoteleia typical set:
\[
B_\epsilon^{(n)}(\underline{x}^n,\underline{y}^n) = \bigcup_{(\tilde{x}^n,\tilde{y}^n)\in E_\epsilon^{(n)}(\underline{x}^n,\underline{y}^n)} S_\epsilon^{(n)}(\tilde{x}^n,\tilde{y}^n).
\]
Consequently, its cardinality satisfies
\begin{equation}
|B_\epsilon^{(n)}(\underline{x}^n,\underline{y}^n)| \approx |E_\epsilon^{(n)}(\underline{x}^n,\underline{y}^n)| \cdot |S_\epsilon^{(n)}(\tilde{x}^n,\tilde{y}^n)|,
\end{equation}
where the approximation holds for any \(\tilde{x}^n,\tilde{y}^n\) in the isoteleia typical set, and the product reflects the composition \(g_{XY}^n = f_{XY}^n \circ e_{XY}^n\).
\end{corollary}

\begin{remark}
The jointly reified typical set \(B_\epsilon^{(n)}(\underline{x}^n,\underline{y}^n)\) groups together all syntactically jointly typical pairs that share the same pragmatic interpretation. Its decomposition into synonymous and isoteleia typical sets reveals the two-level abstraction: the isoteleia mapping merges semantic differences that do not affect the optimal actions, while the synonymous mapping merges syntactic differences that do not affect the meaning. The cardinality of a jointly reified typical set is approximately \(2^{n(H(X,Y)-H_p(\underline{X},\underline{Y}))}\), which is larger than a single synonymous typical set by a factor of \(2^{n(H_s(\tilde{X},\tilde{Y})-H_p(\underline{X},\underline{Y}))}\), reflecting the additional syntactic freedom gained when only the pragmatic purpose needs to be preserved. When randomly selecting an independent pair of syntactic sequences, the probability that they fall into the same jointly reified typical set (i.e., are pragmatically jointly typical) is about \(2^{-nI^p(\underline{X};\underline{Y})}\), which is not larger than the classical joint typicality probability \(2^{-nI(X;Y)}\), reflecting the pragmatic capacity gain.
\end{remark}

\subsection{Pragmatic Channel Coding Theorem}

We establish the fundamental limit for pragmatic transmission over noisy channels using a two-level code: the encoder selects a pragmatic index (optimal action) and then a syntactic codeword from its reified typical set; the decoder recovers only the index, tolerating syntactic errors. Via random coding and joint typical decoding based on the jointly reified typical set, we prove that the pragmatic channel capacity is the maximum achievable rate. We also connect this result to the pragmatic cost of information.

\begin{definition}[Pragmatic Channel Code]
An \((M_p, M_r, n)\) pragmatic channel code consists of:
\begin{enumerate}
\item A pragmatic index set \(\mathcal{I}_p = \{1,2,\dots,M_p\}\), where each index \(i_p\in\mathcal{I}_p\) corresponds to a distinct pragmatic purpose (optimal action class).
\item A reification index set \(\mathcal{I}_r = \{1,2,\dots,M_r\}\), where each index \(i_r\in\mathcal{I}_r\) enumerates the syntactic codewords within each reified typical set.
\item An encoding function \(\phi: \mathcal{I}_p \times \mathcal{I}_r \to \mathcal{X}^n\) that maps each pair \((i_p, i_r)\) to a codeword \(x^n(i_p, i_r)\) such that for each pragmatic index \(i_p\), the set of codewords \(\{x^n(i_p, i_r): i_r\in\mathcal{I}_r\}\) lies entirely within the corresponding reified typical set \(B_\epsilon^{(n)}(\underline{x}_{i_p}^n, \underline{y}_{i_p}^n)\) for some typical output pragmatic sequence \(\underline{y}_{i_p}^n\). (The decoder will recover \(i_p\) without needing the exact \(i_r\).)
\item A decoding function \(\psi: \mathcal{Y}^n \to \mathcal{I}_p\) that maps the received sequence to a pragmatic index.
\end{enumerate}
The pragmatic code rate is \(R_p = \frac{1}{n}\log_2 M_p\) (prabits per channel use), and the reification rate is \(R_r = \frac{1}{n}\log_2 M_r\) (bits per channel use). The total syntactic rate is
\begin{equation}
R_{\text{tot}} = R_p + R_r = \frac{1}{n}\log_2 (M_p M_r). \label{eq:tot-rate}
\end{equation}
For a given code, the error probability is defined as
\begin{equation}
P_e^{(n)} = \Pr\{\psi(Y^n) \neq i_p\}, \label{eq:chan-error}
\end{equation}
where \(i_p\) is the intended pragmatic index, and \(Y^n\) is the channel output when the corresponding codeword is transmitted.
\end{definition}

\begin{figure}[htbp]
\setlength{\abovecaptionskip}{0.cm}
\setlength{\belowcaptionskip}{-0.cm}
\centering{\includegraphics[scale=1]{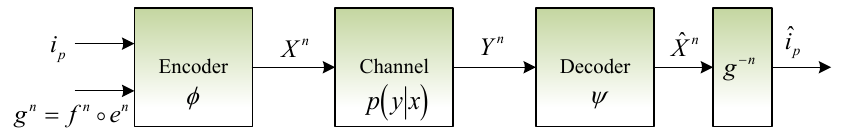}}
\caption{Block diagram of pragmatic channel coding. The encoder selects a pragmatic index \(i_p\) and a reification index \(i_r\), then transmits the corresponding codeword. The decoder recovers only the pragmatic index \(\hat{i}_p\), ignoring syntactic variations within the reified typical set.}
\label{fig:pragmatic_channel_coding}
\end{figure}

Figure \ref{fig:pragmatic_channel_coding} depicts the process of pragmatic channel coding. The encoder first selects a pragmatic index \(i_p\in\mathcal{I}_p\) corresponding to the intended optimal action class. Given \(i_p\), the encoder then chooses a reification index \(i_r\in\mathcal{I}_r\), which selects a specific syntactic codeword from the reified typical set \(B_\epsilon^{(n)}(\underline{x}_{i_p}^n,\underline{y}_{i_p}^n)\). The transmitted codeword \(x^n(i_p,i_r)\) is sent over the channel. The decoder recovers only the pragmatic index \(\hat{i}_p\) via jointly typical decoding based on the jointly reified typical set, ignoring syntactic variations within the reified typical set. Thus, the two-level structure separates the pragmatic purpose from its syntactic realizations, enabling reliable transmission at rates up to the pragmatic channel capacity.

\begin{theorem}[Pragmatic Channel Coding Theorem]\label{theorem:Pragmatic-Channel-Coding-Theorem}
Let \(C_p\) be the pragmatic channel capacity defined in Eq.~(\ref{eq:pragmatic-capacity}) of Section~\ref{subsec:pragmatic-capacity}, i.e.,
\begin{equation}
C_p = \max_{p(x)}\max_{g_{XY}} I^p(\underline{X};\underline{Y}),
\end{equation}
where \(I^p(\underline{X};\underline{Y}) = H(X) + H(Y) - H_p(\underline{X},\underline{Y})\) is the up pragmatic mutual information. For any \(\epsilon>0\) and total rate \(R_{\text{tot}}\):
\begin{enumerate}
\item \textbf{Achievability:} If \(R_{\text{tot}} < C_p\), then for sufficiently large \(n\), there exists an \((2^{n(R_p+R_r)}, n)\) pragmatic channel code with error probability \(P_e^{(n)} < \epsilon\).
\item \textbf{Converse:} If \(R_{\text{tot}} > C_p\), then for any \((2^{n(R_p+R_r)}, n)\) pragmatic channel code, the error probability \(P_e^{(n)}\) is bounded away from zero for sufficiently large \(n\).
\end{enumerate}
Consequently, the pragmatic cost of information \(\mathrm{CoI}_p(R)\) is finite for \(R \le C_p\) and becomes unbounded for \(R > C_p\), establishing \(C_p\) as the fundamental threshold for resource-limited pragmatic communication.
\end{theorem}

\begin{proof}
\textbf{Achievability:} Fix \(\epsilon>0\) and choose rates \(R_p, R_r\) such that \(R_{\text{tot}} = R_p + R_r < C_p\). By the definition of \(C_p\), there exists an input distribution \(p(x)\) and a joint reification mapping \(g_{XY}\) such that \(I^p(\underline{X};\underline{Y}) > R_{\text{tot}}\). Let \(R = I^p(\underline{X};\underline{Y}) - \delta\) for some small \(\delta>0\), so that \(R_{\text{tot}} < R < C_p\).

We construct a random codebook \(\mathcal{C}\) of size \(2^{n(R_p+R_r)}\) by drawing each codeword independently according to \(P(X^n)=\prod_{k=1}^n P(X_k)\). These codewords are then partitioned uniformly into \(2^{nR_p}\) groups, each corresponding to a pragmatic index \(i_p\), with \(2^{nR_r}\) codewords per group. This partition is based on the reification mapping: all codewords in group \(i_p\) are chosen from the reified typical set \(B_\epsilon^{(n)}(\underline{x}_{i_p}^n, \underline{y}_{i_p}^n)\) for some typical output pragmatic sequence.

The encoder, upon receiving a pragmatic message \(i_p\), selects one of the \(2^{nR_r}\) codewords in group \(i_p\) (any choice is sufficient, as all share the same pragmatic purpose) and transmits it. The decoder uses the jointly typical decoding rule based on the jointly reified typical set: upon receiving \(Y^n\), it searches for the unique pragmatic index \(i_p\) such that there exists a codeword \(x^n\) in group \(i_p\) with \((x^n, Y^n)\in B_\epsilon^{(n)}(\underline{x}_{i_p}^n, \underline{y}_{i_p}^n)\), i.e., the pair is jointly reified typical. If no such index exists, or if multiple exist, an error is declared.

We now analyze the error probability. Assume codeword \(x^n(1,1)\) (from group \(i_p=1\)) is transmitted. Define the following events:
\begin{itemize}
\item \(E_0\): The transmitted codeword and the received sequence are not jointly reified typical. By the jointly reified typical set properties (Theorem~\ref{theorem:Pragmatically-Joint-AEP-Trans}), the probability of this event tends to 0 as \(n\to\infty\), i.e., \(P(E_0)\le \epsilon/2\) for large \(n\).
\item \(E_i\) for \(i_p\ne 1\): There exists some codeword in a wrong pragmatic group that is jointly reified typical with the received sequence. For a fixed wrong group \(i_p\), the number of codewords is \(2^{nR_r}\), and each is independent of the received sequence. By the probability bound in Theorem~\ref{theorem:Pragmatically-Joint-AEP-Trans}, the probability that a given codeword from a wrong group is jointly typical with \(Y^n\) is at most \(2^{-n(I^p(\underline{X};\underline{Y})-3\epsilon)}\). Hence, by the union bound,
\[
P(E_i) \le 2^{nR_r} \cdot 2^{-n(I^p(\underline{X};\underline{Y})-3\epsilon)} = 2^{-n(I^p(\underline{X};\underline{Y}) - R_r - 3\epsilon)}.
\]
\end{itemize}

Summing over all \(2^{nR_p}-1\) wrong groups gives
\begin{equation}
    \begin{aligned}
    P\left(\bigcup_{i_p\ne 1} E_i\right) & \le 2^{nR_p} \cdot 2^{-n(I^p(\underline{X};\underline{Y}) - R_r - 3\epsilon)} \\ &= 2^{-n(I^p(\underline{X};\underline{Y}) - R_p - R_r - 3\epsilon)}\\ &= 2^{-n(I^p(\underline{X};\underline{Y}) - R_{\text{tot}} - 3\epsilon)}.
    \end{aligned} \label{eq:chan-error-bound}
\end{equation}
Since \(R_{\text{tot}} < I^p(\underline{X};\underline{Y})\), this probability tends to 0 as \(n\to\infty\). Therefore, the total error probability \(P_e^{(n)} \le P(E_0) + \sum_{i_p\ne 1}P(E_i) < \epsilon\) for sufficiently large \(n\). This proves the achievability.

\textbf{Converse:} Assume a code with total rate \(R_{\text{tot}} > C_p\) and error probability \(P_e^{(n)}\). Let \(W_p\) be the uniformly distributed pragmatic message, and let \(\hat{W}_p\) be the decoder's estimate. By Fano's inequality,
\[
H(W_p|\hat{W}_p) \le 1 + P_e^{(n)} \log M_p = 1 + P_e^{(n)} n R_p.
\]
Using the data processing inequality and the fact that \(\underline{X}^n\) is a function of \(W_p\) and the codebook, we have
\[
n R_p = H(W_p) \le I(W_p; Y^n) + H(W_p|Y^n) \le I(\underline{X}^n; Y^n) + 1 + P_e^{(n)} n R_p.
\]
Now, by the definition of pragmatic capacity and the chain rule for up pragmatic mutual information (Theorem~\ref{theorem:Chain-Rule-Up-Pragmatic-MI}), we have \(I(\underline{X}^n; Y^n) \le I^p(\underline{X}^n; \underline{Y}^n) \le n C_p\). Thus,
\[
n R_p \le n C_p + 1 + P_e^{(n)} n R_p \implies R_p (1 - P_e^{(n)}) \le C_p + \frac{1}{n}.
\]
Since the total rate \(R_{\text{tot}} = R_p + R_r\), and \(R_r\ge 0\), we get
\begin{equation}
R_{\text{tot}} - C_p \le P_e^{(n)} R_p + \frac{1}{n} + R_r. \label{eq:chan-converse}
\end{equation}
Letting \(n\to\infty\) and assuming \(P_e^{(n)}\to 0\), we obtain \(R_{\text{tot}} \le C_p\), a contradiction. Hence, if \(R_{\text{tot}} > C_p\), the error probability cannot tend to zero. This proves the converse.

\textbf{Connection to Pragmatic Cost of Information:} The pragmatic cost of information \(\mathrm{CoI}_p(R)\) is the minimum resource required to achieve rate \(R\). Since reliable transmission is possible iff \(R \le C_p\), \(\mathrm{CoI}_p(R)\) is finite for \(R \le C_p\) and increases to infinity as \(R \to C_p\). Thus, \(C_p\) is the threshold beyond which no finite resource can guarantee reliable pragmatic communication.
\end{proof}

\begin{remark}
The pragmatic channel coding theorem generalizes Shannon's classical theorem (when isoteleia and synonymous mappings are trivial, \(C_p = C\)) and the semantic theorem (when only the synonymous mapping is nontrivial, \(C_p = C_s\)). By tolerating errors that do not affect the optimal action, pragmatic coding achieves a larger capacity. The two-level code structure separates the pragmatic purpose from its syntactic realizations, allowing the encoder to exploit the redundancy in syntactic representations while ensuring that the decoder recovers the essential pragmatic information. The jointly reified typical set serves as the foundation for both encoding and decoding, enabling reliable communication at rates up to \(C_p\). The theorem also establishes a direct link between pragmatic capacity and the pragmatic cost of information, highlighting the resource trade-offs inherent in task-oriented communication.
\end{remark}

\section{Pragmatic Lossy Source Coding}
\label{section_IX}

In this section, we investigate the fundamental limits of pragmatic lossy source coding, where the goal is to compress a source while preserving its pragmatic utility rather than its exact syntactic or semantic content. We extend the asymptotic equipartition property to the joint distribution of the source and its reconstruction at the pragmatic level. By introducing the reification mapping for both the source and the reconstruction sequences, we define the jointly reified typical set, which characterizes the equivalence classes of source-reconstruction pairs that share the same pragmatic interpretation. Using random coding and jointly typical encoding based on this typical set, we prove the pragmatic rate-distortion coding theorem, which establishes that the pragmatic rate-distortion function \(R_p(D)\) is the fundamental limit for task-oriented compression. 

\subsection{Pragmatic Distortion and Jointly Typical Set}

To establish the fundamental limits of pragmatic lossy source coding, we first extend the asymptotic equipartition property to the joint distribution of the source sequence and its reconstruction at the pragmatic level. Unlike the semantic approach, which relies on the joint synonymous mapping, we directly employ the reification mapping, which is the composition of the synonymous and isoteleia mappings, to aggregate syntactic sequences into pragmatic equivalence classes. This approach provides a direct characterization of the equivalence of source and reconstruction pairs with respect to the optimal terminal actions.

Consider a discrete memoryless source \(X\) with alphabet \(\mathcal{X}\) and probability mass function \(P(X)\). Let \(\tilde{X}\) be the associated semantic variable and \(\underline{X}\) be the pragmatic variable induced by the isoteleia mapping \(e_X:\underline{\mathcal{X}}\to 2^{\tilde{\mathcal{X}}}\) and the synonymous mapping \(f_X:\tilde{\mathcal{X}}\to 2^{\mathcal{X}}\), with the reification mapping \(g_X = f_X\circ e_X:\underline{\mathcal{X}}\to 2^{\mathcal{X}}\). Similarly, let \(\hat{X}\) be the reconstruction variable with associated semantic \(\tilde{\hat{X}}\) and pragmatic \(\underline{\hat{X}}\), induced by the reification mapping \(g_{\hat{X}}:\underline{\hat{\mathcal{X}}}\to 2^{\hat{\mathcal{X}}}\).

For lossy source coding, the encoder produces a reconstruction sequence \(\hat{X}^n\) based on a test channel \(P(\hat{X}^n|X^n)\). The pragmatic distortion measure \(d_p(\underline{x}, \hat{\underline{x}})\) quantifies the loss of utility when a pragmatic symbol \(\underline{x}\) is represented by \(\hat{\underline{x}}\). In practice, it can be defined as the utility loss:
\begin{equation}
d_p(\underline{x},\hat{\underline{x}}) = \max_{a}\mathbb{E}[U(X,a)|\underline{x}] - \max_{a}\mathbb{E}[U(X,a)|\hat{\underline{x}}],
\end{equation}
or any other non-negative function that reflects the degradation in decision quality. The average pragmatic distortion is then
\begin{equation}
\bar{d}_p = \mathbb{E}\left[d_p(\underline{X},\underline{\hat{X}})\right] = \sum_{\underline{x},\hat{\underline{x}}} P(\underline{x},\hat{\underline{x}}) d_p(\underline{x},\hat{\underline{x}}).
\end{equation}

We now define the relevant typical sets. Let \(P(X,\hat{X}) = P(X)P(\hat{X}|X)\) denote the joint distribution induced by the test channel. The syntactically jointly typical set \(A_\epsilon^{(n)}\) is defined as in classical information theory \cite{Classicpaper_Shannon,Book_Cover}:
\begin{equation}
\begin{aligned}
A_\epsilon^{(n)} \triangleq \Big\{(x^n,\hat{x}^n)\in\mathcal{X}^n\times\hat{\mathcal{X}}^n :\;&
\left|-\frac{1}{n}\log P(x^n)-H(X)\right|<\epsilon,\\
& \left|-\frac{1}{n}\log P(\hat{x}^n)-H(\hat{X})\right|<\epsilon,\\
& \left|-\frac{1}{n}\log P(x^n,\hat{x}^n)-H(X,\hat{X})\right|<\epsilon\Big\}.
\end{aligned}
\end{equation}

For the pragmatic layer, we define the pragmatically jointly typical set, which aggregates probabilities over reified equivalence classes.

\begin{definition}[Pragmatically Jointly Typical Set for Source and Reconstruction]
The pragmatically jointly typical set \(\underline{A}_\epsilon^{(n)}\) is the set of pragmatic sequence pairs \((\underline{x}^n,\underline{\hat{x}}^n)\in\underline{\mathcal{X}}^n\times\underline{\hat{\mathcal{X}}}^n\) such that
\begin{equation}
\begin{aligned}
\underline{A}_\epsilon^{(n)} \triangleq \Big\{(\underline{x}^n,\underline{\hat{x}}^n):\;&
\left|-\frac{1}{n}\log P(\underline{x}^n)-H_p(\underline{X})\right|<\epsilon,\\
& \left|-\frac{1}{n}\log P(\underline{\hat{x}}^n)-H_p(\underline{\hat{X}})\right|<\epsilon,\\
& \left|-\frac{1}{n}\log P(\underline{x}^n,\underline{\hat{x}}^n)-H_p(\underline{X},\underline{\hat{X}})\right|<\epsilon\Big\},
\end{aligned}
\end{equation}
where
\begin{equation}
P(\underline{x}^n,\underline{\hat{x}}^n) = \prod_{k=1}^n P(\underline{x}_k,\underline{\hat{x}}_k), \quad
P(\underline{x}_k,\underline{\hat{x}}_k) = \sum_{(x_k,\hat{x}_k)\in g_{X\hat{X}}(\underline{x}_k,\underline{\hat{x}}_k)} P(x_k,\hat{x}_k),
\end{equation}
with \(g_{X\hat{X}} = g_X \times g_{\hat{X}}\) being the joint reification mapping for the source and reconstruction. The marginal pragmatic probabilities are obtained by summing over the other variable.
\end{definition}

For a given pragmatic pair, the jointly reified typical set collects all syntactic source reconstruction pairs that share that pragmatic interpretation and are syntactically jointly typical.

\begin{definition}[Jointly Reified Typical Set for Source and Reconstruction]
For a given pragmatic typical pair \((\underline{x}^n,\underline{\hat{x}}^n)\in\underline{A}_\epsilon^{(n)}\), the jointly reified typical set \(B_\epsilon^{(n)}(\underline{x}^n,\underline{\hat{x}}^n)\) is defined as the set of syntactic pairs \((x^n,\hat{x}^n)\in\mathcal{X}^n\times\hat{\mathcal{X}}^n\) such that the following conditions hold:
\begin{equation}
\begin{aligned}
B_\epsilon^{(n)}(\underline{x}^n,\underline{\hat{x}}^n) = \Big\{(x^n,\hat{x}^n):\;&
\left|-\frac{1}{n}\log P(x^n)-H(X)\right|<\epsilon,\\
& \left|-\frac{1}{n}\log P(\hat{x}^n)-H(\hat{X})\right|<\epsilon,\\
& \left|-\frac{1}{n}\log P(x^n,\hat{x}^n)-H(X,\hat{X})\right|<\epsilon,\\
& \left|-\frac{1}{n}\log \frac{P(x^n)}{P(\underline{x}^n)}-\bigl(H(X)-H_p(\underline{X})\bigr)\right|<\epsilon,\\
& \left|-\frac{1}{n}\log \frac{P(\hat{x}^n)}{P(\underline{\hat{x}}^n)}-\bigl(H(\hat{X})-H_p(\underline{\hat{X}})\bigr)\right|<\epsilon,\\
& \left|-\frac{1}{n}\log \frac{P(x^n,\hat{x}^n)}{P(\underline{x}^n,\underline{\hat{x}}^n)}-\bigl(H(X,\hat{X})-H_p(\underline{X},\underline{\hat{X}})\bigr)\right|<\epsilon,\\
& (x^n,\hat{x}^n)\in g_{X\hat{X}}^n(\underline{x}^n,\underline{\hat{x}}^n)
\Big\}.
\end{aligned}
\end{equation}
The first three conditions ensure syntactic joint typicality. The fourth and fifth conditions enforce individual reification typicality for the source and reconstruction sequences. The sixth condition enforces joint reification typicality. The last condition guarantees that the syntactic pair maps to the given pragmatic pair under the joint reification mapping.
\end{definition}

Under the sequential joint reification mapping \(g_{X\hat{X}}^n\), the syntactically jointly typical set \(A_\epsilon^{(n)}\) can be partitioned into jointly reified typical sets:
\begin{equation}
A_\epsilon^{(n)} = \bigcup_{(\underline{x}^n,\underline{\hat{x}}^n)\in\underline{A}_\epsilon^{(n)}} B_\epsilon^{(n)}(\underline{x}^n,\underline{\hat{x}}^n),
\end{equation}
with disjointness for distinct pragmatic typical pairs.

\begin{figure*}[htbp]
\setlength{\abovecaptionskip}{0.cm}
\setlength{\belowcaptionskip}{-0.cm}
  \centering{\includegraphics[scale=0.9]{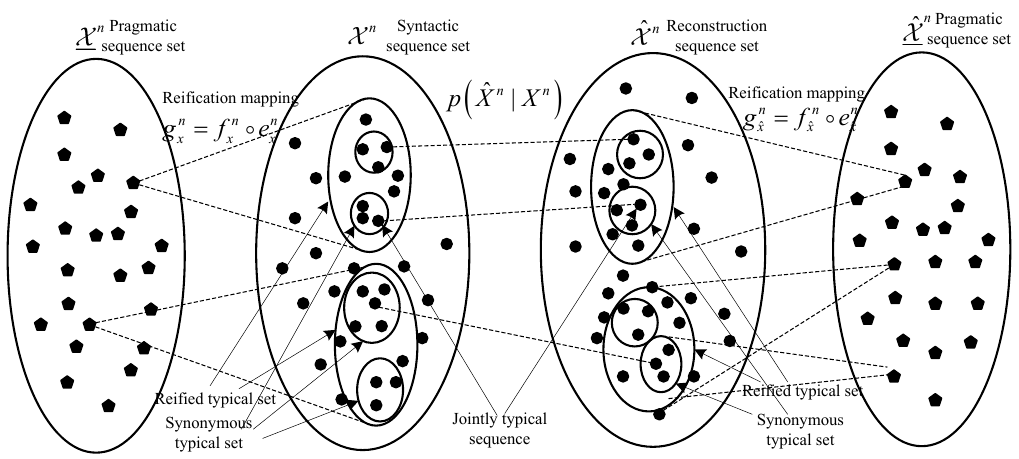}}
  \caption{Hierarchical typical sets for pragmatic lossy source coding. The joint reification mapping $g_{X\hat{X}}^n = f_{X\hat{X}}^n \circ e_{X\hat{X}}^n$ partitions the syntactic joint typical set $A_\epsilon^{(n)}$ into disjoint jointly reified typical sets $B_\epsilon^{(n)}(\underline{x}^n,\underline{\hat{x}}^n)$, each corresponding to a pragmatic pair in $\underline{A}_\epsilon^{(n)}$. Each $B_\epsilon^{(n)}$ is further decomposed into synonymous typical sets $S_\epsilon^{(n)}$ over the isoteleia typical sets $E_\epsilon^{(n)}$. }
\label{fig:joint_typical_set_mapping_for_lossy_source_coding}
\end{figure*}

As shown in Fig.~\ref{fig:joint_typical_set_mapping_for_lossy_source_coding}, the joint reification mapping $g_{X\hat{X}}^n = f_{X\hat{X}}^n \circ e_{X\hat{X}}^n$ first collapses syntactic sequence pairs into pragmatic equivalence classes according to the optimal actions, yielding the pragmatically jointly typical set $\underline{A}_\epsilon^{(n)}$. For each pragmatic pair, the preimage is the jointly reified typical set $B_\epsilon^{(n)}$, which collects all syntactic pairs sharing the same pragmatic purpose. Within $B_\epsilon^{(n)}$, the isoteleia mapping further groups semantic interpretations that lead to the same action, while the synonymous mapping collects syntactic realizations of each meaning. Thus, $B_\epsilon^{(n)}$ is the union of synonymous typical sets over the isoteleia typical set. This hierarchical coarsening reduces the number of equivalence classes from $|A_\epsilon^{(n)}|$ to $|\underline{A}_\epsilon^{(n)}|$, consistent with the entropy hierarchy $H_p(\underline{X},\underline{\hat{X}}) \le H_s(\tilde{X},\hat{\tilde{X}}) \le H(X,\hat{X})$, and provides the basis for the pragmatic rate-distortion theorem.

We now establish the AEP properties of these typical sets. The syntactically joint AEP is the classical result \cite{Classicpaper_Shannon,Book_Cover}; we restate it for reference.

\begin{theorem}[Syntactically Joint AEP]
Let \((X^n,\hat{X}^n)\) be a sequence pair of length \(n\) drawn i.i.d. according to \(P(x^n,\hat{x}^n)\). Then for sufficiently large \(n\):
\begin{enumerate}
\item \(\Pr\{(X^n,\hat{X}^n)\in A_\epsilon^{(n)}\} > 1-\epsilon\).
\item \((1-\epsilon)2^{n(H(X,\hat{X})-\epsilon)} \le |A_\epsilon^{(n)}| \le 2^{n(H(X,\hat{X})+\epsilon)}\).
\item If \((\dot{X}^n,\hat{\dot{X}}^n)\sim P(x^n)P(\hat{x}^n)\) are independent sequences with the same marginals, then
\[
(1-\epsilon)2^{-n(I(X;\hat{X})+3\epsilon)} \le \Pr\{(\dot{X}^n,\hat{\dot{X}}^n)\in A_\epsilon^{(n)}\} \le 2^{-n(I(X;\hat{X})-3\epsilon)}.
\]
\end{enumerate}
\end{theorem}

The pragmatically joint AEP follows by applying the weak law of large numbers to the aggregated pragmatic probabilities.

\begin{theorem}[Pragmatically Joint AEP for Compression]\label{theorem:pragmatically-Joint-AEP-Compr}
Let \((\underline{X}^n,\underline{\hat{X}}^n)\) be a pragmatic sequence pair of length \(n\) drawn i.i.d. according to \(P(\underline{x}^n,\underline{\hat{x}}^n)\), with the associated syntactic sequence pair \((X^n,\hat{X}^n)\) under the joint reification mapping \(g_{X\hat{X}}^n\). Then for sufficiently large \(n\):
\begin{enumerate}
\item \(\Pr\{(\underline{X}^n,\underline{\hat{X}}^n)\in\underline{A}_\epsilon^{(n)}\} > 1-\epsilon\).
\item For any \((\underline{x}^n,\underline{\hat{x}}^n)\in\underline{A}_\epsilon^{(n)}\),
\[
2^{-n(H_p(\underline{X},\underline{\hat{X}})+\epsilon)} \le P(\underline{x}^n,\underline{\hat{x}}^n) \le 2^{-n(H_p(\underline{X},\underline{\hat{X}})-\epsilon)}.
\]
\item \((1-\epsilon)2^{n(H_p(\underline{X},\underline{\hat{X}})-\epsilon)} \le |\underline{A}_\epsilon^{(n)}| \le 2^{n(H_p(\underline{X},\underline{\hat{X}})+\epsilon)}\).
\item If \((\dot{X}^n,\hat{\dot{X}}^n)\sim P(x^n)P(\hat{x}^n)\) are independent sequences with the same marginals as \(X^n\) and \(\hat{X}^n\), and \((\dot{\underline{X}}^n,\hat{\dot{\underline{X}}}^n)\) are the corresponding pragmatic sequences, then
\[
(1-\epsilon)2^{-n(I_p(\underline{X};\underline{\hat{X}})+3\epsilon)} \le \Pr\{(\dot{\underline{X}}^n,\hat{\dot{\underline{X}}}^n)\in\underline{A}_\epsilon^{(n)}\} \le 2^{-n(I_p(\underline{X};\underline{\hat{X}})-3\epsilon)}.
\]
\end{enumerate}
\end{theorem}
The details of the proof are provided in Appendix~\ref{a6:pragmatic-AEP-Compr-proof}.

Finally, we establish the properties of the jointly reified typical set.

\begin{theorem}[Properties of the Jointly Reified Typical Set]
For any pragmatic typical pair \((\underline{x}^n,\underline{\hat{x}}^n)\in\underline{A}_\epsilon^{(n)}\), for sufficiently large \(n\):
\begin{enumerate}
\item For any \((x^n,\hat{x}^n)\in B_\epsilon^{(n)}(\underline{x}^n,\underline{\hat{x}}^n)\),
\[
2^{-n(H(X,\hat{X})-H_p(\underline{X},\underline{\hat{X}})+\epsilon)} \le \frac{P(x^n,\hat{x}^n)}{P(\underline{x}^n,\underline{\hat{x}}^n)} \le 2^{-n(H(X,\hat{X})-H_p(\underline{X},\underline{\hat{X}})-\epsilon)}.
\]
\item 
\[
2^{n(H(X,\hat{X})-H_p(\underline{X},\underline{\hat{X}})-\epsilon)} \le |B_\epsilon^{(n)}(\underline{x}^n,\underline{\hat{x}}^n)| \le 2^{n(H(X,\hat{X})-H_p(\underline{X},\underline{\hat{X}})+\epsilon)}.
\]
\end{enumerate}
\end{theorem}

\begin{proof}
Property (1) follows directly from the definition of \(B_\epsilon^{(n)}(\underline{x}^n,\underline{\hat{x}}^n)\). The cardinality bounds are obtained by the same summation argument as in Theorem~\ref{theorem:Reified-Typical-Set} of Section~\ref{subsection:JAEP-JRTS}.
\end{proof}

\begin{remark}
The jointly reified typical set \(B_\epsilon^{(n)}(\underline{x}^n,\underline{\hat{x}}^n)\) groups together all syntactically jointly typical source-reconstruction pairs that share the same pragmatic interpretation. Its definition directly employs the reification mapping \(g_{X\hat{X}}^n = f_{X\hat{X}}^n \circ e_{X\hat{X}}^n\), which is the composition of the joint synonymous and joint isoteleia mappings. Consequently, the jointly reified typical set can be decomposed as the union of jointly synonymous typical sets over the isoteleia typical set:
\[
B_\epsilon^{(n)}(\underline{x}^n,\underline{\hat{x}}^n) = \bigcup_{(\tilde{x}^n,\hat{\tilde{x}}^n)\in E_\epsilon^{(n)}(\underline{x}^n,\underline{\hat{x}}^n)} S_\epsilon^{(n)}(\tilde{x}^n,\hat{\tilde{x}}^n),
\]
where \(E_\epsilon^{(n)}(\underline{x}^n,\underline{\hat{x}}^n)\) is the isoteleia typical set and \(S_\epsilon^{(n)}(\tilde{x}^n,\hat{\tilde{x}}^n)\) is the jointly synonymous typical set. The cardinality of a jointly reified typical set is approximately \(2^{n(H(X,\hat{X})-H_p(\underline{X},\underline{\hat{X}}))}\), which is larger than a single synonymous typical set by a factor of \(2^{n(H_s(\tilde{X},\hat{\tilde{X}})-H_p(\underline{X},\underline{\hat{X}}))}\), reflecting the additional syntactic freedom gained when only the pragmatic purpose needs to be preserved. 
\end{remark}

\subsection{Pragmatic Rate-Distortion Coding Theorem}

We now investigate the fundamental limit of pragmatic lossy source coding. As depicted in Fig.~\ref{fig:pragmatic_lossy_coding}, with the help of the reification mapping $g^n = f^n \circ e^n$, a pragmatic index $i_p$ is mapped into the syntactic source sequence $X^n(i)$. Considering the distortion requirement, the encoder selects a suitable reconstruction codeword $\hat{X}^n(j)$ from a pragmatic class to represent the source, then sends the indices to the receiver. The decoder outputs the reconstruction, and after demapping recovers the estimated pragmatic index.

\begin{figure}[htbp]
\centering
\includegraphics[scale=0.9]{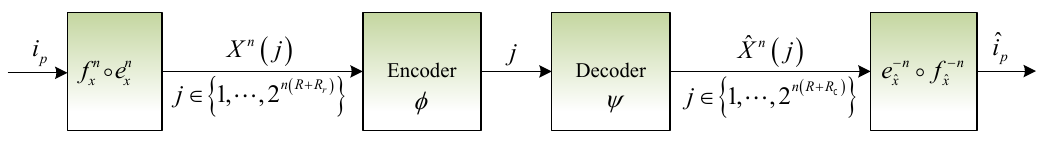}
\caption{Block diagram of pragmatic lossy source coding.}
\label{fig:pragmatic_lossy_coding}
\end{figure}

For a lossy source coding system, let $p(\hat{X}|X)$ be the test channel and $\mathcal{X}$, $\hat{\mathcal{X}}$ denote the source and reconstruction syntactic alphabets. The conditional distribution for the $n$-th extension is
\begin{equation}
p(\hat{x}^n | x^n) = \prod_{k=1}^{n} p(\hat{x}_k | x_k). \label{eq:test_channel}
\end{equation}

\begin{definition}[Pragmatic Lossy Source Code]
An $(M, n)$ code for pragmatic lossy source coding consists of the following parts:
\begin{enumerate}[(1)]
   \item A pragmatic index set $\mathcal{I}_p = \{1,\dots,M_p\}$ and a source reification index set $\mathcal{I}_r = \{1,\dots,M_r\}$; a reconstruction reification index set $\mathcal{I}_c = \{1,\dots,M_c\}$. The syntactic source and reconstruction index sets are $\mathcal{I} = \{1,\dots,M\}$ and $\mathcal{I}' = \{1,\dots,M'\}$.
   \item A reification mapping $g_X^n: \underline{\mathcal{X}}^n \to \mathcal{X}^n$ generates the set of source sequences, namely the pragmatic sourcebook $\mathcal{S} = \{X^n(1),\dots,X^n(M')\}$. This book is partitioned into reified source subsets $\mathcal{S}_r$ according to the pragmatic index.
   \item An encoding function $\phi: \mathcal{X}^n \to \hat{\mathcal{X}}^n$ generates the reconstruction codebook\\ $\mathcal{C} = \{\hat{X}^n(1),\dots,\hat{X}^n(M)\}$, partitioned into pragmatic codeword subsets $\mathcal{C}_p$ via the reification mapping $g_{\hat{X}}^n$.
   \item A decoding function $\psi: \hat{\mathcal{X}}^n \to \hat{\mathcal{X}}^n$ outputs the reconstruction $\hat{X}^n$.
   \item After demapping, $g_{\hat{X}}^n(\hat{X}^n) = \hat{i}_p$, the estimated pragmatic index is obtained. Both $\psi$ and $g_{\hat{X}}^n$ are deterministic.
\end{enumerate}
\end{definition}

The quotient sets $\mathcal{S} / g_X^n = \{\mathcal{S}_r\}$ and $\mathcal{C} / g_{\hat{X}}^n = \{\mathcal{C}_p\}$ have sizes $|\mathcal{S}/g_X^n| = M_r$ and $|\mathcal{C}/g_{\hat{X}}^n| = M_p$. Define the pragmatic code rate $R = \frac{1}{n}\log_2 M_p$, the source reification rate $R_r = \frac{1}{n}\log_2 M_r$, and the reconstruction reification rate $R_c = \frac{1}{n}\log_2 M_c$. We assume each reified set has equal size: $|\mathcal{S}_r| = M'/M_r = 2^{nR_r}$ and $|\mathcal{C}_p| = M/M_p = 2^{nR_c}$. The total syntactic rate is $R_{\mathrm{tot}} = R + R_c$.

\begin{theorem}[Pragmatic Rate-Distortion Coding Theorem]
\label{thm:pragmatic_RD}
Given an i.i.d. syntactic source $X \sim p(x)$ with associated pragmatic variable $\underline{X}$ under the reification mapping $g_X$, and a bounded pragmatic distortion measure $d_p(\underline{x}, \hat{\underline{x}})$, define the pragmatic rate-distortion function as
\begin{equation}
R_p(D) = \min_{g_X, g_{\hat{X}}} \min_{\substack{p(\hat{x}|x):\\ \mathbb{E}[d_p(\underline{X},\hat{\underline{X}})] \le D}} I_p(\underline{X};\hat{\underline{X}}), \label{eq:RpD}
\end{equation}
where $I_p(\underline{X};\hat{\underline{X}}) = H_p(\underline{X}) + H_p(\hat{\underline{X}}) - H(X,\hat{X})$ is the down pragmatic mutual information.

For any $\epsilon > 0$:
\begin{enumerate}
   \item \textbf{Achievability:} If $R > R_p(D)$, then for sufficiently large $n$, there exists an $(2^{n(R+R_c)}, n)$ code such that the average pragmatic distortion satisfies $\bar{d}_p \le D + \epsilon$.
   \item \textbf{Converse:} If $R < R_p(D)$, then for any $(2^{n(R+R_c)}, n)$ code, the average distortion $\bar{d}_p > D - \epsilon$ for sufficiently large $n$.
\end{enumerate}

Moreover, the pragmatic Value of Information (VoI) achieved by the reconstruction is maximized (i.e., equals the VoI obtained from the full pragmatic variable $\underline{X}$) if and only if $R \ge R_p(D)$. If $R < R_p(D)$, the achievable pragmatic VoI is strictly less than the full value, and the gap is determined by the rate shortfall.
\end{theorem}

\begin{proof}
\textbf{Achievability:}
Fix $\epsilon > 0$ and choose a test channel $p(\hat{x}|x)$ such that $\mathbb{E}[d_p(\underline{X},\hat{\underline{X}})] \le D/(1+\epsilon)$. Let $I_p = I_p(\underline{X};\hat{\underline{X}})$ and select a pragmatic rate $R$ satisfying$R_p(D) < R < I_p - 3\epsilon$.

Generate a pragmatic sourcebook $\mathcal{S}$ of size $2^{n(R+R_r)}$ by drawing sequences independently from $p(x^n) = \prod_{k=1}^n p(x_k)$. Partition it uniformly into $2^{nR}$ groups $\mathcal{S}_r(i_r)$ according to the reification mapping $g_X^n$, each of size $2^{nR_r}$, such that $\mathcal{S}_r(i_r) \subset B_\epsilon^{(n)}(\underline{x}_{i_r})$ for some representative pragmatic sequence.

\sloppy Generate a reconstruction codebook $\mathcal{C}$ of size $2^{n(R+R_c)}$ independently from $p(\hat{x}^n) = \prod_{k=1}^n p(\hat{x}_k)$, where $p(\hat{x}) = \sum_x p(x) p(\hat{x}|x)$. Partition it uniformly into $2^{nR}$ groups $\mathcal{C}_p(i_p)$ according to $g_{\hat{X}}^n$, each of size $2^{nR_c}$, with $\mathcal{C}_p(i_p) \subset B_\epsilon^{(n)}(\hat{\underline{x}}_{i_p})$.

Given a source sequence $X^n$, the encoder first finds the unique source reification group $\mathcal{S}_r(w_r)$ containing it. Then, based on the distortion criterion, it searches for a pragmatic index $w_p$ and a codeword $\hat{X}^n(w_p, l) \in \mathcal{C}_p(w_p)$ such that $(X^n, \hat{X}^n(w_p,l)) \in B_\epsilon^{(n)}$ (jointly reified typical). If more than one such pair exists, choose the smallest index; if none, set $w_p = 1$ and $l = 1$. The encoder sends the indices $(w_p, l)$. Similar to the classical random coding argument \cite{Book_ElGamal}, we analyze the probability of encoding error and the expected distortion.

Define the error event $\mathcal{E}$ as the event that the chosen pair is not jointly reified typical. It can be decomposed into
\[
\mathcal{E} = \mathcal{E}_1 \cup \mathcal{E}_2,
\]
where$\mathcal{E}_1 = \{ X^n \notin A_\epsilon^{(n)}(X^n) \}$,
and 
\[
\mathcal{E}_2 = \left\{ X^n \in \mathcal{S}_r(w_r),\; \hat{X}^n(w_p,l) \in \mathcal{C}_p(w_p),\; (X^n,\hat{X}^n(w_p,l)) \notin B_\epsilon^{(n)},\; \forall w_p,l \right\}.
\]

By the AEP, $\Pr(\mathcal{E}_1) \to 0$ as $n\to\infty$. For $\mathcal{E}_2$, using the fact that $\mathcal{S}_r(w_r) \subset B_\epsilon^{(n)}(\underline{x}_{w_r})$ and $\mathcal{C}_p(w_p) \subset B_\epsilon^{(n)}(\hat{\underline{x}}_{w_p})$, we have
\begin{align}
P(\mathcal{E}_2) &= \sum_{x^n \in A_\epsilon^{(n)}} p(x^n) \prod_{w_p=1}^{2^{nR}} \Pr\left( (X^n, \hat{X}^n(w_p,l)) \notin B_\epsilon^{(n)} \,\middle|\, X^n=x^n \right) \notag \\
&\le \sum_{x^n \in A_\epsilon^{(n)}} p(x^n) \left[ 1 - \Pr\left( (X^n, \hat{X}^n(1,l)) \in B_\epsilon^{(n)} \right) \right]^{2^{nR}}. \label{eq:E2_step}
\end{align}

Since $x^n \in A_\epsilon^{(n)}$ and $\hat{X}^n(1,l)$ is drawn independently from $p(\hat{x}^n)$, by the pragmatically joint AEP (Theorem~\ref{theorem:pragmatically-Joint-AEP-Compr}), the probability that the pair is jointly reified typical is bounded below by
\begin{equation}
\Pr\left( (X^n, \hat{X}^n(1,l)) \in B_\epsilon^{(n)} \right) \ge (1-\epsilon) 2^{-n(I_p + 3\epsilon)}, \label{eq:joint_prob_lower}
\end{equation}
for sufficiently large $n$. This bound follows from the fact that the probability that an independent pair falls into the pragmatically jointly typical set $\underline{A}_\epsilon^{(n)}$ is at least $(1-\epsilon)2^{-n(I_p+3\epsilon)}$, and $B_\epsilon^{(n)} \subseteq \underline{A}_\epsilon^{(n)}$.

Using $(1-x)^m \le e^{-mx}$ for $x\in[0,1]$ and $m\ge 0$, we obtain
\begin{align}
P(\mathcal{E}_2) &\le \left( 1 - (1-\epsilon)2^{-n(I_p+3\epsilon)} \right)^{2^{nR}} \notag \\
&\le \exp\left( - (1-\epsilon) 2^{n(R - I_p - 3\epsilon)} \right). \label{eq:E2_exp}
\end{align}
Since $R < I_p - 3\epsilon$, the exponent tends to $-\infty$, hence $P(\mathcal{E}_2) \to 0$. Therefore the total encoding error probability $P_e^{(n)} \to 0$.

Conditioned on the success event $\mathcal{E}^c$ (i.e., the pair is jointly reified typical), the average distortion is bounded by
\[
\mathbb{E}[d_p(\underline{X}^n, \hat{\underline{X}}^n) \mid \mathcal{E}^c] \le (1+\epsilon) \mathbb{E}[d_p(\underline{X},\hat{\underline{X}})] \le D.
\]
Let $d_{\max}$ be the maximum distortion. The overall expected distortion satisfies
\begin{equation}
\mathbb{E}[d_p(\underline{X}^n, \hat{\underline{X}}^n)] \le P_e^{(n)} d_{\max} + (1-P_e^{(n)}) D \le D + \epsilon,
\end{equation}
for sufficiently large $n$. This proves achievability.

\textbf{Converse:}
Assume a sequence of $(2^{n(R+R_c)}, n)$ codes with average distortion $\bar{d}_p \le D$. Let $W_p$ be the uniform pragmatic message (index $i_p$), and let $\hat{W}_p$ be its estimate at the decoder. We will show that $R \ge R_p(D)$ must hold.

First, by the data processing inequality and the fact that $W_p \to X^n \to \hat{X}^n$ forms a Markov chain (where $\hat{X}^n$ is the reconstruction), we have
\begin{equation}
I(W_p; X^n) \ge I(\hat{\underline{X}}^n(W_p); X^n), \label{eq:data1}
\end{equation}
where $\hat{\underline{X}}^n(W_p)$ denotes the pragmatic reconstruction associated with the chosen codeword. Since $W_p$ is a function of $\hat{\underline{X}}^n(W_p)$, we have $H(W_p) \le I(W_p; X^n) + H(W_p \mid X^n)$. However, a more direct route is used in \cite{Paper_SIT}.

We proceed as follows:
\begin{align}
nR &= H(W_p) \ge I(W_p; X^n) \label{eq:step1} \\ 
   &\ge I(\hat{\underline{X}}^n(W_p); X^n) \label{eq:step2} \\
   &\ge I_p(\underline{X}^n; \hat{\underline{X}}^n). \label{eq:step3}
\end{align}
Inequality \eqref{eq:step2} follows from the data processing inequality because $\hat{\underline{X}}^n$ is a function of $W_p$ (through the codebook). Inequality \eqref{eq:step3} is a pragmatic analogue of Lemma 10 in \cite{Paper_SIT,Book_SIT}, which states that the mutual information between the syntactic source and the pragmatic reconstruction is at least the down pragmatic mutual information between the pragmatic source and the pragmatic reconstruction. Indeed, by the definition of down pragmatic mutual information and the fact that the reconstruction variable is a function of the channel output, one can verify that $I(X^n; \hat{\underline{X}}^n) \ge I_p(\underline{X}^n; \hat{\underline{X}}^n)$ (this follows from the entropy hierarchy and the data processing inequality).

Since the source is memoryless, the down pragmatic mutual information satisfies
\begin{equation}
I_p(\underline{X}^n; \hat{\underline{X}}^n) = \sum_{k=1}^n I_p(\underline{X}_k; \hat{\underline{X}}_k) \label{eq:sum_Ip},
\end{equation}
because the pragmatic variables are also memoryless due to the deterministic reification mapping.

By the definition of the pragmatic rate-distortion function $R_p(D)$, for each $k$, we have
\begin{equation}
I_p(\underline{X}_k; \hat{\underline{X}}_k) \ge R_p\left( \mathbb{E}[d_p(\underline{X}_k, \hat{\underline{X}}_k)] \right). \label{eq:each_Rp}
\end{equation}
Summing over $k$ and using the convexity of $R_p(D)$, we obtain
\begin{equation}
\sum_{k=1}^n I_p(\underline{X}_k; \hat{\underline{X}}_k) \ge \sum_{k=1}^n R_p\left( \mathbb{E}[d_p(\underline{X}_k, \hat{\underline{X}}_k)] \right) \ge n R_p\left( \frac{1}{n}\sum_{k=1}^n \mathbb{E}[d_p(\underline{X}_k, \hat{\underline{X}}_k)] \right). \label{eq:convex}
\end{equation}
The average distortion over the entire sequence is $\bar{d}_p = \frac{1}{n}\sum_{k=1}^n \mathbb{E}[d_p(\underline{X}_k, \hat{\underline{X}}_k)] \le D$ by assumption. Since $R_p(D)$ is non-increasing, $R_p(\bar{d}_p) \ge R_p(D)$. Therefore,
\begin{equation}
I_p(\underline{X}^n; \hat{\underline{X}}^n) \ge n R_p(\bar{d}_p) \ge n R_p(D). \label{eq:final_Ip}
\end{equation}
Combining \eqref{eq:step1}–\eqref{eq:step3} with \eqref{eq:final_Ip}, we obtain
\begin{equation}
nR \ge n R_p(D) \quad \Longrightarrow \quad R \ge R_p(D).
\end{equation}
Thus, if $R < R_p(D)$, the assumption $\bar{d}_p \le D$ leads to a contradiction. Hence, for any code with $R < R_p(D)$, the average distortion must satisfy $\bar{d}_p > D - \epsilon$ for sufficiently large $n$. This completes the converse.

\textbf{Connection to Pragmatic Value of Information}
The pragmatic Value of Information achieved by the reconstruction $\hat{\underline{X}}$ is
\begin{equation}
\mathrm{VoI}_p(\hat{\underline{X}}) = \mathbb{E}_{X,\hat{X}}\left[U(X, a^*(\hat{\underline{X}}))\right] - U_0, \label{eq:voi}
\end{equation}
where $a^*(\hat{\underline{X}})$ is the optimal action based on the reconstructed pragmatic class. The full VoI (as if $\underline{X}$ were directly observed) is attained if and only if the reconstruction preserves all pragmatic distinctions, which requires $I_p(\underline{X};\hat{\underline{X}}) = H_p(\underline{X})$. This occurs only when $R \ge H_p(\underline{X})$. More generally, for a given distortion level $D$, the full VoI is attainable iff $R \ge R_p(D)$. If $R < R_p(D)$, the loss in pragmatic information reduces the achievable utility.
\end{proof}

\begin{remark}
The pragmatic rate-distortion theorem generalizes both the classical rate-distortion theorem \cite{Classicpaper_Shannon, Book_Cover} and the semantic rate-distortion theorem \cite{Paper_SIT, Book_SIT}. When the reification mapping $g$ is trivial (identity), $R_p(D) = R(D)$; when only the synonymous mapping $f$ is non-trivial, we recover $R_p(D) = R_s(D)$. In general, the hierarchy $R_p(D) \le R_s(D) \le R(D)$ holds. Although our code uses three levels (pragmatic + two reification rates), the reification mapping $g = f \circ e$ implicitly includes the semantic layer. The extra reification rates $R_r$ and $R_c$ allow flexible trade-offs: setting them to zero achieves the pure pragmatic compression limit $R_p(D)$, while increasing them preserves additional syntactic or semantic details, recovering the classical or semantic limits. This modularity is consistent with the design principles of pragmatic information systems.
\end{remark}

\section{Pragmatic Information Measure of Continuous Message}
\label{section_X}

In this section, we extend the pragmatic information measures to continuous messages, completing the theoretical framework with both bilateral and single-sided characterizations. Building upon the isoteleic volume formalism, we derive the pragmatic capacity and rate-distortion functions for Gaussian channels and sources, respectively. More importantly, we introduce the corresponding pragmatic cost of information (CoI) and pragmatic value of information (VoI) for both bilateral and single-sided cases, establishing the full duality spectrum: CoI is the inverse of the achievable rate functions, while VoI is the Legendre-Fenchel dual of the rate-distortion functions. These results provide a complete information-theoretic and decision-theoretic characterization of task-oriented communication over continuous alphabets, encompassing perceptual (observation-to-action) and expressive (action-to-reconstruction) asymmetric scenarios.

\subsection{Pragmatic Entropy and Pragmatic Mutual Information for Continuous Message}
\label{sec:pragmatic_continuous}

In this subsection, we extend the pragmatic information measures, such as pragmatic entropy and pragmatic mutual information, to the domain of continuous messages. Unlike the syntactic layer, which deals with exact signal values, or the semantic layer, which deals with meanings, the pragmatic layer is concerned with the consequences of actions. For continuous signals, the equivalence is determined by whether different realizations lead to the same optimal action, as defined by a utility function. This gives rise to the concept of \emph{isoteleic volumes}, which are the continuous analogues of the pragmatic equivalence classes (fibers) introduced in Section~\ref{section_II}.

The coding theorems established in Sections \ref{section_VII}--\ref{section_IX} are developed for discrete, finite alphabets, where the syntactic, semantic, and pragmatic layers are characterized by quotient spaces of finite cardinality. For continuous alphabets, the pragmatic equivalence classes induced by the isoteleia mapping become uncountable sets, and the discrete entropy measures must be replaced by their differential counterparts. To bridge the discrete and continuous formulations, we adopt the standard quantization argument: partition the continuous space \(\mathcal{W}\) into small cells of volume \(\Delta\), apply the discrete theory to the resulting finite alphabet, and then take the limit \(\Delta \to 0\). Under this limiting procedure, the quotient-space structure of the isoteleia mapping converges to the continuous notion of \emph{isoteleic volumes} \(\Omega\), which quantify the average measure of signal space that maps to the same optimal action. The entropy hierarchy \(H_p \le H_s \le H\) is preserved in the limit, with the discrete pragmatic entropy giving way to the differential pragmatic entropy defined below. Consequently, all subsequent Gaussian channel and source characterizations are consistent with—and indeed inherit the operational meaning of—the discrete coding theorems proved earlier.

To formalize pragmatic information for continuous messages, we first define the isoteleia mapping for continuous variables and introduce the notion of isoteleic volumes.

\begin{definition}[Isoteleic Volume]
Let $W$ be a continuous random variable with alphabet $\mathcal{W} \subseteq \mathbb{R}^d$ and probability density function $p(w)$. Let $\tilde{\mathcal{W}}$ be the associated semantic alphabet, and let $\underline{\mathcal{W}}$ be the pragmatic alphabet. The pragmatic variable $\underline{W}$ is induced by the isoteleia mapping $e: \underline{\mathcal{W}} \to 2^{\tilde{\mathcal{W}}}$, which, in turn, is derived from a utility function $U(s, a)$ defined on states and actions.

For each pragmatic symbol $\underline{w}_i \in \underline{\mathcal{W}}$, its \emph{isoteleic volume} $\Omega_i$ is defined as the total measure of the set of all semantic (and hence syntactic) realizations that map to the same optimal action. Formally,
\begin{equation}
\Omega_i \triangleq \left| \left\{ w \in \mathcal{W} : a^*(w) = \underline{w}_i \right\} \right|,
\end{equation}
where $a^*(w)$ is the optimal action for the realization $w$, defined by
\begin{equation}
a^*(w) = \arg \max_{a \in \mathcal{A}} \mathbb{E}_{S|w}[U(S, a)].
\end{equation}
The collection $\{ \Omega_i \}_{i=1}^{|\underline{\mathcal{W}}|}$ forms a partition of the support of $W$, i.e., $\bigcup_i \Omega_i = \mathcal{W}$ and $\Omega_i \cap \Omega_j = \emptyset$ for $i \neq j$.
\end{definition}

Analogous to the average synonymous length $S$ in semantic information theory, we define the \emph{average isoteleic volume} as the expectation of the isoteleic volume under the distribution of $W$.

\begin{definition}[Average Isoteleic Volume]
The average isoteleic volume $\Omega$ is defined as
\begin{equation}
\Omega \triangleq \mathbb{E}_{w \sim p(w)} \left[ \left| \Omega_{a^*(w)} \right| \right] = \sum_{i=1}^{|\underline{\mathcal{W}}|} \Pr(\underline{w}_i) \cdot |\Omega_i|,
\end{equation}
where $\Pr(\underline{w}_i) = \int_{\Omega_i} p(w) dw$ is the probability mass of the pragmatic class $\underline{w}_i$.
\end{definition}

The average isoteleic volume $\Omega$ quantifies the ``compression" achieved by grouping different continuous signals into the same pragmatic class. A larger $\Omega$ indicates that more distinct signals lead to the same optimal action, thus reducing decision uncertainty.

\subsubsection{Pragmatic Entropy for Continuous Messages}
\leavevmode\newline\indent
We now define the pragmatic entropy for a continuous random variable. The definition is derived by quantizing the continuous space and taking the limit, analogous to the derivation of differential entropy.

\begin{definition}[Pragmatic Entropy for Continuous Variables]
Let $W$ be a continuous random variable with probability density function $p(w)$, supported on $\mathcal{W} \subseteq \mathbb{R}^d$. Let $\underline{W}$ be the associated pragmatic variable induced by the isoteleia mapping, with average isoteleic volume $\Omega$. The pragmatic entropy of $\underline{W}$ is defined as
\begin{equation}\label{eq:pragmatic-entropy-cont}
H_p(\underline{W}) \triangleq - \int_{\mathcal{W}} p(w) \log p(w) \, dw - \log \Omega,
\end{equation}
where $\log \Omega = \sum_{i} \Pr(\underline{w}_i) \log |\Omega_i|$ is the logarithm of the average isoteleic volume.
\end{definition}

\begin{remark}
\sloppy The first term in Eq. (\ref{eq:pragmatic-entropy-cont}) is the classical differential entropy $h(W) = - \int p(w) \log p(w) dw$. The second term, $\log \Omega$, accounts for the reduction in uncertainty due to pragmatic abstraction. Unlike the discrete case, the pragmatic differential entropy $H_p(\underline{W})$ is not an absolute measure of information but a relative one, defined up to a reference measure. To maintain a meaningful non-negative pragmatic entropy (in prabits), the average isoteleic volume $\Omega$ must satisfy:
$\log \Omega \le h(W)$. When the isoteleia mapping is trivial (i.e., each signal leads to a unique action, $\Omega=1$), the pragmatic entropy reduces to the differential entropy: $H_p(\underline{W}) = h(W)$. 
\end{remark}

\begin{corollary}[Lower Bound on Pragmatic Entropy]
The pragmatic entropy is lower bounded by
\begin{equation}
H_p(\underline{W}) \geq h(W) - \log \Omega.
\end{equation}
Equality holds when the isoteleic volumes are chosen proportionally to the probability mass of the pragmatic classes, i.e., when the partition is optimal.
\end{corollary}

\begin{example}[Uniform Distribution]
Let $W$ be uniformly distributed over an interval of length $L$, i.e., $p(w) = 1/L$ for $w \in [0, L]$. Suppose the isoteleia mapping partitions this interval into $K$ equal isoteleic volumes, each of size $\Omega_i = L/K$. Then the average isoteleic volume is $\Omega = L/K$. The pragmatic entropy is
\begin{equation}
H_p(\underline{W}) = -\int_0^L \frac{1}{L} \log \frac{1}{L} dw - \log \frac{L}{K} = \log K \text{ prabits}.
\end{equation}
This result is analogous to the semantic entropy for a uniformly distributed source, but here the interpretation is in terms of optimal actions rather than meanings.
\end{example}

\begin{example}[Gaussian Distribution]
Let $W \sim \mathcal{N}(0, \sigma^2)$ be a Gaussian random variable. Under an optimal isoteleia mapping with average isoteleic volume $\Omega$, the pragmatic entropy is
\begin{align}
H_p(\underline{W}) &= \mathbb{E}[-\log p(W)] - \log \Omega \nonumber \\
&= \frac{1}{2} \log (2\pi e \sigma^2) - \log \Omega \nonumber \\
&= \frac{1}{2} \log \left( \frac{2\pi e \sigma^2}{\Omega^2} \right) \text{ prabits}.
\end{align}
This shows that the pragmatic entropy of a Gaussian source decreases logarithmically with the average isoteleic volume.
\end{example}

The definitions can be extended to multiple continuous variables. Let $(W, V)$ be a pair of continuous random variables with joint density $p(w, v)$, and let $(\underline{W}, \underline{V})$ be their associated pragmatic variables, induced by a joint isoteleia mapping with average isoteleic volume $\Omega_{wv}$.

\begin{definition}[Joint Pragmatic Entropy for Continuous Variables]
The joint pragmatic entropy of $(\underline{W}, \underline{V})$ is defined as
\begin{equation}
H_p(\underline{W}, \underline{V}) \triangleq - \int_{\mathcal{W} \times \mathcal{V}} p(w, v) \log p(w, v) \, dw \, dv - \log \Omega_{wv},
\end{equation}
where $\Omega_{wv}$ is the average joint isoteleic volume.
\end{definition}

Similarly, the conditional pragmatic entropy of $\underline{V}$ given $W$ is defined as
\begin{equation}
H_p(\underline{V} | W) \triangleq - \int_{\mathcal{W} \times \mathcal{V}} p(w, v) \log p(v|w) \, dw \, dv - \log \Omega_v,
\end{equation}
where $\Omega_v$ is the average isoteleic volume for $V$.

\subsubsection{Pragmatic Mutual Information for Continuous Variables}
\leavevmode\newline\indent
We now define the pragmatic mutual information for continuous variables. As in the discrete case, there are two variants: the \emph{up pragmatic mutual information} and the \emph{down pragmatic mutual information}.

\begin{definition}[Up Pragmatic Mutual Information for Continuous Variables]
Let $(W, V)$ be a pair of continuous random variables with joint density $p(w, v)$. Let $\underline{W}$ and $\underline{V}$ be their pragmatic counterparts, induced by a joint isoteleia mapping with average isoteleic volumes $\Omega_w$ and $\Omega_v$. The up pragmatic mutual information is defined as
\begin{align}\label{eq:UP-Pragmatic-MI-Cont}
I^p(\underline{W}; \underline{V}) &\triangleq H(W) + H(V) - H_p(\underline{W}, \underline{V}) \nonumber \\
&= - \int_{\mathcal{W} \times \mathcal{V}} p(w, v) \log \frac{p(w) p(v)}{p(w, v)} \, dw \, dv + \log (\Omega_w \Omega_v).
\end{align}
\end{definition}

\begin{definition}[Down Pragmatic Mutual Information for Continuous Variables]
The down pragmatic mutual information is defined as
\begin{align}\label{eq:Down-Pragmatic-MI-Cont}
I_p(\underline{W}; \underline{V}) &\triangleq H_p(\underline{W}) + H_p(\underline{V}) - H(W, V) \nonumber \\
&= - \int_{\mathcal{W} \times \mathcal{V}} p(w, v) \log \frac{p(w) p(v)}{p(w, v)} \, dw \, dv - \log (\Omega_w \Omega_v).
\end{align}
\end{definition}

\begin{remark}
The first term in both (\ref{eq:UP-Pragmatic-MI-Cont}) and (\ref{eq:Down-Pragmatic-MI-Cont}) is the classical mutual information $I(W; V)$. The difference lies in the second term. If $\mathbb{E}[\log (\Omega_w \Omega_v)] \ge 0$, then the up pragmatic mutual information is greater than or equal to the classical mutual information, while the down pragmatic mutual information is less than or equal to it. In practice, the down pragmatic mutual information may be negative; in such cases, one can take its positive part, $(I_p(\underline{W}; \underline{V}))^+$.
\end{remark}

The relationship between these measures is summarized by the following hierarchy, which mirrors the discrete case:
\begin{equation}
I_p(\underline{W}; \underline{V}) \le I(W; V) \le I^p(\underline{W}; \underline{V}).
\end{equation}

\begin{example}[Gaussian Variables]
Let $(W, V)$ be jointly Gaussian with correlation coefficient $\rho$ and marginal variances $\sigma_w^2$ and $\sigma_v^2$. The classical mutual information is $I(W; V) = -\frac{1}{2} \log (1 - \rho^2)$. Under a joint isoteleia mapping with average volumes $\Omega_w$ and $\Omega_v$, the up and down pragmatic mutual informations are
\begin{align}
I^p(\underline{W}; \underline{V}) &= -\frac{1}{2} \log (1 - \rho^2) + \log(\Omega_w \Omega_v), \\
I_p(\underline{W}; \underline{V}) &= -\frac{1}{2} \log (1 - \rho^2) - \log(\Omega_w \Omega_v).
\end{align}
This demonstrates how pragmatic abstraction can either enhance or reduce the measured information, depending on whether the coarsening is applied to the joint entropy or to the marginal entropies.
\end{example}

This completes the extension of pragmatic entropy and mutual information to continuous messages. In the following subsections, we will apply these measures to derive the pragmatic capacity of the Gaussian channel and the pragmatic rate-distortion function for Gaussian sources, along with their associated cost and value counterparts.

\subsection{Pragmatic Capacity of Gaussian Channel and Pragmatic Cost of Information}
\label{sec:pragmatic_gaussian_capacity}

We now investigate the pragmatic capacity of the additive white Gaussian noise (AWGN) channel and its dual—the pragmatic cost of information. The received signal at time $k$ is modeled as
\begin{equation}
Y_k = X_k + Z_k,
\end{equation}
with $X_k$ being the channel input, $Z_k \sim \mathcal{N}(0, \sigma^2)$ i.i.d. Gaussian noise, and $Y_k$ the output. The channel is subject to an average power constraint $\mathbb{E}[X^2] \leq P$.

Let $\underline{X}$ and $\underline{Y}$ be the pragmatic variables associated with the input and output, respectively, induced by the isoteleia mappings $e_X$ and $e_Y$ (and their compositions with the synonymous mappings) according to a given utility function. The pragmatic channel capacity is defined as the maximum achievable rate of pragmatic information transmission, i.e.,
\begin{equation}
C_p \triangleq \max_{p(x)} \max_{g_{XY}} I^p(\underline{X}; \underline{Y}),
\end{equation}
where $g_{XY} = f_{XY} \circ e_{XY}$ is the joint reification mapping, and $I^p(\underline{X}; \underline{Y}) = H(X) + H(Y) - H_p(\underline{X}, \underline{Y})$ is the up pragmatic mutual information.

For the Gaussian channel, the input distribution that maximizes the classical mutual information is $X \sim \mathcal{N}(0, P)$. Under the pragmatic framework, we assume that the joint isoteleia mapping partitions the input-output space such that the average isoteleic volumes for $X$ and $Y$ are $\Omega_x$ and $\Omega_y$, respectively. Without loss of generality, we consider a symmetric case where $\Omega_x = \Omega_y = \Omega$, with $\Omega \ge 1$ being the average isoteleic volume. This volume represents the average measure of the set of signals that lead to the same optimal terminal action.

Using the continuous version of the up pragmatic mutual information derived in (\ref{eq:UP-Pragmatic-MI-Cont}), we have
\begin{align}
I^p(\underline{X}; \underline{Y}) &= I(X; Y) + \log(\Omega_x \Omega_y) \nonumber \\
&= \frac{1}{2} \log \left( 1 + \frac{P}{\sigma^2} \right) + \log(\Omega^2) \nonumber \\
&= \frac{1}{2} \log \left( \Omega^4 \left( 1 + \frac{P}{\sigma^2} \right) \right).
\end{align}

Thus, the pragmatic capacity of the Gaussian channel (per real dimension) is given by
\begin{equation}
 C_p = \frac{1}{2} \log \left( \Omega^4 \left( 1 + \frac{P}{\sigma^2} \right) \right) \text{ prabits per channel use}. 
\label{eq:gaussian-capacity}
\end{equation}

By the capacity-cost duality established in Section \ref{section_VI}, the bilateral pragmatic cost of information is the inverse of the capacity function. From Eq.~\eqref{eq:gaussian-capacity}, solving for $P$ as a function of $R$ yields:
\begin{equation}
 \mathrm{CoI}_p(R) = \sigma^2 \left( \frac{2^{2R}}{\Omega^4} - 1 \right) \quad \text{for } R \ge 0.
\label{eq:bilateral-coi}
\end{equation}
This represents the minimum transmit power required to achieve a pragmatic rate $R$ (in prabits per channel use) over the AWGN channel.

For asymmetric systems where only one side is coarsened, we have two additional cost measures, that is, perceptual/expressive pragmatic cost of information.
For the perceptual path (observation $Y$ carries information about pragmatic action $\underline{X}$), the prospective perceptual achievable rate from Section \ref{subsec:single-sided} is:
\begin{equation}
C_p^{\text{p-perc}}(P) = \frac{1}{2} \log \left( 1 + \frac{P}{\sigma^2} \right) + \log \Omega = \frac{1}{2} \log \left( \Omega^2 \left( 1 + \frac{P}{\sigma^2} \right) \right).
\end{equation}
The inverse gives the perceptual pragmatic cost of information:
\begin{equation}
 \mathrm{CoI}_p^{\text{p-perc}}(R) = \sigma^2 \left( \frac{2^{2R}}{\Omega^2} - 1 \right).
\label{eq:perc-coi}
\end{equation}
This is the minimum sensing power required to extract $R$ prabits of pragmatic information from the observation, when the observation itself is not coarsened.

For the expressive path (pragmatic label $\underline{Y}$ conveys information about source $X$), the prospective expressive achievable rate is:
\begin{equation}
C_p^{\text{p-expr}}(P) = \frac{1}{2} \log \left( \Omega^2 \left( 1 + \frac{P}{\sigma^2} \right) \right),
\end{equation}
which is identical to the perceptual case under symmetry. Thus,
\begin{equation}
 \mathrm{CoI}_p^{\text{p-expr}}(R) = \sigma^2 \left( \frac{2^{2R}}{\Omega^2} - 1 \right).
\label{eq:expr-coi}
\end{equation}

\begin{remark}[Comparison of CoI Measures]
The three CoI measures satisfy the hierarchy:
\begin{equation}
\mathrm{CoI}_p(R) \le \mathrm{CoI}_p^{\text{p-perc}}(R) = \mathrm{CoI}_p^{\text{p-expr}}(R),
\end{equation}
since $\Omega^4 \ge \Omega^2$ for $\Omega \ge 1$. This confirms the general principle: bilateral pragmatic abstraction (coarsening both sides) is more efficient than single-sided abstraction, requiring less power to achieve the same pragmatic rate.
\end{remark}

We now extend the result to the band-limited Gaussian channel. Consider a Gaussian channel with bandwidth $B$ (Hz), two-sided power spectral density $N_0/2$, and signal power constraint $P$. The classical Shannon capacity for this channel is
\begin{equation}
C = B \log \left( 1 + \frac{P}{N_0 B} \right) \quad \text{bits per second}.
\end{equation}

By the Shannon-Nyquist sampling theorem, a band-limited signal of duration $T$ can be represented by $2BT$ independent samples per real dimension. The total pragmatic capacity over the $2BT$ dimensions gives:
\begin{equation}
 C_p = B \log \left( \Omega^4 \left( 1 + \frac{P}{N_0 B} \right) \right) \text{ prabits per second}. 
\label{eq:bandlimited-capacity}
\end{equation}

The corresponding bilateral CoI (power per second) is:
\begin{equation}
 \mathrm{CoI}_p(R) = N_0 B \left( \frac{2^{R/B}}{\Omega^4} - 1 \right) \quad \text{(watts)}.
\label{eq:bandlimited-coi}
\end{equation}

For the single-sided cases, the perceptual and expressive CoI become:
\begin{equation}
\mathrm{CoI}_p^{\text{p-perc}}(R) = \mathrm{CoI}_p^{\text{p-expr}}(R) = N_0 B \left( \frac{2^{R/B}}{\Omega^2} - 1 \right).
\end{equation}

\begin{example}[Numerical comparison of pragmatic cost measures]
Let $\sigma^2 = 1$ for the discrete-time case, and for the band-limited case set $B = 1$ Hz, $N_0 = 1$. Take $\Omega = 2$. The table below compares the minimum required power (in units of $\sigma^2$ or watts) for various target rates $R$ (in prabits per channel use or prabits/s) under bilateral and single-sided coarsening. Negative values indicate that zero power suffices (the rate lies below the zero-power capacity).
\end{example}

\begin{table}[htbp]
\centering
\caption{CoI values for different scenarios and rates.}
\label{tab:coi_examples}
\begin{tabular}{|c|c|c|c|c|}
\hline
\textbf{Scenario} & \textbf{CoI expression} & $R=1$ & $R=2$ & $R=3$ \\
\hline
Discrete bilateral & $\sigma^2(2^{2R}/\Omega^4 - 1)$ & $-0.75$ (0) & $0$ (0) & $3$ \\
Discrete perceptual/expressive & $\sigma^2(2^{2R}/\Omega^2 - 1)$ & $0$ (0) & $3$ & $15$ \\
Band-limited bilateral & $N_0 B(2^{R/B}/\Omega^4 - 1)$ & $-0.875$ (0) & $-0.75$ (0) & $-0.5$ (0) \\
Band-limited perceptual/expressive & $N_0 B(2^{R/B}/\Omega^2 - 1)$ & $-0.5$ (0) & $0$ (0) & $1$ \\
\hline
\end{tabular}
\end{table}

For the discrete bilateral case, the zero-power capacity is $C_p(0) = \frac{1}{2}\log(16) = 2$ prabits per channel use, so any $R \le 2$ requires zero power. For the discrete perceptual path, $C_p^{\text{p-perc}}(0) = \frac{1}{2}\log(4) = 1$ prabit, thus $R=1$ is free but $R=2$ costs $3$ units. In the band-limited case with $B=1$, the bilateral zero-power capacity is $C_p(0) = \log(16) = 4$ prabits/s, hence $R=3$ is still free; the perceptual zero-power capacity is $\log(4) = 2$ prabits/s, so $R=2$ is free and $R=3$ requires $1$ W. These numerical results illustrate that bilateral abstraction (coarsening both source and observation) yields the greatest power savings for a given pragmatic rate.

\begin{remark}
The pragmatic capacity and CoI formulas reveal that the isoteleic volume $\Omega$ acts as a power amplification factor. In the band-limited case, $\Omega^4$ effectively multiplies the signal-to-noise ratio, while $\Omega^2$ multiplies it for single-sided cases. This provides a clear design guideline: to minimize transmit power for a given pragmatic rate, one should exploit pragmatic abstraction on both the source and the observation sides whenever possible.
\end{remark}

\subsection{Pragmatic Rate Distortion of Gaussian Source and Pragmatic Value of Information}
\label{sec:pragmatic_gaussian_ratedistortion}

As a dual to the channel capacity problem, we now investigate the pragmatic rate-distortion function for a Gaussian source and its corresponding pragmatic value of information. This function characterizes the minimum rate required to represent a continuous source such that the distortion, measured at the pragmatic level, does not exceed a given threshold.

Consider a Gaussian source $X \sim \mathcal{N}(0, P)$ with variance $P$. Let $\hat{X}$ be the reconstruction of $X$, and let $\underline{X}$ and $\underline{\hat{X}}$ be the corresponding pragmatic variables induced by the isoteleia mappings $e_X$ and $e_{\hat{X}}$ (and their compositions with synonymous mappings). The distortion measure is defined at the pragmatic level as the mean squared error between the optimal actions corresponding to the source and its reconstruction:
\begin{equation}
d_p(\underline{x}, \hat{\underline{x}}) \triangleq \left( a^*(\underline{x}) - a^*(\hat{\underline{x}}) \right)^2,
\end{equation}
where $a^*(\underline{x})$ is the optimal action for the pragmatic class $\underline{x}$. For the Gaussian source, we adopt the standard mean squared error (MSE) distortion measure, which is widely used and has a well-known rate-distortion characterization. Specifically, we consider the distortion constraint $
\mathbb{E}\left[ (X - \hat{X})^2 \right] \leq D$,
which, under appropriate isoteleia mappings, translates to a pragmatic distortion constraint.

The pragmatic rate-distortion function is defined as
\begin{equation}
R_p(D) \triangleq \min_{g_X, g_{\hat{X}}} \min_{p(\hat{x}|x): \mathbb{E}[d_p(\underline{X}, \underline{\hat{X}})] \leq D} I_p(\underline{X}; \underline{\hat{X}}),
\end{equation}
where $I_p(\underline{X}; \underline{\hat{X}}) = H_p(\underline{X}) + H_p(\underline{\hat{X}}) - H(X, \hat{X})$ is the down pragmatic mutual information.

For a Gaussian source with variance $P$, the classical rate-distortion function is $R(D) = \frac{1}{2} \log (P/D)$ for $0 \le D \le P$. Under the pragmatic framework, the isoteleia mapping groups signals that lead to the same optimal action, effectively reducing the ``volume" of the source space that needs to be distinguished. Let $\Omega_x$ and $\Omega_{\hat{x}}$ be the average isoteleic volumes for the source and reconstruction, respectively. Assuming symmetry, we set $\Omega_x = \Omega_{\hat{x}} = \Omega$, where $\Omega \ge 1$ is the average isoteleic volume.

Using the continuous expression for the down pragmatic mutual information derived in Eq. (\ref{eq:Down-Pragmatic-MI-Cont}), we have
\begin{align}
I_p(\underline{X}; \underline{\hat{X}}) &= I(X; \hat{X}) - \log(\Omega_x \Omega_{\hat{x}}) \nonumber \\
&\ge \frac{1}{2} \log \left( \frac{P}{D} \right) - \log(\Omega^2) \nonumber \\
&= \frac{1}{2} \log \left( \frac{P}{\Omega^4 D} \right),
\end{align}
where the inequality follows from the fact that for any test channel satisfying the MSE constraint $\mathbb{E}[(X-\hat{X})^2] \le D$, the classical mutual information satisfies $I(X; \hat{X}) \ge \frac{1}{2} \log(P/D)$. The equality is achieved when the test channel is Gaussian and the isoteleia mapping is optimal (i.e., partitions the space proportionally to the probability mass of the pragmatic classes).

Thus, the pragmatic rate-distortion function for a Gaussian source is given by
\begin{equation}
 R_p(D) = 
\begin{cases}
\dfrac{1}{2} \log \left( \dfrac{P}{\Omega^4 D} \right), & 0 \le D \le \dfrac{P}{\Omega^4}, \\[6pt]
0, & D > \dfrac{P}{\Omega^4}.
\end{cases} 
\label{eq:gaussian-rd}
\end{equation}

By the rate-distortion-value duality established in Section \ref{section_VI}, the bilateral pragmatic value of information is the Legendre-Fenchel dual of the rate-distortion function. For the Gaussian quadratic utility case, the value-rate function $\Phi_p(R)$ is obtained by solving the distortion-rate relationship.

From Eq.~\eqref{eq:gaussian-rd}, the distortion as a function of rate $R$ is:
\begin{equation}
D_p(R) = \frac{P}{\Omega^4} 2^{-2R}, \quad 0 \le R < \infty.
\end{equation}
For a quadratic utility $U(X,A) = -(X-A)^2$, the maximum expected utility at rate $R$ is:
\begin{equation}
\Phi_p(R) = -D_p(R) = -\frac{P}{\Omega^4} 2^{-2R}.
\end{equation}
The baseline utility (without information) is $U_0 = -P$. Therefore, the bilateral pragmatic value of information is:
\begin{equation}
 \mathrm{VoI}_p(R) = \Phi_p(R) - U_0 = P \left( 1 - \frac{2^{-2R}}{\Omega^4} \right).
\label{eq:bilateral-voi}
\end{equation}
This is the maximum utility gain achievable when both the source and the reconstruction are coarsened pragmatically.

For asymmetric systems, we have two single-sided value measures corresponding to the perceptual and expressive rate-distortion functions from Section \ref{subsec:single-sided}. Due to symmetry, these two cases yield identical expressions. For either side, the rate-distortion function is
\begin{equation}
R_p^{\text{side}}(D) = \frac{1}{2} \log \left( \frac{P}{\Omega^2 D} \right), \quad 0 \le D \le \frac{P}{\Omega^2},
\end{equation}
with corresponding distortion-rate function $D_p^{\text{side}}(R) = \frac{P}{\Omega^2} 2^{-2R}$. The pragmatic value of information for both the perceptual and expressive paths is therefore
\begin{equation}
\mathrm{VoI}_p^{\text{perc}}(R) = \mathrm{VoI}_p^{\text{expr}}(R) = P \left( 1 - \frac{2^{-2R}}{\Omega^2} \right).
\label{eq:perc-voi}
\end{equation}
We shall refer to either of these as the single-sided VoI.

\begin{remark}[Comparison of VoI Measures]
The three VoI measures satisfy the hierarchy:
\begin{equation}
\mathrm{VoI}_p^{\text{perc}}(R) = \mathrm{VoI}_p^{\text{expr}}(R) \le \mathrm{VoI}_p(R),
\end{equation}
since $1/\Omega^2 \ge 1/\Omega^4$ for $\Omega \ge 1$. This is consistent with the general principle that bilateral pragmatic abstraction (coarsening both the source and the reconstruction) yields higher utility gain for a given rate than single-sided abstraction, because it discards more decision-irrelevant information, allowing the same rate to be used more efficiently for the task.
\end{remark}

\begin{example}[Numerical comparison of VoI and distortion for Gaussian source]
Let $P = 1$ and $\Omega = 2$. The table below compares the pragmatic value of information $\mathrm{VoI}(R)$ and the corresponding minimal distortion $D(R)$ for bilateral, perceptual, and expressive coarsening at selected rates $R$ (in prabits). Recall that under symmetry, perceptual and expressive values are identical.
\end{example}

\begin{table}[htbp]
\centering
\caption{VoI and distortion for Gaussian source with $P=1$, $\Omega=2$.}
\label{tab:voi_examples}
\begin{tabular}{|c|c|c|c|c|}
\hline
\textbf{Scenario} & \textbf{Measure} & $R=0$ & $R=1$ & $R=2$ \\
\hline
\multirow{2}{*}{Bilateral} & $\mathrm{VoI}_p(R)$ & $0.9375$ & $0.984375$ & $0.996094$ \\
 & $D_p(R) = P/\Omega^4 \cdot 2^{-2R}$ & $1/16=0.0625$ & $1/64 \approx 0.015625$ & $1/256 \approx 0.003906$ \\
\hline
\multirow{2}{*}{Perceptual / Expressive} & $\mathrm{VoI}_p^{\text{perc/expr}}(R)$ & $0.75$ & $0.9375$ & $0.984375$ \\
 & $D_p^{\text{perc/expr}}(R) = P/\Omega^2 \cdot 2^{-2R}$ & $1/4=0.25$ & $1/16=0.0625$ & $1/64 \approx 0.015625$ \\
\hline
\end{tabular}
\end{table}

At $R=1$, bilateral coarsening yields a VoI of $0.984$ utils and distortion $0.0156$, while single-sided coarsening gives $0.9375$ utils and distortion $0.0625$ — a fourfold reduction in distortion for the same rate. Equivalently, to achieve the single-sided distortion of $0.0625$ with bilateral coarsening, only $R=0$ is needed (since $D_p(0)=0.0625$), saving $1$ prabit. As $R\to\infty$, both VoI approach $1$ (perfect action reconstruction) and distortion tends to $0$. The classical case $\Omega=1$ is recovered by setting $\Omega=1$, which gives $\mathrm{VoI}(R)=1-2^{-2R}$, strictly smaller than any pragmatic VoI for finite $R$.

This example quantifies the benefit of coarsening both ends: bilateral pragmatic compression achieves lower distortion and higher utility gain than single-sided compression at the same rate, highlighting the efficiency of task-oriented abstraction.

\begin{remark}
The closed-form expressions for CoI and VoI derived above assume that the action space is continuous Euclidean, enabling the use of MSE as the distortion measure. For discrete action spaces (e.g., classification tasks), the corresponding CoI and VoI functions must be derived using discrete distortion measures such as Hamming distance or utility loss, which generally do not admit such simple closed forms. The general methodology, however, remains the same: CoI is the inverse of the achievable rate functions, and VoI is the Legendre-Fenchel dual of the rate-distortion functions, both extending naturally to the single-sided cases via the perceptual and expressive mutual informations.
\end{remark}

The pragmatic cost of information and pragmatic value of information for Gaussian channels and sources complete the bilateral and single-sided duality spectrum in the continuous domain, providing a comprehensive framework for task-oriented communication system design with explicit resource-utility trade-offs.

\section{Joint Optimization of Pragmatic Information System}
\label{section_XI}

In this section, we present the pragmatic source-channel coding theorem, which extends the classical separation principle to the pragmatic domain and establishes that reliable transmission is possible iff the pragmatic entropy (or pragmatic rate-distortion function) does not exceed the pragmatic channel capacity. This condition is equivalent to the existence of a scheme that achieves the full pragmatic VoI while keeping the pragmatic CoI within the resource budget, thereby bridging information-theoretic limits with decision-theoretic objectives. We then formulate a unified optimization over communication rate, control law, and pragmatic abstraction in closed-loop systems, and derive optimality conditions via the pragmatic Lagrangian for co-designing communication and control under resource constraints.

\subsection{Pragmatic Source-Channel Coding}
\label{subsec:pragmatic-source-channel}

We now consider the joint source-channel coding problem in the pragmatic domain, where the goal is to transmit a syntactic source over a noisy channel such that the receiver can recover the pragmatic intention rather than the exact source symbols. This problem generalizes Shannon's classical source-channel coding theorem by relaxing the requirement of exact symbol reconstruction and allowing errors that do not affect the optimal action.

Consider a discrete memoryless syntactic source \(W\) with alphabet \(\mathcal{W}\) and probability mass function \(P(W)\). Let \(\tilde{W}\) and \(\underline{W}\) be the associated semantic and pragmatic variables induced by the synonymous mapping \(f:\tilde{\mathcal{W}}\to 2^{\mathcal{W}}\) and the isoteleia mapping \(e:\underline{\mathcal{W}}\to 2^{\tilde{\mathcal{W}}}\), respectively, with the reification mapping \(g = f\circ e:\underline{\mathcal{W}}\to 2^{\mathcal{W}}\). The source generates an i.i.d. sequence \(W^n\) according to \(P(W^n)=\prod_{k=1}^n P(W_k)\).

Let the channel be a discrete memoryless channel with input alphabet \(\mathcal{X}\), output alphabet \(\mathcal{Y}\), and transition probability \(P(Y|X)\). The pragmatic channel capacity \(C_p\) is defined as in Section~\ref{subsec:pragmatic-capacity}: $C_p \triangleq \max_{p(x)} \max_{g_{XY}} I^p(\underline{X};\underline{Y})$,
where \(I^p(\underline{X};\underline{Y})\) is the up pragmatic mutual information, and \(g_{XY}\) is the joint reification mapping.

\begin{figure}[htbp]
\setlength{\abovecaptionskip}{0.cm}
\setlength{\belowcaptionskip}{-0.cm}
\centering{\includegraphics[scale=1]{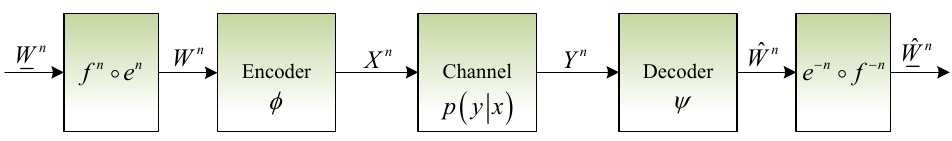}}
\caption{Block diagram of pragmatic source channel coding.}
\label{fig:pragmatic-source-channel}
\end{figure}

Figure~\ref{fig:pragmatic-source-channel} illustrates the block diagram of pragmatic source-channel coding. At the transmitter side, the syntactic source \(W^n\) first passes through the reification mapping \(g^n = f^n \circ e^n\), which collapses semantically equivalent syntactic symbols into pragmatic classes \(\underline{W}^n\) representing the optimal terminal actions. These pragmatic classes are then encoded into channel codewords \(X^n\) via a pragmatic channel encoder, which may exploit the structure of the pragmatic equivalence classes to achieve higher rates than classical coding. The codewords are transmitted over the noisy channel, producing the received sequence \(Y^n\). At the receiver side, the pragmatic channel decoder recovers an estimate of the pragmatic index \(\hat{\underline{W}}^n\) from the channel output, ignoring syntactic and semantic variations that do not affect the terminal action. The demapping operation \(g^{-n}\) then maps the estimated pragmatic index to a representative syntactic sequence \(\hat{W}^n\), which is passed to the destination. The entire system is evaluated by the error probability \(P_e^{(n)} = \Pr\{g^{-n}(\psi_n(Y^n)) \neq g^{-n}(W^n)\}\), i.e., the probability that the recovered pragmatic intention differs from the true one. This architecture embodies the separation principle: the pragmatic source code (reification mapping) and the pragmatic channel code can be designed independently while preserving asymptotic optimality.

A pragmatic source-channel code of block length \(n\) consists of:
\begin{enumerate}
    \item An encoding function \(\phi_n:\mathcal{W}^n \to \mathcal{X}^n\) that maps each source sequence to a channel input sequence, possibly exploiting the reification mapping \(g^n\) to first map the source to its pragmatic class;
    \item A decoding function \(\psi_n:\mathcal{Y}^n \to \mathcal{W}^n\) that maps each channel output to an estimated source sequence (or directly to an estimated pragmatic index), followed by demapping \(g^n\) to recover the pragmatic intention.
\end{enumerate}

The error probability is defined as the probability that the decoded pragmatic intention differs from the true pragmatic intention:
\begin{equation}
P_e^{(n)} \triangleq \Pr\left\{ g^{-n}(\psi_n(Y^n)) \neq g^{-n}(W^n) \right\}.
\end{equation}

We now state the fundamental theorem of pragmatic source-channel coding.

\begin{theorem}[Pragmatic Source-Channel Coding Theorem]
\label{thm:source-channel}
Let \(W\) be a discrete memoryless syntactic source with associated pragmatic variable \(\underline{W}\) having pragmatic entropy \(H_p(\underline{W})\), and let \(p(y|x)\) be a discrete memoryless channel with pragmatic capacity \(C_p\).

\begin{enumerate}
    \item \textbf{Lossless Case:}
        If $H_p(\underline{W}) < C_p$, then for any \(\epsilon>0\), there exists, for sufficiently large \(n\), a sequence of pragmatic source-channel codes with error probability \(P_e^{(n)} < \epsilon\).
        On the contrary, if $H_p(\underline{W}) > C_p$, then for any sequence of pragmatic source-channel codes, the error probability is bounded away from zero for sufficiently large \(n\).
        
    \item \textbf{Lossy Case:}
        For a given pragmatic distortion measure \(d_p(\underline{w},\hat{\underline{w}})\) and distortion constraint \(D\), if $R_p(D) < C_p$, then there exists a sequence of pragmatic source-channel codes such that \(\mathbb{E}[d_p(\underline{W},\hat{\underline{W}})] \le D\).
        On the contrary, if $R_p(D) > C_p$, then no such sequence exists.
\end{enumerate}
\end{theorem}

\begin{proof}
We outline the proof following the standard separation principle, adapted to the pragmatic domain.

\textit{Achievability:} Assume \(H_p(\underline{W}) < C_p\). Choose \(\epsilon > 0\) such that \(H_p(\underline{W}) + \epsilon < C_p - \epsilon\). By the pragmatic lossless source coding theorem (Theorem~\ref{theorem:Pragmatic-Lossless-Source-Coding-Theorem}), there exists a pragmatic source code that maps the source sequences to pragmatic indices at rate \(R_s < H_p(\underline{W}) + \epsilon\) with error probability less than \(\epsilon/2\). By the pragmatic channel coding theorem (Theorem~\ref{theorem:Pragmatic-Channel-Coding-Theorem}), since \(R_s < C_p\), there exists a pragmatic channel code that transmits these indices over the channel with error probability less than \(\epsilon/2\). By the union bound, the overall error probability is less than \(\epsilon\). The decoder first recovers the pragmatic index from the channel output, then outputs a representative syntactic sequence from the corresponding reified typical set.

\textit{Converse:} Suppose a sequence of codes exists with \(P_e^{(n)} \to 0\). Following the standard Fano's inequality argument adapted to the pragmatic domain, we have
\begin{equation}
    H_p(\underline{W}) \le \frac{1}{n} I(\underline{W}^n; \underline{Y}^n) + \delta_n,
\end{equation}
where \(\delta_n \to 0\) as \(P_e^{(n)} \to 0\). By the definition of pragmatic channel capacity and the data processing inequality,
\begin{equation}
    \frac{1}{n} I(\underline{W}^n; \underline{Y}^n) \le C_p.
\end{equation}
Thus, \(H_p(\underline{W}) \le C_p\), which contradicts \(H_p(\underline{W}) > C_p\). The proof for the lossy case follows analogously using the pragmatic rate-distortion function \(R_p(D)\).
\end{proof}

The pragmatic source-channel coding theorem has a direct interpretation in terms of the pragmatic value of information (VoI) and pragmatic cost of information (CoI) established in Section~\ref{section_VI}.

The condition \(H_p(\underline{W}) < C_p\) (or \(R_p(D) < C_p\)) can be equivalently stated as the existence of a rate \(R\) such that:
\begin{equation}
    \mathrm{VoI}_p(R) = \mathrm{VoI}_p^{\max} \quad \text{and} \quad \mathrm{CoI}_p(R) \le P_{\max},
\end{equation}
where \(\mathrm{VoI}_p^{\max}\) is the maximum achievable pragmatic value (the utility gain obtained when all pragmatic distinctions are preserved), and \(P_{\max}\) is the available physical resource budget (e.g., maximum transmit power). This follows from the monotonicity of VoI and CoI:

\begin{itemize}
    \item The pragmatic value of information \(\mathrm{VoI}_p(R)\) is non-decreasing in \(R\) and saturates at its maximum when \(R \ge H_p(\underline{W})\) (for lossless transmission) or \(R \ge R_p(D)\) (for lossy transmission).
    \item The pragmatic cost of information \(\mathrm{CoI}_p(R)\) is non-decreasing and convex in \(R\), representing the minimum resource required to achieve rate \(R\).
\end{itemize}

Therefore, the achievability condition \(H_p(\underline{W}) < C_p\) guarantees that there exists a rate \(R\) such that:
\begin{equation}
    \mathrm{VoI}_p(R) = \mathrm{VoI}_p^{\max} \quad \text{and} \quad \mathrm{CoI}_p(R) \le P_{\max}.
\end{equation}
Conversely, if \(H_p(\underline{W}) > C_p\), then even at the maximum achievable rate \(C_p\), the system cannot transmit enough information to preserve all pragmatic distinctions, resulting in a loss of pragmatic value:
\begin{equation}
    \mathrm{VoI}_p(C_p) < \mathrm{VoI}_p^{\max}.
\end{equation}
This gap represents the \emph{pragmatic value loss} due to the channel bottleneck.

\begin{corollary}[Pragmatic Value-Cost Characterization]
\sloppy A pragmatic communication system achieves the full pragmatic value of information if and only if the transmission rate \(R\) satisfies:
\begin{equation}
    R \ge H_p(\underline{W}) \quad \text{(lossless)} \quad \text{or} \quad R \ge R_p(D) \quad \text{(lossy)},
\end{equation}
and the available physical resources \(P_{\max}\) satisfy:
\begin{equation}
    \mathrm{CoI}_p(R) \le P_{\max}.
\end{equation}
\end{corollary}

\begin{remark}
The pragmatic source-channel coding theorem provides a rigorous theoretical foundation for task-oriented communication. Unlike classical source-channel coding, which requires the syntactic entropy \(H(W)\) to be less than the syntactic channel capacity \(C\), the pragmatic version relaxes this condition to \(H_p(\underline{W}) \le C_p\). Since \(H_p(\underline{W}) \le H(W)\) and \(C_p \ge C\), the pragmatic condition is strictly weaker than the classical one. This explains the performance gains observed in semantic and pragmatic communication systems: by tolerating syntactic and semantic errors that do not affect the optimal action, the system can operate reliably in regimes where classical communication would fail.
\end{remark}

\begin{example}[Autonomous Driving Revisited]
Consider the two-vehicle autonomous driving scenario from Sections~III and IV, where the syntactic source has entropy \(H(W) \approx 1.9893\) bits per symbol, while the pragmatic entropy is \(H_p(\underline{W}) \approx 0.9928\) prabits per symbol. Suppose the communication channel has syntactic capacity \(C = 1.5\) bits per channel use. Classical source-channel coding would fail because \(H(W) > C\). However, if the pragmatic capacity \(C_p \ge 0.9928\) (which holds for any channel with \(C_p \ge C\)), the pragmatic source-channel coding theorem guarantees reliable transmission of the optimal action (Go/Stop) even though the exact traffic light symbols cannot be reliably transmitted. This demonstrates the power of pragmatic abstraction in overcoming physical communication bottlenecks.
\end{example}

\subsubsection{Separation Principle and the Imperative of Joint Optimization}

The pragmatic source-channel coding theorem (Theorem~\ref{thm:source-channel}) is proved via a separation architecture, which builds on two independent results:
\begin{enumerate}
    \item The pragmatic lossless source coding theorem (Section~\ref{section_VII}) guarantees that the pragmatic source can be compressed to rate \(H_p(\underline{W})\) (or \(R_p(D)\) in the lossy case) without losing any pragmatically relevant information.
    \item The pragmatic channel coding theorem (Section~\ref{section_VIII}) guarantees that any rate up to \(C_p\) can be reliably transmitted over the noisy channel.
\end{enumerate}
Concatenating the pragmatic source code and the pragmatic channel code yields an end-to-end system that achieves reliable pragmatic transmission \emph{if and only if} the source rate does not exceed the channel capacity, i.e.,
\[
H_p(\underline{W}) \le C_p \quad \text{(lossless)} \quad \text{or} \quad R_p(D) \le C_p \quad \text{(lossy)}.
\]
This constitutes the \emph{pragmatic separation principle}—the asymptotic bedrock of the theory.

In the classical information-theoretic setting where the blocklength \(n\) tends to infinity and no constraints are imposed on encoding/decoding complexity, the separation principle is strictly optimal. It allows the system designer to optimize the pragmatic source code (e.g., feature extraction, action classification) and the pragmatic channel code (e.g., modulation, error correction) independently, while guaranteeing that the concatenated system achieves the fundamental limit. This modularity is a direct consequence of the asymptotic equipartition property and the joint typicality arguments used in the proofs, and it is fully consistent with the backward-compatibility requirement of Principle~III in Section~\ref{section_II}.

The optimality of separation, however, is an asymptotic statement. In practical systems, where the blocklength \(n\) is finite, or when the encoder and decoder are subject to computational complexity, delay, or memory constraints, the separation architecture is \emph{strictly suboptimal} with respect to the end-to-end task utility. In such regimes, the optimal code cannot be factored into independent source and channel components; instead, the encoder must jointly consider the source statistics, the channel characteristics, and the decision utility to minimize the overall distortion or maximize the expected reward. This fundamental gap between asymptotic theory and engineering reality is precisely the motivation for the joint optimization framework that follows.

This observation directly justifies the ``joint optimization'' framework developed in Sections~\ref{section_VI}. While the separation principle identifies the ultimate performance bounds (capacities and rate-distortion functions), the \emph{optimal allocation of finite resources}—such as transmit power, latency, bandwidth, or model size—requires solving the coupled Lagrangian:
\[
\mathcal{L}_{\mathrm{global}}(R;\lambda) = \Phi_p(R) - \lambda \, \mathrm{CoI}_p(R),
\]
where the same rate \(R\) simultaneously determines the coarsening level of the isoteleia mapping (source-side abstraction) and the error-protection strength (channel-side coding). The cross-layer coupling therefore does not lie in the \emph{code structure}, which can be separated asymptotically, but in the \emph{rate-resource trade-off}, which must be jointly optimized to maximize the net benefit. Thus, the separation principle provides the theoretical bedrock, while the Lagrangian framework delivers the practical tool for engineering finite-dimensional, resource-constrained systems.

\begin{remark}
The pragmatic separation principle generalizes both the classical separation theorem (when the isoteleia and synonymous mappings are trivial) and the semantic separation theorem (when only the synonymous mapping is non-trivial). It unifies all previous results as special cases and offers new insights into the design of task-oriented communication systems, while the accompanying joint optimization framework addresses the real-world constraints that lie beyond the asymptotic ideal.
\end{remark}

\subsubsection{Gaussian Source over Gaussian Channel}
\label{subsec:gaussian-joint-analysis}

To illustrate the joint optimization framework in a concrete setting, we now specialize to the canonical case of a Gaussian source transmitted over an additive white Gaussian noise (AWGN) channel. This example demonstrates how the pragmatic cost of information (CoI) and pragmatic value of information (VoI) jointly determine the optimal operating point of a task-oriented communication system.

Consider a Gaussian source \(X \sim \mathcal{N}(0, P)\) with variance \(P\), transmitted over an AWGN channel with noise variance \(\sigma^2\) and power constraint \(\mathbb{E}[|X|^2] \le P_{\text{tx}}\). The channel output is \(Y = X + Z\), where \(Z \sim \mathcal{N}(0, \sigma^2)\). Let the isoteleic volume be \(\Omega \ge 1\), representing the pragmatic abstraction that groups signals leading to the same optimal action.

From Sections~\ref{sec:pragmatic_gaussian_capacity} and \ref{sec:pragmatic_gaussian_ratedistortion}, we have the following closed-form expressions for the continuous Gaussian case:

\begin{itemize}
    \item \textbf{Pragmatic channel capacity} (per real dimension):
    \begin{equation}
        C_p(P_{\text{tx}}) = \frac{1}{2} \log \left( \Omega^4 \left( 1 + \frac{P_{\text{tx}}}{\sigma^2} \right) \right) \quad \text{prabits per channel use}.
        \label{eq:gaussian-pragmatic-channel-capacity}
    \end{equation}
    
    \item \textbf{Bilateral pragmatic cost of information} (inverse of capacity):
    \begin{equation}
        \mathrm{CoI}_p(R) = \sigma^2 \left( \frac{2^{2R}}{\Omega^4} - 1 \right) \quad \text{for } R \ge 0.
        \label{eq:gaussian-coi}
    \end{equation}
    
    \item \textbf{Pragmatic rate-distortion function}:
    \begin{equation}
        R_p(D) = \frac{1}{2} \log \left( \frac{P}{\Omega^4 D} \right), \quad 0 \le D \le \frac{P}{\Omega^4}.
        \label{eq:gaussian-pragmatic-RD}
    \end{equation}
    
    \item \textbf{Bilateral pragmatic value of information} (for quadratic utility \(U(X,A) = -(X-A)^2\)):
    \begin{equation}
        \mathrm{VoI}_p(R) = P \left( 1 - \frac{2^{-2R}}{\Omega^4} \right).
        \label{eq:gaussian-voi}
    \end{equation}
\end{itemize}

The pragmatic source-channel coding theorem (Theorem~\ref{thm:source-channel}) states that reliable transmission of pragmatic information is possible if and only if
\begin{equation}
    R_p(D) \le C_p(P_{\text{tx}}).
    \label{eq:gaussian-feasibility}
\end{equation}

Substituting the closed-form expressions from (\ref{eq:gaussian-pragmatic-channel-capacity}) and (\ref{eq:gaussian-pragmatic-RD}), the feasibility condition becomes
\begin{equation}
    \frac{1}{2} \log \left( \frac{P}{\Omega^4 D} \right) \le \frac{1}{2} \log \left( \Omega^4 \left( 1 + \frac{P_{\text{tx}}}{\sigma^2} \right) \right).
\end{equation}

Simplifying, we obtain the equivalent condition:
\begin{equation}
    D \ge \frac{P}{\Omega^8 \left( 1 + P_{\text{tx}} / \sigma^2 \right) }.
    \label{eq:gaussian-distortion-constraint}
\end{equation}
Here the factor \(\Omega^8\) results from the product of the pragmatic gains at both ends: the source--reconstruction pair contributes \(\Omega^4\) through the rate--distortion function, and the input--output pair contributes another \(\Omega^4\) through the channel capacity, under the assumption of equal isoteleic volumes for both the source and the channel.

This inequality reveals the fundamental trade-off: for a given transmit power \(P_{\text{tx}}\) and pragmatic abstraction \(\Omega\), the minimum achievable pragmatic distortion is
\begin{equation}
    D_{\min}(P_{\text{tx}}) = \frac{P}{\Omega^8 \left( 1 + P_{\text{tx}} / \sigma^2 \right)}.
    \label{eq:gaussian-min-distortion}
\end{equation}

Equivalently, for a target distortion \(D\), the minimum required transmit power is
\begin{equation}
    P_{\text{tx}}^{\min}(D) = \sigma^2 \left( \frac{P}{\Omega^8 D} - 1 \right).
    \label{eq:gaussian-min-power}
\end{equation}

The joint feasibility condition can also be expressed directly in terms of VoI and CoI. Recall that the full pragmatic value \(\mathrm{VoI}_p^{\max} = P\) is achieved when \(R \ge R_p(D)\) (i.e., the rate is sufficient to meet the distortion requirement). The required rate to achieve a given VoI level is obtained by inverting (\ref{eq:gaussian-voi}):
\begin{equation}
    R_{\mathrm{VoI}}(V) = \frac{1}{2} \log \left( \frac{1}{\Omega^4 (1 - V/P)} \right), \quad 0 \le V \le P.
\end{equation}

The system achieves the target VoI \(V\) if and only if the available power satisfies
\begin{equation}
    \mathrm{CoI}_p(R_{\mathrm{VoI}}(V)) \le P_{\text{tx}}.
\end{equation}

Substituting the CoI expression (\ref{eq:gaussian-coi}), this becomes
\begin{equation}
    \sigma^2 \left( \frac{2^{2 R_{\mathrm{VoI}}(V)}}{\Omega^4} - 1 \right) \le P_{\text{tx}}.
\end{equation}

Since \(2^{2 R_{\mathrm{VoI}}(V)} = 1 / (\Omega^4 (1 - V/P))\), we obtain
\begin{equation}
    \sigma^2 \left( \frac{1}{\Omega^8 (1 - V/P)} - 1 \right) \le P_{\text{tx}}.
\end{equation}

Solving for \(V\), the maximum achievable pragmatic value for a given transmit power is
\begin{equation}
    \mathrm{VoI}_p^{\max}(P_{\text{tx}}) = P \left( 1 - \frac{1}{\Omega^8 (1 + P_{\text{tx}} / \sigma^2)} \right).
    \label{eq:gaussian-max-voi}
\end{equation}

This expression directly quantifies the \emph{value-cost trade-off}: the utility gain increases with transmit power but saturates at \(P\) as \(P_{\text{tx}} \to \infty\). The pragmatic abstraction \(\Omega\) amplifies the effective signal-to-noise ratio, reducing the power required to achieve a given utility level. For example, when \(\Omega = 1\) (no pragmatic abstraction), (\ref{eq:gaussian-max-voi}) reduces to the classical value-cost relationship:
\begin{equation}
    \mathrm{VoI}^{\max}(P_{\text{tx}}) = P \left( 1 - \frac{1}{1 + P_{\text{tx}} / \sigma^2} \right) = \frac{P P_{\text{tx}}}{P_{\text{tx}} + \sigma^2},
\end{equation}
which is the well-known SNR-dependent utility gain for a Gaussian channel with quadratic utility.

\begin{example}
Let \(P = 1\), \(\sigma^2 = 1\), and consider three values of the isoteleic volume: \(\Omega = 1\) (classical), \(\Omega = 1.5\), and \(\Omega = 2\). For a target VoI of \(V = 0.9\) (i.e., 90\% of the maximum utility), the required transmit powers are:
\begin{align}
    \Omega = 1 &: \quad P_{\text{tx}} = \sigma^2 \left( \frac{1}{\Omega^8 (1 - V/P)} - 1 \right) = \frac{1}{0.1} - 1 = 9 \text{ (units)}, \\
    \Omega = 1.5 &: \quad P_{\text{tx}} = \frac{1}{1.5^8 \cdot 0.1} - 1 \approx \frac{1}{25.63 \cdot 0.1} - 1 \approx -0.61 \text{ (zero power suffices)},\\
    \Omega = 2 &: \quad P_{\text{tx}} = \frac{1}{2^8 \cdot 0.1} - 1 = \frac{1}{25.6} - 1 \approx -0.96 \text{ (zero power suffices)}.
\end{align}
\end{example}

With \(\Omega = 1.5\), the pragmatic abstraction alone (\(P_{\text{tx}} = 0\)) already achieves \(\mathrm{VoI} = P(1 - 1/\Omega^8) = 1 - (2/3)^8 \approx 0.961\), exceeding the target 0.9. This demonstrates that pragmatic abstraction can dramatically reduce or eliminate the need for transmit power in task-oriented communication.

For asymmetric systems where only one side is coarsened, the perceptual and expressive CoI/VoI functions from Section~\ref{subsec:single-sided-cost-value} yield analogous joint feasibility conditions. For the perceptual path, the achievable VoI under power constraint is
\begin{equation}
    \mathrm{VoI}_p^{\text{perc, max}}(P_{\text{tx}}) = P \left( 1 - \frac{1}{\Omega^4 (1 + P_{\text{tx}} / \sigma^2)} \right),
\end{equation}
which is strictly less than the bilateral case (\ref{eq:gaussian-max-voi}) for any finite \(\Omega > 1\) and \(P_{\text{tx}}\), since the effective SNR amplification is \(\Omega^4\) rather than \(\Omega^8\). This quantifies the performance loss when only the source or only the reconstruction is coarsened, confirming the benefits of coarsening both sides of the communication link.

\begin{remark}
The Gaussian example illustrates the full power of the pragmatic information-theoretic framework: it provides a unified, analytically tractable model for jointly optimizing communication resources (power), decision performance (distortion/utility), and pragmatic abstraction (isoteleic volume). The closed-form expressions reveal the precise quantitative relationships among these quantities, enabling system designers to make principled trade-offs in task-oriented communication system design.
\end{remark}

\subsection{Communication-Control-Decision Joint Optimization}
\label{subsec:joint-optimization-dynamic}

We now extend the Lagrangian dual framework from the static source-channel coding setting to the general closed-loop system where communication, control, and decision-making are jointly optimized. In many practical systems—such as autonomous vehicles, robotic control, and industrial automation—the communication rate, control policy, and pragmatic abstraction must be co-designed to maximize the overall task performance under resource constraints. We present two formal theorems that characterize the optimal trade-offs in both static (single-step) and dynamic (multi-step) settings, with explicit optimality conditions that bridge information-theoretic limits and decision-theoretic objectives.

\subsubsection{Static Joint Optimization: Lagrangian Approach and Single-Sided Extensions}
\label{subsubsec:static-joint}
\leavevmode\newline\indent
We first consider a single-step decision problem where the system observes a state \(S \in \mathcal{S}\), takes an action \(A \in \mathcal{A}\), and receives a utility \(U(S,A)\). The observation is obtained through a communication channel that conveys a pragmatic message \(\underline{Y}\) about the state, with a rate \(R\) that determines the quality of the pragmatic information. The receiver chooses a policy \(\pi:\underline{\mathcal{Y}} \to \mathcal{A}\) that maps the pragmatic observation to an action. The joint optimization problem is to maximize the expected utility subject to a communication cost constraint.

\begin{theorem}[Static Joint Optimization]
\label{thm:static-joint}
For a single-step decision problem with state \(S\), action \(A\), and utility \(U(S,A)\), let the pragmatic observation \(\underline{Y}\) be obtained at rate \(R\) satisfying \(I_p(\underline{S};\underline{Y}) \le R\), and let the communication cost be given by the pragmatic cost of information \(\mathrm{CoI}_p(R)\). The Lagrangian relaxation of the constrained utility maximization problem
\begin{equation}
\max_{\pi, \, P(\underline{Y}|S)} \; \mathbb{E}_{S,\underline{Y}}[U(S, \pi(\underline{Y}))] \quad \text{subject to} \quad \mathrm{CoI}_p(R) \le P_{\max}
\end{equation}
yields the \emph{global pragmatic Lagrangian}
\begin{equation}
\mathcal{L}_{\mathrm{global}}(\pi, R; \lambda) = \mathbb{E}_{S,\underline{Y}}[U(S, \pi(\underline{Y}))] - \lambda \, \mathrm{CoI}_p(R),
\label{eq:global-lagrangian-static-thm}
\end{equation}
where \(\lambda > 0\) is the Lagrange multiplier (shadow price of communication resources). The optimal policy \(\pi^*\) and optimal rate \(R^*\) satisfy the following necessary and sufficient conditions:
\begin{align}
\pi^*(\underline{Y}) &= \arg\max_{a} \mathbb{E}_{S|\underline{Y}}[U(S,a)], \label{eq:static-policy-opt} \\
\lambda \, \mathrm{CoI}_p'(R^*) &= \frac{\partial}{\partial R} \mathbb{E}_{S,\underline{Y}}[U(S, \pi^*(\underline{Y}))], \label{eq:static-rate-opt}
\end{align}
where the derivative in (\ref{eq:static-rate-opt}) is understood in the sense of marginal utility gain with respect to the rate. Moreover, the value of the Lagrangian at the optimum equals the maximum net benefit, that is behavioral capacity:
\begin{equation}
\max_{\pi, R} \mathcal{L}_{\mathrm{global}}(\pi, R; \lambda) = \max_{P(\underline{Y}|S)} \left\{ \mathbb{E}_{S,\underline{Y}}[U(S, a^*(\underline{Y}))] - \lambda \, \mathrm{CoI}_p(R) \right\},
\end{equation}
with \(a^*(\underline{Y}) = \arg\max_a \mathbb{E}_{S|\underline{Y}}[U(S,a)]\).
\end{theorem}

\begin{proof}[Sketch of Proof]
The Lagrangian is concave in \(\pi\) and \(R\) under mild regularity conditions (e.g., \(U\) bounded, \(\mathrm{CoI}_p\) convex). The first-order conditions follow from the concavity and the fact that the optimal policy is the Bayes decision rule. The marginal utility gain in (\ref{eq:static-rate-opt}) is obtained by differentiating the optimal value with respect to \(R\), which equals the rate of improvement of the expected utility when the constraint on mutual information is relaxed. The global optimum is achieved when the marginal benefit of increasing the rate equals its marginal cost, scaled by the shadow price.
\end{proof}

\begin{corollary}[Single-Sided Lagrangian Extensions]
\label{cor:single-sided-lagrangian}
For asymmetric systems where only one side of the communication link is coarsened, the corresponding single-sided Lagrangians are obtained by substituting the appropriate CoI functions:
\begin{itemize}
 \sloppy   \item \textbf{Perceptual pragmatic Lagrangian}:  
    \(\displaystyle \mathcal{L}_{\mathrm{perc}}(\pi, R; \lambda) = \mathbb{E}_{S,\underline{Y}}[U(S, \pi(\underline{Y}))] - \lambda \, \mathrm{CoI}_p^{\text{perc}}(R)\),  
    where \(\mathrm{CoI}_p^{\text{perc}}(R)\) is the perceptual cost (sensing power for rate \(R\)).
    \item \textbf{Expressive pragmatic Lagrangian}:  
    \(\displaystyle \mathcal{L}_{\mathrm{expr}}(\pi, R; \lambda) = \mathbb{E}_{S,\underline{Y}}[U(S, \pi(\underline{Y}))] - \lambda \, \mathrm{CoI}_p^{\text{expr}}(R)\),  
    where \(\mathrm{CoI}_p^{\text{expr}}(R)\) is the expressive cost (transmit power for rate \(R\)).
\end{itemize}
The optimality conditions for each case are identical to (\ref{eq:static-policy-opt}) and (\ref{eq:static-rate-opt}) with the corresponding CoI function.
\end{corollary}

\begin{example}[Linear Quadratic Gaussian (LQG) with Perceptual Cost]
Consider a scalar state \(S \sim \mathcal{N}(0, \sigma_s^2)\), quadratic utility \(U(S,A) = -(S-A)^2\), and perceptual cost \(\mathrm{CoI}_p^{\text{perc}}(R) = \sigma^2 (2^{2R}/\Omega^2 - 1)\) with isoteleic volume \(\Omega\). The Lagrangian \(\mathcal{L}_{\mathrm{perc}}(P_{\text{tx}}; \lambda) = -\sigma_s^2 / (1 + \Omega^2 P_{\text{tx}}/\sigma^2) - \lambda P_{\text{tx}}\) yields the optimal power
\begin{equation}
P_{\text{tx}}^* = \max\left\{ \sigma^2 \left( \frac{\sqrt{\sigma_s^2}}{\Omega \sqrt{\lambda}} - 1 \right), 0 \right\}.
\end{equation}
This closed-form solution illustrates the trade-off between control performance and communication cost: the optimal power decreases as the shadow price \(\lambda\) increases, and the pragmatic abstraction \(\Omega\) amplifies the effective signal-to-noise ratio.
\end{example}

\begin{example}[Gaussian Source over Gaussian Channel: Closed-Form Bilateral Pragmatic Lagrangian]
\label{ex:gaussian-bilateral-lagrangian}
We now instantiate Theorem~\ref{thm:static-joint} for the canonical case of a Gaussian source transmitted over an AWGN channel, where both the source and the observation are coarsened pragmatically (bilateral abstraction). This provides a closed-form verification of the static joint optimization conditions.
\end{example}

Let the state (source) be \(X \sim \mathcal{N}(0, P)\) and the pragmatic observation be \(\underline{Y}\) obtained through an AWGN channel with noise variance \(\sigma^2\) and transmit power \(P_{\text{tx}}\). The isoteleic volume \(\Omega \ge 1\) captures the pragmatic equivalence. From Section~\ref{subsec:gaussian-joint-analysis}, the bilateral pragmatic value and cost functions are (cf.~\eqref{eq:gaussian-voi} and~\eqref{eq:gaussian-coi}):
\begin{align}
\Phi_p(R) &= P\left(1 - \frac{2^{-2R}}{\Omega^4}\right), \label{eq:ex-phi} \\
\mathrm{CoI}_p(R) &= \sigma^2\left(\frac{2^{2R}}{\Omega^4} - 1\right), \label{eq:ex-coi}
\end{align}
where \(R\) is the pragmatic rate (prabits per channel use). Substituting into the global Lagrangian of Theorem~\ref{thm:static-joint} (after maximizing over the policy \(\pi\), which yields \(\Phi_p(R)\)), we obtain the explicit Lagrangian:
\begin{equation}
\mathcal{L}(R) = \Phi_p(R) - \lambda \, \mathrm{CoI}_p(R)
= P\left(1 - \frac{2^{-2R}}{\Omega^4}\right) - \lambda \sigma^2\left(\frac{2^{2R}}{\Omega^4} - 1\right).
\label{eq:ex-lagrangian}
\end{equation}
Equivalently, symmetrically:
\begin{equation}
\mathcal{L}(R) = (P + \lambda \sigma^2) - \frac{1}{\Omega^4}\left( P \cdot 2^{-2R} + \lambda \sigma^2 \cdot 2^{2R} \right).
\label{eq:ex-lagrangian-sym}
\end{equation}

\begin{figure}[htbp]
    \centering
    \subfloat[Net benefit maximization: \(\Phi_p(R)\), \(\lambda\mathrm{CoI}_p(R)\), and \(\mathcal{L}(R)\). Parameters: \(\Omega=2\), \(\lambda=0.5\), \(P/\sigma^2=1\).]
    {
        \includegraphics[height=0.3\textheight, keepaspectratio]{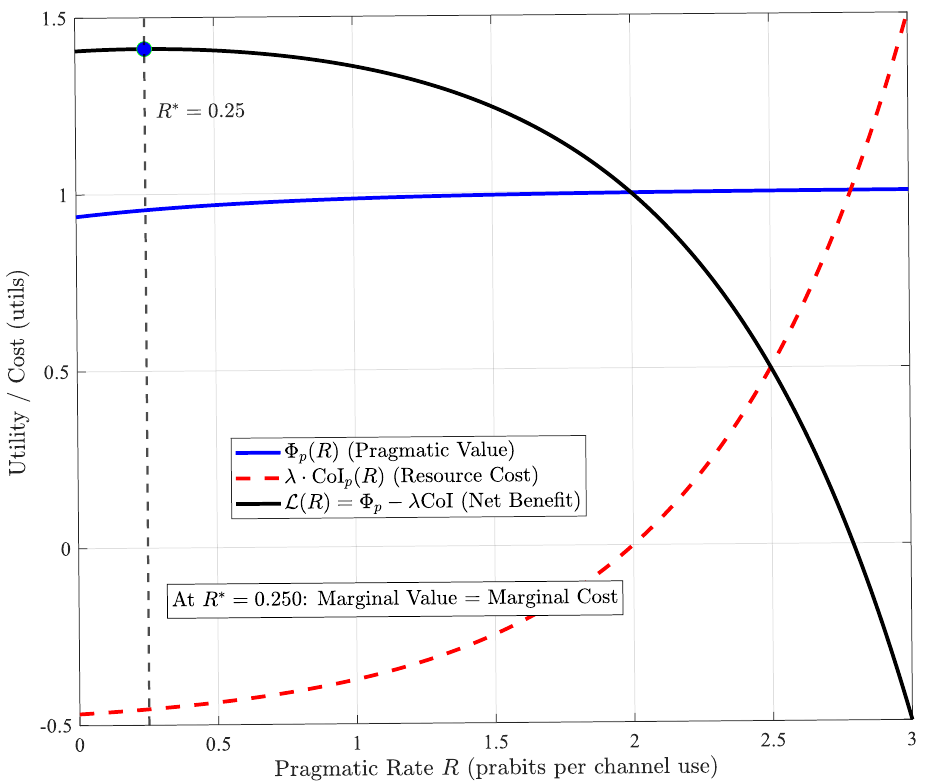}
        \label{fig:pragmatic_Lagrangian}
    }
    \hfill
    \subfloat[Comparison of \(\mathcal{L}(R)\) for different \(\Omega\). Parameters: \(P/\sigma^2=1\), \(\lambda=0.5\).]
    {
        \includegraphics[height=0.3\textheight, keepaspectratio]{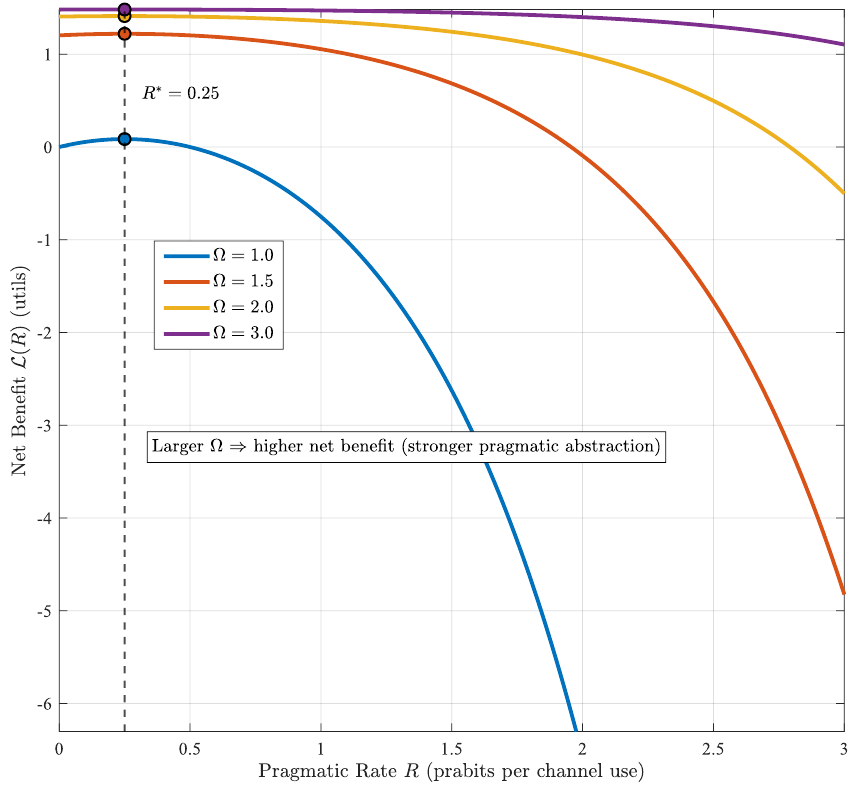}
        \label{fig:net_benefit_vs_Omega}
    }
    \caption{Global pragmatic Lagrangian and the effect of pragmatic abstraction.}
    \label{fig:lagrangian_combined}
\end{figure}

Figure~\ref{fig:lagrangian_combined} illustrates the global Lagrangian structure. Subfigure (a) shows the pragmatic value, the weighted cost, and the resulting net benefit, with the optimal rate \(R^*\) marked where marginal value equals marginal cost. Subfigure (b) demonstrates that increasing \(\Omega\) uniformly shifts the net benefit curve upward while leaving \(R^*\) unchanged, confirming that pragmatic abstraction amplifies the effective signal-to-noise ratio without altering the optimal communication rate.

Applying the first-order optimality condition from Theorem~\ref{thm:static-joint} (\(\mathcal{L}'(R^*) = 0\)):
\begin{equation}
\frac{d\mathcal{L}}{dR} = \frac{2\ln 2}{\Omega^4}\left( P \cdot 2^{-2R^*} - \lambda \sigma^2 \cdot 2^{2R^*} \right) = 0,
\label{eq:ex-first-order}
\end{equation}
we obtain the optimal rate:
\begin{equation}
R^* = \frac{1}{4} \log_2\left( \frac{P}{\lambda \sigma^2} \right)^+,
\label{eq:ex-optimal-rate}
\end{equation}
where \((x)^+ = \max\{x, 0\}\) ensures non-negativity. The maximized Lagrangian (net benefit) is then:
\begin{equation}
\mathcal{L}^* = \mathcal{L}(R^*) = P + \lambda \sigma^2 - \frac{2\sqrt{P \lambda \sigma^2}}{\Omega^4},
\label{eq:ex-max-lagrangian}
\end{equation}
provided \(P > \lambda \sigma^2\); otherwise \(R^* = 0\) and \(\mathcal{L}^* = 0\) (the system opts not to communicate).

This maximized net benefit is precisely the behavioral capacity \(\mathcal{E}_p(\lambda)\) defined in Definition~\ref{def:pragmatic-efficiency-bound}. Substituting \(\lambda = (P/\sigma^2) 2^{-4R^*}\) into \eqref{eq:ex-max-lagrangian} yields the behavioral capacity as a function of the optimal rate:
\begin{equation}
\mathcal{E}_p(R^*) = P \left( 1 + 2^{-4R^*} - \frac{2 \cdot 2^{-2R^*}}{\Omega^4} \right), \quad R^* \ge 0.
\label{eq:ex-Ep-vs-R}
\end{equation}
Equation \eqref{eq:ex-max-lagrangian} and \eqref{eq:ex-Ep-vs-R} explicitly characterize the behavioral capacity in terms of the shadow price and the optimal communication rate, respectively.

\begin{figure}[htbp]
    \centering
    \subfloat[Behavioral capacity \(\mathcal{E}_p\) vs. optimal pragmatic rate \(R^*\). Larger \(\Omega\) yields higher capacity, saturating at \(P\). Parameters: \(P=1\), \(\sigma^2=1\).]{
        \includegraphics[height=0.3\textheight, keepaspectratio]{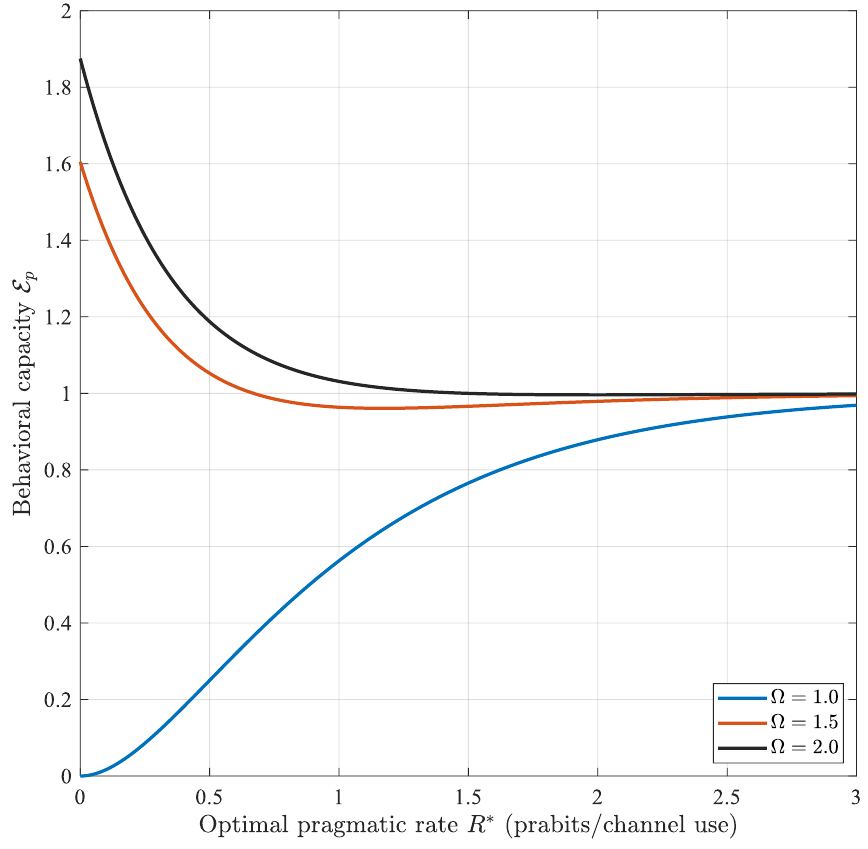}
        \label{fig:Ep-vs-Rstar}
    }
    \hfill
    \subfloat[Behavioral capacity \(\mathcal{E}_p\) vs. shadow price \(\lambda\). Capacity decreases with \(\lambda\) and vanishes at \(\lambda=P/\sigma^2\). Larger \(\Omega\) reduces sensitivity. Parameters: \(P=1\), \(\sigma^2=1\).]{
        \includegraphics[height=0.3\textheight, keepaspectratio]{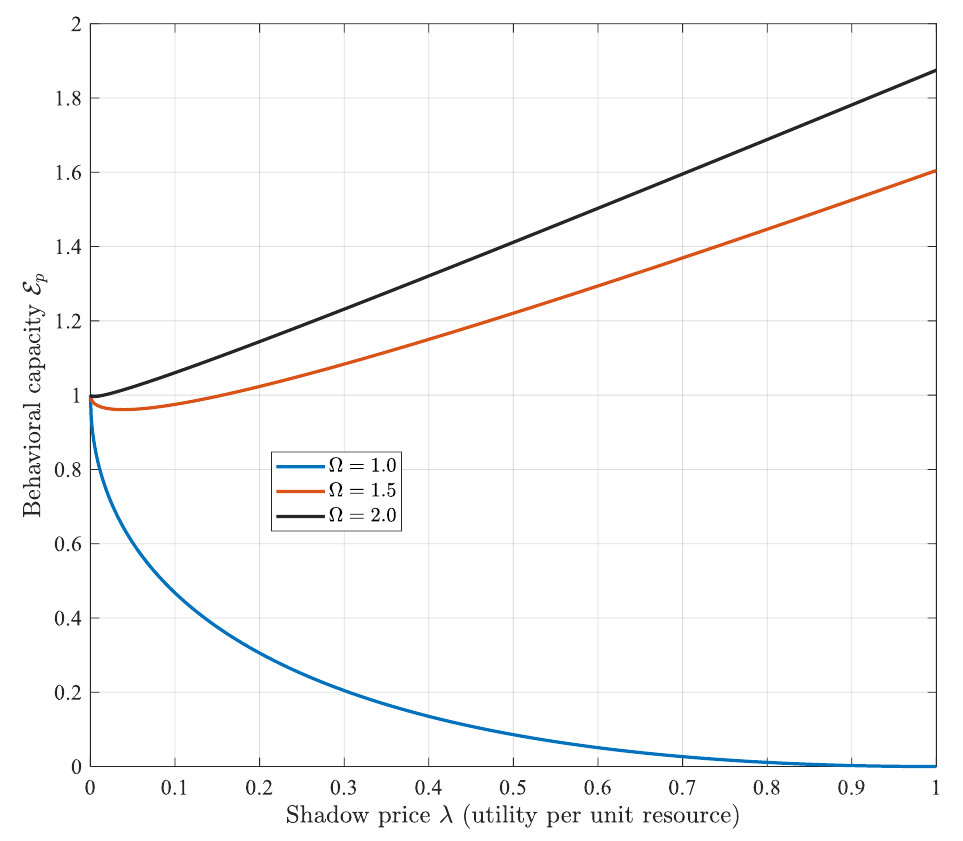}
        \label{fig:Ep-vs-lambda}
    }
    \caption{Behavioral capacity as a function of optimal rate and resource shadow price.}
    \label{fig:Ep_combined}
\end{figure}

Figure~\ref{fig:Ep_combined} further characterizes the behavioral capacity. Subfigure (a) shows how the capacity increases with the invested communication rate, with diminishing returns and saturation at the perfect-information utility \(P\). A larger \(\Omega\) shifts the curve upward, indicating that the same rate yields higher net utility under stronger pragmatic abstraction. Subfigure (b) reveals the decay of capacity as resources become more expensive; the zero-capacity threshold marks the point where communication ceases. Larger \(\Omega\) flattens the decline, demonstrating that abstraction provides robustness against rising resource costs.

This closed-form solution explicitly confirms the two optimality conditions of Theorem~\ref{thm:static-joint}:
\begin{enumerate}
    \item The optimal policy \(\pi^*\) is the Bayes decision rule (implicit in the definition of \(\Phi_p(R)\));
    \item The marginal value equals the marginal cost: \(\Phi_p'(R^*) = \lambda \, \mathrm{CoI}_p'(R^*)\), which is precisely the balance expressed in \eqref{eq:ex-first-order}.
\end{enumerate}

Moreover, the expression \eqref{eq:ex-max-lagrangian} reveals the four-way trade-off among source power \(P\), channel noise \(\sigma^2\), resource shadow price \(\lambda\), and pragmatic abstraction \(\Omega\). When \(\Omega = 1\), it recovers the classical Gaussian benchmark; when \(\Omega > 1\), the effective SNR is amplified by \(\Omega^4\) on both sides, drastically reducing the power needed to achieve a given net benefit. 

\subsubsection{Dynamic Joint Optimization: Sequential Decision Making}
\label{subsubsec:dynamic-joint}
\leavevmode\newline\indent
In dynamic systems, the state evolves over time according to a controlled Markov process, and the decision-maker must balance immediate rewards against future consequences. The system is modeled as a Partially Observable Markov Decision Process (POMDP) with communication constraints. At each time step \(t\), the system state \(S_t\) evolves as \(S_{t+1} \sim P(\cdot|S_t, A_t)\), and the decision-maker receives a pragmatic observation \(\underline{Y}_t\) about \(S_t\) at a rate \(R_t\) that incurs a cost \(\mathrm{CoI}_p(R_t)\). The goal is to find a policy \(\pi = \{\pi_t\}_{t\ge 0}\) mapping the history of pragmatic observations to actions, maximizing the expected discounted sum of utilities subject to a total communication resource budget.

\begin{theorem}[Dynamic Joint Optimization for POMDP]
\label{thm:dynamic-joint}
For an infinite-horizon discounted POMDP with state space \(\mathcal{S}\), action space \(\mathcal{A}\), observation space \(\mathcal{Y}\), transition kernel \(P(s'|s,a)\), observation kernel \(P(y|s', R)\) (where the communication rate \(R\) affects the observation quality), utility \(U(s,a)\), discount factor \(\gamma \in (0,1)\), and pragmatic communication cost \(\mathrm{CoI}_p(R)\) per step, the Lagrangian relaxation of the constrained expected discounted utility maximization
\begin{equation}
\max_{\pi, \{R_t\}} \; \mathbb{E}\left[ \sum_{t=0}^{\infty} \gamma^t U(S_t, A_t) \right] \quad \text{subject to} \quad \sum_{t=0}^{\infty} \gamma^t \mathrm{CoI}_p(R_t) \le C_{\text{total}}
\end{equation}
yields the \emph{dynamic pragmatic Lagrangian}
\begin{equation}
\mathcal{L}_{\mathrm{dyn}}(\pi, \{R_t\}; \lambda) = \mathbb{E}\left[ \sum_{t=0}^{\infty} \gamma^t \left( U(S_t, A_t) - \lambda \, \mathrm{CoI}_p(R_t) \right) \right],
\label{eq:dynamic-lagrangian-thm}
\end{equation}
where \(\lambda > 0\) is the Lagrange multiplier for the total resource constraint. 

Let \(b \in \Delta(\mathcal{S})\) denote the belief state (the posterior distribution over the state given the history of actions and observations). The optimal policy and rate schedule satisfy the Bellman equation \cite{Bellman1954} in belief-state form:
\begin{equation}
\label{eq:bellman-belief}
V(b) = \max_{a, R \ge 0} \left\{
\sum_{s} b(s) U(s,a) - \lambda \, \mathrm{CoI}_p(R)
+ \gamma \sum_{s, s', y} b(s) P(s'|s,a) P(y|s', R) \, V(b_{y}')
\right\},
\end{equation}
where \(b_{y}'\) is the updated belief state after observing \(y\) (and taking action \(a\) with rate \(R\)):
\begin{equation}
b_{y}'(s') = \frac{P(y|s', R) \sum_{s} P(s'|s,a) b(s)}
 {\sum_{\tilde{s}} P(y|\tilde{s}, R) \sum_{s} P(\tilde{s}|s,a) b(s)}.
\label{eq:belief-update}
\end{equation}

The optimal action \(a^*\) and optimal rate \(R^*\) at belief \(b\) are characterized by the following first-order conditions:
\begin{align}
a^* &= \arg\max_a \left\{ \sum_s b(s) U(s,a) + \gamma \sum_s b(s) \sum_{s'} P(s'|s,a) \sum_y P(y|s', R^*) \, V(b_{y}') \right\}, \label{eq:dynamic-action-opt} \\
\lambda \, \mathrm{CoI}_p'(R^*) &= \frac{\partial}{\partial R} \left( \gamma \sum_s b(s) \sum_{s'} P(s'|s,a^*) \sum_y P(y|s', R) \, V(b_{y}') \right) \Bigg|_{R=R^*}, \label{eq:dynamic-rate-opt}
\end{align}
where \(b_{y}'\) is given by (\ref{eq:belief-update}) and depends on \(a^*\) and \(R\). The derivative in (\ref{eq:dynamic-rate-opt}) accounts for the improvement in future value due to better state estimation (sharper belief) resulting from a higher rate \(R\).
\end{theorem}

\begin{proof}[Sketch of Proof]
The proof follows from the principle of optimality applied to the belief-state MDP, which is equivalent to the original POMDP. Before applying the Bellman equation, we first establish its well-definedness. Define the Bellman operator $\mathcal{T}$ on the space of bounded continuous functions on $\Delta(\mathcal{S})$ by
\[
(\mathcal{T}V)(b) \triangleq \max_{a, R \ge 0} \left\{
\sum_{s} b(s) U(s,a) - \lambda \, \mathrm{CoI}_p(R)
+ \gamma \sum_{s,s',y} b(s) P(s'|s,a) P(y|s',R) \, V(b_y')
\right\}.
\]
For any two bounded functions $V_1, V_2$ and any belief $b$, the transition kernels sum to $1$, which gives
\[
|(\mathcal{T}V_1)(b) - (\mathcal{T}V_2)(b)|
\le \gamma \sup_{a,R} \sum_{s,s',y} b(s) P(s'|s,a) P(y|s',R) \, |V_1(b_y') - V_2(b_y')|
\le \gamma \|V_1 - V_2\|_\infty.
\]
Taking the supremum over $b$ yields $\|\mathcal{T}V_1 - \mathcal{T}V_2\|_\infty \le \gamma \|V_1 - V_2\|_\infty$. Since $\gamma < 1$, $\mathcal{T}$ is a contraction. By the Banach fixed-point theorem, there exists a unique fixed point $V^*$ in the space of bounded continuous functions, ensuring that the belief-state Bellman equation (\ref{eq:bellman-belief}) is well-defined. The dynamic isoteleia mapping $e_{V^*}$ defined through $V^*$ is then uniquely determined and recovers the static mapping $e$ when $\gamma = 0$.

With this established, the augmented reward process $U(S_t, A_t) - \lambda \mathrm{CoI}_p(R_t)$ leads to the belief-state Bellman equation (\ref{eq:bellman-belief}), where the expectation over the next belief accounts for both state transition and observation likelihood. The optimality conditions (\ref{eq:dynamic-action-opt}) and (\ref{eq:dynamic-rate-opt}) are obtained by differentiating the right-hand side of (\ref{eq:bellman-belief}) with respect to $a$ and $R$, respectively, under suitable regularity conditions (e.g., concavity of the value function and convexity of $\mathrm{CoI}_p$). The coupling between action and rate appears through the belief update, as both influence the future belief distribution.
\end{proof}

To avoid a conceptual gap between the static hierarchical framework (Section \ref{section_II}--\ref{section_IX}) and the dynamic POMDP formulation, we emphasize that the pragmatic variable $\underline{S}$ appearing in the Bellman equation (\ref{eq:belief-update}) is induced by the \emph{dynamic} isoteleia mapping $e_V$ defined in Definition \ref{definition:Dynamic-Isoteleia-Mapping}. The resulting pragmatic equivalence classes are therefore functions of the future value function $V$. In the steady-state case, $e_V$ collapses to the static mapping $e$ used in the coding theorems, ensuring that the hierarchical entropy inequalities remain valid as limiting bounds.

\begin{remark}[Interpretation and Design Insights]
The dynamic Lagrangian framework provides several key insights for system design:
\begin{enumerate}
    \item \textbf{Adaptive Resource Allocation}: The optimal rate \(R^*\) varies with the belief state, allocating more resources when the belief is diffuse (high uncertainty) or when future rewards are particularly sensitive to estimation accuracy. This yields a \emph{value-of-information-guided} scheduling policy.
    \item \textbf{Belief-Dependent Optimization}: Unlike the original erroneous formulation, the correct belief-state Bellman equation explicitly incorporates the observation model and the impact of \(R\) on the quality of future beliefs. The action and rate decisions are coupled through the belief update, reflecting the inherent trade-off between control and sensing/communication.
    \item \textbf{Connection to CoI and VoI}: The term \(\lambda \mathrm{CoI}_p(R)\) represents the instantaneous communication cost, while the future value \(V(b_{y}')\) encodes the long-term utility gain from improved state information. The optimal policy balances immediate and future value, mirroring the economic principle of marginal benefit equals marginal cost.
\end{enumerate}
\end{remark}

\begin{example}[Linear Quadratic Gaussian (LQG) with Linear Communication Cost]
Consider a scalar linear system \(S_{t+1} = a S_t + b A_t + W_t\) with \(W_t \sim \mathcal{N}(0, \sigma_w^2)\), quadratic utility \(U(S_t,A_t) = -q S_t^2 - r A_t^2\), and communication cost \(\mathrm{CoI}_p(R_t) = c R_t\) (linear in rate). Assume Gaussian observations with precision proportional to \(R_t\), i.e., \(Y_t = S_t + Z_t\), \(Z_t \sim \mathcal{N}(0, \sigma^2/R_t)\). The belief state is Gaussian, characterized by its mean and variance. The optimal control law remains linear in the estimated state (separation principle holds for LQG with quadratic cost and Gaussian noise). The optimal rate \(R^*\) is constant over time when the system is stationary, and can be found by solving a Riccati equation augmented with the rate penalty. The steady-state rate decreases with the shadow price \(\lambda\) and depends on the system dynamics through the controllability and observability parameters. This example demonstrates how the dynamic Lagrangian framework provides a tractable method for co-designing communication and control in practical systems.
\end{example}

The dynamic joint optimization framework completes the pragmatic information-theoretic spectrum, providing a unified mathematical tool for designing task-oriented communication and control systems that optimally trade off information value, physical resource consumption, and long-term performance.

\section{Conclusions}
\label{section_XII}

\subsection{Summary of Contributions}

We have established a mathematical theory of pragmatic information that connects communication, control, and decision-making within a coherent framework. The theory is built upon the isoteleia mapping $e:\underline{\mathcal{W}}\to 2^{\tilde{\mathcal{W}}}$, which formalizes equifinality---the principle that distinct semantic paths converging to the same optimal action are pragmatically equivalent. This induces a three-tier hierarchy of syntactic, semantic, and pragmatic information, with each successive abstraction discarding task-irrelevant distinctions.

The main contributions are threefold. First, we have developed a complete set of information-theoretic measures at the pragmatic level---entropy $H_p(\underline{W})$, up/down mutual information $I^p(\underline{X};\underline{Y})$ and $I_p(\underline{X};\underline{Y})$, channel capacity $C_p$, and rate-distortion function $R_p(D)$---and established the fundamental hierarchies $H_p \le H_s \le H$, $C_p \ge C_s \ge C$, and $R_p \le R_s \le R$. Second, we have introduced pragmatic value of information (VoI) and pragmatic cost of information (CoI) as decision-theoretic and economic duals to rate-distortion and capacity, respectively, and formulated a Lagrangian dual framework for cross-layer optimization. The pragmatic efficiency functional $\mathcal{E}_p(\lambda) = \sup_R[\Phi_p(R) - \lambda\,\mathrm{CoI}_p(R)]$ defines a behavioral capacity for resource-constrained intelligent systems under the stated assumptions. Third, we have proved three coding theorems---lossless source coding, channel coding, and rate-distortion coding---that generalize Shannon's classical results and their semantic counterparts. We have further extended the framework to continuous messages, deriving closed-form Gaussian expressions, and formulated joint optimization of communication, control, and decision-making in both static and dynamic settings.

\subsection{Engineering Implications and Deployment Pathways}

The pragmatic information theory developed in this paper may have implications for a range of engineering disciplines and applications. By shifting the focus from symbol fidelity to the effectiveness of information in guiding actions, the framework offers a mathematical language for task-oriented communication, networked control, autonomous systems, and embodied AI, and suggests conceptual directions for human-machine interaction and social-cognitive systems. The Value of Information (VoI) and Cost of Information (CoI) offer a unified economic and decision-theoretic basis for resource allocation, enabling system designers to balance utility gains against physical resource consumption. The Lagrangian dual framework provides a principled approach for cross-layer optimization, bridging the gap between information-theoretic limits and engineering practice. The framework also offers new perspectives on emerging challenges such as semantic communication, edge AI, intelligent transportation, and human learning, where information must be evaluated by its effectiveness in achieving goals rather than its syntactic accuracy.

Looking toward engineering practice, a critical path for deploying the proposed theory lies in the realization of the isoteleia mapping---the core mathematical object that groups semantically distinct signals into pragmatically equivalent action classes. In real-world scenarios where the utility function $U$ is not explicitly specified, the isoteleia mapping can be acquired implicitly through large-scale data-driven learning. Specifically, deep neural networks can serve as universal approximators of the optimal action function, effectively inducing the equivalence relation from data without requiring an analytic utility model. This paradigm aligns naturally with inverse reinforcement learning and reinforcement learning from human feedback (RLHF), where the network learns to map high-dimensional observations (semantic inputs) to optimal action classes (pragmatic outputs) by minimizing task-oriented loss functions.

Furthermore, in practical deployment, the isoteleia mapping should not be a fixed, static structure. It must adapt dynamically to changing environmental contexts and evolving task objectives. For instance, in autonomous driving, the mapping from a visual scene to the action ``brake'' may depend critically on the current speed, road condition, and surrounding traffic. This context-dependent behavior suggests an extension of the isoteleia mapping to a conditional form depending on the contextual state or scene encoding. In terms of neural implementation, this can be achieved through context-gating mechanisms or conditional neural networks that modulate the mapping based on the current environment.

Equally important is the principle of determinism in decision-making. For safety-critical systems, the mapping from observation to action should be deterministic: a given input must produce a unique, consistent output. This ensures reproducibility, facilitates fault diagnosis, and aligns with the $\arg\max$ operation in the decision axiom. In practice, this can be realized by deploying deterministic policies at inference time, or by using techniques such as Gumbel-Softmax with temperature annealing to ensure that the learned mapping converges to a hard, deterministic decision boundary.

To concretize these ideas into a deployable pipeline, we envision an end-to-end pragmatic information system consisting of three functional stages: (i) an offline training phase, where a deep neural network is trained on large-scale task-specific data to approximate the conditional isoteleia mapping; (ii) a context-adaptive inference phase, where the network dynamically adjusts its decision boundaries based on real-time contextual inputs; and (iii) an online adaptation mechanism that continuously refines the mapping through reinforcement learning or online fine-tuning, enabling the system to respond to non-stationary environments. This pipeline directly addresses the practical challenges of implementing the theory in real-world applications, providing a bridge between the mathematical abstraction of the isoteleia mapping and the constraints of physical hardware and real-time computation.

\subsection{Limitations and Future Directions}

Several limitations of the current framework should be acknowledged. The theory assumes stationarity and ergodicity, which may not hold in social, economic, or biological systems where statistical properties evolve over time. The coding theorems are asymptotic in nature; finite-blocklength analyses and the corresponding dispersion bounds remain to be developed. The framework primarily addresses single-agent or centralized settings; multi-agent systems with conflicting or adversarial utilities require game-theoretic extensions. The dynamic optimization relies on belief-state Bellman equations that are computationally intractable in high dimensions, necessitating approximate solution methods. Finally, while deep neural networks offer a practical path for learning isoteleia mappings, they introduce approximation errors and generalization gaps not captured by the current theoretical guarantees.

Looking forward, several directions merit further investigation. The extension to non-stationary and non-ergodic processes calls for time-varying entropies and online adaptive coding strategies. Finite-blocklength analysis would provide dispersion bounds and practical guidelines for system design. Multi-agent and networked pragmatic information systems---including collaborative autonomous fleets, adversarial communication, and competitive economic interactions---require the development of game-theoretic equilibrium concepts and incentive-compatible protocols. The integration of deep learning with pragmatic coding demands provably efficient algorithms with finite-sample guarantees and convergence rates. Approximate dynamic programming techniques, such as deep reinforcement learning and function approximation, can be brought to bear on the high-dimensional belief-state optimization problem. Beyond these technical directions, the framework opens avenues for cross-disciplinary applications in cognitive science, linguistics, and economics, where information exchange is inherently goal-directed and context-dependent.

The mathematical theory of pragmatic information provides a unified language for goal-directed information processing under resource constraints. While its full consequences will be revealed through future theoretical and empirical advances, the framework lays a solid foundation for this emerging direction in information theory.


\appendix
\subsection{Proof of Theorem \ref{thm:pragmatic-typical-properties}}
\label{app:a1_pragmatic}

\begin{proof}
Property (1) follows directly from the definition of \(\underline{A}_\epsilon^{(n)}\). By Theorem \ref{thm:three-tier-aep}, the probability of the event \(\underline{W}^n\in\underline{A}_\epsilon^{(n)}\) tends to 1 as \(n\to\infty\). Thus for any \(\epsilon>0\), there exists \(n_0\) such that for all \(n\ge n_0\), \(\Pr\{\underline{W}^n\in\underline{A}_\epsilon^{(n)}\}>1-\epsilon\), proving property (2). Summing over the set \(\underline{A}_\epsilon^{(n)}\), using property (2) and the bounds from property (1), we obtain
\begin{equation}
(1-\epsilon)2^{n(H_p(\underline{W})-\epsilon)} \le |\underline{A}_\epsilon^{(n)}| \le 2^{n(H_p(\underline{W})+\epsilon)},
\end{equation}
which proves property (3).
\end{proof}

\subsection{Proof of Theorem \ref{thm:isoteleia-typical-properties}}
\label{app:a2_isoteleia}

\begin{proof}
Property (1) follows directly from the definition of \(E_\epsilon^{(n)}(\underline{w}^n)\). To prove the left inequality of property (2), we write
\begin{align}
P(\underline{w}^n) &= \sum_{\tilde{w}^n\in E_\epsilon^{(n)}(\underline{w}^n)} P(\tilde{w}^n) \nonumber \\
&\le \sum_{\tilde{w}^n\in E_\epsilon^{(n)}(\underline{w}^n)} P(\underline{w}^n) 2^{-n(H_s(\tilde{W})-H_p(\underline{W})-\epsilon)} \nonumber \\
&= P(\underline{w}^n)\,|E_\epsilon^{(n)}(\underline{w}^n)|\,2^{-n(H_s(\tilde{W})-H_p(\underline{W})-\epsilon)}.
\end{align}
Dividing by \(P(\underline{w}^n)>0\) yields
\begin{equation}
|E_\epsilon^{(n)}(\underline{w}^n)| \ge 2^{n(H_s(\tilde{W})-H_p(\underline{W})-\epsilon)}.
\end{equation}
For the right inequality, we have
\begin{align}
P(\underline{w}^n) &= \sum_{\tilde{w}^n\in E_\epsilon^{(n)}(\underline{w}^n)} P(\tilde{w}^n) \nonumber \\
&\ge \sum_{\tilde{w}^n\in E_\epsilon^{(n)}(\underline{w}^n)} P(\underline{w}^n) 2^{-n(H_s(\tilde{W})-H_p(\underline{W})+\epsilon)} \nonumber \\
&= P(\underline{w}^n)\,|E_\epsilon^{(n)}(\underline{w}^n)|\,2^{-n(H_s(\tilde{W})-H_p(\underline{W})+\epsilon)},
\end{align}
so
\begin{equation}
|E_\epsilon^{(n)}(\underline{w}^n)| \le 2^{n(H_s(\tilde{W})-H_p(\underline{W})+\epsilon)}.
\end{equation}
This completes the proof.
\end{proof}

\subsection{Proof of Theorem \ref{theorem:Reified-Typical-Set}}
\label{app:a3_reified}

\begin{proof}
Property (1) follows directly from the definition of \(B_\epsilon^{(n)}(\underline{w}^n)\). To prove the left inequality of property (2), we write
\begin{align}
P(\underline{w}^n) &= \sum_{w^n\in B_\epsilon^{(n)}(\underline{w}^n)} P(w^n) \nonumber \\
&\le \sum_{w^n\in B_\epsilon^{(n)}(\underline{w}^n)} P(\underline{w}^n) 2^{-n(H(W)-H_p(\underline{W})-\epsilon)} \nonumber \\
&= P(\underline{w}^n)\,|B_\epsilon^{(n)}(\underline{w}^n)|\,2^{-n(H(W)-H_p(\underline{W})-\epsilon)}.
\end{align}
Dividing by \(P(\underline{w}^n)>0\) gives
\begin{equation}
|B_\epsilon^{(n)}(\underline{w}^n)| \ge 2^{n(H(W)-H_p(\underline{W})-\epsilon)}.
\end{equation}
For the right inequality,
\begin{align}
P(\underline{w}^n) &= \sum_{w^n\in B_\epsilon^{(n)}(\underline{w}^n)} P(w^n) \nonumber \\
&\ge \sum_{w^n\in B_\epsilon^{(n)}(\underline{w}^n)} P(\underline{w}^n) 2^{-n(H(W)-H_p(\underline{W})+\epsilon)} \nonumber \\
&= P(\underline{w}^n)\,|B_\epsilon^{(n)}(\underline{w}^n)|\,2^{-n(H(W)-H_p(\underline{W})+\epsilon)},
\end{align}
so
\begin{equation}
|B_\epsilon^{(n)}(\underline{w}^n)| \le 2^{n(H(W)-H_p(\underline{W})+\epsilon)}.
\end{equation}
This completes the proof.
\end{proof}

\subsection{Proof of Theorem \ref{theorem:Pragmatically-Joint-AEP-Trans}}
\label{app:a4_Pragmatic-Joint-AEP-Trans}

\begin{proof}
Properties (1)-(3) follow directly from the definition of \(\underline{A}_\epsilon^{(n)}\) and the weak law of large numbers applied to the i.i.d. pragmatic sequence pair, with the aggregated probabilities over the reified equivalence classes. For property (4), since \(\dot{X}^n\) and \(\dot{Y}^n\) are independent, we have
\[
\Pr\{(\dot{\underline{X}}^n,\dot{\underline{Y}}^n)\in\underline{A}_\epsilon^{(n)}\} = \sum_{(\underline{x}^n,\underline{y}^n)\in\underline{A}_\epsilon^{(n)}} P(\underline{x}^n)P(\underline{y}^n).
\]
Using the bounds on \(P(\underline{x}^n)\) and \(P(\underline{y}^n)\) from the entropy hierarchy \(H_p(\underline{X})\le H_s(\tilde{X})\le H(X)\) and the size bound on \(\underline{A}_\epsilon^{(n)}\), we obtain
\[
\Pr\{(\dot{\underline{X}}^n,\dot{\underline{Y}}^n)\in\underline{A}_\epsilon^{(n)}\} \le 2^{n(H_p(\underline{X},\underline{Y})+\epsilon)}2^{-n(H(X)-\epsilon)}2^{-n(H(Y)-\epsilon)} = 2^{-n(I^p(\underline{X};\underline{Y})-3\epsilon)}.
\]
Similarly,
\[
\Pr\{(\dot{\underline{X}}^n,\dot{\underline{Y}}^n)\in\underline{A}_\epsilon^{(n)}\} \ge (1-\epsilon)2^{n(H_p(\underline{X},\underline{Y})-\epsilon)}2^{-n(H(X)+\epsilon)}2^{-n(H(Y)+\epsilon)} = (1-\epsilon)2^{-n(I^p(\underline{X};\underline{Y})+3\epsilon)}.
\]
This proves property (4).
\end{proof}

\subsection{Proof of Theorem \ref{theorem:Jointly-Reified-Typical-Set}}
\label{app:a5_Jointly-Reified-Typical-Set}
\begin{proof}
Property (1) follows directly from the definition of \(B_\epsilon^{(n)}(\underline{x}^n,\underline{y}^n)\). To prove the left inequality of property (2), we write
\[
\begin{aligned}
P(\underline{x}^n,\underline{y}^n) &= \sum_{(x^n,y^n)\in B_\epsilon^{(n)}(\underline{x}^n,\underline{y}^n)} P(x^n,y^n) \\
&\le \sum_{(x^n,y^n)\in B_\epsilon^{(n)}(\underline{x}^n,\underline{y}^n)} P(\underline{x}^n,\underline{y}^n) 2^{-n(H(X,Y)-H_p(\underline{X},\underline{Y})-\epsilon)} \\
&= P(\underline{x}^n,\underline{y}^n)\,|B_\epsilon^{(n)}(\underline{x}^n,\underline{y}^n)|\,2^{-n(H(X,Y)-H_p(\underline{X},\underline{Y})-\epsilon)}.
\end{aligned}
\]
Dividing by \(P(\underline{x}^n,\underline{y}^n)>0\) yields the lower bound. The right inequality follows similarly by reversing the inequality.
\end{proof}

\subsection{Proof of Theorem \ref{theorem:pragmatically-Joint-AEP-Compr}}
\label{a6:pragmatic-AEP-Compr-proof}

\begin{proof}
Properties (1)-(3) follow directly from the definition of \(\underline{A}_\epsilon^{(n)}\) and the weak law of large numbers applied to the i.i.d. pragmatic sequence pair, with the aggregated probabilities over the reified equivalence classes. For property (4), since \(\dot{X}^n\) and \(\hat{\dot{X}}^n\) are independent, we have
\[
\Pr\{(\dot{\underline{X}}^n,\hat{\dot{\underline{X}}}^n)\in\underline{A}_\epsilon^{(n)}\} = \sum_{(\underline{x}^n,\underline{\hat{x}}^n)\in\underline{A}_\epsilon^{(n)}} P(\underline{x}^n)P(\underline{\hat{x}}^n).
\]
Using the bounds on \(P(\underline{x}^n)\) and \(P(\underline{\hat{x}}^n)\) from the entropy hierarchy \(H_p(\underline{X})\le H_s(\tilde{X})\le H(X)\) and the size bound on \(\underline{A}_\epsilon^{(n)}\), we obtain the desired exponential bounds.
\end{proof}


\begin{thebibliography}{99}
%Classic Information Theory
\bibitem{Classicpaper_Shannon}
C. E. Shannon, ``A Mathematical Theory of Communication," \emph{The Bell System Technical Journal}, vol. 27, pp. 379-423, 623-656, July, Oct., 1948.

\bibitem{Shannon_Weaver}
C. E. Shannon and W. Weaver, \emph{The Mathematical Theory of Communication}, The University of Illinois Press, 1949.

%Cybernetics Theory
\bibitem{Wiener1948}
N. Wiener, \emph{Cybernetics: Or Control and Communication in the Animal and the Machine}. 
Cambridge, MA, USA: The Technology Press, and New York, USA: John Wiley \& Sons, 1948. 
%Cambridge, MA, USA: The MIT Press, 1948.

\bibitem{Tsien1954}
H. S. Tsien, \emph{Engineering Cybernetics}.
New York, USA: McGraw-Hill, 1954.
%Englewood Cliffs, NJ, USA: Prentice-Hall, 1954.

\bibitem{Bellman1954}
R. Bellman, ``The Theory of Dynamic Programming,''
\emph{Bulletin of the American Mathematical Society}, vol. 60, pp. 503--515, 1954.

\bibitem{Kalman1960}
R. E. Kalman, ``A New Approach to Linear Filtering and Prediction Problems,''
\emph{Journal of Basic Engineering, Transactions of the ASME}, vol. 82, no. 1, pp. 35--45, Mar. 1960.

%Decision Theory
\bibitem{Paper_VoI}
R. L. Stratonovich, ``On value of information," Izvestiya of USSR Academy of Sciences, Technical Cybernetics, no. 5, pp. 3-12, 1965.

\bibitem{Book_VoI}
R. L. Stratonovich, Theory of Information and its Value. Cham, Switzerland: Springer, 2020.

\bibitem{Semantic_Weaver}
W. Weaver, ``Recent Contributions to the Mathematical Theory of Communication," \emph{ETC: A Review of General Semantics}, pp. 261-81, 1953.

%Semantic Information Theory
\bibitem{Semantic_Carnap}
R. Carnap, Y. Bar-Hillel, ``An Outline of A Theory of Semantic Information," \emph{RLE Technical Reports 247}, Research Laboratory of Electronics, Massachusetts Institute of Technology, Cambridge MA,1952.

\bibitem{Semantic_Floridi}
L. Floridi, ``Outline of A Theory of Strongly Semantic Information," \emph{Minds and machines}, vol. 14, no. 2, pp. 197-221, 2004.

\bibitem{Entropy_Luca}
A. De Luca, S. Termini, ``A Definition of A Non-probabilistic Entropy In The Setting of Fuzzy Sets Theory," \emph{Information and Control}, vol. 20, pp. 301-312, 1972.

\bibitem{Fuzzy_Luca}
A. De Luca, S. Termini, ``Entropy of L-Fuzzy Sets," \emph{Information and Control}, vol. 24, pp. 55-73, 1974.

\bibitem{Wuweiling}
W. Wu, ``General Source and General Entropy," \emph{Journal of Beijing University of Posts and Telecommunications}, vol. 5, no. 1, pp. 29-41, 1982.

\bibitem{Zhong1986}
Y.~Zhong, ``Comprehensive measure of information", \textit{Journal of Beijing University of Posts and Telecommunications}, vol.~9, no.~2, pp.~12--19, 1986.

\bibitem{Zhong1998}
——, ``Intelligence oriented comprehensive information theory——In memory of the 50th anniversary of Shannon information theory", \textit{Journal of Beijing University of Posts and Telecommunications}, vol.~21, no.~4, pp.~1--6, 1998.

\bibitem{Lu1994}
C.~Lu, ``Coding significance of generalized entropy and generalized mutual information", \textit{Journal of China Institute of Communications}, vol.~15, no.~6, pp.~37--44, 1994.

\bibitem{Lu2025}
——, ``A semantic generalization of Shannon's information theory and applications,'' \textit{Entropy}, vol.~27, no.~5, p.~461, 2025.

\bibitem{Semantic_Bao}
J. Bao, P. Basu, M. Dean, et al., ``Towards A Theory of Semantic Communication,'' IEEE Network Science Workshop, West Point, NY, USA, Jun. 2011.

\bibitem{RateD_Liu}
J. Liu, W. Zhang, and H. V. Poor, ``A Rate-Distortion Framework for Characterizing Semantic Information," 2021 IEEE International Symposium
on Information Theory (ISIT), Melbourne, Australia,2021.

\bibitem{SideInfo_Guo}
T. Guo, Y. Wang, et al., ``Semantic Compression with Side Information: A Rate-Distortion Perspective," arXiv preprint arXiv:2208.06094, 2022.

\bibitem{Theory_Shao}
Y. Shao, Q. Cao, and D. Gunduz, ``A Theory of Semantic Communication," arXiv preprint arXiv:2212.01485, 2022.

\bibitem{Theory_Tang}
J. Tang, Q. Yang, and Z. Zhang, ``Information-Theoretic Limits on Compression of Semantic Information," arXiv preprint arXiv:2306.02305, 2023.

\bibitem{Paper_SIT}
K. Niu and P. Zhang, ``A mathematical theory of semantic communication." Journal on Communications, Vol. 45, No. 6, pp. 7-59, June 2024. \url{https://www.joconline.com.cn/en/article/doi/10.11959/j.issn.1000-436x.2024111/}

\bibitem{Book_SIT}
——, \emph{The Mathematical Theory of Semantic Communication}, Springer Nature, 2025. \url{https://link.springer.com/book/10.1007/978-981-96-5132-0}

\bibitem{Paper_Beyond_Shannon}
P. Zhang , K. Niu, Z. Liang et al. ``Beyond shannon: Semantic information theory and methodology," \emph{IEEE Transactions on Network Science and Engineering}, vol. 13, pp. 8062-8079, 2026.

%Pragmatic information
\bibitem{Weizsacker1972}
E.~von Weizs\"acker and C.~von Weizs\"acker,
``Wiederaufnahme der begrifflichen Frage: Was ist Information?,''
\textit{Nova Acta Leopoldina}, vol.~37, no.~206, pp.~535--555, 1972.

\bibitem{Gernert2006}
D.~Gernert, ``Pragmatic information: Historical exposition and general overview,''
\textit{Mind and Matter}, vol.~4, no.~2, pp.~141--167, 2006.

\bibitem{Weinberger2002}
E.~D.~Weinberger, ``A theory of pragmatic information and its application to the quasi-species model of biological evolution,''
\textit{BioSystems}, vol.~66, no.~3, pp.~105--119, Aug.-Sept. 2002.

\bibitem{Weinberger2024}
E.~D.~Weinberger, ``Towards a theory of pragmatic information,'' \textit{arXiv preprint}, arXiv:2403.12324, 2024.

%Semantic Communication
\bibitem{Survey_Shi}
G. Shi, Y. Xiao, Y. Li, et al., ``From Semantic Communication to Semantic-aware Networking: Model, Architecture, and Open Problems," \emph{IEEE Communications Magazine}, vol. 59, no. 8, pp. 44-50, 2021.

\bibitem{Survey_Qin}
Z. Qin, X. Tao, J. Lu, et al., ``Semantic Communications: Principles and Challenges," arXiv preprint arXiv:2201.01389, 2021.

\bibitem{Survey_Xie}
H. Xie, Z. Qin, X. Tao, et al., ``Task-oriented Multi-user Semantic Communications," \emph{IEEE Journal on Selected Areas in Communications}, vol. 40, no. 9, pp. 2584-2597, 2022.

\bibitem{Survey_Gunduz}
D. G\"und\"uz, Z. Qin, et al., ``Beyond Transmitting Bits: Context, Semantics, and Task-oriented Communications," \emph{IEEE Journal on Selected Areas in Communications}, vol. 41, no. 1, pp. 5-41, Jan. 2023.

\bibitem{Semantic_ZhangPing}
P. Zhang, W. Xu, H. Gao, et al., ``Toward Wisdom-Evolutionary and Primitive-Concise 6G: A New Paradigm of Semantic Communication Networks," \emph{Engineering}, vol. 8, no. 1, pp. 60-73, 2022.

\bibitem{Semantic_Niu}
K. Niu, J. Dai, S. Yao, et al., ``A paradigm shift toward semantic communications," \emph{IEEE Communications Magazine}, vol. 60, no. 11, pp. 113-119, 2022.

%Communication and Control Covergence
\bibitem{Xiao2003Joint}
L. Xiao, M. Johansson, H. Hindi, S. Boyd, and A. Goldsmith,
``Joint optimization of communication rates and linear systems,''
\emph{IEEE Trans. Autom. Control}, vol. 48, no. 1, pp. 148--153, Jan. 2003.

\bibitem{Park2018WNCS}
P. Park, S. C. Ergen, C. Fischione, C. Lu, and K. H. Johansson,
``Wireless network design for control systems: A survey,''
\emph{IEEE Commun. Surveys Tuts.}, vol. 20, no. 2, pp. 978--1013, Second Quarter 2018.

\bibitem{Wang2023WNCS}
Y. Wang, S. Wu, C. Lei, J. Jiao, and Q. Zhang,
``A review on wireless networked control system: The communication perspective,''
\emph{IEEE Internet Things J.}, vol. 11, no. 5, pp. 7499--7524, Mar. 2024.

\bibitem{Zhao2019CoDesign}
G. Zhao, M. A. Imran, Z. Pang, Z. Chen, and L. Li,
``Toward real-time control in future wireless networks: Communication-control co-design,''
\emph{IEEE Commun. Mag.}, vol. 57, no. 2, pp. 138--144, Feb. 2019.

\bibitem{Negi2022Optimal}
N. Negi and A. Chakrabortty,
``Optimal co-designs of communication and control in bandwidth-constrained cyber-physical systems,''
\emph{Automatica}, vol. 142, art. no. 110288, Aug. 2022.

\bibitem{Lu2023Survey}
Z. Lu and G. Guo,
``Control and communication scheduling co-design for networked control systems: A survey,''
\emph{Int. J. Syst. Sci.}, vol. 54, no. 1, pp. 189--203, 2023.

\bibitem{Tishby2011}
N. Tishby and D. Polani, ``Information Theory of Decisions and Actions,''
in \emph{Perception-Action Cycle}, New York, NY, USA: Springer, 2011, pp. 601--636.
%Books
\bibitem{Book_RandomSet}
I. Molchanov, \emph{Theory of Random Sets}, 2nd edition. Probability Theory and Stochastic Modelling, Springer, 2017.

\bibitem{Book_Cover}
T. Cover and J. Thomas, \emph{Elements of Information Theory}, New York: Wiley, 1991.

\bibitem{Book_ElGamal}
A. El Gamal and Y. H. Kim, \emph{Network Information Theory}, Cambridge: Cambridge University Press, 2011.

\end{thebibliography}
\end{document}